\documentclass[a4paper,11pt]{amsart}
\usepackage[active]{srcltx}

\usepackage{graphicx}
\usepackage{amsmath}
\usepackage{amssymb}
\usepackage{hyperref}
\usepackage{longtable}

\newtheorem{thm}{Theorem}[section]
\newtheorem{lem}[thm]{Lemma}%
\newtheorem{prop}[thm]{Proposition}%
\newtheorem{cor}[thm]{Corollary}%
\theoremstyle{remark}
\newtheorem{remark}{Remark}[section] %

\theoremstyle{plain}

\numberwithin{equation}{section}

\def\RR{{\mathbb R}}

\def\ZZ{{\mathbb Z}}

\def\veca{{\text{\boldmath$a$}}}
\def\vecb{{\text{\boldmath$b$}}}

\def\vecc{{\text{\boldmath$c$}}}
\def\vecd{{\text{\boldmath$d$}}}
\def\vece{{\text{\boldmath$e$}}}
\def\vech{{\text{\boldmath$h$}}}
\def\veck{{\text{\boldmath$k$}}}

\def\vecell{{\text{\boldmath$\ell$}}}
\def\vecm{{\text{\boldmath$m$}}}

\def\vecq{{\text{\boldmath$q$}}}
\def\vecQ{{\text{\boldmath$Q$}}}
\def\vecp{{\text{\boldmath$p$}}}
\def\vecr{{\text{\boldmath$r$}}}
\def\vecs{{\text{\boldmath$s$}}}
\def\vecS{{\text{\boldmath$S$}}}

\def\vecu{{\text{\boldmath$u$}}}
\def\vecv{{\text{\boldmath$v$}}}
\def\vecV{{\text{\boldmath$V$}}}
\def\vecw{{\text{\boldmath$w$}}}
\def\vecx{{\text{\boldmath$x$}}}

\def\vecy{{\text{\boldmath$y$}}}
\def\vecz{{\text{\boldmath$z$}}}
\def\vecalf{{\text{\boldmath$\alpha$}}}

\def\veceta{{\text{\boldmath$\eta$}}}
\def\vecomega{{\text{\boldmath$\omega$}}}

\def\vecxi{{\text{\boldmath$\xi$}}}

\def\vecnull{{\text{\boldmath$0$}}}

\def\scrA{{\mathcal A}}
\def\scrB{{\mathcal B}}

\def\scrI{{\mathcal I}}

\def\scrK{{\mathcal K}}
\def\scrL{{\mathcal L}}

\def\fC{{\mathfrak C}}

\def\fZ{{\mathfrak Z}}

\def\id{\operatorname{id}}

\def\C{\operatorname{C{}}}

\def\L{\operatorname{L{}}}

\def\GL{\operatorname{GL}}

\def\S{\operatorname{S{}}}

\def\SL{\operatorname{SL}}
\def\ASL{\operatorname{ASL}}
\def\AGL{\operatorname{AGL}}
\def\SO{\operatorname{SO}}

\def\T{\operatorname{T{}}}

\def\sgn{\operatorname{sgn}}

\def\vol{\operatorname{vol}}

\def\Area{\operatorname{Area}}

\def\SLSL{\SL(d,\ZZ)\backslash\SL(d,\RR)}

\def\trans{\,^\mathrm{t}\!}

\def\Onder#1#2#3#4#5{#1 \setbox0=\hbox{$#1$}\setbox1=\hbox{$#2$}
       \dimen0=.5\wd0 \dimen1=\dimen0 \dimen2=\dp0 \dimen3=\dimen2
       \advance\dimen0 by .5\wd1 \advance\dimen0 by -#4
       \advance\dimen1 by -.5\wd1 \advance\dimen1 by -#4
       \advance\dimen2 by -#3 \advance\dimen2 by \ht1
       \advance\dimen2 by 0.3ex \advance\dimen3 by #5
        \kern-\dimen0\raisebox{-\dimen2}[0ex][\dimen3]{\box1}
       \kern\dimen1}
\def\utilde#1{\Onder{#1}{\mbox{\char'176}}{0pt}{0ex}{0.5ex}} %
\def\myutilde#1{\Onder{#1}{\mbox{$\sim$}}{0pt}{0ex}{0.5ex}} %
\newcommand{\EE}{D}
\newcommand{\myX}{X_1}

\newcommand{\clowA}{{2}}

\newcommand{\clowC}{{5}}
\newcommand{\clowD}{{3}}
\newcommand{\clowE}{{4}}
\newcommand{\clowF}{{6}}
\newcommand{\clowG}{{7}}
\newcommand{\clowH}{{1}}
\newcommand{\clowI}{{8}}
\newcommand{\clowJ}{{9}}
\newcommand{\clowK}{{10}}
\newcommand{\clowL}{{11}}
\newcommand{\clowM}{{12}}
\newcommand{\clowN}{{13}}
\newcommand{\clowO}{{14}}
\newcommand{\clowP}{{15}}
\newcommand{\clowQ}{{16}}
\newcommand{\clowR}{{17}}
\newcommand{\clowT}{{18}}

\newcommand{\clowW}{{20}}
\newcommand{\clowX}{{21}}
\newcommand{\clowY}{{22}}

\newcommand{\llgg}{\asymp}

\newcommand{\R}{\mathbb{R}}
\newcommand{\Z}{\mathbb{Z}}

\newcommand{\HS}{{{\S'_1}^{d-1}}}
\renewcommand{\aa}{\mathsf{a}}
\newcommand{\kk}{\mathsf{k}}
\newcommand{\nn}{\mathsf{n}}
\newcommand{\sfrac}[2]{{\textstyle \frac {#1}{#2}}}
\newcommand{\col}{\: : \:}
\newcommand{\Si}{\mathcal{S}}
\newcommand{\bn}{\mathbf{0}}
\newcommand{\ta}{\utilde{a}}
\newcommand{\tu}{\utilde{u}}
\newcommand{\tM}{\myutilde{M}}
\newcommand{\tkk}{\utilde{\mathsf{k}}}

\newcommand{\ve}{\varepsilon}
\newcommand{\F}{\mathcal{F}}
\newcommand{\FG}{\mathcal{G}}
\newcommand{\FC}{\mathcal{C}}
\newcommand{\matr}[4]{\left( \begin{matrix} #1 & #2 \\ #3 & #4 \end{matrix} \right) }
\newcommand{\smatr}[4]{\bigr( \begin{smallmatrix} #1 & #2 \\ #3 & #4 \end{smallmatrix} \bigr) }

\title{The periodic Lorentz gas in the Boltzmann-Grad limit: Asymptotic estimates}
\author{Jens Marklof}
\author{Andreas Str\"ombergsson}
\address{School of Mathematics, University of Bristol,
Bristol BS8 1TW, U.K.\newline
\rule[0ex]{0ex}{0ex} \hspace{8pt}{\tt j.marklof@bristol.ac.uk}}
\address{Department of Mathematics, Box 480, Uppsala University,
SE-75106 Uppsala, Sweden\newline
\rule[0ex]{0ex}{0ex} \hspace{8pt}{\tt astrombe@math.uu.se}}
\date{\today}
\thanks{J.M.\ is supported by a Royal Society Wolfson Research Merit Award.
A.S.\ is a Royal Swedish Academy of Sciences Research Fellow supported by
a grant from the Knut and Alice Wallenberg Foundation.}

\begin{document}

\begin{abstract}
The dynamics of a point particle in a periodic array of spherical scatterers converges, in the limit of small scatterer size, to a random flight process, whose paths are piecewise linear curves generated by a Markov process with memory two. The corresponding transport equation is distinctly different from the linear Boltzmann equation observed in the case of a random configuration of scatterers. In the present paper we provide asymptotic estimates for the transition probabilities of this Markov process. Our results in particular sharpen previous upper and lower bounds on the distribution of free path lengths obtained by Bourgain, Golse and Wennberg. 
\end{abstract}

\enlargethispage{15pt}
\maketitle
\vspace{-5pt}\thispagestyle{empty}
\begin{footnotesize}
\tableofcontents
\end{footnotesize}

\section{Introduction}

The linear Boltzmann equation (also referred to as the Boltzmann-Lorentz equation or kinetic Lorentz equation) is one of the fundamental transport equations that describe the macroscopic dynamics of a dilute gas in matter. The equation was postulated by Lorentz in 1905 \cite{Lorentz} by considering a gas of non-interacting point particles moving in an infinite, fixed array of hard sphere scatterers. Crucially, Lorentz assumed that in the limit of small scatterers (Boltzmann-Grad limit) consecutive collisions become independent of each other and are solely determined by the single-scatterer cross section. Lorentz' heuristic derivation of the linear Boltzmann equation was put on a rigorous footing in the case of a {\em random} scatterer configuration in the seminal papers by Gallavotti \cite{Gallavotti69}, Spohn \cite{Spohn78}, and Boldrighini, Bunimovich and Sinai \cite{Boldrighini83}. On the other hand, our recent studies of {\em periodic} scatterer configurations \cite{partI}, \cite{partII} show that in this case the Boltzmann-Grad limit is governed by a transport equation which is distinctly different from the linear Boltzmann equation. One of the main features is here that (contrary to Lorentz' assumption) consecutive collisions are no longer independent: The collision kernel of our transport equation does not only depend on the particle velocity before and after each collision, but also on the flight time until the next collision and the velocity thereafter. The collision kernel is thus significantly more complicated than in the linear Boltzmann equation, and explicit formulas are so far only known in dimension $d=2$ \cite{partIII}; cf.~also \cite{Bykovskii09}, \cite{Caglioti08} for different approaches. The objective of the present paper is to focus on dimension $d\geq 3$ and derive asymptotic estimates for the collision kernel for small and large inter-collision times. These estimates yield in particular precise asymptotics for the distribution of the free path length in the periodic Lorentz gas, and thus improve the upper and lower bounds obtained by Bourgain, Golse and Wennberg \cite{Bourgain98}, \cite{Golse00}. 

\begin{figure}
\begin{center}
\framebox{
\begin{minipage}{0.4\textwidth}
\unitlength0.1\textwidth
\begin{picture}(10,10)(0,0)
\put(0.5,1){\includegraphics[width=0.9\textwidth]{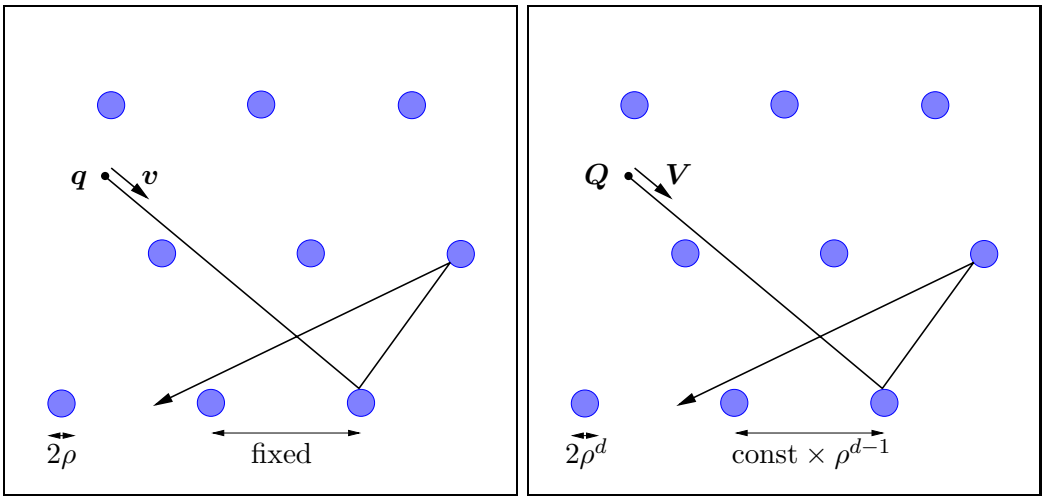}}
\put(1,6.4){$\vecq$} \put(2.5,6.4){$\vecv$}
\put(0.5,0.5){$2\rho$} \put(4.8,0.5){fixed}
\end{picture}
\end{minipage}
}
\framebox{
\begin{minipage}{0.4\textwidth}
\unitlength0.1\textwidth
\begin{picture}(10,10)(0,0)
\put(0.5,1){\includegraphics[width=0.9\textwidth]{lorentzgas.eps}}
\put(0.8,6.4){$\vecQ$} \put(2.5,6.4){$\vecV$}
\put(0.4,0.5){$2\rho^d$} \put(3.9,0.5){$\text{const}\times\rho^{d-1}$}
\end{picture}
\end{minipage}
}
\end{center}
\caption{Left: The periodic Lorentz gas in ``microscopic'' coordinates---the lattice $\scrL$ remains fixed as the radius $\rho$ of the scatterer tends to zero. Right: The periodic Lorentz gas in ``macroscopic'' coordinates ---both the lattice constant and the radius of each scatter tend to zero, in such a way that the mean free path length remains finite.} \label{figLorentz}
\end{figure}

\subsection{The Boltzmann-Grad limit of the periodic Lorentz gas}

To explain the setting of our results in more detail,
let us fix a euclidean lattice $\scrL\subset\RR^d$, and assume (without loss of generality) that its fundamental cell has volume one. We denote by $\scrK_\rho\subset\RR^d$ the complement of the set $\scrB_\rho^d + \scrL$ (the ``billiard domain''), and $\T^1(\scrK_\rho)=\scrK_\rho\times\S_1^{d-1}$ its unit tangent bundle (the ``phase space''), with $\vecq(t)\in\scrK_\rho$ the position and $\vecv(t)\in\S_1^{d-1}$ the velocity of the particle at time $t$. Here $\scrB_\rho^d$ 
\label{scrBDEF}
denotes the open ball of radius $\rho$, centered at the origin, and $\S_1^{d-1}$ the unit sphere. The dynamics of a particle in the Lorentz gas is defined as the motion with unit speed along straight lines, and specular reflection at the balls $\scrB_\rho^d+\scrL$. We may in fact also permit other scattering processes, such as the scattering map of a Muffin-tin Coulomb potential; cf.~\cite{partII} for details. 
A dimensional argument shows that in the Boltzmann-Grad limit $\rho\to 0$ the free path length scales like $\rho^{-(d-1)}$, i.e., the inverse of the total scattering cross section of an individual scatterer. It is therefore natural to rescale space and time by introducing the macroscopic coordinates (see Figure \ref{figLorentz})
\begin{equation}\label{macQV}
	\big(\vecQ(t),\vecV(t)\big) = \big(\rho^{d-1} \vecq(\rho^{-(d-1)} t),\vecv(\rho^{-(d-1)} t)\big) .
\end{equation}
The time evolution of a particle with initial data $(\vecQ,\vecV)$ is then described by the billiard flow
\begin{equation}\label{LP}	(\vecQ(t),\vecV(t))=F_{t,\rho}(\vecQ,\vecV) .
\end{equation}
Since the speed of our particle is a constant of motion we may assume without loss of generality that $\|\vecV\|=1$. For notational reasons it is convenient to extend the dynamics to the inside of each scatterer trivially, i.e., set $F_{t,\rho}=\id$ whenever $\vecQ$ is inside the scatterer. That is, the relevant phase space is now the unit tangent bundle of $\RR^d$, which will be denoted by $\T^1(\RR^d)$.

Let us fix a probability measure $\Lambda$ on $\T^1(\RR^d)$. For random initial data $(\vecQ_0,\vecV_0)$ with respect to $\Lambda$, we can then view the billiard flow $\{F_{t,\rho} : t> 0\}$ as a stochastic process. The central result of \cite{partI}, \cite{partII} is that, if $\Lambda$ is absolutely continuous with respect to Lebesque measure, the billiard flow converges in the Boltzmann-Grad limit to a random flight process $\{\Xi(t):t > 0\}$, which is defined as the flow with unit speed along a random piecewise linear curve, whose path segments $\vecS_1,\vecS_2,\vecS_2,\ldots\in\RR^d$ are generated by a Markov process with memory two. Specifically, if we set $\xi_j=\|\vecS_j\|$ and $\vecV_{j-1} =\frac{\vecS_j}{\|\vecS_j\|}$ for $j=1,2,3,\ldots$,
then the distribution of the first $n$ path segments is given by the probability density
\begin{multline}
	\Lambda'(\vecQ_0,\vecV_0) p(\vecV_0,\xi_1,\vecV_1) p_\vecnull(\vecV_0,\vecV_1,\xi_2,\vecV_2)\cdots \\ \cdots p_\vecnull(\vecV_{n-3},\vecV_{n-2},\xi_{n-1},\vecV_{n-1}) \int_{\S_1^{d-1}} p_\vecnull(\vecV_{n-2},\vecV_{n-1},\xi_n,\vecV_n) d\!\vol_{\S_1^{d-1}}(\vecV_n) ,
\end{multline}
see Theorem 1.3 and Section 4 in \cite{partII}.

Before describing the transition probability densities $p$ and $p_\vecnull$
in more detail, let us explain the relation of our limiting stochastic
process $\Xi(t)$ with the macroscopic dynamics of a particle cloud discussed
earlier. The time evolution of an initial particle density $f_0\in\L^1(\T^1(\RR^d))$ in the Lorentz gas with fixed scatterer radius $\rho$ is given by $f_t= L_\rho^t f_0$, where $L_\rho^t$ is the Liouville operator defined by
\begin{equation}
	[L_\rho^t f_0](\vecQ,\vecV) := f_0\big(F_{-t,\rho}(\vecQ,\vecV)\big) .
\end{equation}
The existence of the limiting stochastic process $\Xi(t)$ implies that for every $t>0$ there exists a linear operator $L^t:\L^1(\T^1(\RR^d))\to\L^1(\T^1(\RR^d))$, such that for every $f_0\in\L^1(\T^1(\RR^d))$ and
any set $\scrA\subset\T^1(\RR^d)$ with boundary of Lebesgue measure zero,
\begin{equation}
	\lim_{\rho\to 0} \int_{\scrA} [L^t_\rho f_0](\vecQ,\vecV)\, d\vecQ\, d\!\vol_{\S_1^{d-1}}(\vecV) 
	= \int_{\scrA} [L^t f_0](\vecQ,\vecV) \, d\vecQ\, d\!\vol_{\S_1^{d-1}}(\vecV) .
\end{equation}
If $f_0$ is $\C^1$ %
then its image %
under the limit operator $L^t$ is given by %
\begin{equation}
	[L^t f_0](\vecQ,\vecV)=\int_{\R_{>0}\times \S_1^{d-1}}  f(t,\vecQ,\vecV,\xi,\vecV_+)\, d\xi\,
d\!\vol_{\S_1^{d-1}}(\vecV_+),
\end{equation}
where $f$ is the unique solution of the differential equation
\begin{multline} \label{FPKEQ}
	\big[ \partial_t + \vecV\cdot\nabla_\vecQ - \partial_\xi \big] f(t,\vecQ,\vecV,\xi,\vecV_+) 
	\\ = \int_{\S_1^{d-1}}  f(t,\vecQ,\vecV_0,0,\vecV)
p_{\vecnull}(\vecV_0,\vecV,\xi,\vecV_+) \,
d\!\vol_{\S_1^{d-1}}(\vecV_0) 
\end{multline}
subject to the initial condition 
$f(0,\vecQ,\vecV,\xi,\vecV_+)= f_0(\vecQ,\vecV) p(\vecV,\xi,\vecV_+)$. 
Equation \eqref{FPKEQ} corresponds to the Fokker-Planck-Kolmogorov equation of our limiting stochastic process $\Xi(t)$, cf.~Section 6.3 of \cite{partII}, and may be viewed as a generalization of the linear Boltzmann equation, cf.~\cite[Section 3]{Marklof09}.

\begin{figure}
\begin{center}
\begin{minipage}{0.49\textwidth}
\unitlength0.1\textwidth
\begin{picture}(10,8)(0,0)
\put(0.5,1){\includegraphics[width=0.9\textwidth]{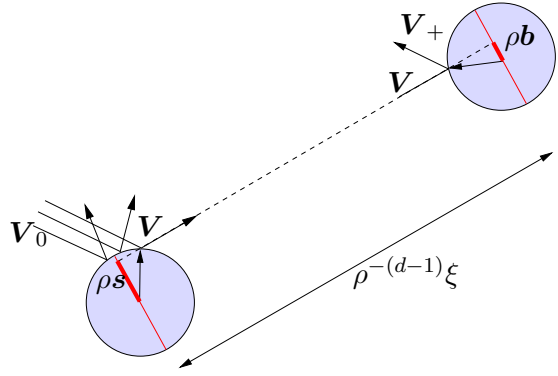}}
\put(0.1,3.2){$\vecV_0$}
\put(2.3,3.3){$\vecV$} 
\put(6.6,5.8){$\vecV$} 
\put(1.6,2.4){$\rho \vecs$}
\put(6,2.5){$\rho^{-(d-1)}\xi$}
\put(8.6,6.6){$\rho \vecb$}
\put(6.8,6.8){$\vecV_+$} 
\end{picture}
\end{minipage}
\end{center}
\caption{Two consecutive collisions in the Lorentz gas. \label{figTP}}
\end{figure}

We can express the probability densities $p(\vecV,\xi,\vecV_+)$ and $p_\vecnull(\vecV_0,\vecV,\xi,\vecV_+)$ as
\begin{equation}
	p(\vecV,\xi,\vecV_+) 
	=\sigma(\vecV,\vecV_+)\, \Phi\big(\xi,\vecb(\vecV,\vecV_+)\big),
\end{equation}
\begin{equation}
	p_{\vecnull}(\vecV_0,\vecV,\xi,\vecV_+) 
	=\sigma(\vecV,\vecV_+)\, \Phi_\vecnull\big(\xi,\vecb(\vecV,\vecV_+),
	-\vecs(\vecV,\vecV_0)\big)
\end{equation}
where $\sigma(\vecV,\vecV_+)$ is the differential cross section, $\Phi\big(\xi,\vecb(\vecV,\vecV_+)\big)$ is the limiting probability density (as $\rho\to 0$) of hitting, from a generic point in $\T^1(\RR^d)$, the first scatterer at time $\rho^{-(d-1)}\xi$ (in microscopic units) with impact parameter $\vecb(\vecV,\vecV_+)$, and $\Phi_\vecnull\big(\xi,\vecb(\vecV,\vecV_+),-\vecs(\vecV,\vecV_0)\big)$ is the limiting probability density of hitting, from a given scatterer
with exit parameter $\vecs(\vecV,\vecV_0)$, the next scatterer at time
$\rho^{-(d-1)}\xi$ with impact parameter $\vecb(\vecV,\vecV_+)$ (cf.~Figure \ref{figTP}).

\begin{remark} 
If the scattering map is given by specular reflection (as in the original Lorentz gas), we have explicitly $\sigma(\vecV,\vecV_+)= \frac 14\|\vecV-\vecV_+\|^{3-d}$
for the scattering cross section, and 
\begin{equation}
\vecs(\vecV,\vecV_0)= -\frac{(\vecV_0 K(\vecV))_\perp}{\|\vecV_0 K(\vecV)-\vece_1\|}, \qquad \vecb(\vecV,\vecV_+)= \frac{(\vecV_+ K(\vecV))_\perp}{\|\vecV_+ K(\vecV)-\vece_1\|},
\end{equation}
for the exit and impact parameters. 
Here $\vecx_\perp$
denotes the orthogonal projection of $\vecx\in\RR^d$ onto $\vece_1^{\perp}=\{0\}\times\RR^{d-1}$,
and for each $\vecV\in\S_1^{d-1}$ we have fixed a rotation
$K(\vecV)\in\SO(d)$ with $\vecV K(\vecV)=\vece_1$.
\end{remark}

The study of the asymptotic properties of the probability densities $\Phi(\xi,\vecw)$ and $\Phi_\bn(\xi,\vecw,\vecz)$ for $\xi\to\infty$ and $\xi\to 0$ are the core objectives of the present paper. Precise formulas for $\Phi$ and $\Phi_\bn$ in terms of natural probability measures on the homogeneous space 
$\SLSL$ are given in 
Section \ref{PHIDEFSEC} below. 
At this point we list the following useful facts:

\begin{itemize}
	\item[(A)] If $d\geq 3$, the functions 
\begin{equation}
	\Phi: \RR_{>0}\times \scrB_1^{d-1} \to [0,1], \qquad \Phi_\vecnull: \RR_{>0}\times \scrB_1^{d-1}\times \scrB_1^{d-1} \to [0,1]
\end{equation}
are continuous.
	\item[(B)] $\Phi(\xi,\vecw)$ depends only on $\xi$ and $\|\vecw\|$; we set 
\begin{equation}\label{PHIXISIMPWDEF}
	\Phi(\xi,w):=\Phi(\xi,\vecw)
\end{equation}
with $w=\|\vecw\|\in [0,1)$.
	\item[(C)] $\Phi_\bn(\xi,\vecw,\vecz)$ depends only on 
$\xi,\|\vecw\|,\|\vecz\|,\varphi(\vecw,\vecz)$ the angle between $\vecw,\vecz$; we set 
\begin{equation}\label{PHIXIWZPHIDEF}
	\Phi_\bn(\xi,w,z,\varphi):=\Phi_\bn(\xi,\vecw,\vecz)
\end{equation}
with $w=\|\vecw\|,z=\|\vecz\|\in [0,1)$ and $\varphi=\varphi(\vecw,\vecz)\in [0,\pi]$.
\item[(D)] $\Phi_\bn(\xi,\vecw,\vecz)=\Phi_\bn(\xi,\vecz,\vecw)$ and thus $\Phi_\bn(\xi,w,z,\varphi)=\Phi_\bn(\xi,z,w,\varphi)$.
\item[(E)] We have the formulas 
\begin{align} \label{PHIFROMPHIZERO}
\Phi(\xi,\vecw)=\int_\xi^\infty 
\int_{\scrB_1^{d-1}} \Phi_\bn(\eta,\vecw,\vecz)\, d\vecz\,d\eta,
\qquad \lim_{\xi\to 0}\Phi(\xi,\vecw)=1.
\end{align}
\end{itemize}

(A)--(D) follow from \cite[Remark 4.5]{partI} and are proved in Sections 8.1 and 8.2 of that paper. (E) follows from \cite[Remark 6.2]{partII} and \cite[(4.16)]{partI}.

\begin{remark}
In dimension $d=2$ we have the following explicit formula for the transition probability \cite{partIII}:
\begin{equation}\label{Xp}	\Phi_\vecnull(\xi,\vecw,\vecz)=\frac{6}{\pi^2}\Upsilon\Bigl(1+\frac{\xi^{-1}-\max(|\vecw|,|\vecz|)-1}{|\vecw+\vecz|}\Bigr)
\end{equation}
with
\begin{equation}
	\Upsilon(x)=
\begin{cases} 
0 & \text{if }x\leq 0\\
x & \text{if }0<x<1\\
1 & \text{if }1\leq x,
\end{cases}
\end{equation}
The same formula has recently been found independently by Caglioti and Golse \cite{Caglioti08} and by Bykovskii and Ustinov \cite{Bykovskii09}, using different methods based on continued fractions.
\end{remark}

\begin{remark}
All of the relations stated in (A)--(E) above are also valid for the Lorentz gas with {\em random} scatterer configuration. Here the fundamental function $\Phi_\vecnull$ is given by the explicit formula
\begin{equation}\label{ranPhi}
	\Phi_\bn(\xi,\vecw,\vecz) = \exp(-v_{d-1} \xi),
\end{equation}
where $v_{d-1}=\pi^{\frac{d-1}2}\Gamma(\frac{d+1}2)^{-1}$ denotes (throughout this paper) the volume of the unit ball in $\RR^{d-1}$. Note that relation (E) implies with \eqref{ranPhi} that in the random setting $\Phi(\xi,\vecw)=\Phi_\bn(\xi,\vecw,\vecz)$.
\end{remark}

\subsection{Asymptotic estimates for $\xi$ small}

Returning to the setting of the periodic Lorentz gas, we first state the asymptotic formulas for $\Phi$ and $\Phi_\bn$ as $\xi\to 0$.
Our main result in this direction is the following.
\begin{thm}\label{PHI0ZEROSMALLTHM}
For $\xi>0$ and $\vecw,\vecz\in\scrB_1^{d-1}$,
\begin{align}\label{PHI0ZEROSMALLTHMRES}
\frac{1-2^{d-1}v_{d-1} \xi}{\zeta(d)}\leq\Phi_\bn(\xi,\vecw,\vecz)\leq\frac{1}{\zeta(d)} .
\end{align}
\end{thm}

That is, $\Phi_\bn(\xi,\vecw,\vecz)=\zeta(d)^{-1} + O(\xi)$, where the remainder term is everywhere non-positive, and the implied constant is independent of $\vecw$ and $\vecz$. Note that this estimate is consistent with formula \eqref{Xp} in dimension $d=2$, where we have the exact relation $\Phi_\bn(\xi,\vecw,\vecz)=\frac{6}{\pi^2}=\zeta(2)^{-1}$ for $\xi\leq \frac12$.

\begin{remark}
In the case of a random scatterer configuration, \eqref{ranPhi} yields $\Phi_\bn(\xi,\vecw,\vecz)=1 + O(\xi)$. Comparing this with Theorem \ref{PHI0ZEROSMALLTHM}, we may conclude that the leading-order asymptotics of the transition probability density $\Phi_\bn(\xi,\vecw,\vecz)$ for $\xi\to 0$ is, in both the random and periodic set-up, independent of $\vecw,\vecz$, and given by the relative density of scatterers which are completely visible from a given scatterer: in the random setting, this is the case for all scatterers close to the given scatterer; in the periodic setting the same holds only for scatterers located on visible (or primitive) lattice points, whose relative density is given by $\zeta(d)^{-1}$.
\end{remark}

Theorem \ref{PHI0ZEROSMALLTHM} combined with \eqref{PHIFROMPHIZERO} immediately implies:\enlargethispage{100pt}

\begin{cor}\label{PHI0ZEROSMALLCOR}
For all $\xi>0$ and $\vecw\in\scrB_1^{d-1}$,
\begin{align}
\Phi(\xi,\vecw)
=1-\frac{v_{d-1}}{\zeta(d)}\,\xi+O(\xi^2),
\end{align}
where the remainder term is everywhere non-negative, and the implied constant is independent of $\vecw$.
\end{cor}
\newpage

\begin{figure}
\begin{center}
\framebox{
\begin{minipage}{0.4\textwidth}
\unitlength0.1\textwidth
\begin{picture}(10,6.5)(0,1.2)
\put(0.5,1){\includegraphics[width=0.9\textwidth]{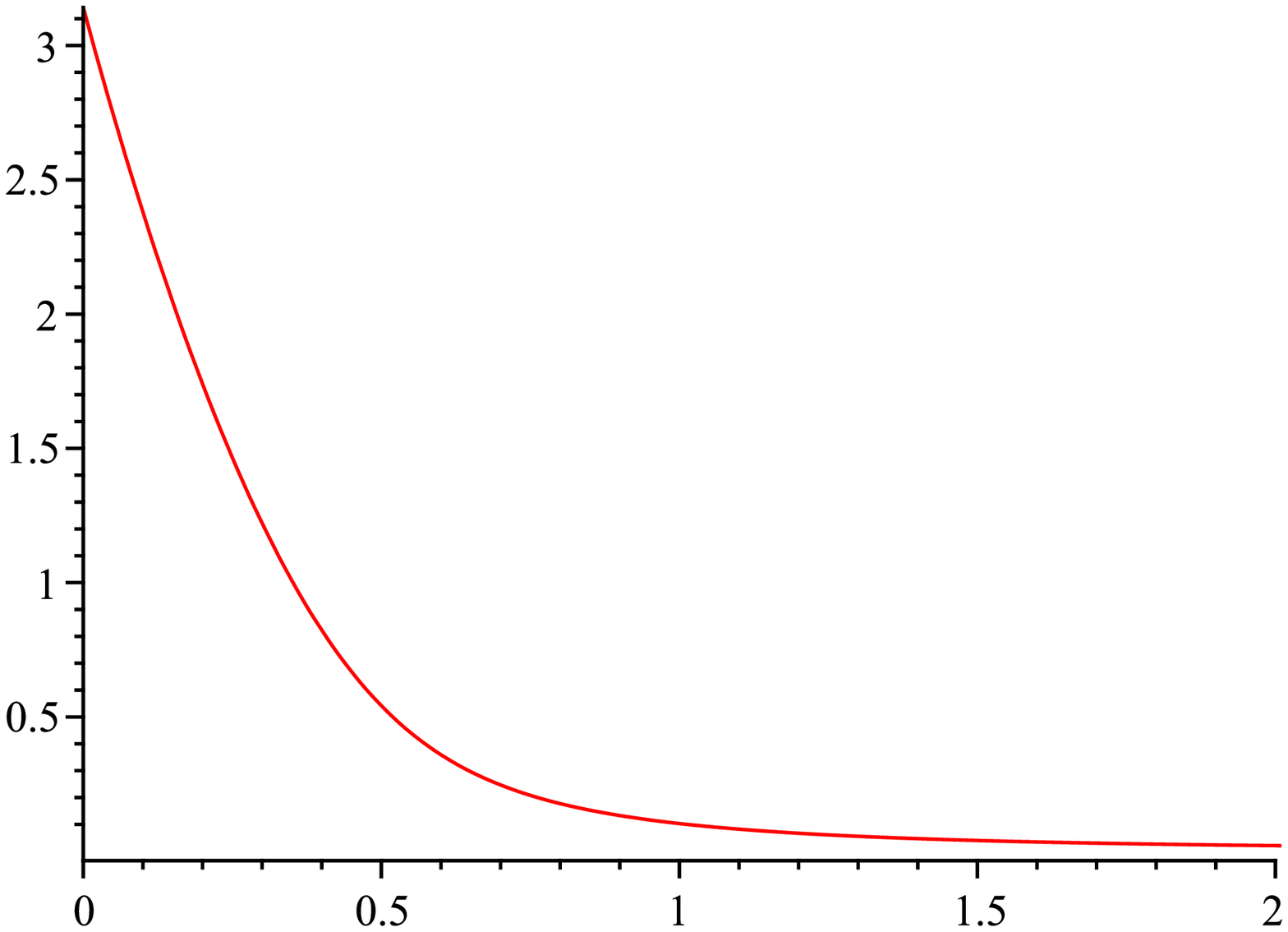}}
\end{picture}
\end{minipage}
}
\framebox{
\begin{minipage}{0.4\textwidth}
\unitlength0.1\textwidth
\begin{picture}(10,6.5)(0,1.2)
\put(0.5,1){\includegraphics[width=0.9\textwidth]{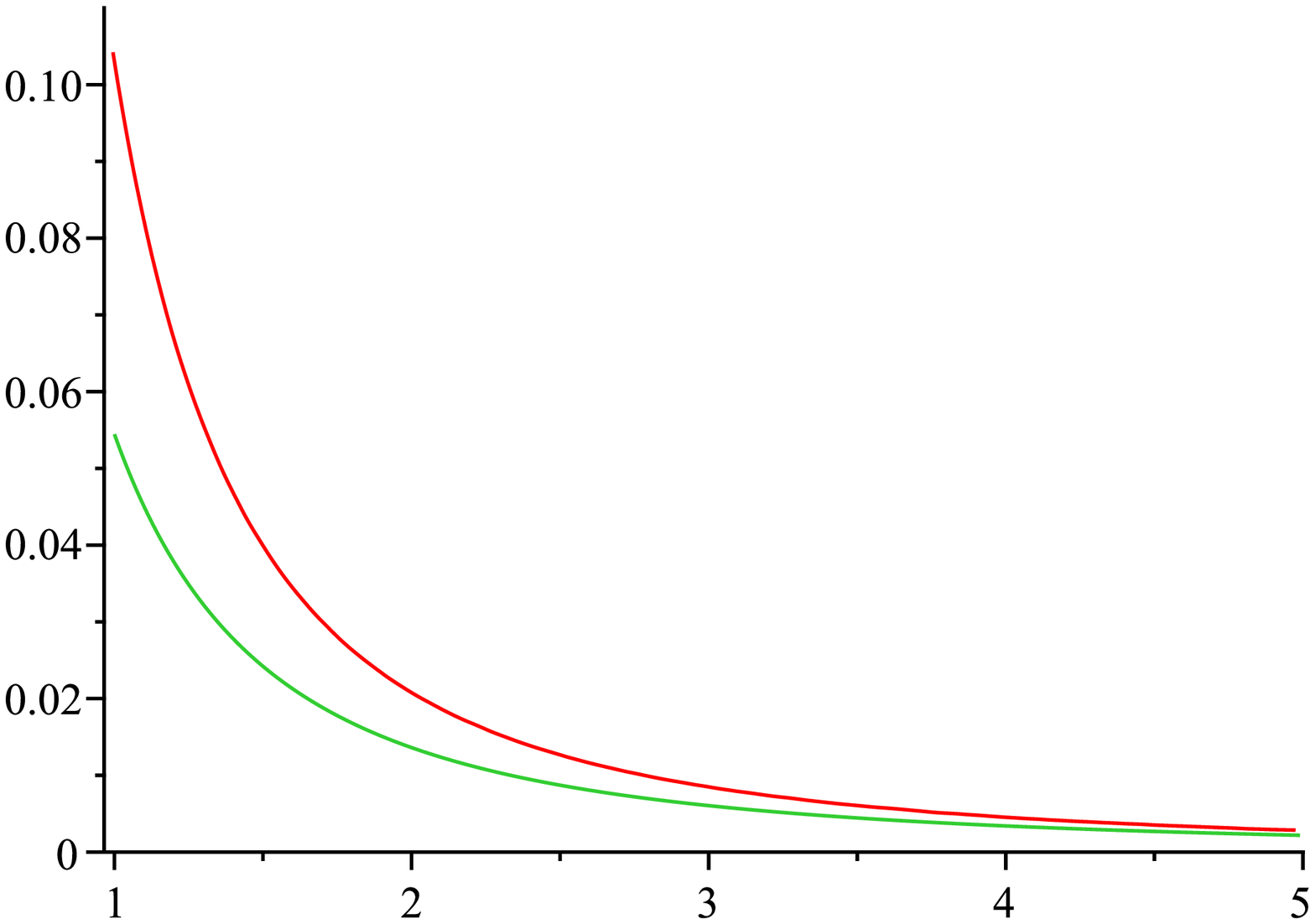}}
\end{picture}
\end{minipage}
}
\end{center}
\caption{Graphs of $\Phi(\xi)$ for $d=3$ in the ranges 
$0<\xi\leq2$ and $1\leq\xi\leq5$,
obtained from numerical computations described in Section \ref{NUMCOMPSEC}.
In the second plot also the asymptotic
$\xi\mapsto\frac{\pi}{48\zeta(3)}\xi^{-2}$ 
from Theorem~\ref{PHIBARXILARGETHM} is shown (the lower curve).}
\label{PHIplot}
\end{figure}

\begin{figure}
\begin{center}
\framebox{
\begin{minipage}{0.4\textwidth}
\unitlength0.1\textwidth
\begin{picture}(10,6.5)(0,1.2)
\put(0.5,1){\includegraphics[width=0.9\textwidth]{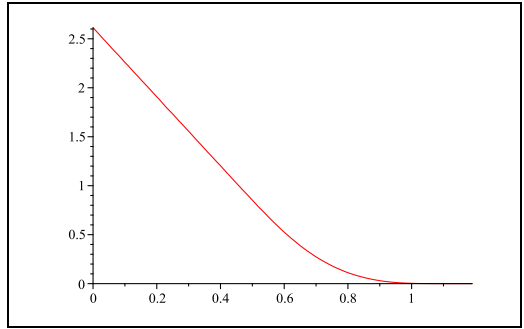}}
\end{picture}
\end{minipage}
}
\end{center}
\caption{Graph of $\Phi_\bn(\xi)$ for $d=3$,
obtained from numerical computations described in Section \ref{NUMCOMPSEC}.}
\label{PHI0plot}
\end{figure}

From our asymptotics for $\Phi(\xi,\vecw)$ and $\Phi_\bn(\xi,\vecw,\vecz)$ we can also derive asymptotics
for the limiting distribution of the free path length.
We have (cf.\ \cite[Remark 4.6]{partI})
\begin{align}\label{PHIXIAVFORMULA2}
\overline{\Phi}_\bn(\xi)=\frac 1{v_{d-1}}
\int_{\scrB_1^{d-1}}\int_{\scrB_1^{d-1}}
\Phi_\bn(\xi,\vecw,\vecz)\,d\vecw\,d\vecz
\end{align}
for the free path length between consecutive collisions,
\begin{align}\label{PHIXIAVFORMULA}
\Phi(\xi)=\int_{\scrB_1^{d-1}}\Phi(\xi,\vecw)\,d\vecw
=v_{d-1} \int_\xi^\infty\overline{\Phi}_\bn(\eta)\,d\eta
\end{align}
for the free path length from a generic initial point inside the
billiard domain, and
\begin{align}\label{PHIXIAVFORMULA3}
\Phi_\bn(\xi)=\int_{\scrB_1^{d-1}}\Phi_\bn(\xi,\vecw,\bn)\,d\vecw
\end{align}
for the free path length of a particle starting at a lattice point
(with the scatterer removed).

Only in dimension $d=2$ we have explicit formulas for the above limiting distributions, thanks to the work of Dahlqvist \cite{Dahlqvist97}, Boca, Gologan and Zaharescu \cite{Boca03}, and Boca and Zaharescu \cite{Boca07}.

Using Theorem \ref{PHI0ZEROSMALLTHM} and Corollary \ref{PHI0ZEROSMALLCOR}
in conjunction with
\eqref{PHIXIAVFORMULA2}, \eqref{PHIXIAVFORMULA}, \eqref{PHIXIAVFORMULA3},
we obtain the following.

\begin{cor}
For $\xi>0$,
\begin{align}\label{0asymp}
\overline{\Phi}_\bn(\xi)=\frac{v_{d-1}}{\zeta(d)}+O(\xi);
\qquad
\Phi_\bn(\xi)=\frac{v_{d-1}}{\zeta(d)}+O(\xi);
\qquad
\Phi(\xi)=v_{d-1}-\frac{v_{d-1}^2}{\zeta(d)}\xi+O(\xi^2).
\end{align}
Here the first two remainder terms are $\leq0$ %
and the last remainder term is $\geq0$, for all $\xi>0$.
\end{cor}

\begin{remark}
Formulas \eqref{PHIXIAVFORMULA2}, \eqref{PHIXIAVFORMULA}, \eqref{PHIXIAVFORMULA3} are also valid in the random setting, and yield
\begin{equation}
	\overline{\Phi}_\bn(\xi)=\Phi_\bn(\xi)=\Phi(\xi)=v_{d-1} \exp(-v_{d-1} \xi)=
	v_{d-1}-v_{d-1}^2\xi+O(\xi^2) .
\end{equation}
These asymtotics are the same as in the periodic setting \eqref{0asymp}, with the relative density of visible lattice points $\zeta(d)^{-1}$ replaced by $1$.
\end{remark}

In the case $d=3$ we are in fact able to compute $\Phi_\bn(\xi,\vecw,\vecz)$
\textit{explicitly} for $\xi$ small.
For $0\leq t<1$ we set
\begin{align}\label{FDEF}
F(t):=\pi-\arccos(t)+t\sqrt{1-t^2}
=\Area\bigl(\bigl\{(x_1,x_2)\in\scrB_1^2\col x_1<t\bigr\}\bigr).
\end{align}
Also let $\xi_1:\scrB_1^2\times\scrB_1^2\to\R$
be the continuous function given by $\xi_1(\vecw,\vecz)=(6V)^{-1}$, 
where $V$ is the largest possible volume of a 
tetrahedron which is contained in the closed cylinder
$[0,1]\times\overline{\scrB_1^2}$ and which has one vertex at
$(0,-\vecz)$ and another at $(1,\vecw)$.
\begin{thm}\label{D3EXPLTHM}
If $d=3$, then for all $\vecw,\vecz\in\scrB_1^2$
and all $0<\xi\leq\xi_1(\vecw,\vecz)$,
\begin{align}\label{D3EXPLTHMRES}
\Phi_\bn(\xi,\vecw,\vecz)=\zeta(3)^{-1}\Bigl(1
-\frac6{\pi^2}F\bigl(\sfrac12\|\vecw-\vecz\|\bigr)\xi\Bigr).
\end{align}
\end{thm}
We will prove in Lemma \ref{XI1BOUNDSLEM} below that
$\frac14<\xi_1(\vecw,\vecz)\leq1$ for all $\vecw,\vecz\in\scrB_1^2$,
where both the lower and the upper bound are sharp.
In particular the formula \eqref{D3EXPLTHMRES} is always true when
$0<\xi\leq\frac14$.

Combining Theorem \ref{D3EXPLTHM} with \eqref{PHIFROMPHIZERO}
we will prove the following:
\begin{cor}\label{D3EXPLTHMCOR1}
If $d=3$, then for all $\vecw\in\scrB_1^2$ and all
$0<\xi\leq\min\bigl(\frac1{2(1+\|\vecw\|)},\frac2{3\sqrt3}\bigr)$:
\begin{align}\label{D3EXPLTHMCOR1RES1}
\Phi(\xi,\vecw)=1-\frac{\pi}{\zeta(3)}\xi
+\frac6{\pi^2\zeta(3)}G(\|\vecw\|)\xi^2,
\end{align}
where $G:[0,1]\to\R_{>0}$ is the function
\begin{align}\label{D3EXPLTHMCOR1GDEF}
G(w)=
\pi\int_0^{1-w}F(\sfrac12r)r\,dr
+\int_{1-w}^{1+w}F(\sfrac12r)\arccos\Bigl(\frac{w^2+r^2-1}{2wr}\Bigr)
r\,dr.
\end{align}
\end{cor}

The function $G(w)$ is continuous and strictly increasing, and satisfies
$G(0)=\frac{\pi(4\pi+3\sqrt3)}{16}$ and $G(1)=\frac5{16}\pi^2+1$.

Furthermore, using Theorem \ref{D3EXPLTHM} in conjunction with
\eqref{PHIXIAVFORMULA2}, \eqref{PHIXIAVFORMULA}, \eqref{PHIXIAVFORMULA3},
we will prove:
\begin{cor}\label{D3EXPLTHMCOR2}
If $d=3$, then
\begin{align*}
& \overline{\Phi}_\bn(\xi)=\frac{\pi}{\zeta(3)}-
\frac{3\pi^2+16}{\pi^2\zeta(3)}\xi
&& \text{for all }\:0<\xi\leq\frac14;
\\
&\Phi(\xi)=\pi-\frac{\pi^2}{\zeta(3)}\xi
+\frac{3\pi^2+16}{2\pi\zeta(3)}\xi^2
&& \text{for all }\:0<\xi\leq\frac14;
\\
&\Phi_\bn(\xi)=\frac{\pi}{\zeta(3)}-
\frac{3(4\pi+3\sqrt3)}{4\pi\zeta(3)}\xi
&& \text{for all }\:0<\xi\leq\frac2{3\sqrt3}(=0.3849...).
\end{align*}
\end{cor}

Figures \ref{PHIplot} and \ref{PHI0plot} show graphs of $\Phi(\xi)$ and
$\Phi_\bn(\xi)$ for $d=3$, obtained by numerical computations
described in Section \ref{NUMCOMPSEC}.

\subsection{Asymptotic estimates for $\xi$ large}

The case of $\xi\to\infty$ is much more complicated than $\xi$ small.
Among other things, difficulties are caused by the fact that
for large $\xi$, $\Phi_\bn(\xi,w,z,\varphi)$ \textit{vanishes} 
unless both $w$ and $z$ are near $1$.

The following theorem gives an asymptotic formula for 
$\Phi_\bn(\xi,w,z,\varphi)$ as $\xi\to\infty$, for 
$\varphi$ small. The case of small $\varphi$ is in a natural sense the most important one.
Indeed, it was seen in \cite{lprob} that, for given large $\xi$,
the function $\Phi_\bn(\xi,w,z,\varphi)$ takes its largest values
when $\varphi$ is small, and also that it has its largest support with
respect to $w,z$ in this case.
In particular, when integrating $\Phi_\bn(\xi,\vecw,\vecz)$
over $\vecz\in\scrB_1^{d-1}$,
the main contribution 
comes from $\vecz$ with %
$\varphi(\vecw,\vecz)\ll\xi^{-\frac1d+\ve}$. 
This follows directly from
Theorem \ref{CYLINDER2PTSMAINTHM} and Proposition \ref{PHI0SUPPORTTHM}
below (cf.\ also \cite[Cor.\ 1.9]{lprob}).

\begin{thm}\label{PHI0XILARGETHM}
There exists a continuous and uniformly bounded function
$F_{\bn,d}:\R_{>0}\times\R_{>0}\times\R_{\geq0}\to\R_{\geq0}$
such that
\begin{align}\label{PHI0XILARGETHMRES}
\Phi_\bn(\xi,w,z,\varphi)=
\xi^{-2+\frac 2d}F_{\bn,d}\Bigl(\xi^{\frac 2d}(1-z),\xi^{\frac 2d}(1-w),
\xi^{\frac1d}\varphi\Bigr)
+O(E),
\end{align}
for all $\xi>0,w,z\in[0,1),\varphi\in[0,\frac\pi2)$,
where the error term is
\begin{align}\label{PHI0XILARGETHMEDEF}
E=\begin{cases}
\xi^{-2} &\text{if }\:d=2,
\\
\xi^{-2}\log(2+\min(\xi,\varphi^{-1}))&\text{if }\:d=3,
\\
\min\bigl(\xi^{-2},\xi^{-3+\frac2{d-1}}\varphi^{2-d+\frac2{d-1}}\bigr)
&\text{if }\:d\geq4.\end{cases}
\end{align}
\end{thm}

In dimension $d=2$ we have from \cite[Eq.~(36)]{partIII} (note that $\varphi=0$ or $=\pi$ in this case)
\begin{equation}\label{F02EXPL}
	F_{\bn,2}(t_1,t_2,0)=\frac{3}{\pi^2} (1-\max(t_1,t_2))^+ .
\end{equation}
For $d\geq 3$, we will express $F_{\bn,d}$ as an integral of a function given by the probability that a random lattice in $\R^{d-1}$ is disjoint from a union of
two cut paraboloids (cf.\ Section \ref{XIDM1BASICSSEC} and
\eqref{PHI0XILARGETHMFDDEF} below). 
This integral representation will imply the following properties:

\begin{itemize}
	\item \textbf{Symmetry:} $F_{\bn,d}(t_1,t_2,\alpha)=F_{\bn,d}(t_2,t_1,\alpha)$ (cf.\ \eqref{XIDM1SYMM} and \eqref{PHI0XILARGETHMFDDEF}).

	\item \textbf{Support:}
There is a continuous function 
$\sigma_d:\R_{>0}\times\R_{\geq0}\to\R_{>0}$,
which we will define in terms of a lattice problem in dimension $d-1$
(cf.\ \eqref{SUPPORTprelim4} below), such that
\begin{align}\label{SUPPORTprelim2}
F_{\bn,d}(t_2,t_1,\alpha)>0\Longleftrightarrow
t_1t_2<\sigma_d\Bigl(\frac{t_2}{t_1},\frac{\alpha^4}{t_1t_2}\Bigr).
\end{align}
This function $\sigma_d$ satisfies the symmetry relation
$\sigma_d(r,\alpha)=\sigma_d(r^{-1},\alpha)$,
and we also have
$\sigma_d(r,\alpha)\asymp r\min(1,(r\alpha)^{-\frac1d})$ 
uniformly over $r\in(0,1]$, $\alpha\geq0$ (cf.\ \eqref{SIGMADASYMP} below).
It follows that
there exist constants $c'>c>0$ which only depend on $d$ such that
\begin{align}\label{SUPPORTprelim}
&\max(t_1,t_2)\cdot\max(1,\alpha^{\frac2{d-1}})\geq c'
\Longrightarrow F_{\bn,d}(t_2,t_1,\alpha)=0;
\\\notag
&\max(t_1,t_2)\cdot\max(1,\alpha^{\frac2{d-1}})<c
\Longrightarrow F_{\bn,d}(t_2,t_1,\alpha)>0.
\end{align}
In particular the support of $F_{\bn,d}$
is contained in $(0,c']\times(0,c']\times\R_{\geq0}$, %
and for any given $t_1,t_2$,
$F_{\bn,d}(t_1,t_2,\cdot)$ has compact support in the third variable.

	\item \textbf{Bounds from above:}
$F_{\bn,d}$ is uniformly bounded (as mentioned),
and we have the following more precise bound
(cf.\ \eqref{PHI0XILARGETHMFDDEF} and Lemma \ref{SABVBOUNDLEMCOR} below):
\begin{align}\label{F0DBOUNDABOVE}
F_{\bn,d}(t_1,t_2,\alpha)
\ll\min\bigl(1,\alpha^{-d+\frac2{d-1}}\bigr).
\end{align}

	\item \textbf{Bounds from below:}
For general $d$ it follows from Theorem \ref{PHI0XILARGETHM}
combined with \cite[Prop.\ 7.8]{lprob} that $F_{\bn,d}$ is uniformly
bounded from below for $t_1,t_2,\alpha$ near zero, viz.\ 
there is a small constant $c>0$ which only depends on $d$ such that
\begin{align}\label{F0DBOUNDBELOW1}
\max(t_1,t_2,\alpha)<c
\:\Longrightarrow\: F_{\bn,d}(t_2,t_1,\alpha)>c.
\end{align}
For $d=3$ we have a stronger result (cf.\ \cite[Prop.\ 7.7]{lprob})
which says that
\eqref{F0DBOUNDABOVE} is sharp in a natural sense; 
but we expect that \eqref{F0DBOUNDABOVE} is \textit{not} sharp for $d\geq4$.
\end{itemize}

The lower bound implies that the main term dominates as $\xi\to\infty$ in 
\eqref{PHI0XILARGETHMRES} whenever 
$1-z<c\xi^{-\frac2d}$, $1-w<c\xi^{-\frac2d}$ and $\varphi<c\xi^{-\frac1d}$.
Note that this $z,w,\varphi$-regime contributes a positive portion to the integral \eqref{PHIFROMPHIZERO}.
Beyond this regime, our motivation for Theorem \ref{PHI0XILARGETHM}
is that even though we cannot say {\em exactly} where
the main term dominates, it dominates on a sufficiently large set
so that we get a good asymptotics for $\Phi(\xi,w)$ via
\eqref{PHIFROMPHIZERO}, see Theorem~\ref{PHIXIWASYMPTTHM} below.

Next we give a closely related result on the asymptotic 
shape of the \textit{support} of
$\Phi_\bn$ for $\xi$ large and $\varphi$ small.
Note that by \eqref{SUPPORTprelim2}, the main term in
\eqref{PHI0XILARGETHMRES} is non-zero if and only if
\begin{align}\label{PHI0XILARGETHMRESSUPPCONS}
\xi<(1-z)^{-\frac d4}(1-w)^{-\frac d4}\sigma_d\Bigl(\frac{1-w}{1-z},
\frac{\varphi^4}{(1-z)(1-w)}\Bigr)^{\frac d4}.
\end{align}
However, clearly Theorem \ref{PHI0XILARGETHM} never helps us deduce
$\Phi_\bn=0$, and furthermore
since  (as discussed above) it is difficult to state when the main term dominates the error term in Theorem \ref{PHI0XILARGETHM}, we cannot deduce $\Phi_\bn\neq0$
in any reasonable region either.
By using some intermediate formulas from the proof of
Theorem \ref{PHI0XILARGETHM} (see Section \ref{SUPPORTSECTION} for details) we are, however, able to prove that \eqref{PHI0XILARGETHMRESSUPPCONS} gives in fact a good
approximation of the support of $\Phi_\bn$ for $\xi$ large:

\begin{thm}\label{PHI0SUPPORTASYMPTTHM}
There is a continuous function 
$\xi_0:[0,1)\times[0,1)\times[0,\pi]\to\R_{>0}$ such that
$\Phi_\bn(\xi,w,z,\varphi)>0$ holds if and only if 
$\xi<\xi_0(w,z,\varphi)$.
This function $\xi_0(w,z,\varphi)$ satisfies
\begin{align}\label{PHI0SUPPORTASYMPTTHMRES}
\xi_0(w,z,\varphi)=
(1-z)^{-\frac d4}(1-w)^{-\frac d4}\sigma_d\Bigl(\frac{1-w}{1-z},
\frac{\varphi^4}{(1-z)(1-w)}\Bigr)^{\frac d4}
\hspace{60pt}
\\\notag
\times\Big\{1+O\bigl(\max(1-z,1-w)+\varphi^2\bigr)\Bigr\},
\end{align}
uniformly over all $z,w\in[0,1)$, $\varphi\in[0,\frac\pi2]$.
(The implied constant depends only on $d$.)
\end{thm}

To complement the picture of $\Phi_\bn$ for $\xi$ large, we recall
the uniform bounds on the size and support of $\Phi_\bn$
which the second author proved in \cite{lprob}.
\begin{thm}\label{CYLINDER2PTSMAINTHM}
(\cite[Thm.\ 1.8]{lprob})
Let $d\geq3$. We then have, for all
$\xi>0$, $\vecw,\vecz\in\scrB_1^{d-1}$,
and writing $\varphi=\varphi(\vecw,\vecz)$,
\begin{align}\label{CYLINDER2PTSMAINTHMRES}
\Phi_\bn(\xi,\vecw,\vecz)\ll\begin{cases}
\xi^{-2+\frac2d}\min\Bigl\{1,
(\xi\varphi^{d})^{-1+\frac2{d(d-1)}}\Bigr\}
&\text{if }\:\varphi\leq\frac\pi2
\\
\xi^{-2}\min\Bigl\{1,
(\xi(\pi-\varphi)^{d-2})^{-1+\frac2{d-1}}\Bigr\}
&\text{if }\:\varphi\geq\frac\pi2.
\end{cases}
\end{align}
\end{thm}

\begin{prop}\label{PHI0SUPPORTTHM} (\cite[Prop.\ 1.9]{lprob})
Let $d\geq 3$. We then have for all $z,w\in[0,1)$, $\varphi\in[0,\pi]$,
writing $t:=\max(1-z,1-w)$,
\begin{align}\label{PHI0SUPPORTBOUNDS}
\xi_0(w,z,\varphi)\llgg
\begin{cases}\min(t^{-\frac d2},t^{-\frac{d-1}2}/\varphi)
&\text{if }\varphi\leq\frac\pi 2
\\
\max(t^{-\frac{d-2}2},t^{-\frac{d-1}2}(\pi-\varphi))
&\text{if }\varphi\geq\frac\pi 2,\end{cases}
\end{align}
where the implied constants only depend on $d$.
(If $\varphi=0$ then the right hand side of \eqref{PHI0SUPPORTBOUNDS} 
should be interpreted as $t^{-\frac d2}$.)
\end{prop}

Note that the case $\varphi\leq\frac\pi2$ of 
Theorem \ref{CYLINDER2PTSMAINTHM} can be recovered from
Theorem \ref{PHI0XILARGETHM} and \eqref{F0DBOUNDABOVE};
and the case $\varphi\leq\frac\pi2$ of Proposition \ref{PHI0SUPPORTTHM}
can be recovered from
Theorem \ref{PHI0SUPPORTASYMPTTHM} combined with the
fact that $\sigma_d(r,\alpha)\asymp r\min(1,(r\alpha)^{-\frac1d})$ 
uniformly over $r\in(0,1]$, $\alpha\geq0$.
In \cite{lprob} it was also proved that the upper bound in
Theorem \ref{PHI0SUPPORTASYMPTTHM} is
\textit{sharp} in a natural sense for $d=3$,
and also for general $d\geq3$ if either
$\varphi\ll\xi^{-\frac1d}$ or $\pi-\varphi\ll\xi^{-\frac1{d-2}}$.

\vspace{15pt}

Next we give an asymptotic formula for $\Phi(\xi,w)$ for
$\xi$ large.

\begin{thm}\label{PHIXIWASYMPTTHM}
There exists a bounded continuous
function $F_d:\R_{>0}\to\R_{\geq 0}$ such that
\begin{align}\label{PHIXIWASYMPTTHMRES2}
\Phi(\xi,w)=\xi^{-2+\frac 2d}F_d\Bigl(\xi^{\frac 2d}(1-w)\Bigr)
+O(1)\left.\begin{cases}
\xi^{-2} &\text{if }\:d=2 \\
\xi^{-2}\log(2+\min(\xi,\xi^{-\frac23}(1-w)^{-1}))&\text{if }\:d=3
\\
\xi^{-2}&\text{if }\:d\geq4
\end{cases}\right\},
\end{align}
as $\xi\to\infty$, uniformly over all $0\leq w<1$.
\end{thm}

For $d=2$, we have explicitly \cite{partIII}
\begin{equation}
	F_2(t)=\frac{3}{2\pi^2} \big((1-t)_+\big)^2.
\end{equation}
For general $d\geq 3$, we will express $F_{d}(t)$ as an integral of a function given by the probability that a random lattice in $\R^{d-1}$ is disjoint from a 
cut paraboloid (cf.\ \eqref{PHIXIWASYMPTTHMRES} below). 
This function has the following properties:
The support of %
$F_d(t)$ is exactly the interval
$0<t<\sqrt{\sigma_d(1,0)}$,
where $\sigma_d(r,\alpha)$ is the same function as in
\eqref{SUPPORTprelim2}.
Furthermore $F_d(t)$ stays bounded from below as $t\to0$;
in fact the limit $F_d(0):=\lim_{t\to0}F_d(t)>0$ exists;
cf.\ Proposition \ref{LIMFDT2PROP} below.
This number has a natural interpretation in relation with
Theorem \ref{PHIXIWASYMPTTHM}:
The function $\Phi(\xi,w)$ may be 
extended to a continuous function on $\R_{>0}\times[0,1]$
(cf., e.g., \cite[Lemma 2.3]{lprob}).
Since \eqref{PHIXIWASYMPTTHMRES2} holds uniformly with respect to
$0\leq w<1$ we conclude by letting $w\to1$:
\begin{align*}
\Phi(\xi,1)=\xi^{-2+\frac2d}F_d(0)+
O(1)\left.\begin{cases}
\xi^{-2}&\text{if }\:d=2\\
\xi^{-2}\log\xi&\text{if }\:d=3
\\
\xi^{-2}&\text{if }\:d\geq4
\end{cases}\right\},
\qquad\text{as }\:\xi\to\infty.
\end{align*}

We prove Theorem \ref{PHIXIWASYMPTTHM} in Section \ref{PRELIMINARIESSEC};
this proof also serves as a preparation for the
proof of Theorem \ref{PHI0XILARGETHM}, which we give in
Sections \ref{PARABAPPRSEC}--\ref{PHI0XIWZASYMPTSEC}.
On the other hand it is alternatively 
possible to obtain Theorem \ref{PHIXIWASYMPTTHM}
(except we get a slightly worse error term when $d=3$)
as a \textit{consequence} of Theorem \ref{PHI0XILARGETHM}
and Theorem \ref{CYLINDER2PTSMAINTHM},
via the formula \eqref{PHIFROMPHIZERO}.
In fact this approach even gives an asymptotic formula for
$\frac{\partial}{\partial\xi}\Phi(\xi,\vecz)$ as $\xi\to\infty$;
cf.\   %
Theorem \ref{PPHIXIWASYMPTTHM} in Section \ref{PARTIALPHIXIZSEC}.

As a simple consequence of \eqref{PHIFROMPHIZERO}, 
Theorem \ref{PHI0SUPPORTASYMPTTHM} and Proposition \ref{PHI0SUPPORTTHM},
we obtain a precise understanding of the support of $\Phi(\xi,w)$
as $\xi\to\infty$. 
Recall that $F_d(t)>0$ if and only if $t<\sqrt{\sigma_d(1,0)}$.
\begin{cor}\label{PHIXIWSUPCOR}
There is a continuous function $\xi_0:[0,1)\to \R_{>0}$ such that 
$\Phi(\xi,w)>0$ holds if and only if $\xi<\xi_0(w)$,
and we have
\begin{align}\label{PHIXIWSUPCORRES}
\xi_0(w)=\sigma_d(1,0)^{\frac d4}(1-w)^{-\frac d2}
+O\bigl((1-w)^{1-\frac{d}2}\bigr)
\end{align}
as $w\to 1^-$.
\end{cor}

The above results directly yield asymptotics for the different
limiting distributions of the free path length, \eqref{PHIXIAVFORMULA2}, \eqref{PHIXIAVFORMULA}, \eqref{PHIXIAVFORMULA3}. The first statement concerns the distribution function for the free path length of a trajectory with generic initial condition.

\begin{thm} \label{PHIXILARGETHM}
\begin{align}\label{PHIXILARGETHMRES1}
\Phi(\xi)=
\frac{\pi^{\frac{d-1}2}}{2^{d}d\, \Gamma(\frac {d+3}2)\,\zeta(d)} \xi^{-2}
+O\bigl(\xi^{-2-\frac 2d}\bigr)
\qquad \text{as } \: \xi\to\infty.
\end{align}
\end{thm}

This result sharpens %
the upper bound given by Bourgain, Golse and Wennberg \cite{Bourgain98} 
and the lower bound of Golse and Wennberg \cite{Golse00}.

As to the distribution of the free path length between consecutive collisions, we have:

\begin{thm} \label{PHIBARXILARGETHM}
\begin{align}\label{PHIBARXILARGETHMRES1}
\overline{\Phi}_\bn(\xi)=
\frac{2^{2-d}}{d(d+1)\zeta(d)}\xi^{-3}
+O\bigl(\xi^{-3-\frac 2d}\bigr)
\left.\begin{cases}
1&\text{if }\:d=2 
\\
\log\xi&\text{if }\:d=3
\\
1&\text{if }\:d\geq4
\end{cases}\right\}
\qquad \text{as } \: \xi\to\infty.
\end{align}
\end{thm}

Finally, regarding the distribution of the free path length starting at a lattice point (with the scatterer removed), let us define $\delta^*_d(\fZ)$ 
\label{DELTASTARDEF}
to be the maximal lattice packing density of any cylinder of the form 
$\fZ %
=(0,a)\times\scrB_r^{d-1}$
(cf.\ \cite{Rogers64} or \cite{GL});
of course this number is independent of $a$ and $r$.
\begin{prop} \label{PHIZEROASYMPT}
Let $d\geq 2$ be given. Then $\Phi_\bn(\xi)$ has compact support;
$\Phi_\bn(\xi)>0$ holds if and only if 
$0<\xi<\xi_0(0)=2^{d-1}v_{d-1}^{-1}\delta^*_d(\fZ)$, where
$\xi_0(w)$ is the function from Corollary \ref{PHIXIWSUPCOR}.
\end{prop}
Cf.\ \cite{Zong05} for a brief listing of what is known about 
$\delta^*_d(\fZ)$ for general $d$, in particular note that
$\delta^*_{d-1}(\scrB)\leq\delta^*_d(\fZ)\leq\delta_{d-1}(\scrB)$,
where $\delta_{d-1}(\scrB)$ 
is the maximal packing density of balls in dimension $d-1$,
and $\delta^*_{d-1}(\scrB)$ is the corresponding lattice packing density.
For dimension $d\leq 4$ we have explicitly: 
\begin{itemize}
	\item For $d=2$, we have $\delta^*_2(\fZ)=1$, which implies $\xi_0(0)=1$; cf.~\cite[(30)]{partIII}.
	\item For $d=3$, we have $\delta^*_3(\fZ)=\frac{\pi}{\sqrt{12}}$, which implies $\xi_0(0)=\frac{2}{\sqrt3}$.
	\item For $d=4$, we have $\delta^*_4(\fZ)=\frac{\pi}{\sqrt{18}}$
(cf.~\cite{Woods58}), which implies $\xi_0(0)=\sqrt2$.
\end{itemize}

\section{\texorpdfstring{Asymptotics of $\Phi_\bn(\xi,\vecw,\vecz)$ for $\xi$ small}{Asymptotics of Phi0(xi,w,z) for xi small}}

\subsection{Recollection of definitions}   %
\label{PHIDEFSEC}

We begin by recalling the formulas of the probability densities
$\Phi(\xi,\vecw)$ and $\Phi_\bn(\xi,\vecw,\vecz)$ derived in \cite{partI}.

Throughout this paper we will write $G=\SL(d,\R)$, $\Gamma=\SL(d,\Z)$ and
$X_1=\Gamma\backslash G$. \label{GMUX1DEF}
We will view $X_1$ as the space of lattices in $\R^d$ of covolume one,
by letting $M\in G$ correspond to the lattice 
$\Z^dM=\{\vecv M\col \vecv\in\Z^d\}$.
We define $\mu$ 
to be the measure on $\myX$ coming from Haar measure on
$G$, normalized to be a probability measure.
We will sometimes write $G^{(d)}$, $\Gamma^{(d)}$, $X_1^{(d)}$ and $\mu^{(d)}$ 
for $G$, $\Gamma$, $X_1$ and $\mu$,
if we need to emphasize the dimension.

We denote by $\fZ(c_1,c_2,r)$ the cylinder
\begin{align}\label{FZC1C2RDEF}
\fZ(c_1,c_2,r)=(c_1,c_2)\times\scrB_r^{d-1}
=\bigl\{(x_1,\ldots,x_d)\in\RR^d \col  c_1 < x_1 < c_2, 
\|(x_2,\ldots,x_d)\|<r \big\}.
\end{align}
The function $\Phi:\R_{>0}\times\scrB_1^{d-1}\to[0,1]$ can be expressed as the probability that a random lattice in $X_1$ does not intersect the cylinder $\fZ(0,\xi,1)+(0,\vecw)$, i.e., 
\begin{align}\label{PHIXIWDEF}
\Phi(\xi,\vecw)=\mu\bigl(\bigl\{M\in \myX\col
\Z^dM\cap(\fZ(0,\xi,1)+(0,\vecw))=\emptyset\bigr\}\bigr),
\end{align}
cf.\ \cite[(8.32)]{partI}.

Next, for each $\vecy\in\R^d\setminus\{\bn\}$ we set
\begin{align}\label{X1YDEF}
X_1(\vecy)=\{M\in X_1\col\vecy\in\Z^dM\}.
\end{align}
This space carries a natural probability measure $\nu_\vecy$,
\label{NUYDEF}
the properties of which are discussed in
\cite[Sec.\ 7]{partI} and \cite[Sec.\ 5]{lprob}.
The function $\Phi_\bn:\R_{>0}\times\scrB_1^{d-1}\times\scrB_1^{d-1}
\to[0,1]$ is then given by
\begin{align}\label{PHI0ZERODEF}
\Phi_\bn(\xi,\vecw,\vecz)=
\nu_\vecy\bigl(\bigl\{M\in X_1(\vecy) \col 
\Z^dM\cap(\fZ(0,\xi,1)+(0,\vecz))=\emptyset\bigr\}\bigr),
\end{align}
where $\vecy=(\xi,\vecz+\vecw)$.

\subsection{\texorpdfstring{Proof of Theorem \ref*{PHI0ZEROSMALLTHM}}{Proof of Theorem 1.1}}

We now give the proof of Theorem \ref{PHI0ZEROSMALLTHM},
viz.\ the fact that
\begin{align*}
\zeta(d)^{-1}\bigl(1-2^{d-1}v_{d-1}\xi\bigr)\leq
\Phi_\bn(\xi,\vecw,\vecz)\leq\zeta(d)^{-1}
\end{align*}
for all $\xi>0$, $\vecw,\vecz\in\scrB_1^{d-1}$.

Using the $G$-invariance of $\nu_\vecy$ (\cite[Lemma 7.2]{partI}) 
we may rewrite \eqref{PHI0ZERODEF} as
\begin{align}\label{PHI0ZEROSMALLTHMPF1}
\Phi_\bn(\xi,\vecw,\vecz)=
\nu_\vecy\bigl(\bigl\{M\in X_1(\vecy) \col 
\Z^dM\cap\fZ=\emptyset\bigr\}\bigr),
\end{align}
with
\begin{align*}
\fZ=\xi^{\frac1d}\bigl(\fZ(0,1,1)+(0,\vecz)\bigr);    
\qquad
\vecy=\xi^{\frac1d}(1,\vecz+\vecw).
\end{align*}
We will keep these choices of $\fZ$ and $\vecy$ from now on.
Recall the splitting 
\begin{align*}
X_1(\vecy)=\sqcup_{\veck\in S}X_1(\veck,\vecy)
\end{align*}
where 
\begin{align}\label{X1KYDEF}
X_1(\veck,\vecy)=\{\Gamma M\in X_1\col M\in G,\:\veck M=\vecy\},
\end{align}
and where we may take $S=\{k\vece_1\col k\in\Z_{>0}\}$,
cf.\ \cite[(7.11)--(7.12)]{partI}.
Note that if $M\in\SL(d,\R)$ is such that $k\vece_1 M=\vecy$ with $k\geq 2$
then $\vece_1 M=k^{-1}\vecy\in\fZ$;
hence \eqref{PHI0ZEROSMALLTHMPF1} may be rewritten as
\begin{align}\label{PHI0ZEROSMALLTHMPF1simpl}
\Phi_\bn(\xi,\vecw,\vecz)=
\nu_{\vecy}\bigl(\bigl\{M\in X_1(\vece_1,\vecy)\col
\Z^dM\cap\fZ=\emptyset\bigr\}\bigr).
\end{align}
Writing $\vecy=(y_1,\ldots,y_d)$ (thus $y_1=\xi^{\frac1d}$) we now set
\begin{align*}
M'=\begin{pmatrix} y_1 & y_2 & \cdots & y_d
\\ & y_1^{-\frac 1{d-1}} & &
\\ & & \ddots &
\\ & & & y_1^{-\frac 1{d-1}}
\end{pmatrix}\in G,
\end{align*}
so that $\vece_1M'=\vecy$.
Then by \cite[(7.15)]{partI} we have
\begin{align}\label{PHI0ZEROSMALLTHMPF2}
X_1(\vece_1,\vecy)=\bigl((\Gamma(1)\cap H)\setminus H\bigr)M'
\end{align}
where 
\begin{align}\label{HDEF}
H=\{g\in G\col\vece_1g=\vece_1\}
=\Bigl\{\matr 1\bn{\trans\vecv}A\col\vecv\in\R^{d-1},\:A\in G^{(d-1)}\Bigr\},
\end{align}
and the restriction of the measure $\nu_\vecy$ 
to $X_1(\vece_1,\vecy)$ corresponds to the measure
$\zeta(d)^{-1}d\vecv\,d\mu^{(d-1)}(A)$ on
$(\Gamma(1)\cap H)\setminus H$ under \eqref{PHI0ZEROSMALLTHMPF2}.
We also know that a fundamental domain for
$(\Gamma(1)\cap H)\setminus H$ is given by
$\{\smatr 1\bn{\trans\vecv}A\col \vecv\in[0,1)^{d-1},\:A\in\F_{d-1}\}$,
where $\F_{d-1}$ is any fundamental domain for \label{SCRFDEF}
$\Gamma^{(d-1)}\backslash G^{(d-1)}$.
Hence
\begin{align}\label{PHI0ZEROSMALLTHMPF3}
\Phi_\bn(\xi,\vecw,\vecz)=
\zeta(d)^{-1}\int_{\F_{d-1}}\int_{[0,1)^{d-1}}
I\biggl(\Z^d\matr 1\bn{\trans\vecv}A M'\cap\fZ=\emptyset\biggr)
\,d\vecv\,d\mu^{(d-1)}(A),
\end{align}
where $I(\cdot)$ is the indicator function.
This relation immediately implies
\begin{align}\label{PHI0ZEROSMALLTHMPF4}
\Phi_\bn(\xi,\vecw,\vecz)\leq\zeta(d)^{-1},
\end{align}
thus proving the upper bound in \eqref{PHI0ZEROSMALLTHMRES}.
On the other hand, we claim that if $A$ is any matrix in
$G^{(d-1)}$ with the property that
$\|\vecm A\|>2\xi^{\frac1{d-1}}$ for all $\vecm\in\Z^{d-1}\setminus\{\bn\}$,
then $\Z^d\smatr1\bn{\trans\vecv}A M'\cap\fZ=\emptyset$ holds for all
$\vecv\in\R^{d-1}$.
Indeed, note that for any $\vecm\in\Z^{d-1}$, $j\in\Z$ we have
\begin{align}
(j,\vecm)\matr 1\bn{\trans\vecv}A M'
&=\bigl(j+\vecm\trans\vecv,\vecm A\bigr)M'
=(j+\vecm\trans\vecv)\vecy+y_1^{-\frac1{d-1}}(0,\vecm A),
\end{align}
and if this vector lies in $\fZ$ then its $\vece_1$-component must lie in
$(0,\xi^{\frac1d})$, viz.\ $0<j+\vecm\trans\vecv<1$.
Hence if we write $\alpha=j+\vecm\trans\vecv$ and
let $p:\R^d\to\R^{d-1}$ be the projection
$(x_1,\ldots,x_d)\mapsto(x_2,\ldots,x_d)$, it follows that
\begin{align*}
\Bigl\|p\Bigl((j,\vecm)\smatr 1\bn{\trans\vecv}A M'
-\xi^{\frac1d}(0,\vecz)\Bigr)\Bigr\|
=\xi^{\frac1d}\Bigl\|\alpha(\vecz+\vecw)-\vecz+
\xi^{-\frac1{d-1}}\vecm A\Bigr\| 
\hspace{50pt}
\\
\geq\xi^{\frac1d}\Bigl(\xi^{-\frac1{d-1}}\|\vecm A\|
-\alpha\|\vecw\|-(1-\alpha)\|\vecz\|\Bigr)
>\xi^{\frac1d}\Bigl(\xi^{-\frac1{d-1}}\|\vecm A\|-1\Bigr).
\end{align*}
Now if $\vecm\in\Z^{d-1}\setminus\{\bn\}$ and if $A$ has the
stated property then the above distance is $>\xi^{\frac1d}$ and hence
$(j,\vecm)\smatr 1\bn{\trans\vecv}A M'\notin\fZ$.
Furthermore if $\vecm=\bn$ then $0<j+\vecm\trans\vecv<1$ is impossible
and we again conclude
$(j,\vecm)\smatr 1\bn{\trans\vecv}A M'\notin\fZ$.
This proves the claim.

Using the claim just proved together with \eqref{PHI0ZEROSMALLTHMPF3}
we obtain
\begin{align*}
\zeta(d)^{-1}-\Phi_\bn(\xi,\vecw,\vecz)
\leq\zeta(d)^{-1}\mu^{(d-1)}\bigl(\bigl\{A\in \myX^{(d-1)}\col
\Z^{d-1}A\cap\scrB_{2\xi^{1/(d-1)}}^{d-1}\neq\{\bn\}\bigr\}\bigr),
\end{align*}
and by a well-known bound 
(cf.\ e.g.\ \cite[Lemma 2.2]{lprob} or \cite[p.\ 167]{Schmidt59})
the right hand side is
\begin{align*}
\leq\zeta(d)^{-1}\vol\bigl(\scrB_{2\xi^{1/(d-1)}}^{d-1}\bigr)
=\zeta(d)^{-1}2^{d-1}v_{d-1}\xi.
\end{align*}
This bound together with \eqref{PHI0ZEROSMALLTHMPF4}
complete the proof of Theorem \ref{PHI0ZEROSMALLTHM}.
\hfill$\square$

\subsection{\texorpdfstring{A parametrization of $X_1(\vece_1,\vecy)$ for $d=3$}{A parametrization of X1(e1,y) for d=3}}

We now turn to the case $d=3$ where we will prove the
explicit formula for $\Phi_\bn(\xi,\vecw,\vecz)$ for $\xi$ small
stated in Theorem \ref{D3EXPLTHM}.
As a preparation we first give a parametrization of $X_1(\vece_1,\vecy)$.
Let $\vecy=(y_1,y_2,y_3)\in\R^3\setminus\{\bn\}$ be given, and fix a matrix
$M_\vecy\in G=G^{(3)}$ with $\vece_1M_\vecy=\vecy$.
Recall that
\begin{align}\label{X1E1Y}
X_1(\vece_1,\vecy)=(\Gamma(1)\cap H)\backslash HM_\vecy
\end{align}
(cf.\ \cite[(7.15)]{partI}).
Now an arbitrary matrix $M$ in $HM_\vecy$ has the form
\begin{align}\label{MYQP}
M=\begin{pmatrix}\vecy\\\vecq\\\vecp\end{pmatrix}
=\begin{pmatrix}y_1&y_2&y_3\\q_1&q_2&q_3\\p_1&p_2&p_3\end{pmatrix}
\in HM_\vecy,
\end{align}
where $\vecq=(q_1,q_2,q_3)$ and $\vecp=(p_1,p_2,p_3)$ 
are two real vectors satisfying
$\vecq\cdot(\vecp\times\vecy)=1$, with ``$\times$'' denoting vector product.
Such a pair of vectors may be parametrized by
$\langle \vecp,\vecx\rangle\in(\R^3\setminus\R\vecy)\times\R^2$,
via the map
\begin{align}\label{QYPX}
\vecq
=\vecq_{\vecy,\vecp}(\vecx)
:=x_1\vecy+x_2\vecp+\|\vecp\times\vecy\|^{-2}\vecp\times\vecy
\qquad (\text{where }\:\vecx=(x_1,x_2)).
\end{align}
Let us write $[\vecp,\vecx]_{\vecy}$ for the matrix
$M$ obtained in this way.
We have thus exhibited a (surjective) diffeomorphism
\begin{align*}
(\R^3\setminus\R\vecy)\times\R^2\ni\langle\vecp,\vecx\rangle
\mapsto[\vecp,\vecx]_{\vecy}\in HM_\vecy.
\end{align*}
Note that the lattice corresponding to $[\vecp,\vecx]_{\vecy}$ is
\begin{align}\label{PXYLATTICE}
\Z^3[\vecp,\vecx]_{\vecy}=\Z\vecy+\Z\vecp+\Z\vecq_{\vecy,\vecp}(\vecx).
\end{align}

Recall that $\nu_\vecy$ is a probability measure on $X_1(\vecy)$;
by restriction this gives a measure on $X_1(\vece_1,\vecy)$
(with $\nu_\vecy(X_1(\vece_1,\vecy))=\zeta(d)^{-1}$);
we will write $\nu_\vecy$ also for the %
lift of this measure to $HM_\vecy$.
\begin{lem}\label{MHYJACLEM}
Given $\vecy\in\R^3\setminus\{\bn\}$,
the measure $\nu_\vecy$ on $HM_\vecy$ takes the following form in 
the $[\vecp,\vecx]_{\vecy}$-parametrization:
\begin{align}\label{MHYJACLEMRES}
d\nu_\vecy=\frac6{\pi^2\zeta(3)}\,d\vecp\,d\vecx.
\end{align}
\end{lem}
\begin{proof}
Let $\mu_H$ be the Haar measure on $H$ normalized so that 
$\mu_H((\Gamma(1)\cap H)\backslash H)=1$;
thus $d\mu_H=d\vecv\,d\mu^{(2)}(A)$ in the coordinates
$\matr1\bn{\trans\vecv}A\in H$ ($\vecv\in\R^2$, $A\in G^{(2)}$).
By definition $\nu_\vecy$ on $HM_{\vecy}$
is the measure which corresponds to
$\zeta(3)^{-1}\mu_H$ on $H$ under $h\mapsto hM_{\vecy}$
(cf.\ \cite[Sec.\ 7.1]{partI}).
Recall that both the set $HM_\vecy$ and the measure $\nu_\vecy$ on this set
are independent of the choice of $M_\vecy\in G$
(subject to $\vece_1M_\vecy=\vecy$).
In particular if $\vecy$ is fixed and $R\in\SO(3)$ is a fixed rotation
then we may choose $M_{\vecy R}=M_\vecy R$;
using this together with 
$\vecq_{\vecy R,\vecp R}(\vecx)=\vecq_{\vecy,\vecp}(\vecx)R$
and the fact that $\vecp\mapsto\vecp R$ preserves the Lebesgue measure 
$d\vecp$, one checks that if \eqref{MHYJACLEMRES} holds for
$\vecy$ then it also holds with $\vecy$ replaced by $\vecy R$.

Hence it suffices to prove \eqref{MHYJACLEMRES} in the case
$\vecy=y\vece_1$, $y>0$, and we may then assume 
$M_\vecy$ to be the diagonal matrix
$M_\vecy=\text{diag}\bigl[y,y^{-\frac12},y^{-\frac12}\bigr]$.
In this case we compute that 
$h=[\vecp,\vecx]_\vecy M_\vecy^{-1}=\matr1\bn{\trans\vecv}A\in H$, with
\begin{align*}
\vecv=(v_1,v_2)=y^{-1}(yx_1+x_2p_1,p_1)
\end{align*}
and, introducing variables %
$\nu>0$ and $\vartheta\in\R/2\pi\Z$ through
$(p_2,p_3)=\nu^{-\frac12}(\sin\vartheta,\cos\vartheta)$,
\begin{align*}
A=y^{\frac12}\matr{x_2p_2+\frac{p_3}{y(p_2^2+p_3^2)}}{x_2p_3-\frac{p_2}{y(p_2^2+p_3^2)}}{p_2}{p_3}
=\matr1{x_2}01\matr{(\nu/y)^{\frac12}}00{(\nu/y)^{-\frac12}}
\matr{\cos\vartheta}{-\sin\vartheta}{\sin\vartheta}{\cos\vartheta}.
\end{align*}
Note that the last matrix product is the Iwasawa decomposition
of $G^{(2)}=\SL(2,\R)$, in terms of which
the (normalized) Haar measure takes the form
\begin{align*}
d\mu^{(2)}(A)=\frac{3y}{\pi^2}\,\frac{dx_2\,d\nu\,d\vartheta}{\nu^2}
\end{align*}
(cf.\ \eqref{SLDZHAAR} below for the case of general $d$).
Hence
\begin{align*}
d\mu_H(h)=d\vecv\,d\mu^{(2)}(A)
=\frac3{\pi^2}\frac{dp_1\,dx_1\,dx_2\,d\nu\,d\vartheta}{\nu^2}
=\frac6{\pi^2}\,dp_1\,dp_2\,dp_3\,dx_1\,dx_2
=\frac6{\pi^2}\,d\vecp\,d\vecx.
\end{align*}
\end{proof}

\subsection{\texorpdfstring{Proof of Theorem \ref*{D3EXPLTHM}}{Proof of Theorem 1.4}}

We keep $d=3$.
Let $\xi>0$ and $\vecw,\vecz\in\scrB_1^2$ be given,
and assume $\xi<\xi_1(\vecw,\vecz)$. 
(The case $\xi=\xi_1(\vecw,\vecz)$ then follows by continuity.)
Recall that
\begin{align}\label{PHI0ZEROSMALLTHMPF1repeat}
\Phi_\bn(\xi,\vecw,\vecz)=
\nu_\vecy\bigl(\bigl\{M\in X_1(\vece_1,\vecy) \col 
\Z^3M\cap\fZ=\emptyset\bigr\}\bigr)
\end{align}
where $\fZ=\fZ(0,\xi,1)+(0,\vecz)$ and $\vecy=(\xi,\vecz+\vecw)$.
If $\vecw\neq\vecz$ then we let $\ell\subset\R^2$ be the line along the
(unique) chord in 
$\scrB_1^2$ with has midpoint $\frac12(\vecw-\vecz)$, and let
$V\subset\R^3$ be the affine plane
\begin{align*}
V=\bigl\{(x_1,x_2,x_3)\col x_1\in\R,\: (x_2,x_3)\in\vecz+\ell\bigr\}.
\end{align*}
Finally let $V^+\subset\R^3$ be that open halfspace which has boundary $V$
and which contains the axis of $\fZ$, viz.\ $\R\times\{\vecz\}\subset V^+$.
If $\vecw=\vecz$ then we modify this definition by letting 
$\ell$ be an arbitrary diameter of $\scrB_1^2$ and $V^+$ be any of the two
open halfspaces determined by $V$.

Now consider the map
\begin{align}\label{D3EXPLTHMDIRECTPF1}
J:(V^+\cap\fZ\setminus\R\vecy)\times[0,1)^2\ni
\langle\vecp,\vecx\rangle\mapsto 
(\Gamma(1)\cap H)[\vecp,\vecx]_\vecy\in X_1(\vece_1,\vecy).
\end{align}

Let us first prove that the image of $J$ equals,
up to a set of $\nu_\vecy$-measure zero,
\begin{align}
\bigl\{M\in X_1(\vece_1,\vecy)\col\Z^3M\cap\fZ\neq\emptyset\bigr\}.
\end{align}
Indeed, every $M\in X_1(\vece_1,\vecy)$ in the image of $J$
clearly satisfies $\Z^3M\cap\fZ\neq\emptyset$,
since $\vecp\in\Z^3[\vecp,\vecx]_\vecy$.
On the other hand, if $M$ is any given element in $X_1(\vece_1,\vecy)$
which satisfies $\Z^3M\cap\fZ\neq\emptyset$
and also $\Z^3M\cap\fZ\cap V=\emptyset$
(this latter condition holds for $\nu_\vecy$-almost all
$M\in X_1(\vece_1,\vecy)$),
then we will prove that   %
$M=J(\langle\vecp,\vecx\rangle)$
for some $\langle\vecp,\vecx\rangle\in
(V^+\cap\fZ\setminus\R\vecy)\times[0,1)^2$.
To this end, among the finitely many points in $\Z^3M\cap\fZ$ we
pick one which has minimal distance to the line $\R\vecy$,
and call it $\vecp'$. Note that $\vecp'\notin V$,
due to our assumption $\Z^3M\cap\fZ\cap V=\emptyset$.
If $\vecp'\in V^+$ then set $\vecp:=\vecp'$;
otherwise set $\vecp:=\vecy-\vecp'$.
In both cases $\vecp\in V^+\cap\fZ$ must hold
(this follows from the fact that among the two regions into which
the line $\ell$ splits the unit disc $\scrB_1^2$,
the smaller one is mapped into the larger one by reflection in the point
$\frac12(\vecw-\vecz)$);
also $\vecp\in\Z^3M$ and $\vecp$ has the same distance
as $\vecp'$ to the line $\R\vecy$.
Note also that $\vecp\notin\R\vecy$, since
$\vecp\in\fZ$ forces $0<\vecp\cdot\vece_1<\xi$
while all points in $\R\vecy\cap\Z^3M=\Z\vecy$ have $\vece_1$-coordinates
in $\Z\xi$.
Now $\Z^3M\cap(\R\vecy+\R\vecp)=\Z\vecy+\Z\vecp$, for otherwise
there would exist a point $\vecr\in\Z^3M\setminus\{\bn,\vecy,\vecp\}$ 
lying in the triangle $\triangle\bn\vecy\vecp$
(the convex hull of $\bn,\vecy,\vecp$);
this point $\vecr$ would belong to $\fZ$ since $\fZ$ is convex,
and $\vecr$ would also lie closer to $\R\vecy$ than $\vecp$ does,
thus causing a contradiction.
It follows from $\Z^3M\cap(\R\vecy+\R\vecp)=\Z\vecy+\Z\vecp$
that there exists a point
$\vecq\in\Z^3M$ satisfying $\Z^3M=\Z\vecy+\Z\vecp+\Z\vecq$
and $\vecq\cdot(\vecp\times\vecy)=1$.
Then $\vecq=\vecq_{\vecy,\vecp}(\vecx)$ for some $\vecx\in\R^2$
(cf.\ \eqref{QYPX}), and replacing 
$\vecq$ by $\vecq+n_1\vecy+n_2\vecp$ with appropriate
$n_1,n_2\in\Z$ we may assume $\vecx\in[0,1)^2$.
Now $\Z^3M=\Z\vecy+\Z\vecp+\Z\vecq_{\vecy,\vecp}(\vecx)
=\Z[\vecp,\vecx]_\vecy$ (cf.\ \eqref{PXYLATTICE}), and this implies 
$M=(\Gamma(1)\cap H)[\vecp,\vecx]_\vecy=J(\langle\vecp,\vecx\rangle)$,
thus completing the proof of our claim.

Next we prove that $J$ %
is injective.
Thus assume $J(\langle\vecp,\vecx\rangle)=J(\langle\vecp',\vecx'\rangle)$
for some $\vecp,\vecp'\in V^+\cap\fZ\setminus\R\vecy$, %
$\vecx,\vecx'\in [0,1)^2$.
Then
\begin{align}\label{D3EXPLTHMDIRECTPF3}
\Z\vecy+\Z\vecp+\Z\vecq_{\vecy,\vecp}(\vecx)
=\Z\vecy+\Z\vecp'+\Z\vecq_{\vecy,\vecp'}(\vecx')
\end{align}
(cf.\ \eqref{PXYLATTICE}).
We now claim that 
\begin{align}\label{D3EXPLTHMDIRECTPF2}
\fZ\cap\bigl(\R\vecy+\R\vecp+n\vecq_{\vecy,\vecp}(\vecx)\bigr)=\emptyset
\qquad\text{for all }\: n\in\Z\setminus\{0\}.
\end{align}
Indeed, assume that there exists a point
$\vecr\in \fZ\cap(\R\vecy+\R\vecp+n\vecq_{\vecy,\vecp}(\vecx))$, for some
non-zero integer $n$. %
Then 
the tetrahedron with vertices
$\bn,\vecy,\vecp,\vecr$ has volume $\frac16|n|\geq\frac16$
(cf.\ \eqref{QYPX});
hence after a scaling and a translation we obtain a tetrahedron
which is contained in the closed cylinder $\overline{\fZ(0,1,1)}
=[0,1]\times\overline{\scrB_1^2}$, which has one vertex at
$(0,-\vecz)$ and another at $(1,\vecw)$,
and which has volume $\geq(6\xi)^{-1}$.
This is impossible, since $\xi<\xi_1(\vecw,\vecz)$
(recall the definition of $\xi_1(\vecw,\vecz)$ given just before
the statement of Theorem \ref{D3EXPLTHM});
hence \eqref{D3EXPLTHMDIRECTPF2} is proved.

Now $\vecp'\in\Z\vecy+\Z\vecp+\Z\vecq_{\vecy,\vecp}(\vecx)$,
$\vecp'\in\fZ$ and \eqref{D3EXPLTHMDIRECTPF2} imply
$\vecp'\in\Z\vecy+\Z\vecp$.
Similarly $\vecp\in\Z\vecy+\Z\vecp'$.
It follows that $\vecp'=\ve\vecp+m\vecy$ for some $\ve=\pm1$, $m\in\Z$.
If $\ve=-1$ then since both $\vece_1\cdot\vecp$ and $\vece_1\cdot\vecp'$ 
lie in the interval $(0,\xi)$ 
we must have $m=1$, viz.\ $\vecp'=\vecy-\vecp$.
It follows that the midpoint of the line segment between $\vecp$
and $\vecp'$ is $\frac12\vecy\in V$, and this contradicts the assumption that 
both $\vecp$ and $\vecp'$ lie in $V^+$.
Hence we must have $\ve=1$; and by again using
$\vece_1\cdot\vecp$, $\vece_1\cdot\vecp'\in(0,\xi)$
we get $m=0$, viz.\ $\vecp'=\vecp$.
Finally using \eqref{QYPX}, \eqref{D3EXPLTHMDIRECTPF3} 
and $\vecx,\vecx'\in[0,1)^2$
we see that also $\vecx'=\vecx$ must hold.
This completes the proof %
that $J$ is injective.

It follows that $J$ is a diffeomorphism of
$(V^+\cap\fZ\setminus\R\vecy)\times(0,1)^2$
onto an open subset of full ($\nu_\vecy$-)measure in
$\{M\in X_1(\vece_1,\vecy)\col\Z^3M\cap\fZ\neq\emptyset\}$.
Hence by \eqref{PHI0ZEROSMALLTHMPF1repeat} and Lemma \ref{MHYJACLEM},
\begin{align*}
\Phi_\bn(\xi,\vecw,\vecz)
=\nu_\vecy\bigl(X_1(\vece_1,\vecy)\bigr)
-\frac6{\pi^2\zeta(3)}\int_{\fZ\cap V^+}\int_{(0,1)^2}
d\vecp\,d\vecx
=\frac1{\zeta(3)}-\frac6{\pi^2\zeta(3)}F(\sfrac12\|\vecw-\vecz\|)\xi.
\end{align*}
This completes the proof of Theorem \ref{D3EXPLTHM}.
\hfill
$\square$ $\square$ $\square$

\vspace{5pt}

Using methods that are beyond the scope of the present paper we 
are able to prove that the function $\xi_1(\vecw,\vecz)$ gives the
true range of validity of the formula \eqref{D3EXPLTHMRES},
i.e.\ for any $\vecw,\vecz\in\scrB_1^2$,
there exist triples $\langle\xi,\vecw',\vecz'\rangle$ arbitrarily near
$\langle\xi_1(\vecw,\vecz),\vecw,\vecz\rangle$ at which
\eqref{D3EXPLTHMRES} fails.

We next give sharp lower and upper bounds on $\xi_1(\vecw,\vecz)$.
Recall that we defined (on p.~\pageref{FDEF}) $\xi_1(\vecw,\vecz)$ 
for $\vecw,\vecz\in\scrB_1^2$ as $\xi_1(\vecw,\vecz)=(6V)^{-1}$, 
where $V$ is the largest possible volume of a 
tetrahedron which is contained in the closed cylinder
$\overline{\fZ(0,1,1)}=[0,1]\times\overline{\scrB_1^2}$ 
and which has one vertex at
$(0,-\vecz)$ and another at $(1,\vecw)$.
By applying the same formula to arbitrary
$\vecw,\vecz\in\overline{\scrB_1^2}$ we obtain an extension
of $\xi_1$ to a continuous
function $\overline{\scrB_1^2}\times\overline{\scrB_1^2}\to\R$;
we will write $\xi_1$ also for this extension.
\begin{lem}\label{XI1BOUNDSLEM}
We have
\begin{align*}
\frac14\leq\xi_1(\vecw,\vecz)\leq1\qquad
\text{for all }\:\vecw,\vecz\in\overline{\scrB_1^2}.
\end{align*}
Here $\xi_1(\vecw,\vecz)=\frac14$ holds if and only if
$\|\vecw\|=\|\vecz\|=1$ and $\vecw\cdot\vecz=0$;
and $\xi_1(\vecw,\vecz)=1$ holds if and only if $\vecw=\vecz=\bn$.
\end{lem}
Note that this implies that
$\frac14<\xi_1(\vecw,\vecz)\leq1$ for all $\vecw,\vecz\in\scrB_1^2$,
both bounds being sharp.
\begin{proof}
Let $V_0$ be the maximal volume of a tetrahedron contained
in $\overline{\fZ(0,1,1)}$.
This volume is clearly attained; let us fix 
$T\subset\overline{\fZ(0,1,1)}$ to be a tetrahedron of volume $V_0$.
By a simple variational argument, varying the vertices of $T$ one at a time,
we see that $T$ may be continuously deformed, keeping its volume $V_0$ fixed,
into a tetrahedron $T'$ which has all its vertices lying on the two
circles $\{0\}\times\S_1^1$ and $\{1\}\times\S_1^1$.
Clearly each of these circles must contain at least one vertex.
If each circle contains two vertices,
say $\veca,\vecb\in\{0\}\times\S_1^1$ and $\vecc,\vecd\in\{1\}\times\S_1^1$,
then the same type of variational argument also shows that
both vectors $\vecc-\vece_1$ and $\vecd-\vece_1$ must be orthogonal to
$\vecb-\veca$, and both $\veca,\vecb$ must be orthogonal to $\vecd-\vecc$.
In other words, the line segment $\veca\vecb$ must be a
diameter of $\{0\}\times\S_1^1$, and the line segment $\vecc\vecd$ 
must be that diameter of $\{1\}\times\S_1^1$ whose direction
is orthogonal to $\veca\vecb$.
We compute that any such tetrahedron $T'$ has volume $\frac23$.
On the other hand if one circle contains three of the vertices,
say $\veca,\vecb,\vecc$, %
then $\triangle\veca\vecb\vecc$ must be an equilateral triangle,
and $\vol(T')=\frac{\sqrt3}4<\frac23$, a contradiction.
It follows that $V_0=\frac23$.
Note also that for a tetrahedron $T'$ with 
$\veca,\vecb\in\{0\}\times\S_1^1$ and $\vecc,\vecd\in\{1\}\times\S_1^1$
and volume $V_0=\frac23$ 
(i.e.\ with $\veca\vecb$ and $\vecc\vecd$ being diameters 
whose directions are orthogonal),
any perturbation of $\veca$ inside $\overline{\fZ(0,1,1)}$
with $\vecb,\vecc,\vecd$ fixed makes $\vol(T')$ \textit{strictly smaller}
(since the plane through $\veca$
orthogonal to $(\vecb-\vecd)\times(\vecc-\vecd)$ contains \textit{only}
$\veca$ in its intersection with $\overline{\fZ(0,1,1)}$).
This implies that also the original tetrahedron $T$ necessarily had
two vertices on $\{0\}\times\S_1^1$ and two vertices on $\{1\}\times\S_1^1$,
i.e.\ we have proved that a tetrahedron with vertices
$\veca,\vecb,\vecc,\vecd\in\overline{\fZ(0,1,1)}$
attains the maximal volume $V_0=\frac23$
\textit{if and only if,} up to a renaming of the vertices,
$\veca\vecb$ is a diameter of $\{0\}\times\S_1^1$ and
$\vecc\vecd$ is a diameter of $\{0\}\times\S_1^1$ whose direction
is orthogonal to $\veca\vecb$.

This result immediately implies that 
$\xi_1(\vecw,\vecz)\geq\frac14$ for all $\vecw,\vecz\in\overline{\scrB_1^2}$,
with equality if and only if $\|\vecw\|=\|\vecz\|=1$ and $\vecw\cdot\vecz=0$.

Next, given any $\vecw,\vecz\in\overline{\scrB_1^2}$, let
us consider the tetrahedron $T$ which has vertices
$(0,-\vecz)$, $(1,\vecw)$, $(1,\veca)$, $(1,\vecb)$,
with $\veca,\vecb\in\overline{\scrB_1^1}$ chosen so as to maximize the area
of the triangle $\triangle\veca\vecb\vecw$.
A simple variational argument, varying $\veca$ and $\vecb$ one at a time,
shows that these $\veca,\vecb$ must satisfy
$\veca,\vecb\in\S_1^1$ and
$\veca\cdot(\vecb-\vecw)=\vecb\cdot(\veca-\vecw)=0$, and
$\veca$ and $\bn$ must lie on the same side of the line $\vecb\vecw$,
and $\vecb$ and $\bn$ must lie on the same side of the line $\veca\vecw$.
If $\vecw\neq\bn$ then this determines 
$\veca\vecb$ to be the unique chord of $\S_1^1$ 
with midpoint $-\frac{f(\|\vecw\|)}{\|\vecw\|}\vecw$ where
$f(w)=\frac2{\sqrt{w^2+8}+w}$,
and we get %
\begin{align}\label{TETR31MAXVOL}
\vol(T)=\sup_{\alpha\in[0,1]}\sfrac13(\alpha+\|\vecw\|)\sqrt{1-\alpha^2}
=\sfrac13\Bigl(f(\|\vecw\|)+\|\vecw\|\Bigr)\sqrt{1-f(\|\vecw\|)^2}.
\end{align}
On the other hand if $\vecw=\bn$ then the maximal area of
$\triangle\veca\vecb\vecw$ is attained if and only if
$\veca,\vecb\in\S_1^1$ are orthogonal, and then
$\vol(T)=\frac16$, i.e.\ \eqref{TETR31MAXVOL} still holds.
It is clear from the first expression in \eqref{TETR31MAXVOL} that
$\vol(T)$ is a strictly increasing function of $\|\vecw\|\in[0,1]$; 
in particular we have $\vol(T)\geq\frac16$, with equality if and only if
$\vecw=\bn$.
This implies that $\xi_1(\vecw,\vecz)\leq1$, where equality is possible
only if $\vecw=\bn$.
Similarly, by instead taking $T$ to have vertices
$(0,-\vecz)$, $(1,\vecw)$, $(0,\veca)$, $(0,\vecb)$,
we see that $\xi_1(\vecw,\vecz)=1$ can only hold if $\vecz=\bn$.

To complete the proof of the lemma 
it now only remains to prove that $\xi_1(\bn,\bn)=1$.
Thus let $T$ be a tetrahedron with vertices $\veca,\vecb,\vecc,\vecd$
and of maximal volume subject to $\vecc=(0,\bn)=\bn$, $\vecd=(1,\bn)=\vece_1$
and $\veca,\vecb\in\overline{\fZ(0,1,1)}$.
By the same type of variational argument as in the first half of this proof
we may assume $\veca,\vecb\in(\{0\}\times\S_1^1)\cup(\{1\}\times\S_1^1)$.
If $\veca,\vecb$ lie on the \textit{same} circle then
as in the discussion leading to \eqref{TETR31MAXVOL} we get
$\vol(T)=\frac16$.
On the other hand if $\veca,\vecb$ lie on distinct circles,
say $\veca\in\{0\}\times\S_1^1$ and $\vecb\in\{1\}\times\S_1^1$,
then similarly as in the first half of this proof we must have
$(\vecb-\vece_1)\cdot\veca=0$,
and this implies $\vol(T)=\frac16$, again.
Hence $\xi_1(\bn,\bn)=1$.
\end{proof}

\subsection{\texorpdfstring{Explicit formulas for $\Phi(\xi,\vecw)$, 
$\overline{\Phi}_\bn(\xi)$, $\Phi_\bn(\xi)$ and $\Phi(\xi)$ for $d=3$, $\xi$ small}{Explicit formulas for Phi(xi,w), Phi0*(xi) and Phi(xi) for d=3, xi small}}

\begin{proof}[Proof of Corollary \ref{D3EXPLTHMCOR1}]
Using \eqref{PHIFROMPHIZERO} and Theorem \ref{D3EXPLTHM} we have,
whenever $0<\xi\leq\inf_{\vecz\in\scrB_1^2}\xi_1(\vecw,\vecz)$,
\begin{align*}
\Phi(\xi,\vecw)=1-\int_0^\xi\int_{\scrB_1^2}\Phi_\bn(\eta,\vecw,\vecz)
\,d\vecz\,d\eta
=1-\frac{\pi}{\zeta(3)}\xi
+\frac3{\pi^2\zeta(3)}
\biggl(\int_{\scrB_1^2}F(\sfrac12\|\vecw-\vecz\|)\,d\vecz\biggr)\xi^2.
\end{align*}
Substituting $\vecz=\vecw+r(\cos\omega,\sin\omega)$ %
we get, writing $w=\|\vecw\|\in[0,1)$:
\begin{align*}
\int_{\scrB_1^2}F(\sfrac12\|\vecw-\vecz\|)\,d\vecz
=2\pi\int_0^{1-w}
F(\sfrac12r)r\,dr
+2\int_{1-w}^{1+w}F(\sfrac12r)\arccos\Bigl(\frac{w^2+r^2-1}{2wr}\Bigr)
\,r\,dr
=2G(w)
\end{align*}
(cf.\ \eqref{D3EXPLTHMCOR1GDEF}).
Hence \eqref{D3EXPLTHMCOR1RES1} holds.

To see that $G(w)$ is a strictly increasing function of $w$
we may e.g.\ note that, for $0<w<1$,
\begin{align}\label{D3EXPLTHMCOR1PF1}
G'(w)=\int_{1-w}^{1+w}\frac{(r^2-1-w^2)r}
{w\sqrt{(1+w)^2-r^2}\sqrt{r^2-(1-w)^2}}F(\sfrac12r)\,dr,
\end{align}
and here
\begin{align*}
\int_{1-w}^{1+w}\frac{(r^2-1-w^2)r}
{w\sqrt{(1+w)^2-r^2}\sqrt{r^2-(1-w)^2}}\,dr
=\Bigl[-\frac1{2w}\sqrt{(1+w)^2-r^2}\sqrt{r^2-(1-w)^2}\Bigr]_{r=1-w}^{r=1+w}
=0,
\end{align*}
with $\frac{(r^2-1-w^2)r}{w\sqrt{(1+w)^2-r^2}\sqrt{r^2-(1-w)^2}}$
being negative for $r\in(1-w,\sqrt{1-w^2})$ and positive for
$r\in(\sqrt{1-w^2},1+w)$.
Furthermore $F(\frac12r)$ is strictly increasing for $0\leq r\leq2$.
Hence it follows from \eqref{D3EXPLTHMCOR1PF1} that 
$G'(w)>0$ for all $0<w<1$, as desired.

The formulas $G(0)=\frac{\pi(4\pi+3\sqrt3)}{16}$ and 
$G(1)=\frac5{16}\pi^2+1$ follow by straighforward 
computations directly from the definitions;
for example for the computation of $G(1)$ one uses the fact that
the following is a primitive function of $F(\frac12r)r\arccos(\frac r2)$:
\begin{align*}
\sfrac14(3-2r^2)\arccos(\sfrac12r)^2
+\Bigl(\sfrac1{16}r(r^2+6)\sqrt{4-r^2}+\sfrac12\pi(r^2-2)\Bigr)
\arccos(\sfrac12r)
\hspace{50pt}
\\
+\sfrac1{64}(16+12r^2+r^4)-\sfrac14\pi r\sqrt{4-r^2}.
\end{align*}

To complete the proof of Corollary \ref{D3EXPLTHMCOR1} it now
only remains to prove that
\begin{align}\label{D3EXPLTHMCOR1PF2}
\inf_{\vecz\in\scrB_1^2}\xi_1(\vecw,\vecz)
=\min\Bigl(\frac1{2(1+\|\vecw\|)},\frac2{3\sqrt3}\Bigr).
\end{align}
It is clear from the definition of $\xi_1(\vecw,\vecz)$
that $\inf_{\vecz\in\scrB_1^2}\xi_1(\vecw,\vecz)=(6V)^{-1}$,
where $V$ is the largest possible volume of a 
tetrahedron which is contained in the closed cylinder
$[0,1]\times\overline{\scrB_1^2}$ and which has one vertex at
$(1,\vecw)$. By simple variational arguments of the same type as in 
the proof of Lemma \ref{XI1BOUNDSLEM} we see that this volume $V$ is
attained either for a tetrahedron whose three other vertices
lie on $\{0\}\times\S_1^1$ and form an equilateral triangle,
or else for a tetrahedron with vertices $(1,\vecw),\veca,\vecb,\vecc$
where $\veca\in\{1\}\times\S_1^1$, the line segment between $(1,\vecw)$
and $\veca$ contains the point $(1,\bn)$,
and $\vecb\vecc$ is that diameter of $\{0\}\times\S_1^1$
whose direction is orthogonal to $\veca-(1,\vecw)$.
A tetrahedron of the first type has volume
$\frac{\sqrt3}4$ and a tetrahedron of the second type has volume
$\frac13(1+\|\vecw\|)$;
hence \eqref{D3EXPLTHMCOR1PF2} follows.
\end{proof}

\begin{proof}[Proof of Corollary \ref{D3EXPLTHMCOR2}]
By \eqref{PHIXIAVFORMULA2},
Theorem \ref{D3EXPLTHM} and Lemma \ref{XI1BOUNDSLEM} we have,
whenever $0<\xi\leq\frac14$,
\begin{align*}
\overline{\Phi}_\bn(\xi)
=\frac 1{\pi}\int_{\scrB_1^2}\int_{\scrB_1^2}
\Phi_\bn(\xi,\vecw,\vecz)\,d\vecw\,d\vecz
=\frac{\pi}{\zeta(3)}-
\frac6{\pi^3\zeta(3)}\biggl(\int_{\scrB_1^2}\int_{\scrB_1^2}
F(\sfrac12\|\vecw-\vecz\|)\,d\vecz\,d\vecw\biggr)\xi.
\end{align*}
Writing $\vecz=\vecw+r(\cos\omega,\sin\omega)$
and $w=\|\vecw\|$ as in the proof of Corollary \ref{D3EXPLTHMCOR1} we have
\begin{align*}
\int_{\scrB_1^2}\int_{\scrB_1^2}
F(\sfrac12\|\vecw-\vecz\|)\,d\vecz\,d\vecw
=2\pi\int_0^1\int_0^{1+w}F(\sfrac12r)\left.\begin{cases}
2\pi&\text{if }r<1-w\\
2\arccos(\frac{w^2+r^2-1}{2wr})&\text{if }r\geq1-w\end{cases}\right\}
\,r\,dr\,w\,dw
\\
=2\pi\int_0^2 F(\sfrac12r)\int_{\max(0,r-1)}^1\left.\begin{cases}
2\pi&\text{if }w<1-r\\
2\arccos(\frac{w^2+r^2-1}{2wr})&\text{if }w\geq 1-r\end{cases}\right\}
\,w\,dw\,r\,dr.
\end{align*}
Here the inner integral equals $2(\pi-F(\frac12r))$,
and thus the above expression evaluates to
\begin{align*}
4\pi\int_0^2F(\sfrac12r)(\pi-F(\sfrac12r))r\,dr
=4\pi(\sfrac18\pi^2+\sfrac23).
\end{align*}
(The last step is by a straightforward computation, which is much 
simplified by re-using facts from the computation of $G(0)$ and $G(1)$
in the proof of Corollary \ref{D3EXPLTHMCOR1}.)
Hence we obtain the formula for $\overline\Phi_\bn(\xi)$ stated in
Corollary \ref{D3EXPLTHMCOR2}.
Furthermore,
using $\Phi(\xi)=\pi(1-\int_0^\xi\overline\Phi_\bn(\eta)\,d\eta)$
(cf.\ \eqref{PHIXIAVFORMULA} and \eqref{PHIFROMPHIZERO}),
we also obtain the stated formula for $\Phi(\xi)$,
for all $0<\xi\leq\frac14$.

Finally, by \eqref{PHIXIAVFORMULA3} and Theorem \ref{D3EXPLTHM}
we have, whenever $0<\xi\leq\inf_{\vecw\in\scrB_1^2}\xi_1(\vecw,\bn)$,
\begin{align*}
\Phi_\bn(\xi)%
=\frac{\pi}{\zeta(3)}-\frac{12}{\pi\zeta(3)}
\biggl(\int_0^1F(\sfrac12r)r\,dr\biggr)\xi
=\frac{\pi}{\zeta(3)}-
\frac{3(4\pi+3\sqrt3)}{4\pi\zeta(3)}\xi.
\end{align*}
This gives the stated formula for $\Phi_\bn(\xi)$,
since $\inf_{\vecw\in\scrB_1^2}\xi_1(\vecw,\bn)
=\inf_{\vecz\in\scrB_1^2}\xi_1(\bn,\vecz)
=\frac2{3\sqrt3}$ by \eqref{D3EXPLTHMCOR1PF2}.
\end{proof}

\subsection{Numerical computations for $d=3$}\label{NUMCOMPSEC}

We now describe how the graphs of $\Phi(\xi)$ and $\Phi_\bn(\xi)$
in Figures \ref{PHIplot} and \ref{PHI0plot}
were obtained.
For $d=3$ and any $0\leq\alpha<\beta$ we have, 
by \cite[(4.3) and (3.8)]{partI},
\begin{align}\label{NCINTPHI}
\int_\alpha^\beta\Phi(\xi)\,d\xi=\int_{X_1^{(3)}}\int_{\R^3/\Z^3}
I\Bigl(\Bigl\{(\Z^3+\vecx)M\cap\fZ(0,\alpha,1)=\emptyset,\:
(\Z^3+\vecx)M\cap\fZ(0,\beta,1)\neq\emptyset\Bigr\}\Bigr)\,d\vecx\,d\mu(M).
\end{align}
The first curve in Figure \ref{PHIplot} was obtained by
using this formula to evaluate
$\delta^{-1}\int_{n\delta}^{(n+1)\delta}\Phi(\xi)\,d\xi$ for
$\delta=0.02$, $n=0,1,\ldots,99$;
taking this to be an approximation of $\Phi((n+\frac12)\delta)$,
and drawing the piecewise linear curve connecting these points.
We stress that $\Phi(\xi)$ as well as $\Phi_\bn(\xi)$ are known to be
continuous and decreasing functions, cf.\ \eqref{PHIXIAVFORMULA},
\eqref{PHIXIAVFORMULA3} and \cite[Lemma 7.11]{lprob}.
The second curve in Figure \ref{PHIplot} was obtained similarly,
using $\delta=0.05$, $n=20,\ldots,99$.

In order to evaluate %
\eqref{NCINTPHI} numerically,
the integral over $X_1$ was replaced by an average over the Hecke points
corresponding to a large prime $p$, shifted by a fixed rotation;
viz.\ $M=p^{-\frac13}Tk$ with fixed $k\in\SO(3)$ and 
$T$ running through the set
\begin{align*}
S(p)=\left\{\begin{pmatrix}p&&\\&1&\\&&1\end{pmatrix},
\begin{pmatrix}1&a&\\&p&\\&&1\end{pmatrix},
\begin{pmatrix}1&&a\\&1&b\\&&p\end{pmatrix}\col
a,b\in\{0,1,2,\ldots,p-1\}\right\}.
\end{align*}
Also the integral over $\R^3/\Z^3$ was replaced by an average
over the $m^3$ points in $\vecx_0+m^{-1}\Z^3/\Z^3$,
for some fixed $\vecx_0\in\R^3/\Z^3$ and $m\in\Z_{>0}$.
For any fixed  $\vecx_0,k$,
this approximation is known to approach the correct value
as $p,m\to\infty$;
cf., e.g., \cite{Sarnak91}, \cite{Chiu}, \cite{COU}.
In our numerical experiments we noted that the rate of convergence 
seems to be improved by taking $\vecx_0$
irrational and also taking $k$ to be ``sufficiently generic''.
For the curves in Figure \ref{PHIplot} we used
$p=1511$, $m=20$, 
\begin{align*}
k=\begin{pmatrix}
\cos(1/2)&\sin(1/2)&0\\-\sin(1/2)&\cos(1/2)&0\\0&0&1
\end{pmatrix}
\begin{pmatrix}
1&0&0\\0&\cos1&\sin1\\0&-\sin1&\cos1
\end{pmatrix}
\begin{pmatrix}
\cos(3/2)&\sin(3/2)&0\\-\sin(3/2)&\cos(3/2)&0\\0&0&1
\end{pmatrix}
\end{align*}
and $\vecx_0=(\sqrt2,\sqrt3,\sqrt5)$.
We did not prove any error bounds for our approximation;
however as an indication of the error we mention that 
for $n$ with $(n+1)\delta\leq\frac14$,
the value which we obtained for
$\delta^{-1}\int_{n\delta}^{(n+1)\delta}\Phi(\xi)\,d\xi$ 
always differed by less than $0.003$ from the known exact values of
both $\delta^{-1}\int_{n\delta}^{(n+1)\delta}\Phi(\xi)\,d\xi$ 
and $\Phi((n+\frac12)\delta)$ 
(cf.\ Corollary \ref{D3EXPLTHMCOR2}).
Also repeated runs with other choices of $p,m,k,\vecx_0$ indicate that
our values for $\delta^{-1}\int_{n\delta}^{(n+1)\delta}\Phi(\xi)\,d\xi$
are correct to within an absolute error $<0.003$, for all $n$.

Similarly, for Figure \ref{PHI0plot} we used the formula
(\cite[(4.3) and (3.8)]{partI})
\begin{align}\label{NCINTPHI0}
\int_\alpha^\beta\Phi_\bn(\xi)\,d\xi=\int_{X_1^{(3)}}
I\Bigl(\Bigl\{\Z^3M\cap\fZ(0,\alpha,1)=\emptyset,\:
\Z^3M\cap\fZ(0,\beta,1)\neq\emptyset\Bigr\}\Bigr)\,d\mu(M)
\end{align}
to evaluate $\int_{n\delta}^{(n+1)\delta}\Phi_0(\xi)\,d\xi$ for
$\delta=0.02$, $n=0,1,\ldots,59$.
For this case our experiments suggest that, for a given number of
sample points $M=p^{-\frac13}Tk$, we get a significantly better
approximation of the $X_1$-integral by running
$T$ through a \textit{random} subset of $S(p)$ with $p$ quite large, 
than by running $T$ through \textit{all} of $S(p)$ for a $p$ of
more modest size.
The curve in Figure \ref{PHI0plot} was obtained by using $p=10^9+7$,
$k$ as above,
and letting $T$ run through $5.4\cdot10^8$ randomly
choosen points from $S(p)$.
Comparison against the known $\Phi_\bn$-values for $\xi\leq0.38..$
(cf.\ Corollary \ref{D3EXPLTHMCOR2}),
as well as comparisons versus the results of using other random seeds
and/or other choices of $k$ and $p$, indicate that our values for
$\delta^{-1}\int_{n\delta}^{(n+1)\delta}\Phi_\bn(\xi)\,d\xi$
are correct to within an absolute error $<0.005$.

Regarding the support of $\Phi_0(\xi)$, recall from 
Proposition \ref{PHIZEROASYMPT} and the ensuing comments 
that $\Phi_\bn(\xi)>0$ holds if and only if
$\xi<\xi_0(0)=2/\sqrt3=1.15470\ldots$.
Our numerics show that the function $\Phi_\bn(\xi)$
approaches zero quite quickly as $\xi$ approaches $2/\sqrt3$,
and the largest value of
$\sup\{\alpha>0\col\Z^3M\cap\fZ(0,\alpha,1)=\emptyset\}$
which we saw among our sample points $M$ was
$\alpha\approx1.132$.

\section{\texorpdfstring{Asymptotics for $\Phi(\xi,\vecw)$ as $\xi\to\infty$}{Asymptotics for Phi(xi,w) as xi tends to infinity}}
\label{PRELIMINARIESSEC}

In this section we prove Theorem \ref{PHIXIWASYMPTTHM}
on the asymptotic size of $\Phi(\xi,w)$ as $\xi\to\infty$. 
Along the way we prove several lemmas which will also be useful later
in our proof of Theorem \ref{PHI0XILARGETHM}
in Sections \ref{PARABAPPRSEC}--\ref{PHI0XIWZASYMPTSEC}.
Note that a second %
proof of Theorem \ref{PHIXIWASYMPTTHM} will
be given in Section \ref{PARTIALPHIXIZSEC}, where
we deduce Theorem \ref{PHIXIWASYMPTTHM} 
as a consequence of Theorem \ref{PHI0XILARGETHM},
using the integration formula \eqref{PHIFROMPHIZERO}.

\subsection{Preliminaries: Iwasawa decomposition and Siegel domains}
\label{SIEGELSEC}

Recall that we write $G=\SL(d,\R)$.
Let $A$ be the subgroup of %
diagonal matrices with positive entries
\begin{align}\label{ADEF}
\aa(a)=\begin{pmatrix} a_1 & & \\ & \ddots & \\ & & a_d \end{pmatrix}
\in G, \qquad a_j>0,
\end{align}
and let $N$ be the subgroup of upper triangular matrices
\begin{align}\label{NUDEF}
\nn(u)=\begin{pmatrix} 1 & u_{12} & \cdots & u_{1d}
\\ & \ddots & \ddots & \vdots 
\\ & & \ddots & u_{d-1,d} 
\\ & & & 1
\end{pmatrix} \in G.
\end{align}
Every element $M\in G$ has a unique Iwasawa decomposition
\begin{align}\label{IWASAWA}
M=\nn(u)\aa(a)\kk,
\end{align}
with $\kk\in \SO(d)$.
In these coordinates the Haar measure takes the form
(\cite[p.\ 172]{DRS})
\begin{align} \label{SLDZHAAR}
d\mu(M) = \frac{2^{d-1}\pi^{d(d+1)/4}}{\prod_{j=1}^{d}
\Gamma(\frac{j}2) \prod_{j=2}^d \zeta(j)} \rho(a) d\nn(u) d\aa(a) d\kk
\end{align}
where $d\nn$, $d\aa$, $d\kk$, are (left and right) 
Haar measures of $N$, $A$, $\SO(d)$,
normalized by $d\nn(u)=\prod_{1\leq j<k\leq d} du_{jk}$,
$d\aa(a)=\prod_{j=1}^{d-1} (a_j^{-1}\, da_j)$ and
$\int_{\SO(d)} d\kk=1$. For $\rho(a)$ one has
\begin{align}
\rho(a)=\prod_{1\leq i<j\leq d} \frac{a_j}{a_i}
=\prod_{j=1}^d a_j^{2j-d-1}.
\end{align}

We set $\mathcal{F}_N=\bigl\{u \col u_{jk} \in (-\sfrac 12,\sfrac 12], \:
1\leq j<k\leq d\bigr\}$; then
$\{\nn(u) \col u\in \mathcal{F}_N\}$
is a fundamental region for $(\Gamma\cap N)\backslash N$.
We define the following Siegel set:
\begin{align}\label{SIDEF}
\Si_d:=\Bigl\{\nn(u) \aa(a) \kk\in G \col u\in \mathcal{F}_N,\:
0<a_{j+1} \leq \sfrac{2}{\sqrt 3}a_j \: (j=1,\ldots,d-1), 
\: \kk \in \SO(d) \Bigr\}.
\end{align}
It is known that $\Si_d$ contains a fundamental region for
$\myX=\Gamma\backslash G$, and on the other hand
$\Si_d$ is contained in a finite union of fundamental regions for $\myX$
(\cite{Borel}).  %

Given $M=\nn(u)\aa(a)\kk\in G$, its row vectors are
\begin{align} \label{BKDEF}
\vecb_k=(0,\ldots,0,a_k,a_{k+1}u_{k,k+1},\ldots,a_du_{k,d})\kk,
\qquad k=1,\ldots,d.
\end{align}
Thus $\vecb_1,\ldots,\vecb_d$ is a basis of the lattice $\Z^dM$.
If $M\in\Si_d$ then we see that, for all $k$,
\begin{align}\label{BKINEQ}
||\vecb_k||\leq\sum_{j=1}^d a_j \leq c_\clowH a_1,
\qquad\text{with }\: c_\clowH=c_\clowH^{(d)}:=\sum_{j=0}^{d-1} (2/\sqrt 3)^j.
\end{align}
Throughout the paper we will let $c_1,c_2,\ldots$ denote certain
constants which we fix once and for all and which only depend on $d$ 
(or in some cases are absolute);
the $d$-dependence will mostly be suppressed but if necessary it will be 
made explicit by writing ``$c_j^{(d)}$''.

The bound \eqref{BKINEQ} implies that if $M\in\Si_d$ 
and if the lattice $\Z^dM$ has empty intersection with a large ball,
then $a_1$ must be large:
\begin{lem}\label{A1LARGELEM}
For any $M=\nn(u)\aa(a)\kk\in\Si_d$ 
such that the lattice $\Z^dM$ is disjoint from
some ball of radius $R$ in $\R^d$, we have $a_1\gg R$.
\end{lem}
\begin{proof}
Choose $h_1,\ldots,h_d\in\R$ so that
$\vecp=h_1\vecb_1+\ldots+h_d\vecb_d$ is the center of the given ball. 
Let $n_j$ be the integer nearest to $h_j$.
Then $n_1\vecb_1+\ldots+n_d\vecb_d$ is a lattice point of $\Z^dM$,
and has distance $\leq\frac12\bigl(\|\vecb_1\|+\ldots+\|\vecb_d\|\bigr)
\ll a_1$ to $\vecp$.
This distance must be $>R$; hence $a_1\gg R$.
\end{proof}

\subsection{\texorpdfstring{A parametrization of lattices with $a_1$ large}{A parametrization of lattices with a1 large}}
\label{PARAMETRIZATIONSUBSEC}

The set of lattices in $X_1$ which have a representative $M\in\Si_d$
with $a_1$ larger than some large fixed number $A$,
may in an approximate sense be parametrized by the set
$(A,\infty)\times(\S_1^{d-1}/\pm)\times(\R^{d-1}/\Z^{d-1})
\times X_1^{(d-1)}$.
In this section we prove a version of this fact,
Lemma \ref{FDCONTAINMENTSLEM} below,
which we will make use of several times.

Let us fix a function $f$ (smooth except possibly at one point, say)
$\S^{d-1}_1\to\SO(d)$ such that
$\vece_1 f(\vecv)=\vecv$ for all $\vecv\in S^{d-1}$.
Given $M=\nn(u)\aa(a)\kk\in G$,
the matrices $\nn(u)$, $\aa(a)$ and $\kk$ can be split uniquely as
\begin{align}\label{NAKSPLIT}
\nn(u)=\matr 1\vecu{\trans\bn}{\nn(\tu)}; \qquad
\aa(a)=\matr{a_1}\bn{\trans\bn}{a_1^{-\frac 1{d-1}}\aa(\ta)}; \qquad
\kk=\matr 1\bn{\trans\bn}{\tkk} f(\vecv)
\end{align}
where $\vecu\in\R^{d-1}$, $\nn(\tu)\in N^{(d-1)}$,
$a_1>0$, $\aa(\ta)\in A^{(d-1)}$ and
$\tkk\in\SO(d-1)$, $\vecv\in\S^{d-1}_1$.
We set
\begin{align}\label{M1DEF}
\tM=\nn(\tu)\aa(\ta)\tkk\in G^{(d-1)}
\qquad(\text{recall $G^{(d-1)}=\SL(d-1,\R)$}).
\end{align}
In this way we get a bijection between $G$ and
$\R_{>0}\times\S_1^{d-1}\times\R^{d-1}\times G^{(d-1)}$;
we write $M=[a_1,\vecv,\vecu,\tM]$ for the element in $G$ corresponding to
the 4-tuple
$\langle a_1,\vecv,\vecu,\tM\rangle\in\R_{>0}\times\S_1^{d-1}\times\R^{d-1}\times G^{(d-1)}$.
In particular note that
\begin{align}\notag
\Si_d =\Bigl\{[a_1,\vecv,\vecu,\tM]\in G\col
\tM\in\Si_{d-1},\:
\ta_1\leq \sfrac 2{\sqrt 3}a_1^{\frac d{d-1}}
,\:\vecu\in(-\sfrac 12,\sfrac 12]^{d-1}\Bigr\}\qquad
\\\label{SIDSUBSSIDM1}
\subset\Bigl\{[a_1,\vecv,\vecu,\tM]\in G\col
\tM\in\Si_{d-1},\:\vecu\in(-\sfrac 12,\sfrac 12]^{d-1}\Bigr\}.
\end{align}

One checks by a straightforward
computation using %
\eqref{SLDZHAAR}
that the Haar measure $\mu$ takes the following form
in the parametrization $M=[a_1,\vecv,\vecu,\tM]$:
\begin{align}\label{SLDRSPLITHAAR}
d\mu(M)=\zeta(d)^{-1} \,d\mu^{(d-1)}(\tM)\,
d\vecu \, d\vecv\,\frac{da_1}{a_1^{d+1}},
\end{align}
where $d\vecv$ is the $(d-1)$-dimensional volume measure on $\S_1^{d-1}$.
Note that all of the above claims are valid also for $d=2$,
with the natural interpretation that $\Si_1=\SL(1,\R)=\{1\}$ 
with $\mu^{(1)}(\{1\})=1$.
We will also need to know the explicit expression of 
the lattice $\Z^dM$ in terms of $a_1,\vecv,\vecu,\tM$:
One computes that, for any $\vecm\in\Z^{d-1}$ and $n\in\Z$,
\begin{align}\label{LATTICEINPARAM}
(n,\vecm)M
=na_1\vecv+a_1^{-\frac1{d-1}}\bigl(0,n\vecu\aa(\ta)\tkk+\vecm\tM\bigr)
f(\vecv).
\end{align}
In particular we always have
\begin{align}\label{LATTICECONTAINEMENT}
\Z^dM\subset\bigsqcup_{n\in\Z} \bigl(na_1\vecv+\vecv^\perp\bigr).
\end{align}

Let us fix a subset $\S_\pm^{d-1}\subset\S_1^{d-1}\cap\{x_1\geq 0\}$
which contains exactly one of the vectors $\vecv$ and $-\vecv$
for every $\vecv\in S_1^{d-1}$.
Let us also fix a (set theoretical, measurable) fundamental region
$\F_{d-1}\subset\Si_{d-1}$ for $\Gamma^{(d-1)}\backslash G^{(d-1)}$.
Now for $A>1$ we set
\begin{align} \label{FGDEF}
\FG_A:=\bigl\{[a_1,\vecv,\vecu,\tM]\in G \col a_1>A,\:
\vecv\in\S_\pm^{d-1},\:
\vecu\in(-\sfrac12,\sfrac12]^{d-1},\:
\tM\in\F_{d-1}\bigr\}.
\end{align}
\begin{lem}\label{FDREGIONLONGVECTOR}
If $M,M'\in\FG_A$ satisfy $M'=\gamma M$ for some $\gamma\in\SL(d,\Z)$,
and if $a_2,a_2'<(c_\clowH^{(d-1)})^{-1}A$ in the Iwasawa
decompositions $M=\nn(u)\aa(a)\kk$, $M'=\nn(u')\aa(a')\kk'$, then $M=M'$.
\end{lem}
\begin{proof}
Assume that $M=\nn(u)\aa(a)\kk=[a_1,\vecv,\vecu,\tM]$ and 
$M'=\nn(u')\aa(a')\kk'=[a_1',\vecv',\vecu',\tM']$
satisfy the assumptions of the lemma.
Then $\Z^d M'=\Z^d\gamma M = \Z^d M$, and this lattice has a basis
$\vecb_1,\ldots,\vecb_d$ (cf.\ \eqref{BKDEF}), and also a basis
$\vecb_1',\ldots,\vecb_d'$, the row vectors of $M'$.
Now for each $\vecm=(m_1,\ldots,m_d)\in\Z^d$ with $m_1\neq 0$
the lattice vector $\vecm M$ has length $\geq a_1$,
since $(\vecm M)\cdot(\vece_1\kk)=m_1a_1$.
On the other hand, by a similar argument as in \eqref{BKINEQ},
using $\tM'\in\F_{d-1}\subset\Si_{d-1}$,
we have $||\vecb_j'||\leq{c_\clowH^{(d-1)}} a_2'<A<a_1$ for each $j\geq 2$;
thus $\vecb_j'\in \Z \vecb_2+\cdots+\Z \vecb_d$.
Similarly $\vecb_j\in \Z \vecb_2'+\cdots+\Z \vecb_d'$ for each $j\geq2$.
Hence $\Z \vecb_2'+\cdots+\Z \vecb_d'=\Z \vecb_2+\cdots+\Z \vecb_d$. 
Let $\Pi\subset \R^d$ be the hyperplane spanned by this set of vectors.
Now $a_1^{-1}=\vol\bigl(\Pi/(\Z \vecb_2+\cdots+\Z \vecb_d)\bigr)$
($(d-1)$-dimensional volume),
and similarly for $a_1'$; hence $a_1=a_1'$.
Also $\vecv=\vece_1 f(\vecv)=\vece_1 \kk\in \Pi^{\perp}$ and similarly
$\vecv'\in \Pi^{\perp}$; hence since %
$\vecv,\vecv'\in S^{d-1}_\pm$ we conclude $\vecv=\vecv'$.
Next by \eqref{LATTICEINPARAM},
if $\iota$ denotes the embedding $\iota:\R^{d-1}\ni (x_1,\ldots,x_{d-1})
\mapsto (0,x_1,\ldots,x_{d-1})\in\R^d$, then
\begin{align} \label{FACTONM1LATTICE}
\iota\bigl(\Z^{d-1}\tM\bigr)=
a_1^{\frac 1{d-1}}\iota(\Z^{d-1})Mf(\vecv)^{-1}=
a_1^{\frac 1{d-1}}(\Z \vecb_2+\cdots+\Z \vecb_d)f(\vecv)^{-1},
\end{align}
and similarly for $\Z^{d-1}\tM'$;
hence $\Z^{d-1}\tM=\Z^{d-1}\tM'$,
and since $\tM,\tM'\in\F_{d-1}$ we conclude $\tM=\tM'$.
Hence also $\tu=\tu'$, $\ta=\ta'$ and $\tkk=\tkk'$
(in an obvious notation, cf.\ \eqref{M1DEF}), and we now also
obtain $\aa(a)=\aa(a')$ and $\kk=\kk'$, so that
$M'=\gamma M$ implies $\nn(u')=\gamma \nn(u)$.
But $\nn(u')=\gamma \nn(u)$ together with
$\tu'=\tu$ and $\vecu,\vecu'\in (-\frac 12,\frac 12]^{d-1}$ imply 
$\gamma=I$, thus $M'=M$.
\end{proof}

\begin{lem}\label{FUNDLONGDOMAINLEM}
If $M=\nn(u)\aa(a)\kk\in\Si_d$ has $a_1>A$ and
$a_2<(c_\clowH^{(d-1)})^{-1}A$, then
there is some $\gamma\in\Gamma$ such that $\gamma M\in \FG_A\cap\Si_d$.
\end{lem}
\begin{proof}
Take any $M=\nn(u)\aa(a)\kk=[a_1,\vecv,\vecu,\tM]\in\Si_d$ 
with $a_1>A$ and $a_2<(c_\clowH^{(d-1)})^{-1}A$.
We write $\tM=\nn(\tu)\aa(\ta)\tkk$ as usual.
If $\vecv\notin\S^{d-1}_\pm$ %
then we replace $M$ with $\gamma D M$,
where $D=\text{diag}[-1,1,\ldots,1,-1]\in\Gamma$
and $\gamma\in\Gamma\cap N$ is chosen so that $\gamma D\nn(u)D\in \F_N$;
this new $M$ lies in $\Si_d$ and has the same $\aa(a)$ component as before 
but $\vecv$ negated.
Hence from now on we may assume $\vecv \in S^{d-1}_\pm$.

Take $\gamma_1\in\Gamma^{(d-1)}$ so that $\gamma_1 \tM\in \F_{d-1}$; 
let the Iwasawa decomposition of 
this matrix be $\gamma_1 \tM = \nn(\tu')\aa(\ta')\tkk'$,
and let $\vecw$ be the unique vector in $\Z^{d-1}$ with
$\vecw \nn(\tu') \in -\vecu \aa(\ta)\tkk {\tkk'}^{-1}\aa(\ta')^{-1}
+ (-\frac 12,\frac 12]^{d-1}$. Set
$\gamma=\matr 1{\vecw \gamma_1}{\trans\bn}{\gamma_1}\in\Gamma$.

We now claim $\gamma M\in\FG_A\cap\Si_d$. 
To prove this, first note that
$\gamma M$ has Iwasawa decomposition
\begin{align}\label{GAMMAMIW}
\gamma M = \begin{pmatrix} 1 & 
\vecw \nn(\tu')+\vecu \aa(\ta)\tkk {\tkk'}^{-1}\aa(\ta')^{-1}
\\ {\trans\bn} & \nn(\tu') 
\end{pmatrix}
\begin{pmatrix} a_1 & \bn \\ {\trans\bn} & a_1^{-\frac 1{d-1}} \aa(\ta')
\end{pmatrix}
\begin{pmatrix} 1 & \bn \\ {\trans\bn} & \tkk' \end{pmatrix}
f(\vecv).
\end{align}
From this we see by inspection that $\gamma M\in\FG_A$.
Next, for each $\vecm=(m_1,\ldots,m_{d-1})\in\Z^{d-1}$ with
$m_1\neq 0$ we have $||\vecm \gamma_1 \tM||\geq \ta_1'$,
since $\vecm\gamma_1\tM\cdot\vece_1\tkk'=m_1\ta_1'$.
Hence any basis for the lattice $\Z^{d-1} \tM=\Z^{d-1}\gamma_1 \tM$ 
must have at least one basis vector of length $\geq \ta_1'$.
But %
as in \eqref{BKINEQ}
we see that $\Z^{d-1}\tM$ has a basis where
each basis vector has length $\leq{c_\clowH^{(d-1)}}\ta_1$.
Hence 
$\ta_1'\leq{c_\clowH^{(d-1)}}\ta_1={c_\clowH^{(d-1)}}a_1^{\frac 1{d-1}}a_2
<Aa_1^{\frac 1{d-1}}<a_1^{\frac d{d-1}}$.
Using this fact together with $\nn(\tu')\aa(\ta')\tkk=
\gamma_1 \tM \in \F_{d-1}\subset\Si_{d-1}$ %
we see that $\gamma M\in\Si_d$ 
(cf.\ \eqref{SIDSUBSSIDM1} and \eqref{GAMMAMIW}).
\end{proof}
Let us define
\begin{align}\label{SIDPDEF}
\Si_d':=\Bigl\{[a_1,\vecv,\vecu,\tM]\in \Si_d\col\vecv\in\S_\pm^{d-1}\Bigr\}.
\end{align}
Then $\Si_d'$ contains a fundamental region for
$\Gamma\backslash G$ (viz.\ $\Gamma\Si_d'=G$), 
by the argument in the beginning of the proof of
Lemma \ref{FUNDLONGDOMAINLEM}.
Furthermore we have $\FG_A\cap\Si_d=\FG_A\cap\Si_d'$.
\begin{lem}\label{FDCONTAINMENTSLEM}
There exists a (set-theoretical, measurable) fundamental region
$\F_d$ for $\Gamma\backslash G$ which satisfies
$\F_d\subset\Si_d'$ and
\begin{align}\label{FDCONTAINMENTSLEMRES}
\FG_A\setminus\FC
\:\:\subset\:\: \bigl\{M\in \F_d\col a_1>A\bigr\}
\:\:\subset\:\:\FG_A\cup\FC,
\end{align}
where
\begin{align*}
\FC&:=\bigl\{M\in\Si_d'\cup\FG_A\col a_1>A,\:
a_2\geq(c_\clowH^{(d-1)})^{-1}A\bigr\}
\\
&\subset\Bigl\{[a_1,\vecv,\vecu,\tM]\in G\col
a_1>A,\: a_2\geq(c_\clowH^{(d-1)})^{-1}A,\:
\vecv\in\S_\pm^{d-1},\:
\vecu\in(-\sfrac12,\sfrac12]^{d-1},\:
\tM\in\Si_{d-1}\Bigr\}.
\end{align*}
\end{lem}
\begin{proof}
Let $\F^1=\{M\in \FG_A\col a_2<(c_\clowH^{(d-1)})^{-1}A\}$;
let $\F^2$ be an arbitrary measurable subset of $\FG_A\cap\Si_d'$
which contains exactly one representative
from each $\Gamma$-coset which intersects $\FG_A\cap\Si_d'$ but
does not intersect $\F^1$,
and then let $\F^3$ be an arbitrary measurable subset of $\Si_d'$
which contains exactly one representative
from each $\Gamma$-coset of $G$ which does not intersect $\F^1\cup\F^2$.
($\F^3$ exists since $\Gamma\Si_d'=G$.)
Finally set $\F_d=\F^1\cup\F^2\cup\F^3$.

We have $\F^1\subset\Si_d'$ since $a_2<(c_\clowH^{(d-1)})^{-1}A\leq A<a_1$
implies $a_2<\frac2{\sqrt3}a_1$.
Hence $\F_d\subset\Si_d'$.
By Lemma \ref{FDREGIONLONGVECTOR} any two distinct elements in $\F^1$
are $\Gamma$-inequivalent; hence any two distinct elements in $\F_d$
are $\Gamma$-inequivalent.
On the other hand we have $\Gamma\F_d=G$ by construction;
hence $\F_d$ is a fundamental region for $\Gamma\backslash G$.
Finally the first inclusion in \eqref{FDCONTAINMENTSLEMRES} holds
since $\FG_A\setminus\FC\subset\F^1$ by construction, 
and the second inclusion holds since
$\F^1\cup\F^2\subset\FG_A$ by construction and
$\{M\in\F^3\col a_1>A\}\subset\FC$ by Lemma \ref{FUNDLONGDOMAINLEM}.
\end{proof}

\subsection{\texorpdfstring{An asymptotic formula for $\int_\xi^\infty\Phi(\eta)\,d\eta$}{An asymptotic formula for integral of Phi(eta)}}
\label{INTASYMPTSEC}

In this section we give a short proof of the following asymptotic formula.
\begin{thm}\label{PHIXILARGETHMINT}
For any $d\geq2$, %
\begin{align}\label{PHIXILARGETHMRES2}
\int_\xi^\infty\Phi(\eta)\,d\eta=
\frac{\pi^{\frac{d-1}2}}{2^{d}d\, \Gamma(\frac {d+3}2)\,\zeta(d)} \xi^{-1}
+O\bigl(\xi^{-1-\frac 2d}\bigr)
\qquad \text{as } \: \xi\to\infty.
\end{align}
\end{thm}
Note that 
$\int_\xi^\infty\Phi(\eta)\,d\eta$ gives the limit probability that
a particle in the Lorentz gas starting from a generic initial point
inside the billiard domain
travels length $\geq\xi$ (in macroscopic coordinates)
before its first collision.
In particular already Theorem~\ref{PHIXILARGETHMINT} sharpens 
the upper bound given by Bourgain, Golse and Wennberg \cite{Bourgain98} 
and the lower bound of Golse and Wennberg \cite{Golse00}.
Note that we will later prove Theorem~\ref{PHIXILARGETHM} which gives
an asymptotic formula for $\Phi(\xi)$ itself and which immediately
implies Theorem~\ref{PHIXILARGETHMINT}.
However the following short direct proof of Theorem~\ref{PHIXILARGETHMINT}
gives a first illustration of 
the usage of Lemma \ref{FDCONTAINMENTSLEM} in a simple case.
Along the way we will prove some auxiliary results which we will need 
later anyway.

By definition (cf.\ \cite[(3.8), (4.3)]{partI}),
\begin{align}\label{PHIIRRASYMPTLARGE2PF1}
\int_{\xi}^\infty\Phi(\eta)\,d\eta=
\int_{\Gamma\backslash G}\int_{[0,1)^dM}I\Bigl((\Z^dM+\vecxi)\cap\fZ=\emptyset\Bigr)
\,d\vecxi\,d\mu(M),
\end{align}
where $\fZ=\fZ(0,1,\xi^{\frac1{d-1}})$;
however since the right hand side of \eqref{PHIIRRASYMPTLARGE2PF1}
is invariant under $\fZ\mapsto \fZ T$
where $T$ is an arbitrary volume preserving affine linear map,
we may just as well take
\begin{align}\label{PHIXILARGETHMINTPFFZCHOICE}
\fZ=\xi^{\frac1d}\fZ',\qquad\text{where }\:\fZ':=\fZ(-\sfrac12,\sfrac12,1),
\end{align}
in \eqref{PHIIRRASYMPTLARGE2PF1}.
Then $\fZ$ contains a ball of radius $\frac12\xi^{\frac1d}$
and hence by Lemma \ref{A1LARGELEM} there is a constant $0<c_\clowA<1$ 
which only depends on $d$ such that 
$a_1>A:=c_\clowA\xi^{\frac1d}$ holds for every
$M=\nn(u)\aa(a)\kk\in\Si_d$ for which 
the inner integral in \eqref{PHIIRRASYMPTLARGE2PF1} is nonzero.
From now on we keep $\xi$ so large that $A>1$.
Fix $\F_d\subset\Si_d'$ to be a fundamental region for 
$\Gamma\backslash G$ as in Lemma \ref{FDCONTAINMENTSLEM},
applied with our $A={c_\clowA}\xi^{\frac1d}$.
Now Lemma \ref{FDCONTAINMENTSLEM} together with %
\eqref{SLDRSPLITHAAR} and $a_2=a_1^{-\frac1{d-1}}\ta_1$ imply
\begin{align}\notag
\int_\xi^\infty\Phi(\eta)\,d\eta
=\int_{\FG_A} \int_{[0,1)^dM}
I\Bigl((\Z^dM+\vecxi)\cap\fZ=\emptyset\Bigr)\,d\vecxi\,d\mu(M)
\hspace{100pt}
\\\label{PHIINTFIRSTASYMPT}
+O\biggl(\int_A^\infty\mu^{(d-1)}\Bigl(\Bigl\{\tM\in\Si_{d-1}\col
\ta_1\geq (c_\clowH^{(d-1)})^{-1}Aa_1^{\frac1{d-1}}\Bigr\}\Bigr)
\,\frac{da_1}{a_1^{d+1}}\biggr).
\end{align}
If $d\geq3$ then using \eqref{SIDSUBSSIDM1} and \eqref{SLDRSPLITHAAR} 
for $d-1$ we obtain
\begin{align}\label{A1LARGEMEASURE}
\mu^{(d-1)}\bigl(\bigl\{\tM\in\Si_{d-1}\col\ta_1>T\bigr\}\bigr)
\ll\int_T^\infty \ta_1^{-d}\,d\ta_1\ll T^{1-d}
\end{align}
uniformly over all $T>0$.
Hence the error term in \eqref{PHIINTFIRSTASYMPT} is
$\ll A^{-2d}\ll\xi^{-2}$.
On the other hand if $d=2$ then the error term vanishes for all
sufficiently large $\xi$.

The inner integral in the main term in \eqref{PHIINTFIRSTASYMPT}
remains the same if $[0,1)^dM$ is replaced 
by any other fundamental region for $\R^d/\Z^dM$.
Since $\Z^dM$ is spanned by the vectors $\vecb_1,\ldots,\vecb_d$
(cf.\ \eqref{BKDEF}), one choice of a fundamental region for $\R^d/\Z^dM$ is
$\bigl(\prod_{j=1}^d [-\frac{a_j}2,\frac{a_j}2)\bigr)\kk$.
Hence, using \eqref{FGDEF} and \eqref{SLDRSPLITHAAR}, we obtain
\begin{align}\notag
\int_\xi^\infty\Phi(\eta)\,d\eta=\frac 1{\zeta(d)} 
\int_A^\infty \int_{\HS} 
\int_{(-\frac 12,\frac 12]^{d-1}}
\int_{\F_{d-1}} 
\int_{\prod_{j=1}^d [-\frac{a_j}2,\frac{a_j}2)}
I\Bigl((\Z^d[a_1,\vecv,\vecu,\tM]+\vecw \kk)\cap \fZ=\emptyset \Bigr) 
\\\label{PHIIRRASYMPTLARGE2PF3}
\times d\vecw %
\,  d\mu(\tM)\,d\vecu\,d\vecv\,\frac{da_1 }{a_1^{d+1}}
+O(\xi^{-2}),
\end{align}
where $\HS:=\{\vecv\in\S_1^{d-1}\col v_1>0\}$ \label{HSDEF}
and $\kk=\smatr 1\bn{\trans\bn}{\tkk} f(\vecv)$ (cf.\ \eqref{NAKSPLIT}).

Now for any $a_1$, $\vecv$, $\vecu$, $\tM$ and $\vecw=(w_1,\ldots,w_d)$
appearing in the above integral we have,
by \eqref{LATTICECONTAINEMENT} and since $\vece_1\kk=\vecv$,
\begin{align} \label{LATTICETRANSLCONT}
\Z^d[a_1,\vecv,\vecu,\tM]+\vecw\kk
\subset \bigcup_{n\in\Z} \bigl((na_1+w_1)\vecv+\vecv^\perp\bigr).
\end{align}
\begin{lem}\label{CYLPLANELEMMA}
Let $\vecv\in\HS\setminus\{\vece_1\}$, so that the angle $\varpi$ between
$\vecv$ and $\vece_1$ satisfies $0<\varpi<\frac\pi2$.
Set $\fZ'=\fZ(-\frac12,\frac12,1)$ as in \eqref{PHIXILARGETHMINTPFFZCHOICE}.
Then for any $t\in\R$, the intersection
$\fZ'\cap(t\vecv+\vecv^\perp)$ is nonempty if and only if 
$|t|<\frac12\cos\varpi+\sin\varpi$,
and in this case we have
\begin{align}\notag
\vol_{d-1}\bigl(\fZ'\cap(t\vecv+\vecv^\perp)\bigr)
\asymp\min\Bigl(1,\frac{\frac12\cos\varpi+\sin\varpi-|t|}{\sin\varpi}\Bigr)^{\frac d2-1}
\min\Bigl(1,\frac{\frac12\cos\varpi+\sin\varpi-|t|}{\sin2\varpi}\Bigr)
\hspace{30pt}
\\ \label{CYLPLANEVOL}
\gg\bigl(\sfrac12\cos\varpi+\sin\varpi-|t|\bigr)^{\frac d2}.
\end{align}
\end{lem}
\begin{proof}
After a rotation in the variables $x_2,\ldots,x_d$
we may assume $\vecv=(v_1,v_2,0,\ldots,0)$ with
$v_1=\cos\varpi$ and $v_2=\sin\varpi$.
By symmetry we may also assume $t\geq0$.
Now $t\vecv+\vecv^\perp=\{\vecx\in\R^d\col v_1x_1+v_2x_2=t\}$,
and we see that if $\vecx=(x_1,\ldots,x_d)$ lies in the intersection
${\fZ'}\cap(t\vecv+\vecv^\perp)$ then so does
$(x_1,x_2,0,\ldots,0)$. %
In particular ${\fZ'}\cap(t\vecv+\vecv^\perp)$ is nonempty if and only if
$t<\frac12v_1+v_2$, thus proving the first assertion.

Now assume $0\leq{t}<\frac12v_1+v_2$.
Let $p:\R^d\to\R^{d-1}$ be the projection
$(x_1,\ldots,x_d)\mapsto(x_2,\ldots,x_d)$, and note that
\begin{align*}
p\bigl(\fZ'\cap(t\vecv+\vecv^\perp)\bigr)
=\scrB_1^{d-1}\cap\bigl((A,B)\times\R^{d-2}\bigr) %
\qquad
\text{with }\:
\begin{cases}
A=\max(-1,(t-\frac12v_1)/v_2)
\\
B=\min(1,(t+\frac12v_1)/v_2).
\end{cases}
\end{align*}
Since $B>0$, it follows that
\begin{align*}
\vol_{d-1}\bigl(p\bigl(\fZ'\cap(t\vecv+\vecv^\perp)\bigr)\bigr)
\asymp(1-A)^{\frac d2-1}(B-A).
\end{align*}
But note that $p_{|t\vecv+\vecv^\perp}$ scales volume with a factor
$v_1$; hence
\begin{align*}
\vol_{d-1}\bigl(\fZ'\cap(t\vecv+\vecv^\perp)\bigr)
\asymp(1-A)^{\frac d2-1}\frac{B-A}{v_1}.
\end{align*}
Here $1-A=\min(2,(\frac12v_1+v_2-t)/v_2)$ and
\begin{align}\label{CYLPLANELEMMAPF1}
\frac{B-A}{v_1}=\min\Bigl(\frac{1}{v_2},\frac{\frac12v_1+v_2-t}{v_1v_2},
\frac{\frac12v_1+v_2+t}{v_1v_2},\frac2{v_1}\Bigr).
\end{align}
Here the third entry is redundant in the minimum, since
$\frac12v_1+v_2+t\geq \frac12v_1+v_2-t$.
Also note that the minimum equals its first entry if and only if
$\frac12v_1\leq v_2-t$ and in this case we have
$\frac12v_1\leq v_2$ and thus $v_2\asymp1$ 
(since $v_1=\cos\varpi$, $v_2=\sin\varpi$, $\varpi\in(0,\frac\pi2)$).
Furthermore the minimum equals its fourth entry if and only if
$v_2\leq\frac12v_1-t$ and in this case we necessarily have
$v_1\asymp1$.
Hence $\frac{B-A}{v_1}\asymp\min(1,\frac{\frac12v_1+v_2-t}{v_1v_2})$
and this concludes the proof of the first relation in
\eqref{CYLPLANEVOL}.
The second relation in \eqref{CYLPLANEVOL} is obvious.
\end{proof}

By Lemma \ref{CYLPLANELEMMA}, the cylinder
$\fZ=\xi^{\frac1d}\fZ'$ has nonempty intersection with the hyperplane
$(na_1+w_1)\vecv+\vecv^\perp$ if and only if
$|na_1+w_1|<\xi^{\frac1d}(\frac12\cos\varpi+\sin\varpi)$.
It follows that 
if $\xi^{\frac1d}(\frac12\cos\varpi+\sin\varpi)\leq|w_1|<\frac12a_1$
then none of the
hyperplanes in \eqref{LATTICETRANSLCONT} intersect $\fZ$, 
and thus $(\Z^d[a_1,\vecv,\vecu,\tM]+\vecw\kk)\cap\fZ=\emptyset$, 
independently of $\vecu,$ $\tM$ or $w_2,\ldots,w_d$.
Hence, if we restrict the range of integration in
\eqref{PHIIRRASYMPTLARGE2PF3} further by 
$|w_1|\geq\xi^{\frac1d}(\frac12\cos\varpi+\sin\varpi)$ 
then the resulting integral equals %
\begin{align}
\frac 1{\zeta(d)} \int_A^\infty \int_{\HS} 
\max(0,1-\xi^{\frac1d}(\cos\varpi+2\sin\varpi)a_1^{-1})
\,d\vecv\,\frac{da_1}{a_1^{d+1}}
\hspace{70pt}
\\\notag
=\frac{2\pi^{\frac{d-1}2}}{\zeta(d)\Gamma(\frac{d-1}2)}
\frac{\xi^{-1}}{d(d+1)}
\int_0^{\pi/2} (\cos\varpi+2\sin\varpi)^{-d}(\sin\varpi)^{d-2}\,d\varpi.
\end{align}
Here in the last step we used the fact that 
$\xi^{\frac1d}(\cos\varpi+2\sin\varpi)\geq\xi^{\frac1d}>A$
for all $\varpi\in(0,\frac\pi2)$.
Substituting $x=\cot\varpi$ the integral is seen to equal
$2^{1-d}(d-1)^{-1}$.
Hence we obtain the main term in Theorem \ref{PHIXILARGETHMINT},
and to complete the proof of Theorem \ref{PHIXILARGETHMINT} 
we now only have to prove that, 
if we denote by $\scrI_{rem}$ the remaining integral,
viz.\ the integral in \eqref{PHIIRRASYMPTLARGE2PF3} with range of integration
further restricted by $|w_1|<\xi^{\frac1d}(\frac12\cos\varpi+\sin\varpi)$, 
then
\begin{align}\label{PHIIRRASYMPTLARGE2PF4}
\scrI_{rem}=O(\xi^{-1-\frac2d}),\qquad\text{as }\:\xi\to\infty.
\end{align}

To prove this, for any $a_1,\vecv,\vecu,\tM,\vecw$ appearing in the 
integral \eqref{PHIIRRASYMPTLARGE2PF3}, we
write $\vecw=(w_1,\vecw')$ with
$\vecw'\in\R^{d-1}$;
then $\vecw\kk=w_1\vecv+(0,\vecw'\tkk)f(\vecv)$,
and hence using \eqref{LATTICEINPARAM} we see that 
\begin{align*}
\Z^d[a_1,\vecv,\vecu,\tM]+\vecw\kk\supset
a_1^{-\frac1{d-1}}\iota\Bigl(\Z^{d-1}\tM+a_1^{\frac1{d-1}}\vecw'\tkk\Bigr)
f(\vecv)+w_1\vecv,
\end{align*}
where (here and throughout the rest of the paper) $\iota$ denotes the embedding
\begin{align}\label{IOTADEF}
\iota:\R^{d-1}\ni (x_1,\ldots,x_{d-1})
\mapsto (0,x_1,\ldots,x_{d-1})\in\R^d.
\end{align}
It follows that
\begin{align}\notag
\scrI_{rem}\ll &\int_A^\infty\int_{\HS} 
\int_{|w_1|<\xi^{\frac1d}(\frac12\cos\varpi+\sin\varpi)}
\int_{\F_{d-1}}\int_{\vecw'\in\prod_{j=2}^d[-\frac{a_j}2,\frac{a_j}2)}
\\\label{PHIIRRASYMPTLARGE2PF5}
&I\Bigl(
(\Z^{d-1}\tM+a_1^{\frac1{d-1}}\vecw'\tkk)
\cap
a_1^{\frac1{d-1}}\iota^{-1}\bigl((\fZ-w_1\vecv)f(\vecv)^{-1}\bigr)=\emptyset
\Bigr)\, d\vecw'\,d\mu(\tM)\,dw_1\,d\vecv\,\frac{da_1}{a_1^{d+1}}.
\end{align}
Here note that $\iota^{-1}((\fZ-w_1\vecv)f(\vecv)^{-1})$
is isometric with
$\fZ\cap(w_1\vecv+\vecv^\perp)$,
and for any $\vecv\in\HS$ with $0<\varpi<\frac\pi2$ and
any $w_1$ with $|w_1|<\xi^{\frac1d}(\frac12\cos\varpi+\sin\varpi)$,
the set $\fZ\cap(w_1\vecv+\vecv^\perp)$ is convex and has
$(d-1)$-dimensional volume
$\gg\xi^{\frac{d-1}d}
(\sfrac12\cos\varpi+\sin\varpi-\xi^{-\frac1d}|w_1|)^{\frac d2}$,
by Lemma~\ref{CYLPLANELEMMA}.
Hence for any $a_1>A$ we have
\begin{align*}
\vol_{d-1}\Bigl(a_1^{\frac1{d-1}}\iota^{-1}\bigl((\fZ-w_1\vecv)f(\vecv)^{-1}
\bigr)\Bigr)
\gg\xi(\sfrac12\cos\varpi+\sin\varpi-\xi^{-\frac1d}|w_1|)^{\frac d2}.
\end{align*}
Note also that 
$a_1^{\frac1{d-1}}\bigl(\prod_{j=2}^d[-\frac{a_j}2,\frac{a_j}2)\bigr)\tkk$
is a fundamental domain for $\R^{d-1}/\Z^{d-1}\tM$ for any given
$a_1$ and $\tM$ (cf.\ \eqref{NAKSPLIT}).

\begin{lem}\label{ASLAMBOUNDLEM}
For any $d\geq2$ and any convex subset $\fC\subset\R^d$, we have
\begin{align}\label{ASLAMBOUNDLEMRES}
\int_{\F_d}\int_{\R^d/\Z^dM}I\bigl((\vecx+\Z^dM)\cap\fC=\emptyset\bigr)
\,d\vecx\,d\mu(M)\ll\vol(\fC)^{-1}.
\end{align}
\end{lem}
\begin{proof}
This is a straightforward %
modification of \cite[Lemma 2.5]{lprob}.
In fact %
the bound holds for any 
\textit{measurable} $\fC\subset\R^d$,
as can be proved by using \cite[Thm.\ 2.2]{athreyamargulis},
after first rewriting the left hand side of
\eqref{ASLAMBOUNDLEMRES} as
$\lim_{R\to\infty}\vol(\scrB_R^d)^{-1}
\int_{\scrB_R^d}\int_{\F_d}
I(\Z^dM\cap(\fC-\vecx)=\emptyset)\,d\mu(M)\,d\vecx$.
\end{proof}

If $d\geq3$ then by Lemma \ref{ASLAMBOUNDLEM} (applied with $d-1$) and the
discussion preceding it, we obtain
\begin{align*}
\scrI_{rem} &\ll \int_A^\infty\int_{\HS} 
\int_{|w_1|<\xi^{\frac1d}(\frac12\cos\varpi+\sin\varpi)}
a_1^{-1}\min\Bigl\{1,
\xi^{-1}(\sfrac12\cos\varpi+\sin\varpi-\xi^{-\frac1d}|w_1|)^{-\frac d2}
\Bigr\}\,
dw_1\,d\vecv\,\frac{da_1}{a_1^{d+1}}
\\
&\ll \xi^{-1}\int_0^2\min\Bigl\{1,\xi^{-1}t^{-\frac d2}\Bigr\}\,dt
\ll\xi^{-1-\frac2d},
\end{align*}
i.e.\ we have proved \eqref{PHIIRRASYMPTLARGE2PF4}.
In the remaining case $d=2$ we have
$\F_{d-1}=\{(1)\}=\SL(1,\R)$ 
and hence the integrand in \eqref{PHIIRRASYMPTLARGE2PF5}
vanishes whenever
$\vol_1(a_1\iota^{-1}(\fZ-w_1\vecv)f(\vecv)^{-1})>1$;
hence we obtain $\scrI_{rem}\ll\xi^{-2}$.
This completes the proof of Theorem \ref{PHIXILARGETHMINT}.
\hfill $\square$ $\square$

\subsection{\texorpdfstring{Asymptotics for $\Phi(\xi,\vecw)$: Simplifying the integral}{Asymptotics for Phi(xi,w): Simplifying the integral}}\label{SIMPLINTSEC}
We now start with the proof of the asymptotic formula for $\Phi(\xi,w)$,
Theorem \ref{PHIXIWASYMPTTHM}. 
Note that if $d=2$ then Theorem \ref{PHIXIWASYMPTTHM}
(with $F_2(t)=\frac{3}{2\pi^2}((1-t)_+)^2$)
follows directly from the explicit formula in \cite{partIII}.
Hence we will from now on assume $d\geq3$.

From \eqref{PHIXIWDEF} we have, using the right $G$-invariance of $\mu$,
\begin{align}\label{PHIXIZTHMPF1}
\Phi(\xi,w)
&=\int_{\Gamma\backslash G}
I\bigl(\Z^dM\cap\fZ=\emptyset\bigr)\,d\mu(M),
\end{align}
where 
\begin{align}\label{FZCHOICE}
\fZ:=\xi^{\frac1d}(\fZ(0,1,1)+w\vece_2).
\end{align}
Note in particular that $\fZ$ does not denote the same cylinder as
in the above proof of Theorem~\ref{PHIXILARGETHMINT}.
The choice of $\fZ$ in \eqref{FZCHOICE} will be in force for the rest of 
Section \ref{PRELIMINARIESSEC}.

It follows from Proposition \ref{PHI0SUPPORTTHM}
and \eqref{PHIFROMPHIZERO}
(or more directly from \cite[Cor.\ 1.4]{lprob})
that there is a constant $c_\clowD>0$ which only depends on $d$ such that
$\Phi(\xi,w)=0$ whenever $1-w\geq c_\clowD\xi^{-\frac2d}$.
The function $F_d(t)$ appearing in the right hand side of 
\eqref{PHIXIWASYMPTTHMRES2} will be defined in
\eqref{PHIXIWASYMPTTHMRES} below, and it will be clear from this definition
that $F_d(t)$ vanishes for large $t$ (cf.\ Lemma \ref{XISUPBOUNDLEM}
or \eqref{FDTCOMPSUPP});
hence we may assume that $c_\clowD$ is so large that 
$F_d(t)=0$ for all $t\geq c_\clowD$.
This means that \eqref{PHIXIWASYMPTTHMRES2} is automatic
when $1-w\geq c_\clowD\xi^{-\frac2d}$.
Hence from now on we will %
assume $1-w<c_\clowD\xi^{-\frac2d}$.
By Lemma~\ref{A1LARGELEM} there is a constant $0<c_\clowE<1$ 
which only depends on 
$d$ such that $a_1>A:=c_\clowE\xi^{\frac1d}$ holds for all
$M\in\Si_d$ with $\Z^dM\cap\fZ=\emptyset$.
We will assume from start that 
$\xi>\max(2,(10c_\clowD)^{d/2},c_\clowE^{-d})$;
in particular we have $\frac9{10}<w<1$ and $A>1$.

Applying Lemma \ref{FDCONTAINMENTSLEM} in the same way as in
\eqref{PHIINTFIRSTASYMPT} we get
\begin{align}\notag
\Phi(\xi,w)=\frac 1{\zeta(d)} 
\int_{\HS}\int_A^\infty\int_{(-\frac 12,\frac 12]^{d-1}}
\int_{\F_{d-1}} I\bigl(\Z^d[a_1,\vecv,\vecu,\tM]\cap\fZ
=\emptyset\bigr)\, d\mu^{(d-1)}(\tM)
\,d\vecu\,\frac{da_1}{a_1^{d+1}}\,d\vecv
\\\label{PHI2XIZFIRSTASYMPT}
+O\bigl(\xi^{-2}\bigr),
\end{align}
where (again) $\HS:=\{\vecv\in\S_1^{d-1}\col v_1>0\}$.

We parametrize a dense open subset of $\HS$ as follows
(recall that we are assuming $d\geq3$):
\begin{align}\notag
\vecv&=(v_1,\ldots,v_d)
\\\label{VPARA}
&=\bigl(\cos\varpi,\sin\varpi\cos\omega,
(\sin\varpi\sin\omega)\alpha_1,(\sin\varpi\sin\omega)\alpha_2,
\ldots,(\sin\varpi\sin\omega)\alpha_{d-2}\bigr)
\in\S_1^{d-1},
\end{align}
where $\varpi\in(0,\frac\pi2)$, $\omega\in(0,\pi)$ and
$\vecalf=(\alpha_1,\ldots,\alpha_{d-2})\in\S_1^{d-3}$.
Thus $\varpi$ is the angle between $\vecv$ and $\vece_1$,
and $\omega$ is the angle between $\vecv':=(v_2,\ldots,v_d)$
and $\vece_1$ in $\R^{d-1}$.
The $(d-1)$-dimensional volume measure on $\S_1^{d-1}$ takes the following
form in our parametrization:
\begin{align}\label{DVINPARA}
d\vecv=(\sin\varpi)^{d-2}(\sin\omega)^{d-3}\,d\varpi\,d\omega\,d\vecalf,
\end{align}
where $d\vecalf$ is the $(d-3)$-dimensional volume measure on $\S_1^{d-3}$
(if $d=3$: $d\vecalf$ is the counting measure on $\S_1^0=\{-1,1\}$).

For any $\vecv$, $a_1$, $\vecu$ as above we have, using \eqref{LATTICEINPARAM}
with $n=0$ only:
\begin{align}\label{PHIXIZTHMPF2part}
\int_{\F_{d-1}} I\bigl(\Z^d[a_1,\vecv,\vecu,\tM]\cap\fZ
=\emptyset\bigr)\, d\mu(\tM)
\leq\int_{\F_{d-1}} I\bigl(\Z^{d-1}\tM\cap a_1^{\frac1{d-1}}\fZ_\vecv=
\emptyset\bigr)\, d\mu(\tM),
\end{align}
where
\begin{align}\label{ZVDEF}
\fZ_\vecv:=\iota^{-1}(\fZ f(\vecv)^{-1})\subset\R^{d-1}
\end{align}
(recall the definition of $\iota$ in \eqref{IOTADEF}).
Note that $\fZ_\vecv$ is isometric with $\fZ\cap\vecv^\perp$.
Now $\fZ_\vecv$ contains a $(d-1)$-dimensional relatively open cone of volume
$\gg\xi^{\frac{d-1}d}\omega^2(\sin\omega)^{d-2}$ with $\bn$ in its base
(cf.\ \cite[Lemma 7.1]{lprob}).
Hence by \cite[Cor.\ 1.4]{lprob},
\begin{align}\label{PHIXIZTHMPF2part2}
\int_{\F_{d-1}} I\bigl(\Z^{d-1}\tM\cap a_1^{\frac1{d-1}}\fZ_\vecv=
\emptyset\bigr)\, d\mu(\tM)
\ll\min\Bigl(1,\Bigl(a_1\xi^{\frac{d-1}d}\omega^2(\sin\omega)^{d-2}
\Bigr)^{-\frac{2(d-2)}{d-1}}\Bigr).
\end{align}

Let $0<c_\clowC<1$ be an arbitrary constant. 
(We will later impose some conditions on $c_\clowC$ being sufficiently
small, but it will be clear that it is possible to fix 
$c_\clowC$ as an absolute constant satisfying these conditions.)
Now in the first line of
\eqref{PHI2XIZFIRSTASYMPT} we may restrict the range of $\vecv$ to
\begin{align}\label{SPDEF1}
S':=\bigl\{\vecv\in\HS\col0<\omega<c_\clowC\bigr\},
\end{align}
at the cost of an error which,
by \eqref{PHIXIZTHMPF2part} and \eqref{PHIXIZTHMPF2part2}, is
\begin{align}\notag
\ll\int_{c_\clowC}^\pi\int_A^\infty\min\Bigl(1,
\Bigl(a_1\xi^{\frac{d-1}d}\omega^2(\sin\omega)^{d-2}
\Bigr)^{-\frac{2(d-2)}{d-1}}\Bigr)\,\frac{da_1}{a_1^{d+1}}
\,(\sin\omega)^{d-3}\,d\omega
\hspace{90pt}
\\\label{COMP1}
\ll\xi^{-1}\int_0^{\pi-{c_\clowC}}\min\Bigl(1,\bigl(\xi\tau^{d-2}
\bigr)^{-\frac{2(d-2)}{d-1}}\Bigr)\,\tau^{d-3}\,d\tau
\hspace{50pt}
\bigl\{\text{we substituted $\tau=\pi-\omega$}\bigr\}
\\\notag
\ll\xi^{-1}\int_0^{\xi^{-\frac1{d-2}}}\tau^{d-3}\,d\tau
+\xi^{-1}\int_{\xi^{-\frac1{d-2}}}^4
\xi^{-\frac{2(d-2)}{d-1}}\tau^{-1-\frac{(d-2)(d-3)}{d-1}}\,d\tau
\ll E,
\end{align}
where we denote
\begin{align}\label{EFIRSTDEF}
E:=\begin{cases}\xi^{-2}\log\xi&\text{if }\: d=3
\\
\xi^{-2}&\text{if }\: d\geq4.
\end{cases}
\end{align}
Our goal in Sec.\ \ref{SIMPLINTSEC}--\ref{FURTHERSIMPLSEC} is to prove that
\eqref{PHIXIWASYMPTTHMRES2} holds with the error term replaced by
$O(E)$; then in Sec.\ \ref{SLIGHTIMPRSEC} we will show how to improve
the error term slightly for $d=3$ so as to complete the proof of
Theorem \ref{PHIXIWASYMPTTHM}.

Collecting our bounds so far we have, 
writing $M=[a_1,\vecv,\vecu,\tM]$,
\begin{align}\label{PHIXIZTHMPF2a}
\Phi(\xi,w)=\frac 1{\zeta(d)} 
\int_{S'}\int_A^\infty\int_{(-\frac 12,\frac 12]^{d-1}}
\int_{\F_{d-1}} I\bigl(\Z^dM\cap\fZ
=\emptyset\bigr)\, d\mu(\tM)
\,d\vecu\,\frac{da_1}{a_1^{d+1}}\,d\vecv
+O(E).
\end{align}

\begin{lem}\label{PHIXIZTHMPFLEM1}
For any $a_1>0$ and any $\vecv$ as in \eqref{VPARA} with 
$\varpi\in(0,\frac\pi2)$ and $\omega\in[0,\frac\pi2]$, we have
$(na_1\vecv+\vecv^\perp)\cap\fZ=\emptyset$ for all $n\in\Z\setminus\{0\}$
if and only if $a_1\geq\xi^{\frac1d}\bigl(v_1+wv_2+\|\vecv'\|\bigr)$.
\end{lem}
(Recall that $\vecv':=(v_2,\ldots,v_d)$, thus $\|\vecv'\|=\sin\varpi$.)
\begin{proof}
We have $\fZ=\xi^{\frac1d}(\fZ'+\frac12\vece_1+w\vece_2)$
where $\fZ'=\fZ(-\frac12,\frac12,1)$. %
Hence $(na_1\vecv+\vecv^\perp)\cap\fZ\neq\emptyset$ if and only if
$\fZ'=\fZ(-\frac12,\frac12,1)$ has nonempty intersection with
\begin{align*}
\xi^{-\frac1d}na_1\vecv-\sfrac12\vece_1-w\vece_2+\vecv^\perp
=\bigl(\xi^{-\frac1d}na_1-(\sfrac12v_1+wv_2)\bigr)\vecv+\vecv^\perp.
\end{align*}
By Lemma \ref{CYLPLANELEMMA} 
this holds if and only if 
$\bigl|\xi^{-\frac1d}na_1-(\frac12v_1+wv_2)\bigr|<\frac12v_1+\|\vecv'\|$.
If this inequality holds for some nonzero integer $n$ then it must hold
for some positive integer $n$, since $\frac12v_1+wv_2>0$ by our assumptions.
But for $n>0$ we \textit{always} have
$\xi^{-\frac1d}na_1-(\frac12v_1+wv_2)>-\frac12v_1-\|\vecv'\|$.
Hence we conclude that 
$(\vecv^\perp+na_1\vecv)\cap\fZ=\emptyset$ holds for all $n\in\Z\setminus\{0\}$
if and only if $\xi^{-\frac1d}na_1-(\frac12v_1+wv_2)\geq\frac12v_1+\|\vecv'\|$
for all $n\in\Z_{>0}$,
viz.\ if and only if $a_1\geq\xi^{\frac1d}(v_1+wv_2+\|\vecv'\|)$.
\end{proof}

\begin{lem}\label{CYLPLANELEMMACOR}
For any  $\vecv$ as in \eqref{VPARA} with 
$\varpi\in(0,\frac\pi2)$ and $\omega\in[0,\frac\pi2]$, and any $a_1$ with
$A={c_\clowE}\xi^{\frac1d}<a_1<\xi^{\frac1d}(v_1+wv_2+\|\vecv'\|)$, we have
(writing $M=[a_1,\vecv,\vecu,\tM]$)
\begin{align}\label{PHIXIZTHMPF4}
\int_{(-\frac 12,\frac 12]^{d-1}}
\int_{\F_{d-1}} I\bigl(\Z^dM\cap\fZ=\emptyset\bigr)
\, d\mu(\tM)\,d\vecu
\ll\xi^{-1}\bigl(v_1+wv_2+\|\vecv'\|-\xi^{-\frac1d}a_1\bigr)^{-\frac d2}.
\end{align}
\end{lem}
\begin{proof}
We use \eqref{LATTICEINPARAM} with $n=1$ 
and the fact that for any $\tM\in G^{(d-1)}$,
$(-\frac 12,\frac 12]^{d-1}\aa(\ta)\tkk$
is a fundamental domain for $\R^{d-1}/\Z^{d-1}\tM$.
Writing
\begin{align*}
\fZ_{a_1,\vecv}=
\iota^{-1}\bigl((\fZ-a_1\vecv)f(\vecv)^{-1}\bigr)
\end{align*}
(a convex subset of $\R^{d-1}$)
we then obtain 
that the left hand side of \eqref{PHIXIZTHMPF4} is bounded from above by
\begin{align}
\int_{\F_{d-1}}\int_{\R^{d-1}/\Z^{d-1}\tM}
I\bigl((\vecx+\Z^{d-1}\tM)\cap a_1^{\frac1{d-1}}\fZ_{a_1,\vecv}
=\emptyset\bigr)\,d\vecx\,d\mu(\tM)
\ll a_1^{-1}\vol_{d-1}\bigl(\fZ_{a_1,\vecv}\bigr)^{-1}.
\end{align}
Here the last inequality follows from Lemma \ref{ASLAMBOUNDLEM}.
Now note that $\fZ_{a_1,\vecv}$ is isometric with
$(a_1\vecv+\vecv^\perp)\cap\fZ$, and hence
applying Lemma \ref{CYLPLANELEMMA} 
(cf.\ the proof of Lemma \ref{PHIXIZTHMPFLEM1}) we get
\begin{align*}
\vol_{d-1}(\fZ_{a_1,\vecv})\gg
\xi^{\frac{d-1}d}
\bigl(\sfrac12v_1+\|\vecv'\|
-\bigl|\xi^{-\frac1d}a_1-(\sfrac12v_1+wv_2)\bigr|\bigr)^{\frac d2}
\hspace{50pt}
\\
\gg\xi^{\frac{d-1}d}(v_1+wv_2+\|\vecv'\|-\xi^{-\frac1d}a_1)^{\frac d2}.
\end{align*}
Here the last relation is obvious if $\xi^{-\frac1d}a_1\geq\frac12v_1+wv_2$,
while if $\xi^{-\frac1d}a_1<\frac12v_1+wv_2$ then it follows from
\begin{align*}
\sfrac12v_1+\|\vecv'\|-\bigl|\xi^{-\frac1d}a_1-(\sfrac12v_1+wv_2)\bigr|
=\|\vecv'\|-wv_2+\xi^{-\frac1d}a_1>\xi^{-\frac1d}a_1
\gg1. %
\end{align*}
We now obtain the stated bound. %
\end{proof}

Using \eqref{PHIXIZTHMPF2part}, \eqref{PHIXIZTHMPF2part2} 
and Lemma \ref{CYLPLANELEMMACOR}
it follows that the contribution from all $a_1$ with 
$a_1<\xi^{\frac1d}(v_1+wv_2+\|\vecv'\|)$ in \eqref{PHIXIZTHMPF2a} is,
writing $t=v_1+wv_2+\|\vecv'\|-\xi^{-\frac1d}a_1$
and using $c_\clowC<1$,
\begin{align}\notag
&\ll\xi^{-1}\int_0^1\int_0^4\min\Bigl\{1,\xi^{-1}t^{-\frac d2},
(\xi\omega^d)^{-\frac{2(d-2)}{d-1}}\Bigr\}\,dt\,\omega^{d-3}\,d\omega
\\\label{COMP2}
&\quad\leq\xi^{-1}\int_0^{\xi^{-1/d}}
\biggl(\int_0^{\xi^{-2/d}}\,dt
+\int_{\xi^{-2/d}}^4\xi^{-1}t^{-\frac d2}\,dt\biggr)\,\omega^{d-3}\,d\omega
\\\notag
&\qquad+\xi^{-1}\int_{\xi^{-1/d}}^1
\biggl(\int_0^{\xi^{\frac{2(d-3)}{d(d-1)}}\omega^{\frac{4(d-2)}{d-1}}}
(\xi\omega^d)^{-\frac{2(d-2)}{d-1}}\,dt
+\int_{\xi^{\frac{2(d-3)}{d(d-1)}}\omega^{\frac{4(d-2)}{d-1}}}^\infty
\xi^{-1}t^{-\frac d2}\,dt\biggr)
\,\omega^{d-3}\,d\omega\ll E,
\end{align}
where $E$ is as in \eqref{EFIRSTDEF}.
But for $a_1>\xi^{\frac1d}(v_1+wv_2+\|\vecv'\|)$ 
we have $\Z^d[a_1,\vecv,\vecu,\tM]\cap\fZ=\emptyset$
if and only if $\Z^{d-1}\tM\cap a_1^{\frac1{d-1}}\fZ_\vecv=\emptyset$,
by \eqref{LATTICEINPARAM}, \eqref{LATTICECONTAINEMENT} and
Lemma \ref{PHIXIZTHMPFLEM1}.
Hence %
\begin{align}\label{PHIXIZTHMPF7}
\Phi(\xi,w)=\frac 1{\zeta(d)} 
\int_{S'}\int_{\xi^{\frac1d}(v_1+wv_2+\|\vecv'\|)}^\infty\int_{\F_{d-1}} 
I\Bigl(\Z^{d-1}\tM\cap a_1^{\frac1{d-1}}\fZ_\vecv=\emptyset\Bigr)
\, d\mu(\tM)\,\frac{da_1}{a_1^{d+1}}\,d\vecv+O(E).
\end{align}

Next we wish replace the lower integration limit for $a_1$ by
$\xi^{\frac1d}(v_1+2v_2)$.
Note that for every $\vecv\in S'$ we have
\begin{align}
\bigl|wv_2+\|\vecv'\|-2v_2\bigr|
\leq1-w+\|\vecv'\|-v_2
=1-w+(\sin\varpi)(1-\cos\omega)
<1-w+\omega^2.
\end{align}
Hence using \eqref{PHIXIZTHMPF2part2} 
we see that the error when replacing 
``$\xi^{\frac1d}(v_1+wv_2+\|\vecv'\|)$''
with ``$\xi^{\frac1d}(v_1+2v_2)$'' in \eqref{PHIXIZTHMPF7} is
\begin{align}\label{PHIXIZTHMPF7a}
\ll\xi^{-1}\int_0^1\min\Bigl(1,(\xi\omega^d)^{-\frac{2(d-2)}{d-1}}\Bigr)
\bigl(1-w+\omega^2\bigr)(\sin\omega)^{d-3}\,d\omega
\ll(1-w)\xi^{-2+\frac2d} + E\ll E.
\end{align}
(Here we used our assumption $1-w<c_\clowD\xi^{-\frac2d}$.)
We thus conclude:
\begin{align}\label{PHIXIZTHMPF8}
\Phi(\xi,w)=\frac 1{\zeta(d)} 
\int_{S'}\int_{\xi^{\frac1d}(v_1+2v_2)}^\infty\int_{\F_{d-1}} 
I\Bigl(\Z^{d-1}\tM\cap a_1^{\frac1{d-1}}\fZ_\vecv=\emptyset\Bigr)
\, d\mu(\tM)\,\frac{da_1}{a_1^{d+1}}\,d\vecv+O(E).
\end{align}

Next, for $\vecv\in\HS$ we define $L_\vecv:\R^{d-1}\to\R^{d-1}$ by
\begin{align*}
L_\vecv(\vecx)=p(\iota(\vecx)f(\vecv)),
\end{align*}
where $p:\R^d\to\R^{d-1}$ is the projection
$(x_1,\ldots,x_d)\mapsto(x_2,\ldots,x_d)$.
Then $L_\vecv$ is a linear map with determinant $v_1>0$
(the same as the Jacobian of $p$ restricted to $\vecv^\perp$).
Thus $v_1^{-\frac1{d-1}}L_\vecv$ is a map in $G^{(d-1)}$,
and hence by the invariance of the Haar measure $\mu^{(d-1)}$
we have
\begin{align}\label{MUDM1INV}
\int_{\F_{d-1}} 
I\Bigl(\Z^{d-1}\tM\cap\fC=\emptyset\Bigr)
\, d\mu(\tM)
=\int_{\F_{d-1}} 
I\Bigl(\Z^{d-1}\tM\cap v_1^{-\frac1{d-1}}L_\vecv(\fC)
=\emptyset\Bigr)
\, d\mu(\tM)
\end{align}
for any measurable subset $\fC\subset\R^{d-1}$.
Note also that from the definition of the cylinder $\fZ$ it is clear that
$p(\fZ\cap\vecv^\perp)$ is a subset of the cut ball
$\xi^{\frac1d}\fC_{p(\vecv)}(w)$, where
\begin{align}\label{CHWDEF}
\fC_\vech(w)=\bigl\{\vecx\in\R^{d-1}\col\|\vecx-w\vece_1\|<1,\: 
\vech\cdot\vecx<0\bigr\}.
\end{align}
Hence
\begin{align}\label{LVZVSUBSET}
L_\vecv(\fZ_\vecv)\subset\xi^{\frac1d}\fC_{p(\vecv)}(w),
\end{align}
and using this together with \eqref{MUDM1INV} 
we get
\begin{align}\notag
&\int_{\F_{d-1}} 
I\Bigl(\Z^{d-1}\tM\cap a_1^{\frac1{d-1}}\fZ_\vecv=\emptyset\Bigr)
\, d\mu(\tM)
\\\label{CVRELATION}
&\hspace{50pt}\geq\mu\Bigl(\Bigl\{M\in \myX^{(d-1)}\col
\Z^{d-1}M\cap(\xi^{\frac1d}v_1^{-\frac1{d-1}}a_1^{\frac1{d-1}})\fC_{p(\vecv)}(w)
=\emptyset
\Bigr\}\Bigr).
\end{align}
Note that \textit{if} 
$\vecv^\perp\cap\fZ\subset\{x_1<\frac12\xi^{\frac1d}\}$ (say) 
then $p(\fZ\cap\vecv^\perp)$ is \textit{equal} to
$\xi^{\frac1d}\fC_{p(\vecv)}(w)$;
hence equality holds in \eqref{LVZVSUBSET} and hence 
\textit{equality also holds in \eqref{CVRELATION}.}
We wish to bound the contribution in \eqref{PHIXIZTHMPF8}
from those $\vecv\in S'$ which do \textit{not} satisfy this condition.
\begin{lem}\label{GOODINTLEM}
If $\vecv$ satisfies %
$\varpi\in(0,\frac\pi2)$, $\omega\in(0,1)$
and $\vecv^\perp\cap\fZ\not\subset\{x_1<\frac12\xi^{\frac1d}\}$, then
\begin{align*}
\sfrac\pi2-\varpi\ll1-w+\omega^2.
\end{align*}
\end{lem}
\begin{proof}
The assumptions imply
$\vecv^\perp\cap(\fZ(0,1,1)+w\vece_2)\not\subset\{x_1<\frac12\}$ 
and thus there is some 
$\vecx\in\vecv^\perp\cap(\fZ(0,1,1)+w\vece_2)$ with $x_1\geq\frac12$;
thus $\frac12\leq x_1<1$, $\|p(\vecx)-w\vece_1\|<1$ and $\vecv\cdot\vecx=0$.
Now %
\begin{align*}
0<\sfrac12v_1\leq v_1x_1=-(v_2x_2+\ldots+v_dx_d)
=-\bigl(v_2(x_2-w)+v_3x_3+\ldots+v_dx_d\bigr)-wv_2
\\
\leq\|p(\vecv)\|\cdot\|p(\vecx)-w\vece_1\|-wv_2
\leq\|p(\vecv)\|-wv_2
=(\sin\varpi)(1-w\cos\omega)
\ll 1-w &+\omega^2.
\end{align*}
This gives the desired bound, since
$v_1=\cos\varpi\gg\frac\pi2-\varpi$.
\end{proof}

It follows from Lemma \ref{GOODINTLEM} that the contribution in 
\eqref{PHIXIZTHMPF8} from those $\vecv\in S'$ which satisfy
$\vecv^\perp\cap\fZ\not\subset\{x_1<\xi^{\frac1d}\}$ is
\begin{align}\label{COMP3}
\ll\int_0^1\int_{\substack{\varpi\in(0,\pi/2)\\\frac\pi2-\varpi\ll
1-w+\omega^2}}
\int_{\xi^{\frac1d}(v_1+2v_2)}^\infty\min\Bigl(1,
\Bigl(a_1\xi^{\frac{d-1}d}\omega^d\Bigr)^{-\frac{2(d-2)}{d-1}}\Bigr)
\,\frac{da_1}{a_1^{d+1}}
\,\varpi^{d-2}\,d\varpi\,\omega^{d-3}\,d\omega.
\end{align}
This is bounded above by the expression in 
\eqref{PHIXIZTHMPF7a}, and is thus $\ll E$.
In view of this bound, the inequality
\eqref{CVRELATION} and the fact that equality holds in 
\eqref{CVRELATION} whenever 
$\vecv^\perp\cap\fZ\subset\{x_1<\frac12\xi^{\frac1d}\}$,
we obtain:
\begin{align}\notag
\Phi(\xi,w)=\frac1{\zeta(d)}\int_{S'}\int_{\xi^{\frac1d}(v_1+2v_2)}^\infty
\mu\Bigl(\Bigl\{M\in \myX^{(d-1)}\col
\Z^{d-1}M\cap(\xi^{\frac1d}v_1^{-\frac1{d-1}}a_1^{\frac1{d-1}})\fC_{p(\vecv)}(w)
=\emptyset
\Bigr\}\Bigr)\,\frac{da_1}{a_1^{d+1}}\,d\vecv
\\\label{PHIXIZTHMPF8a}
+O(E).
\end{align}

\subsection{The paraboloid approximation}\label{PARABOLOIDAPPRSEC}

Next we will replace $\fC_{p(\vecv)}(w)$ with a \textit{cut paraboloid}.
We first prove three easy lemmas on the approximation of a part of
a unit ball with a paraboloid.
(In the present section we will only need to take $u=0$ in the
following, but the case of general $u$ is needed later in the treatment
of $\Phi_\bn(\xi,\vecw,\vecz)$.)

For any $u,r$ with $|u|,|r|<1$, we let $P_{u,r}\subset\R^{d-1}$ 
be the paraboloid given by
\begin{align}\label{PUVDEF}
P_{u,r}:=\bigl\{x_1>A(x_2^2+x_3^2+\ldots+x_{d-1}^2)+Bx_2+C\bigr\};\hspace{30pt}
\begin{cases}A=\frac{1+{\EE}}{2\sqrt{1-u^2}}
\\
B=-\frac{{\EE}u}{\sqrt{1-u^2}}
\\
C=\frac{-2+(1+{\EE})u^2}{2\sqrt{1-u^2}},
\end{cases}
\end{align}
where
\begin{align}\label{EDEF}
{\EE}=\biggl(\frac{u+r}{\sqrt{1-u^2}+\sqrt{1-r^2}}\biggr)^2.
\end{align}
This definition is motivated by the following lemma, which is proved by
a direct calculation. %
\begin{lem}\label{PUVDEFLEM}
If $u\neq r$, then $A,B,C$ are the unique real numbers for which the
parabola $x_1=Ax_2^2+Bx_2+C$ is tangent to the unit circle $x_1^2+x_2^2=1$ at
the point $(-\sqrt{1-u^2},u)$ and also intersects the circle at
$(-\sqrt{1-r^2},r)$.
\end{lem}

The following lemma gives the two fundamental facts which we
will use about the relation between $P_{u,r}$ and the unit ball 
$\scrB_1^{d-1}$.
\begin{lem}\label{PARABOLOIDBALLAPPRLEM}
\rule{0pt}{0pt}
\begin{enumerate}
\item[(i)]
For any $0\leq u<1$ we have $\scrB_1^{d-1}\subset P_{u,-u}$.
\item[(ii)]
For any $0\leq u<r<1$ we have
$P_{u,r}\cap\bigl\{x_1<-\sqrt{1-r^2}\bigr\}\subset\scrB_1^{d-1}.$
\end{enumerate}
\end{lem}
\begin{proof}
Straightforward.
\end{proof}
Let us write $\R_+^{d-1}$ for the upper halfspace
\begin{align}\label{RPLUSDEF}
\R_+^{d-1}:=\bigl\{\vech=(h_1,\ldots,h_{d-1})\in\R^{d-1}\col h_1>0\bigr\}.
\end{align}
Also, for any $\vech\in\R^{d-1}$, we write
\begin{align}\label{RHMINUSDEF}
\R_{\vech-}^{d-1}:=\bigl\{\vecx\in\R^{d-1}\col\vecx\cdot\vech<0\bigr\}.
\end{align}
\begin{lem}\label{CUTPARABOLOIDINCUTBALLLEM}
There is an absolute constant $c_\clowF>0$ such that for any
$\frac9{10}\leq w<1$ and $\vech\in\R_+^{d-1}$,
if $c_\clowF(\sqrt{1-w}+\varphi(\vech,\vece_1))\leq r\leq\frac12$ then
\begin{align*}
(w\vece_1+P_{0,r})\cap\R_{\vech-}^{d-1}\subset
\fC_\vech(w).
\end{align*}
\end{lem}
\begin{proof}
We leave out the routine proof, since this lemma is a special case
(obtained when $\vecz=\vecw=w\vece_1$) of Lemma \ref{CUTCONTAINEDLEM2}
which we will prove later.
\end{proof}

With these three lemmas in place we may now approximate the integrand
in \eqref{PHIXIZTHMPF8a} from above and below.
With $c_\clowF>0$ as in Lemma \ref{CUTPARABOLOIDINCUTBALLLEM},
set (for $\vecv$ as in \eqref{VPARA})
\begin{align*}
r=r(\vecv)=c_\clowF(\sqrt{1-w}+\omega)
=c_\clowF\bigl(\sqrt{1-w}+\varphi(p(\vecv),\vece_1)\bigr).
\end{align*}
From now on we assume that $\xi$ is so large that
$c_\clowD\xi^{-\frac2d}<(4c_\clowF)^{-2}$; 
thus $1-w<(4c_\clowF)^{-2}$.
We also require that our constant $c_\clowC$ in 
\eqref{SPDEF1} should satisfy $c_\clowC<(4c_\clowF)^{-1}$.
Then for all $\vecv\in S'$ we have $r=r(\vecv)<\frac12$,
and hence by Lemma \ref{PARABOLOIDBALLAPPRLEM}(i) and
Lemma \ref{CUTPARABOLOIDINCUTBALLLEM},
\begin{align*}
(w\vece_1+P_{0,r})\cap\R^{d-1}_{p(\vecv)-}\subset
\fC_{p(\vecv)}(w)\subset(w\vece_1+P_{0,0})\cap\R^{d-1}_{p(\vecv)-}.
\end{align*}

We transform $(w\vece_1+P_{0,r})\cap\R^{d-1}_{p(\vecv)-}$ 
to a ``standard'' cut paraboloid as follows. Set
\begin{align}\label{PDM1DEF}
P^{d-1}:=
\bigl\{\vecx=(x_1,\ldots,x_{d-1})\in\R^{d-1}\col x_1>x_2^2+\ldots+x_{d-1}^2-1
\bigr\},
\end{align}
and for $\vech\in\R_+^{d-1}$, %
\begin{align}\label{PDM1HDEF}
P^{d-1}_\vech:=P^{d-1}\cap \R_{\vech-}^{d-1}.
\end{align}
Finally, for ${\sigma}\in\R$ and $v>0$,
let $\Xi({\sigma},v)$ be the probability that a random lattice of
covolume $v$ is disjoint from $P_{(1,\sigma,0,\ldots,0)}^{d-1}$,
viz.\
\begin{align}\label{XIDM1CVDEF}
\Xi({\sigma},v)
=\mu^{(d-1)}\bigl(\bigl\{M\in X_1^{(d-1)}\col
(v^{\frac1{d-1}}\Z^{d-1}M)\cap P_{(1,\sigma,0,\ldots,0)}^{d-1}=\emptyset
\bigr\}\bigr).
\end{align}
By an obvious rotational symmetry the last expression is invariant under
replacement of 
$P_{(1,\sigma,0,\ldots,0)}^{d-1}$ by
$P_\vech^{d-1}$ for any $\vech\in\R_+^{d-1}$ satisfying
$\|(h_2,\ldots,h_{d-1})\|=|\sigma| h_1$.
Now one checks by a quick computation that if we let
\begin{align}
T=\text{diag}\bigl[1-w,\underbrace{\sqrt{2\delta(r)(1-w)},\ldots,\sqrt{2\delta(r)(1-w)}}_{d-2\text{ entries}}\bigr]
\in\GL(d-1,\R),
\end{align}
where
\begin{align*}
\delta(r):=\sfrac12\bigl(1+\sqrt{1-r^2}\bigr),
\end{align*}
then
\begin{align}
(w\vece_1+P_{0,r})\cap\R_{p(\vecv)-}^{d-1}=P^{d-1}_{\vech}T,
\qquad\text{with }\:
\vech=\biggl(\frac{\sqrt{1-w}}{\sqrt{2\delta(r)}}v_2,v_3,\ldots,v_d\biggr).
\end{align}
This transformation formula applies also with $0$ in place of $r$.
Note that $T$ scales volume with a 
factor $2^{\frac d2-1}\delta(r)^{\frac d2-1}(1-w)^{\frac d2}$.
Hence we obtain, for all $\vecv\in S'$,
writing $\vecv''=(v_3,\ldots,v_d)$,  %
\begin{align*}
\Xi\biggl(\sqrt{\frac2{1-w}}\frac{\|\vecv''\|}{v_2},
\kappa\xi^{\frac1d} v_1a_1^{-1}\biggr)
\leq
\mu\Bigl(\Bigl\{M\in \myX^{(d-1)}\col
\Z^{d-1}M\cap(\xi^{\frac1d}v_1^{-\frac1{d-1}}a_1^{\frac1{d-1}})\fC_{p(\vecv)}(w)
=\emptyset
\Bigr\}\Bigr)
\hspace{10pt}
\\
\leq
\Xi\biggl(\sqrt{\frac{2\delta(r)}{1-w}}\frac{\|\vecv''\|}{v_2},
\kappa\xi^{\frac1d} \delta(r)^{1-\frac d2}v_1a_1^{-1}\biggr),
\end{align*}
where 
\begin{align}\label{KAPPADEF}
\kappa=2^{1-\frac d2}\xi^{-1}(1-w)^{-\frac d2}.
\end{align}

Using these bounds in \eqref{PHIXIZTHMPF8a} we conclude
\begin{align}\label{PHI2XIZPARABLOWBOUND}
\Phi(\xi,w)\geq
\frac1{\zeta(d)}\int_{S'}\int_{\xi^{\frac1d}(v_1+2v_2)}^\infty
\Xi\biggl(\sqrt{\frac2{1-w}}\frac{\|\vecv''\|}{v_2},
\kappa\xi^{\frac1d} v_1a_1^{-1}\biggr)
\,\frac{da_1}{a_1^{d+1}}\,d\vecv
-O(E)
\end{align}
and
\begin{align}\label{PHI2XIZPARABUPPBOUND}
\Phi(\xi,w)\leq
\frac1{\zeta(d)}\int_{S'}\int_{\xi^{\frac1d}(v_1+2v_2)}^\infty
\Xi\biggl(\sqrt{\frac{2\delta(r)}{1-w}}\frac{\|\vecv''\|}{v_2},
\kappa\xi^{\frac1d}\delta(r)^{1-\frac d2}v_1a_1^{-1}\biggr)
\,\frac{da_1}{a_1^{d+1}}\,d\vecv
+O(E).
\end{align}

\subsection{\texorpdfstring{Further simplification of the integral; proof of Theorem \ref*{PHIXIWASYMPTTHM} except for $d=3$}{Further simplification of the integral}}
\label{FURTHERSIMPLSEC}

We will now simplify further the integral in 
\eqref{PHI2XIZPARABUPPBOUND}.
Recall that $r=r(\vecv)=c_\clowF(\sqrt{1-w}+\omega)$ in 
this integral.
(Note that the integral in \eqref{PHI2XIZPARABLOWBOUND} can be viewed 
as a special case of the integral in \eqref{PHI2XIZPARABUPPBOUND}, 
by replacing $c_\clowF$ by $0$.)
Substituting $a_1=\xi^{\frac1d}v_1\delta(r)^{1-\frac d2}y^{-1}$ in 
the inner integral, and then using the parametrization \eqref{VPARA}
and substituting further %
$\omega=\arctan\rho$ and $\varpi=\arcsin(v_2\sqrt{1+\rho^2})$, we get
\begin{align*}
\frac{\vol(\S_1^{d-3})\xi^{-1}}{\zeta(d)}
\int_0^{\tan c_\clowC}\int_0^{(1+\rho^2)^{-\frac12}}
\biggl(\int_0^{\alpha\frac{v_1}{v_1+2v_2}}
\Xi\biggl(\sqrt{\frac{2\delta(r)}{1-w}}\rho,
\kappa y\biggr)
\, y^{d-1}\,dy\biggr)
\,\frac{v_2^{d-2}}{v_1^{d+1}}\,dv_2
\,\alpha^{-d}\rho^{d-3}\,d\rho,
\end{align*}
where we use the notation %
\begin{align*}
v_1=\sqrt{1-v_2^2(1+\rho^2)},\qquad r=c_\clowF(\sqrt{1-w}+\arctan \rho),\qquad
\alpha=\delta(r)^{1-\frac d2}.
\end{align*}
Note that $r$ and $\alpha$ only depend on the integration variable $\rho$,
not on $v_2$ or $y$.
Changing order between the two inner integrals gives
\begin{align*}
\frac{\vol(\S_1^{d-3})\xi^{-1}}{\zeta(d)}
\int_0^{\tan c_\clowC}\int_0^{\alpha}
\biggl(\int_0^{\frac{1-y/\alpha}{\sqrt{4(y/\alpha)^2+(1+\rho^2)(1-y/\alpha)^2}}}
\frac{v_2^{d-2}}{v_1^{d+1}}\,dv_2\biggr)
\Xi\biggl(\sqrt{\frac{2\delta(r)}{1-w}}\rho,{\kappa}y\biggr)
y^{d-1}\,dy
\,\alpha^{-d}\rho^{d-3}\,d\rho.
\end{align*}
Here note that for any $0\leq x\leq(1+\rho^2)^{-\frac12}$ we have
\begin{align*}
\int_0^x v_1^{-d-1}v_2^{d-2}\,dv_2
=\int_0^x (1-v_2^2(1+\rho^2))^{-\frac{d+1}2}v_2^{d-2}\,dv_2
=\frac{x^{d-1}}{(d-1)(1-(1+\rho^2)x^2)^{(d-1)/2}},
\end{align*}
and using this with 
$x=\frac{1-y/\alpha}{\sqrt{4(y/\alpha)^2+(1+\rho^2)(1-y/\alpha)^2}}$ we 
conclude that the innermost integral in the above expression equals
$2^{1-d}(d-1)^{-1}(\frac{\alpha-y}{y})^{d-1}$.
Hence the whole expression equals
\begin{align*}
\frac{\vol(\S_1^{d-3})2^{1-d}\xi^{-1}}{(d-1)\zeta(d)}
\int_0^{\tan c_\clowC}\int_0^{\alpha}
\Xi\biggl(\sqrt{\frac{2\delta(r)}{1-w}}\rho,
{\kappa}y\biggr)
(\alpha-y)^{d-1}\alpha^{-d}\rho^{d-3}\,dy\,d\rho.
\end{align*}
Next we replace $\rho$ with the new variable $\sigma$ in the outer integral,
via the substitution
\begin{align*}
\sigma=\sqrt{\frac{2\delta(r)}{1-w}}\rho
=\sqrt{\frac{2\,\delta\bigl(c_\clowF(\sqrt{1-w}+\arctan\rho)\bigr)}{1-w}}\,\rho.
\end{align*}
Note that (using $1-w<(4c_\clowF)^{-2}$)
\begin{align*}
\frac{d\sigma}{d\rho}
=\sqrt{\frac2{1-w}}\Bigl(1+O(1-w+\rho^2)\Bigr),
\end{align*}
uniformly over all $\rho\in(0,\tan\frac1{4c_\clowF})$.
Here the implied constant in the big-$O$ term is absolute 
(since $c_\clowF$ is an absolute constant).
Hence so long as $c_\clowC$ is sufficiently small,
our substitution is a strictly increasing (and smooth, thus
with smooth inverse) map from $\rho\in(0,\tan c_\clowC)$ 
to $\sigma\in(0,\sigma_0)$ where
$\sigma_0\asymp(1-w)^{-\frac12}$,
and furthermore, for all such $\rho,\sigma$ we have
\begin{align*}
&\rho=\sqrt{\sfrac12(1-w)}\,\sigma\bigl(1+O(1-w+\rho^2)\bigr)
=\sqrt{\sfrac12(1-w)}\,\sigma\bigl(1+O((1-w)(1+\sigma^2))\bigr);
\\
&\frac{d\rho}{d\sigma}=\sqrt{\sfrac12(1-w)}\,\bigl(1+O(1-w+\rho^2)\bigr)
=\sqrt{\sfrac12(1-w)}\,\bigl(1+O((1-w)(1+\sigma^2))\bigr);
\\
&r=r(\sigma)\ll\sqrt{1-w}+\arctan\rho\ll\sqrt{1-w}(1+\sigma)
\qquad\text{and}\quad r<\sfrac12.
\end{align*}
Hence our expression equals
\begin{align}\notag
\frac{\vol(\S_1^{d-3})2^{2-\frac32d}}{(d-1)\zeta(d)}
\xi^{-1}(1-w)^{\frac d2-1}
\int_0^{\sigma_0}\int_0^{\alpha}
\Xi\bigl(\sigma,{\kappa}y\bigr)
(\alpha-y)^{d-1}\,dy\hspace{100pt}
\\\label{ALMOSTDONE1}
\times \sigma^{d-3}\Bigl(1+O((1-w)(1+\sigma^2))\Bigr)\,d\sigma.
\end{align}

Now to bound the contribution from the error term we will use the
following lemma.
\begin{lem}\label{XICVBOUNDLEM1}
For all ${\sigma}\geq0$ and $v>0$ we have
\begin{align}\label{XICVBOUNDLEM1RES}
\Xi({\sigma},v)\ll\min\Bigl\{1,((1+{\sigma})^{-d}v)^{2-\frac2{d-1}}\Bigr\}.
\end{align}
\end{lem}
\begin{proof}
This follows from \cite[Cor.\ 1.4]{lprob} coupled with 
the fact that $P_{(1,\sigma,0,\ldots,0)}^{d-1}$ contains
a $(d-1)$-dimensional open cone of volume $\gg(1+{\sigma})^d$ with $\bn$ 
in its base.
To give a slightly more detailed argument it is convenient to
use a simple invariance relation which we will prove in 
Section~\ref{XIDM1BASICSSEC}:
By \eqref{XIDTINV} applied with $T=T_{\alpha,\beta}$
(cf.\ \eqref{TALFBETDEF}),
$\alpha=v^{-\frac1d}$ and $\beta=\frac12{\sigma}v^{-\frac1d}$, we have
\begin{align*}
\Xi({\sigma},v)=\mu\bigl(\bigl\{M\in X_1^{(d-1)}\col \Z^{d-1}M\cap 
(P^{d-1}-\vecy)\cap\R_{\vece_1-}^{d-1}=\emptyset\bigr\}\bigr),
\end{align*}
where
\begin{align*}
\vecy=\Bigl(\bigl(1+\sfrac14{\sigma}^2\bigr)v^{-\frac2d}-1\Bigr)\vece_1
+\sfrac12{\sigma}v^{-\frac1d}\vece_2\in P^{d-1}.
\end{align*}
Now $(P^{d-1}-\vecy)\cap\R_{\vece_1-}^{d-1}$ contains the open cone with base
$B=(-\frac12{\sigma} v^{-\frac1d}\vece_2+\scrB_r^{d-1})\cap\vece_1^\perp$ 
of radius
$r=(1+\frac14{\sigma}^2)^{\frac12}v^{-\frac1d}$, and apex $-\vece_1-\vecy$,
and this cone has volume
$\gg (1+\frac14{\sigma}^2)v^{-\frac2d}r^{d-2}\gg(1+{\sigma})^dv^{-1}$.
Hence the desired bound follows from \cite[Cor.\ 1.4]{lprob}.
\end{proof}

Using Lemma \ref{XICVBOUNDLEM1}
we see that the contribution from the error term in 
\eqref{ALMOSTDONE1} is
(since $\alpha\ll1$ over the whole range of integration)
\begin{align}\label{ERRORBDINALMOSTDONE1}
&\ll\xi^{-1}(1-w)^{\frac d2}\int_0^{\sigma_0}
\min\Bigl\{1,(\sigma^{-d}{\kappa})^{2-\frac2{d-1}}\Bigr\}
\sigma^{d-3}(1+\sigma^2)\,d\sigma.
\end{align}
If $d\geq4$ then this is $\ll\xi^{-2}$, even after replacing the
upper integration limit $\sigma_0$ by $\infty$.
(Here we again used our assumption $1-w<c_\clowD\xi^{-\frac2d}$.)
However if $d=3$ then we obtain, using the fact that
$\sigma_0\ll(1-w)^{-\frac12}$,
\begin{align*}
&\ll\xi^{-1}(1-w)^{\frac 32}
\int_0^{{{\kappa}}^{1/3}}
\bigl(1+\sigma^2\bigr)\,d\sigma
+
\xi^{-2}
\int_{{{\kappa}}^{1/3}}^{\max({{\kappa}}^{1/3},\sigma_0)}
\bigl(\sigma^{-3}+\sigma^{-1}\bigr)\,d\sigma
\ll\xi^{-2}\log\xi.
\end{align*}
Hence in all cases the contribution from the error term in 
\eqref{ALMOSTDONE1} is $\ll E$.

Let us also note that for any fixed $0\leq \sigma\leq\sigma_0$, 
\begin{align*}
f(\beta)=\int_0^\beta
\Xi\bigl(\sigma,{\kappa}y\bigr)(\beta-y)^{d-1}\,dy
\end{align*}
is an increasing function of $\beta\geq1$, with derivative
\begin{align*}
f'(\beta)=(d-1)\int_0^\beta
\Xi\bigl(\sigma,{\kappa}y\bigr)(\beta-y)^{d-2}\,dy.
\end{align*}
If $1\leq\beta\ll1$ then this derivative is
$\ll\min\bigl\{1,(\sigma^{-d}{\kappa})^{2-\frac2{d-1}}\bigr\}$,
by Lemma \ref{XICVBOUNDLEM1}.
Since also $\alpha-1 %
\ll(1-w)(1+\sigma^2)$ 
we conclude that the difference caused by replacing
$\alpha$ by $1$ in \eqref{ALMOSTDONE1}
is bounded by exactly the same expression as in 
\eqref{ERRORBDINALMOSTDONE1}.

Hence from \eqref{PHI2XIZPARABUPPBOUND}, we have proved that
\begin{align}\label{PHIXIWASYMPTTHMKLAR1}
\Phi(\xi,w)
\leq
\frac{\vol(\S_1^{d-3})2^{2-\frac32d}}{(d-1)\zeta(d)}
\xi^{-1}(1-w)^{\frac d2-1}
\int_0^\infty\int_0^1
\Xi\bigl(\sigma,{\kappa}y\bigr)
(1-y)^{d-1}\,\sigma^{d-3}\,dy\,d\sigma+O(E).
\end{align}
(We used the fact that the integrand is nonnegative to increase the
$\sigma$-integration range from $(0,\sigma_0)$ to $(0,\infty)$.)

Finally note that the same computations allow us to compute the
integral in \eqref{PHI2XIZPARABLOWBOUND} in \textit{exact} terms;
we thus obtain 
\begin{align*}
\Phi(\xi,w)
\geq
\frac{\vol(\S_1^{d-3})2^{2-\frac32d}}{(d-1)\zeta(d)}
\xi^{-1}(1-w)^{\frac d2-1}
\int_0^{\sqrt{\frac2{1-w}}\tan c_\clowC}\int_0^1
\Xi\bigl(\sigma,{\kappa}y\bigr)
(1-y)^{d-1}\,\sigma^{d-3}\,dy\,d\sigma
-O(E).
\end{align*}
By Lemma \ref{XICVBOUNDLEM1} we may here increase the upper range of
$\sigma$ to $\infty$ at the cost of an error
\begin{align*}
\ll\xi^{-1}(1-w)^{\frac d2-1}\int_{\sqrt{\frac2{1-w}}\tan c_\clowC}^\infty
\min\bigl\{1,\sigma^{-d}{\kappa}\bigr\}
\sigma^{d-3}\,d\sigma
\ll\xi^{-2}(1-w)^{-1}\int_{\sqrt{\frac2{1-w}}\tan c_\clowC}^\infty \sigma^{-3}
\,d\sigma\ll\xi^{-2}.
\end{align*}
Hence, using $\vol(\S_1^{d-3})=\frac{2\pi^{d/2-1}}{\Gamma(d/2-1)}$ 
and ${\kappa}=2^{1-\frac d2}\xi^{-1}(1-w)^{-\frac d2}$, 
we finally obtain
\eqref{PHIXIWASYMPTTHMRES2} with
\begin{align}\label{PHIXIWASYMPTTHMRES}
F_d(t)
=\frac{2^{3(1-\frac{d}2)}\pi^{\frac d2-1}}{(d-1)\Gamma(\frac d2-1)\zeta(d)}
t^{\frac d2-1}\int_0^1\int_0^\infty 
\Xi\bigl(\sigma,2^{1-\frac d2}t^{-\frac{d}{2}}y\bigr)
\,\sigma^{d-3}\,(1-y)^{d-1}\,d\sigma\,dy,
\end{align}
except that we get the error term $O(E)$
(cf.\ \eqref{EFIRSTDEF}), which is slightly worse than the error term in
\eqref{PHIXIWASYMPTTHMRES2} when $d=3$.
Note that Lemma \ref{XICVBOUNDLEM1} implies
that the function $F_d(t)$ is uniformly bounded over $t\in\R_{>0}$.
Furthermore $F_d(t)$ is easily seen to be continuous;
for instance this follows immediately from Lemma \ref{XISUPBOUNDLEM},
\eqref{XI2INXI4} and Lemma \ref{XIDM1CONTLEM} below.

Hence to complete the proof of Theorem \ref{PHIXIWASYMPTTHM} 
it now only remains to improve the error term slightly in the case $d=3$.

\subsection{A slight improvement of the error term for $d=3$}
\label{SLIGHTIMPRSEC}
The following lemma is valid for arbitrary $d\geq3$.
\begin{lem}\label{PHIXIZTHMPF2part2vanishlem}
There is a constant $c_\clowG>1$ which only depends on $d$ such that 
if $\vecv\in\HS\setminus\{\vece_1\}$ 
satisfies either [$\omega\leq\frac\pi2$ and 
$\omega\geq c_\clowG\xi^{-1}(1-w)^{\frac{1-d}2}$]
or [$\omega\geq\frac\pi2$ and 
$c_\clowG\xi^{-\frac1{d-2}}\leq\pi-\omega\leq 
c_\clowG^{-1}\xi(1-w)^{\frac{d-1}2}$],
then the left hand side of \eqref{PHIXIZTHMPF2part2} vanishes
for all $a_1>A=c_\clowE\xi^{\frac1d}$.
\end{lem}
\begin{proof}
For any $a_1>A$, $a_1^{\frac1{d-1}}\fZ_\vecv$ contains a cone of volume
$\gg\xi\omega^2(\sin\omega)^{d-2}$ and edge ratio
$\asymp\min(1,\frac{1-w}{\sin^2\omega})$
(cf.\ \cite[Lem.\ 7.1]{lprob});
hence by \cite[Cor.\ 1.4]{lprob}, there is a constant $c>1$ such that 
the left hand side of \eqref{PHIXIZTHMPF2part2} vanishes whenever
$\min(1,\frac{1-w}{\sin^2\omega})\geq c
(\xi\omega^2(\sin\omega)^{d-2})^{-\frac2{d-1}}$.
The lemma follows from this.
(In the case $\omega\leq\frac\pi2$ one also uses 
$1-w<c_\clowD\xi^{-\frac2d}$.)
\end{proof}

Now if $\xi^{-1}(1-w)^{\frac{1-d}2}$ is sufficiently small
so that $c_\clowG\xi^{-1}(1-w)^{\frac{1-d}2}\leq c_\clowC$ and
$c_\clowG^{-1}\xi(1-w)^{\frac{d-1}2}>\frac\pi2$, then
Lemma \ref{PHIXIZTHMPF2part2vanishlem} says that
\eqref{COMP1} remains a valid bound on the error in question
even if we restrict the integration range for $\tau$ to
$\tau\in(0,c_\clowG\xi^{-\frac1{d-2}})$.
This does not give any improvement if $d\geq4$, but if
$d=3$ it means that the error bound in \eqref{COMP1} is improved from
$\xi^{-2}\log\xi$ down to $\xi^{-2}$.

Keeping the assumption $c_\clowG\xi^{-1}(1-w)^{\frac{1-d}2}\leq c_\clowC$,
Lemma \ref{PHIXIZTHMPF2part2vanishlem} also says that
we may restrict the range of $\omega$ in \eqref{COMP2} to
$\omega\in(0,c_\clowG\xi^{-1}(1-w)^{\frac{1-d}2})$. When 
$d=3$ this means that we improve the bound in 
\eqref{COMP2} down to $\xi^{-2}\log(2+\xi^{-\frac23}(1-w)^{-1})$.
In the same way also the bounds in \eqref{PHIXIZTHMPF7a} and
\eqref{COMP3} can be improved down to 
$\xi^{-2}\log(2+\xi^{-\frac23}(1-w)^{-1})$.

Finally to improve the error bounds in Section \ref{FURTHERSIMPLSEC}
we note the following.
\begin{lem}\label{XISUPBOUNDLEM}
There is a constant $c_\clowI>0$ which only depends on $d$ such 
that $\Xi(\sigma,v)=0$ holds whenever $v\leq c_\clowI(1+\sigma)$, 
$\sigma\geq0$.
\end{lem}
\begin{proof}
This follows by noticing that the cone considered in the proof of
Lemma \ref{XICVBOUNDLEM1} has edge ratio
$\asymp(1+\sigma)^{-2}$,
and volume $\gg(1+\sigma)^dv^{-1}$,
and applying \cite[Cor.\ 1.4]{lprob}.
\end{proof}

Using Lemma \ref{XISUPBOUNDLEM} (and the fact that $\alpha\ll1$ for all
$\sigma\in(0,\sigma_0)$) we see that the
inner integral in \eqref{ALMOSTDONE1} vanishes unless
$\sigma\ll{\kappa}$,
and thus the bound \eqref{ERRORBDINALMOSTDONE1} is valid also if
we integrate over all $\sigma\ll{\kappa}$ instead of
$\sigma\in(0,\sigma_0)$.
Hence for $d=3$ we obtain the improved bound $\xi^{-2}\log(2+{{\kappa}}^{2/3})
\asymp\xi^{-2}\log(2+\xi^{-\frac23}(1-w)^{-1})$.

Using all these improved bounds in 
Sec.\ \ref{SIMPLINTSEC}--\ref{FURTHERSIMPLSEC}
we conclude that, for $d=3$,
\textit{if} $\xi(1-w)$ is sufficiently large
so that $c_\clowG\xi^{-1}(1-w)^{-1}\leq c_\clowC$ and
$c_\clowG^{-1}\xi(1-w)>\frac\pi2$, then
\begin{align*}
\Phi(\xi,w)=\xi^{-\frac43} F_3(\xi^{\frac23}(1-w))
+O\Bigl(\xi^{-2}\log(2+\xi^{-\frac23}(1-w)^{-1})\Bigr).
\end{align*}
We have already proved in Sec.\ \ref{FURTHERSIMPLSEC}
the same formula with the error bound
$\xi^{-2}\log\xi$, valid for \textit{all} $w\in[0,1)$ 
and large $\xi$.
(Recall that both $\Phi(\xi,w)$ and $\xi^{-\frac43} F_3(\xi^{\frac23}(1-w))$
\textit{vanish} when $1-w\geq c_\clowD\xi^{-\frac23}$.)
Taking these facts together it follows that
\eqref{PHIXIWASYMPTTHMRES2} holds.
This concludes the proof of Theorem~\ref{PHIXIWASYMPTTHM}.
\hfill$\square$ $\square$ $\square$

\section{The paraboloid approximation}
\label{PARABAPPRSEC}

\subsection{\texorpdfstring{Definition and basic properties of the general $\Xi$-function}{Definition and basic properties of the general Xi-function}}\label{XIDM1BASICSSEC}
We now introduce a general lattice probability function 
$\Xi(\vecy,\vecy';\vech;v)$
involving two cut paraboloids in $\R^{d-1}$.
This is the function in terms of which we will
later express our asymptotic formula for $\Phi_\bn(\xi,w,z,\varphi)$
as $\xi\to\infty$, cf.\ Theorem \ref{PHI0XILARGETHM} and
\eqref{PHI0XILARGETHMFDDEF}.
We keep $d\geq3$ throughout this section.

Recall 
\begin{align}\notag
P^{d-1}:=
\bigl\{\vecx=(x_1,\ldots,x_{d-1})\in\R^{d-1}\col x_1>x_2^2+\ldots+x_{d-1}^2-1
\bigr\},
\end{align}
For $\vech\in\R_+^{d-1}$ and $\vecy\in\R^{d-1}$ we let
$P^{d-1}_\vech(\vecy)$ be the following cut translate of $P^{d-1}$:
\begin{align}\label{PDM1CYDEF}
P^{d-1}_\vech(\vecy):=
\R_{\vech-}^{d-1}\cap(P^{d-1}-\vecy)
=\bigl\{\vecx\in\R^{d-1}\col\vecx+\vecy\in P^{d-1},\:\vech\cdot\vecx<0\bigr\}.
\end{align}
Thus $P^{d-1}_\vech=P^{d-1}_\vech(\bn)$
(cf.\ \eqref{PDM1HDEF}).
Now for $\vecy,\vecy'\in\R^{d-1}$, $\vech\in\R_+^{d-1}$ and $v>0$, 
we let $\Xi(\vecy,\vecy';\vech;v)$ 
be the probability that 
a random lattice of covolume $v$ has empty intersection with both
$P^{d-1}_\vech\bigl(\vecy\bigr)$ and $P^{d-1}_\vech\bigl(\vecy'\bigr)$,
i.e.
\begin{align}\label{XIDM1DEFnew}
\Xi(\vecy,\vecy';\vech;v):=
\mu\Bigl(\Bigl\{M\in \myX^{(d-1)}\col 
(v^{\frac 1{d-1}}\Z^{d-1}M)\cap\bigl(P^{d-1}_\vech\bigl(\vecy\bigr)
\cup P^{d-1}_\vech\bigl(\vecy'\bigr)\bigr)=\emptyset
\Bigr\}\Bigr).
\end{align}
In the special case $\vecy'=\vecy$ we also write, for short:
\begin{align}\label{XI3DEF}
\Xi(\vecy;\vech;v):=\Xi(\vecy,\vecy;\vech;v).
\end{align}
We will in fact only consider $\Xi(\vecy,\vecy';\vech;v)$
for $\vecy,\vecy'\in\overline{P^{d-1}}$;
actually we will even have $\vecy,\vecy'\in P^{d-1}$
throughout the paper except in Proposition \ref{LIMFDT2PROP} below.

The function $\Xi(\vecy,\vecy';\vech;v)$
satisfies an invariance relation under the 
simultaneous transformation of the couple $\langle\vecy,\vecy'\rangle$
by any affine linear map preserving $P^{d-1}$.
Let us write $\AGL(d-1,\R)$ for the group of non-singular affine linear
transformations of $\R^{d-1}$.
We represent the elements of $\AGL(d-1,\R)$ by pairs
$(M,\vecxi)$ with $M\in\GL(d-1,\R)$ and $\vecxi\in\R^{d-1}$,
where the action of $(M,\vecxi)$ on $\R^{d-1}$ is given by
$\vecy\mapsto \vecy M+\vecxi$.
Now note that for any $T=(M,\vecxi)\in\AGL(d-1,\R)$ which
maps $P^{d-1}$ onto itself we have,
directly from %
\eqref{PDM1CYDEF}, 
\begin{align}\label{PCYMTRANSF}
P_\vech^{d-1}(\vecy)M=P_{\vech\trans M^{-1}}^{d-1}(\vecy T).
\end{align}
(Note here that $\vech\trans M^{-1}\in\R^{d-1}_+$.
Indeed, for all sufficiently large $t>0$ we have
$t\vece_1\in P^{d-1}-\vecxi=P^{d-1}M$ and thus
$t\vece_1M^{-1}\in P^{d-1}$. This implies $\vece_1M^{-1}=s\vece_1$
for some $s>0$, and hence $\vech\trans M^{-1}\cdot\vece_1
=\vech\cdot(\vece_1 M^{-1})%
>0$.)
From \eqref{PCYMTRANSF} and the fact that 
$\mu$ is preserved under $L\mapsto |\det M|^{-\frac 1{d-1}}LM$, we get
\begin{align}\label{XIDTINV}
\Xi(\vecy,\vecy';\vech;v)
=\Xi\bigl(\vecy T,\vecy' T;\vech\trans M^{-1};v|\det M|\bigr).
\end{align}
As a special case of this relation we have
\begin{align}\label{XIDTKINV}
\Xi(\vecy,\vecy';\vech;v)=\Xi(\vecy K,\vecy' K;\vech K;v)
\qquad\text{for any $K\in O(d-1)$ with $\vece_1 K=\vece_1$.}
\end{align}

Next note that for any $\alpha,\beta\in\R$, $\alpha\neq0$, 
the following
affine linear map preserves $P^{d-1}$:
\begin{align}\label{TALFBETDEF}
T_{\alpha,\beta}:=\bigl(M_{\alpha,\beta},
(\alpha^2+\beta^2-1)\vece_1+\beta\vece_2\bigr),
\qquad
M_{\alpha,\beta}=\begin{pmatrix}
\alpha^2 &            \\
2\alpha\beta & \alpha &         \\
 & & \alpha           \\
 & & & \ddots     \\
 & & & & \alpha
\end{pmatrix}.
\end{align}
The set of these maps $T_{\alpha,\beta}$ forms a group, with
multiplication laws
\begin{align*}
T_{\alpha,\beta}T_{\alpha',\beta'}=T_{\alpha\alpha',\beta\alpha'+\beta'};
\qquad
T_{\alpha,\beta}^{-1}=T_{\alpha^{-1},-\alpha^{-1}\beta}.
\end{align*}
Applying \eqref{XIDTINV} repeatedly with $T=K$ as in \eqref{XIDTKINV}
and $T=T_{\alpha,\beta}$ as in \eqref{TALFBETDEF} we see that 
$\langle\vecy,\vecy'\rangle\in(P^{d-1})^2$ may always be transformed to a pair 
of vectors with
$\vecy=\bn$ and %
$\vecy'\in \text{span}\{\vece_1,\vece_2\}$.
It is now natural to define, for $a>0$, $b\in\R$,
\begin{align}\label{XIDM1ABDEF}
\Xi(a,b;\vech;v):=
\Xi\bigl(\bn,\bn T_{a,b};\vech;v)
=\Xi\bigl(\bn,(a^2+b^2-1)\vece_1+b\vece_2;\vech;v\bigr).
\end{align}
We note that this function satisfies the symmetry relation
\begin{align}\label{XIDM1SYMM} 
\Xi(a,b;\vech;v)=
\Xi\bigl(a^{-1},a^{-1}b;(ah_1,-2bh_1-h_2,h_3,\ldots,h_{d-1});a^{-d}v\bigr).
\end{align}
Indeed, applying \eqref{XIDTINV} with $T=T_{a,b}^{-1}$ we get
\begin{align}\notag
\Xi(a,b;\vech;v)
&=\Xi\bigl(\bn T_{a,b}^{-1},\bn;\vech \trans M_{a,b};v|\det M_{a,b}^{-1}|)
\\
&=\Xi\bigl(a^{-1},-a^{-1}b;(a^2h_1,2abh_1+ah_2,ah_3,\ldots,ah_{d-1}),
a^{-d}v\bigr),
\end{align}
and using now \eqref{XIDTKINV} with $K=\text{diag}[1,-1,1,\ldots,1]$
together with the fact that $\Xi(\vecy,\vecy';\vech;v)$ only
depends on the direction of $\vech$ and not its length, 
we obtain \eqref{XIDM1SYMM}.

Note that we are now using the same sign $\Xi$ for four
different but related functions:
$\Xi(\sigma,v)$, $\Xi(\vecy,\vecy';\vech;v)$, $\Xi(\vecy;\vech;v)$ and
$\Xi(a,b;\vech;v)$
(cf.\ \eqref{XIDM1CVDEF}, \eqref{XIDM1DEFnew}, \eqref{XI3DEF} and
\eqref{XIDM1ABDEF}).
There should be no risk of confusing these, since the number or
types (vector/scalar) of the arguments are different in the four cases.
The relation between $\Xi(\sigma,v)$ and the other functions is of course:
\begin{align}\label{XI2INXI4}
\Xi(\sigma,v)=\Xi\bigl(\bn;(1,\sigma,0,\ldots,0);v\bigr)
=\Xi(1,0;(1,\sigma,0,\ldots,0);v\bigr).
\end{align}

The following lemma tells how to bring two points 
$\vecy,\vecy'\in \text{span}\{\vece_1,\vece_2\}\cap P^{d-1}$
to normal position:
\begin{lem}\label{TALFBETGENNORMALIZELEM}
Assume that both
$\vecy=y_1\vece_1+y_2\vece_2$ and $\vecy'=y_1'\vece_1+y_2'\vece_2$
lie in $P^{d-1}$ (viz.,\ $1+y_1-y_2^2>0$ and $1+y_1'-{y_2'}^2>0$).
Set
\begin{align*}
\alpha=\sqrt{1+y_1-y_2^2};\qquad\beta=y_2;\qquad
a=\sqrt{\frac{1+y_1'-{y_2'}^2}{1+y_1-y_2^2}};\qquad
b=\frac{y_2'-y_2}{\sqrt{1+y_1-y_2^2}}.
\end{align*}
Then $\vecy T_{\alpha,\beta}^{-1}=\bn$ and
$\vecy' T_{\alpha,\beta}^{-1}=\bn T_{a,b}$.
\end{lem}
\begin{proof}
This is verified by a direct computation.
\end{proof}

\subsection{\texorpdfstring{Some properties of the function $F_d(t)$}{Some properties of the function Fd(t)}}
\label{PHIXILARGETHMPFSEC}

In this section we prove the properties of the function $F_d(t)$ stated
in the paragraph below Theorem \ref{PHIXIWASYMPTTHM},
and we also derive Theorem~\ref{PHIXILARGETHM} as a consequence of
Theorem \ref{PHIXIWASYMPTTHM} and \eqref{PHIXIAVFORMULA}.

Recall that we have defined, in \eqref{PHIXIWASYMPTTHMRES},
\begin{align*}
F_d(t)
=\frac{2^{3(1-\frac{d}2)}\pi^{\frac d2-1}}{(d-1)\Gamma(\frac d2-1)\zeta(d)}
t^{\frac d2-1}\int_0^1\int_0^\infty 
\Xi\bigl(\sigma,2^{1-\frac d2}t^{-\frac{d}{2}}y\bigr)
\,\sigma^{d-3}\,(1-y)^{d-1}\,d\sigma\,dy.
\end{align*}
We introduce the function (for $a>0$, $b\in\R$)
\begin{align}\label{RHODABDEF}
\rho(a,b):=\inf\bigl\{v>0\col \exists \vech\in\R_+^{d-1}:\:
\Xi(a,b;\vech;v)>0\bigr\}.
\end{align}
Then in particular we have (cf.\ \eqref{XI2INXI4})
\begin{align}\label{RHODDEF}
\rho(1,0)=\inf\bigl\{v>0\col \exists \sigma\in\R_{\geq0}:\:
\Xi(\sigma;v)>0\bigr\}.
\end{align}
It follows easily from \eqref{XIDM1CVDEF} and
\cite[Lemma 2.3]{lprob} that $\Xi(\sigma,v)$ is a continuous function 
of $\sigma$ and $v$.  %
Hence we conclude that, for $t>0$,
\begin{align}\label{FDTCOMPSUPP}
F_d(t)>0\Longleftrightarrow
t<2^{\frac2d-1}\rho(1,0)^{-\frac2d}=\sqrt{\sigma_d(1,0)}
\end{align}
(cf.\ \eqref{SUPPORTprelim4} below for the last relation),
just as stated in the introduction.

Next, it is easy to see that $F_d(t)$ stays bounded from below as $t\to0$.
Indeed, note that
$\vol_{d-1}\bigl(v^{-\frac1{d-1}}P_{(1,\sigma,0,\ldots,0)}^{d-1}\bigr)
\asymp(1+\sigma)^dv^{-1}$;
hence $\Xi(\sigma,v)\gg1$ holds whenever
$(1+\sigma)^dv^{-1}$ is smaller than some positive constant which only
depends on $d$ (cf.\ \cite[Lemma 2.2]{lprob}).
In particular if $t$ is sufficiently small then
the integrand in \eqref{PHIXIWASYMPTTHMRES} is
$\asymp\sigma^{d-3}(1-y)^{d-1}$ for all
$\frac12\leq y\leq1$, $\sigma\ll t^{-\frac 12}$,
and thus $F_d(t)\gg1$.

However we even have that $\lim_{t\to0}F_d(t)$ exists and is positive:
\begin{prop}\label{LIMFDT2PROP}
For every $d\geq3$ we have
\begin{align}\label{LIMFDT2}
\lim_{t\to 0^+}F_d(t)=
\frac{2^{3-d-\frac2d}\pi^{\frac d2-1}\Gamma(d)\Gamma(2-\frac2d)}
{d(d-1)\Gamma(\frac d2-1)\Gamma(d+2-\frac2d)\zeta(d)}
\int_0^\infty \Xi\bigl(\vece_2;\vece_1;v\bigr)\,v^{-2+\frac2d}\,dv.
\end{align}
\end{prop}
\begin{proof}
Using \eqref{XIDTINV} with $T=T_{\alpha,\beta}$,
$\alpha=\frac2{\sqrt{4+\sigma^2}}$ and
$\beta=\frac\sigma{\sqrt{4+\sigma^2}}$ we may rewrite
\eqref{PHIXIWASYMPTTHMRES} as follows:
\begin{align}\label{FDTFORMULA} %
F_d(t)
=K %
t^{\frac d2-1}
\int_0^1\int_0^\infty
\Xi\biggl(\frac\sigma{\sqrt{4+\sigma^2}}\vece_2;\vece_1;
2^{1-\frac d2}t^{-\frac{d}{2}}(1+\sfrac14\sigma^2)^{-\frac d2}y\biggr)
\,\sigma^{d-3}\,(1-y)^{d-1}\,d\sigma\,dy,
\end{align}
where 
$K=\frac{2^{3(1-\frac{d}2)}\pi^{\frac d2-1}}{(d-1)\Gamma(\frac d2-1)\zeta(d)}$.
We now have, uniformly over all
$x\in[0,1]$ and $v>0$ 
\begin{align*}
&\bigl|\Xi(x\vece_2;\vece_1;v)-\Xi(\vece_2;\vece_1;v)\bigr|
\\
&\leq v^{-1}
\max\Bigl\{
\vol_{d-1}\bigl(P^{d-1}_{\vece_1}(x\vece_2)\setminus P^{d-1}_{\vece_1}(\vece_2)\bigr),
\vol_{d-1}\bigl(P^{d-1}_{\vece_1}(\vece_2)\setminus P^{d-1}_{\vece_1}(x\vece_2)\bigr)
\Bigr\}
\ll v^{-1}(1-x)
\end{align*}
(cf.\ \cite[Lemma 2.3]{lprob}), and also
\begin{align}\label{LIMFDT2PROPPF3}
\Xi(x\vece_2;\vece_1;v)\ll\min(1,v^{2-\frac2{d-1}})
\end{align}
(cf.\ \cite[Cor.\ 1.4]{lprob}).
Let us keep $t<1$ for the rest of the proof.
Using the above two bounds we see that the error caused by
replacing $\frac\sigma{\sqrt{4+\sigma^2}}\vece_2$ by $\vece_2$ 
in \eqref{FDTFORMULA} is:
\begin{align}\notag
\ll t^{\frac d2-1}\int_0^1\int_0^\infty 
\min\biggl\{1,(t^{-\frac d2}(1+\sigma)^{-d}y)^{2-\frac2{d-1}},
t^{\frac d2}(1+\sigma)^{d-2}y^{-1}\biggr\}
\sigma^{d-3}(1-y)^{d-1}
\,d\sigma\,dy.
\\\label{LIMFDT2PROPPF1}
\ll t^{\frac d2-1}\int_1^\infty\int_0^1
\min\biggl\{1,(t^{-\frac d2}\sigma^{-d})^{2-\frac2{d-1}},
t^{\frac d2}\sigma^{d-2}y^{-1}\biggr\}\sigma^{d-3}\,dy\,d\sigma.
\end{align}
To compute this we note that
\begin{align*}
\int_0^1\min(A,By^{-1})\,dy=
\begin{cases} B(1+\log(A/B))&\text{if }\: 0<B\leq A
\\
A&\text{if }\: 0<A\leq B.
\end{cases}
\end{align*}
Furthermore note that for $1<\sigma<t^{-\frac12}$ we have
$1<(t^{-\frac d2}\sigma^{-d})^{2-\frac2{d-1}}$ and
$1>t^{\frac d2}\sigma^{d-2}$,
while for $\sigma>t^{-\frac12}$ we have
$1>(t^{-\frac d2}\sigma^{-d})^{2-\frac2{d-1}}$,
and $(t^{-\frac d2}\sigma^{-d})^{2-\frac2{d-1}}>t^{\frac d2}\sigma^{d-2}$
holds if and only if 
$\sigma<t^c$, where
$c:=-\frac{(3d-5)d}{2(3d^2-7d+2)}$
(note $-\frac34\leq c<-\frac12$).
Writing also
$c'=\frac{3d^2-7d+2}{d-1}$ (note $c'\geq4$),
we now obtain that \eqref{LIMFDT2PROPPF1} is
\begin{align*}
\ll t^{\frac d2-1}\biggl(\int_1^{t^{-1/2}}
t^{\frac d2}\sigma^{2d-5}\bigl(1+\log(t^{-\frac d2}\sigma^{2-d})\bigr)\,d\sigma
+\int_{t^{-1/2}}^{t^c}t^{\frac d2}\sigma^{2d-5}
\bigl(1+\log((t^c\sigma^{-1})^{c'})\bigr)\,d\sigma
\hspace{40pt}
\\
+\int_{t^c}^\infty t^{-d+\frac d{d-1}}
\sigma^{-\frac{d^2-3}{d-1}}\,d\sigma\biggr)
\ll t^{\frac{d+1}{3d-1}}.
\end{align*}
Hence we have proved
\begin{align*}
F_d(t)=
K
t^{\frac d2-1}
\int_0^1\int_0^\infty
\Xi\biggl(\vece_2;\vece_1;
2^{1-\frac d2}t^{-\frac{d}{2}}(1+\sfrac14\sigma^2)^{-\frac d2}y\biggr)
\,\sigma^{d-3}\,(1-y)^{d-1}\,d\sigma\,dy
+O\bigl(t^{\frac{d+1}{3d-1}}\bigr).
\end{align*}

Substituting
$y=(1+\frac14\sigma^2)^{\frac d2}x$ the double integral becomes:
\begin{align}\notag %
&\int_0^1\Xi\Bigl(\vece_2;\vece_1;
2^{1-\frac d2}t^{-\frac{d}{2}}x\Bigr)
\int_0^{2\sqrt{x^{-2/d}-1}}
\,\Bigl(1-\bigl(1+\sfrac14\sigma^2\bigr)^{\frac d2}x\Bigr)^{d-1}
\, (1+\sfrac14\sigma^2)^{\frac d2}\sigma^{d-3}\,d\sigma\,dx.
\end{align}
Here the inner integral equals
\begin{align*}
2^{d-3}
\int_1^{x^{-\frac2d}}\bigl(1-\tau^{\frac d2}x\bigr)^{d-1}\tau^{\frac d2}
(\tau-1)^{\frac d2-2}\,d\tau
=2^{d-3}
\int_0^{x^{-\frac2d}}\bigl(1-\tau^{\frac d2}x\bigr)^{d-1}\tau^{d-2}\,d\tau
+O\bigl(x^{-2+\frac4d}\bigr)
\\
=\frac{2^{d-2}\Gamma(d)\Gamma(2-\frac2d)}{d\,\Gamma(d+2-\frac2d)}
x^{-2+\frac2d}\bigl(1+O(x^{\frac2d})\bigr),
\end{align*}
uniformly over $0<x\leq1$.
Using this and the bound \eqref{LIMFDT2PROPPF3}
we obtain \eqref{LIMFDT2}
(and we furthermore see that the rate of convergence is
$O(t^{\frac{d+1}{3d-1}})$).
\end{proof}

We next turn to the proof of Theorem~\ref{PHIXILARGETHM}
using Theorem \ref{PHIXIWASYMPTTHM} and \eqref{PHIXIAVFORMULA}.
The key step in this derivation is the following formula.
\begin{prop}\label{PHIXIWAVERAGEPF3PROP}
\begin{align}\label{PHIXIWAVERAGEPF3}
\int_0^\infty F_d(t)\,dt=\frac{2^{1-d}}{(d-1)d(d+1)\zeta(d)}.
\end{align}
\end{prop}
\begin{proof}[Proof of Theorem~\ref{PHIXILARGETHM} using Proposition \ref{PHIXIWAVERAGEPF3PROP}]
We may evaluate the integral $\int_{\scrB_1^{d-1}}\Phi(\xi,\vecw)\,d\vecw$
by introducing polar coordinates $\vecw=(1-u)\vecomega$ ($0<u<1$, 
$\vecomega\in\S_1^{d-2}$) %
and using Theorem \ref{PHIXIWASYMPTTHM}
together with the fact that $F_d(t)=0$ for all 
$t\geq\sqrt{\sigma_d(1,0)}$
and $\Phi(\xi,w)=0$ unless $1-w\ll\xi^{-\frac2d}$.
This gives, for large $\xi$,
\begin{align}\label{PHIXIWAVERAGEPF1}
\int_{\scrB_1^{d-1}}\Phi(\xi,\vecw)\,d\vecw
=\vol(\S_1^{d-2})\, %
\xi^{-2+\frac2d}\int_0^1 F_d\bigl(u\xi^{\frac2d}\bigr)(1-u)^{d-2}\,du+
O\bigl(\xi^{-2-\frac2d}\bigr).
\end{align}
(Here for $d=3$ we used the fact that, for every fixed $c>0$,
$\int_0^{c\xi^{-\frac23}}\log(2+\xi^{-\frac23}u^{-1})\,du
\ll\xi^{-\frac23}$.)
Replacing the factor $(1-u)^{d-2}$ by $1$ 
in the integral in the right hand 
side of \eqref{PHIXIWAVERAGEPF1} causes an error $\ll\xi^{-2-\frac2d}$,
since $F_d(t)$ is uniformly bounded and of compact support.
Hence we get, via Proposition \ref{PHIXIWAVERAGEPF3PROP},
\begin{align*}
\int_{\scrB_1^{d-1}}\Phi(\xi,\vecw)\,d\vecw
=\vol(\S_1^{d-2})\,\frac{2^{1-d}}{(d-1)d(d+1)\zeta(d)}\xi^{-2}+
O\bigl(\xi^{-2-\frac2d}\bigr).
\end{align*}
This concludes the proof, in view of
\eqref{PHIXIAVFORMULA} and
$\vol(\S_1^{d-2})=2\pi^{\frac{d-1}2}\Gamma(\frac{d-1}2)^{-1}$.
\end{proof}

It now remains to prove Proposition \ref{PHIXIWAVERAGEPF3PROP}.
We need the following lemma.
\begin{lem}\label{PHIXIWASYMPTTHMPFLEM}
\begin{align}
\int_{P^{d-1}}\Xi\bigl(\vecy;\vece_1;1\bigr)\,d\vecy=1.
\end{align}
\end{lem}
\begin{proof}
Let $X_a=X_a^{(d-1)}\cong\ASL(d-1,\Z)\backslash\ASL(d-1,\R)$ 
be the space of affine lattices
(i.e.\ translates of lattices)
of covolume one in $\R^{d-1}$, endowed with its invariant probability measure
$\mu_a$ and standard projection $\pi:X_a\to X_1^{(d-1)}$.
Let $\widetilde{X}_a\subset X_a$ be the set of
$L\in X_a$ for which there is a unique point
$\vecy(L)$ which has minimal $\vece_1$-coordinate among the points in
$L\cap P^{d-1}$
(note that $\#(L\cap P^{d-1})=\infty$ for every $L\in X_a$).
Then $\widetilde{X}_a$ is open and of full measure in $X_a$,
and  $L\mapsto \langle \vecy(L),\pi(L)\rangle$ is a smooth injective map
from $\widetilde{X}_a$ onto a certain subset $\Omega$ of $P^{d-1}\times\myX$,
so that
\begin{align}\label{PHIXIWASYMPTTHMPFLEMPF1}
1=\mu_a(\widetilde{X}_a)=\int_{P^{d-1}}\int_{X_1} 
I\bigl(\langle\vecy,\Z^{d-1}M\rangle\in \Omega\bigr)\,d\mu(M)\,d\vecy.
\end{align}
Now for given $\vecy\in P^{d-1}$ and $M\in X_1$ we have
$\langle\vecy,\Z^{d-1}M\rangle\in \Omega$ if and only if 
$\Z^{d-1}M\cap P_{\vece_1}^{d-1}(\vecy)=\emptyset$ and 
$\Z^{d-1}M\cap (P^{d-1}-\vecy)\cap\{x_1=0\}=\{\bn\}$.
The last condition holds for $\mu$-almost all $M\in \myX$.
Hence the inner integral in \eqref{PHIXIWASYMPTTHMPFLEMPF1} equals
$\Xi(\vecy;\vece_1;1)$,
and this proves the lemma.
\end{proof}

\begin{proof}[Proof of Proposition \ref{PHIXIWAVERAGEPF3PROP}]
Using \eqref{XIDTINV}
with $T=T_{v^{-\frac1d},\frac12{\sigma}v^{-\frac1d}}$,
where $v=2^{1-\frac d2}t^{-\frac d2}y$ we get, for any 
$t,{\sigma}>0$, $0<y<1$,
\begin{align}\notag
\Xi\bigl({\sigma},2^{1-\frac d2}t^{-\frac{d}{2}}y\bigr)
=\Xi\bigl(a\vece_1+b\vece_2;\vece_1;1\bigr)
\end{align}
where (for any fixed $t>0$)
\begin{align}
\begin{cases}
a=a({\sigma},y)
=\bigl(\sfrac14{\sigma}^2+1\bigr)(2^{1-\frac2d}ty^{-\frac2d})-1
\\
b=b({\sigma},y)=2^{-\frac12-\frac1d}{\sigma}t^{\frac12}y^{-\frac1d}.
\end{cases}
\end{align}
One checks that the map $\langle {\sigma},y\rangle\mapsto\langle a,b\rangle$ is
a diffeomorphism from $(0,\infty)\times(0,1)$ onto
$\bigl\{\langle a,b\rangle\col b>0,\:a>b^2-1+2^{1-\frac2d}t\}$,
with the inverse map %
\begin{align}
\begin{cases}
{\sigma}=2b(a+1-b^2)^{-\frac12}
\\
y=2^{\frac d2-1}t^{\frac d2}(a+1-b^2)^{-\frac d2},
\end{cases}
\end{align}
and the Jacobian is
\begin{align}
\left|\frac{\partial({\sigma},y)}{\partial(a,b)}\right|
=2^{\frac d2-1}dt^{\frac d2}(a+1-b^2)^{-\frac{d+3}2}.
\end{align}
Hence \eqref{PHIXIWASYMPTTHMRES} may be rewritten as follows:
\begin{align}\notag
F_d(t)=\frac{2^{-1}d\pi^{\frac d2-1}}{(d-1)\Gamma(\frac d2-1)\zeta(d)}t^{d-1}
\int_0^\infty\int_{b^2-1+2^{1-\frac2d}t}^\infty
\Xi\bigl(a\vece_1+b\vece_2;\vece_1;1\bigr)
\hspace{50pt}
\\
\times
(a+1-b^2)^{-d}b^{d-3}
\Bigl(1-2^{\frac d2-1}t^{\frac d2}(a+1-b^2)^{-\frac d2}\Bigr)^{d-1}
\,da\,db.
\end{align}
Substituting this formula in $\int_0^\infty F_d(t)\,dt$ and then
changing the order of integration and using
\begin{align}
\int_0^{2^{\frac2d-1}(a+1-b^2)}
t^{d-1}\Bigl(1-2^{\frac d2-1}t^{\frac d2}(a+1-b^2)^{-\frac d2}\Bigr)^{d-1}\,dt
=\frac{2^{3-d}(a+1-b^2)^d}{d^2(d+1)},
\end{align}
we get
\begin{align}
\int_0^\infty F_d(t)\,dt=
\frac{2^{2-d}\pi^{\frac d2-1}}{(d-1)d(d+1)\Gamma(\frac d2-1)\zeta(d)}
\int_0^\infty\int_{b^2-1}^\infty
\Xi\bigl(a\vece_1+b\vece_2;\vece_1;1\bigr)
b^{d-3}\,da\,db.
\end{align}
Now use the fact that 
$\Xi\bigl(a\vece_1+b\vece_2;\vece_1;1\bigr)
=\Xi\bigl(a\vece_1+b\vecomega;\vece_1;1\bigr)$
for any $\vecomega\in\{0\}\times\S_1^{d-3}\subset\R^{d-1}$ and 
integrate over this sphere; this gives
\begin{align}
\int_0^\infty F_d(t)\,dt=
\frac{2^{1-d}}{(d-1)d(d+1)\zeta(d)}
\int_{P^{d-1}}\Xi\bigl(\vecy;\vece_1;1\bigr)\,d\vecy,
\end{align}
and now \eqref{PHIXIWAVERAGEPF3} is a consequence of 
Lemma \ref{PHIXIWASYMPTTHMPFLEM}.
\end{proof}

\subsection{\texorpdfstring{On the size and continuity of $\Xi(a,b;\vech;v)$}{On the size and continuity of Xi(a,b;h;v)}}

\begin{lem}\label{XIDm1TRIVBOUNDCOR}
For all $a>0$, $b\in\R$, $\vech=(h_1,\ldots,h_{d-1})\in\R_+^{d-1}$ 
and $v>0$ we have
\begin{align*}
\Xi(a,b;\vech;v)\ll
\min\biggl\{1,\Bigl(\Bigl(1+\frac{\|\vech'\|}{h_1}\Bigr)^{-d}v\Bigr)^{2-\frac2{d-1}},
\hspace{160pt}
\\
\Bigl(\Bigl(a+\frac{\|(2bh_1+h_2,h_3,\ldots,h_{d-1})\|}{h_1}\Bigr)^{-d}v
\Bigr)^{2-\frac2{d-1}}
\biggr\},
\end{align*}
where $\vech'=(h_2,\ldots,h_{d-1})$.
\end{lem}
\begin{proof}
It is clear from the definitions and \eqref{XIDTKINV} that
$\Xi(a,b;\vech;v)\leq\Xi(\bn;\vech;v)=
\Xi(\frac{\|\vech'\|}{h_1},v)$;
hence by Lemma \ref{XICVBOUNDLEM1} we have
\begin{align*}
\Xi(a,b;\vech;v)\ll
\min\biggl\{1,\Bigl(\Bigl(1+\frac{\|\vech'\|}{h_1}\Bigr)^{-d}v\Bigr)^{2-\frac2{d-1}}\biggr\}.
\end{align*}
The lemma follows from this bound by also using \eqref{XIDM1SYMM}.
\end{proof}

\begin{lem}\label{PARABOLOIDINCLLEM}
For any $x,y>0$ and $\alpha\geq\max(\frac yx,1)$ we have
\begin{align}\label{PARABOLOIDINCLLEMRES}
P^{d-1}-(y\vece_1+\vece_2)\subset
\alpha\bigl(P^{d-1}-(x\vece_1+\vece_2)\bigr).
\end{align}
\end{lem}
\begin{proof}
The point $\vecx=(x_1,\ldots,x_{d-1})$
belongs to the set in the left hand side of \eqref{PARABOLOIDINCLLEMRES}
if and only if
$x_1>(x_2+1)^2+x_3^2+\ldots+x_{d-1}^2-(1+y)$,
and it belongs to the set in the right hand side of 
\eqref{PARABOLOIDINCLLEMRES} if and only if
$\frac{x_1}{\alpha}>\bigl(\frac{x_2}{\alpha}+1\bigr)^2
+\bigl(\frac{x_3}{\alpha}\bigr)^2
+\ldots+\bigl(\frac{x_{d-1}}{\alpha}\bigr)^2-(1+x).$
Hence our task is to prove
\begin{align} 
\alpha\bigl((x_2+1)^2+x_3^2+\ldots+x_{d-1}^2-(1+y)\bigr)\geq
(x_2+\alpha)^2+x_3^2+\ldots+x_{d-1}^2-\alpha^2(1+x),
\end{align}
for all $x_2,\ldots,x_{d-1}\in\R$.
This simplifies to
\begin{align} 
(\alpha-1)\bigl(x_2^2+x_3^2+\ldots+x_{d-1}^2\bigr)
+\alpha(\alpha x-y)\geq0,
\end{align}
which is clear when $\alpha\geq\max(\frac yx,1)$.
\end{proof}

\begin{lem}\label{XIDM1INCRLEM}
For any fixed $a,b,\vech$, 
$\Xi(a,b;\vech;v)$ is an increasing function of $v\in\R_{>0}$.
\end{lem}
\begin{proof}
It suffices that $\Xi(\vecy,\vecy';\vech;v)$ is an increasing function
of $v$ for fixed $\vecy,\vecy'\in P^{d-1},\vech\in\R_+^{d-1}$.
This follows directly from the definition \eqref{XIDM1DEFnew}, since
\begin{align*}
v^{-\frac1{d-1}}(P^{d-1}_\vech(\vecy)\cup P^{d-1}_\vech(\vecy'))
\subset
{v'}^{-\frac1{d-1}}(P^{d-1}_\vech(\vecy)\cup P^{d-1}_\vech(\vecy'))
\end{align*}
whenever $v'\leq v$.
\end{proof}

\begin{lem}\label{XIDM1CONTLEM}
$\Xi(a,b;\vech;v)$ is a continuous function on 
$\R_{>0}\times\R\times\R_+^{d-1}\times\R_{>0}$.
\end{lem}
\begin{proof}
It suffices to prove that 
$\Xi(\vecy,\vecy';\vech;v)$ is continuous.
Fix $\langle\vecy,\vecy',\vech,v\rangle\in
P^{d-1}\times P^{d-1}\times\R_+^{d-1}\times\R_{>0}$ and let
$\langle\tilde\vecy,\tilde\vecy',\tilde\vech,\tilde v\rangle$ run through
a sequence of tuples in $P^{d-1}\times P^{d-1}\times\R_+^{d-1}\times\R_{>0}$
tending to $\langle\vecy,\vecy',\vech,v\rangle$.
Then 
\begin{align*}
\limsup_{\langle\tilde\vecy,\tilde\vecy',\tilde\vech,\tilde v\rangle}
\Bigl(
v^{-\frac1{d-1}}(P^{d-1}_\vech(\vecy)\cup P^{d-1}_\vech(\vecy'))
\:\triangle\:
\tilde v^{-\frac1{d-1}}(P^{d-1}_{\tilde\vech}(\tilde\vecy)
\cup P^{d-1}_{\tilde\vech}(\tilde\vecy'))
\Bigr)
\hspace{100pt}
\\
\subset
v^{-\frac1{d-1}}
\partial(P^{d-1}_\vech(\vecy)\cup P^{d-1}_\vech(\vecy')),
\end{align*}
and thus (\cite[Thm.\ 1.2.2]{friedman})
\begin{align*}
\limsup_{\langle\tilde\vecy,\tilde\vecy',\tilde\vech,\tilde v\rangle}
\vol_{d-1}\Bigl(
v^{-\frac1{d-1}}(P^{d-1}_\vech(\vecy)\cup P^{d-1}_\vech(\vecy'))
\:\triangle\:
\tilde v^{-\frac1{d-1}}(P^{d-1}_{\tilde\vech}(\tilde\vecy)
\cup P^{d-1}_{\tilde\vech}(\tilde\vecy'))
\Bigr)=0.
\end{align*}
Hence by \cite[Lemma 2.3]{lprob},
\begin{align*}
\limsup_{\langle\tilde\vecy,\tilde\vecy',\tilde\vech,\tilde v\rangle}
\bigl|\Xi\bigl(\tilde\vecy,\tilde\vecy';\tilde\vech;\tilde v\bigr)
-\Xi\bigl(\vecy,\vecy';\vech;v\bigr)\bigr|=0.
\end{align*}
\end{proof}

\begin{lem}\label{XIdm1CONTLEM3}
There is a constant $c>0$ which only depends on $d$ such that,
for all $0<a\leq10$, $b\in\R$, $v>0$, 
$\vech=(h_1,\ldots,h_{d-1})\in\R^{d-1}_{+}$
and $\tau_a,\tau_b\in[\frac12,\frac32]$, we have
\begin{align}\label{XIdm1CONTLEM3RES}
\Xi(a,b;\vech;v)\leq
\Xi\Bigl(\tau_aa,\tau_bb;\bigl(\tau_b^{-1}h_1,h_2,\ldots,h_{d-1}\bigr);
\bigl(1+c(|\tau_a-1|+|\tau_b-1|)\bigr)v\Bigr).
\end{align}
\end{lem}

\begin{proof}
Set $a'=\tau_aa$, $b'=\tau_bb$ and
$\veck=(\tau_b^{-1}h_1,h_2,\ldots,h_{d-1})$.
It suffices to prove the inequality for $b\neq0$
since the case $b=0$ then follows by continuity 
(Lemma \ref{XIDM1CONTLEM}).
By \eqref{XIDTINV} with $T=T_{-2/b',1}$ we have, for any $v'>0$,
\begin{align*}
\Xi(a',b';\veck;v')
=\mu\Bigl(\Bigl\{M\in \myX\col
\Z^{d-1}M\cap 
w'\bigl(P_{\vecell}^{d-1}(\vecy')\cup 
P_{\vecell}^{d-1}(\vecz')\bigr)
=\emptyset\Bigr\}\Bigr)
\end{align*}
where
$\vecell=\veck\trans M_{-2/b',1}^{-1}$, 
$\vecy'=\sfrac4{{b'}^2}\vece_1+\vece_2$, 
$\vecz'=\sfrac{4{a'}^2}{{b'}^2}\vece_1-\vece_2$ and
$w'=2^{-\frac d{d-1}}|b'|^{\frac d{d-1}}{v'}^{-\frac1{d-1}}$.
Similarly, using the fact that
$\vech\trans M_{-2/b,1}^{-1}\sim\vecell$,
\begin{align*}
\Xi(a,b;\vech;v)
=\mu\Bigl(\Bigl\{M\in \myX\col
\Z^{d-1}M\cap 
w\bigl(P_{\vecell}^{d-1}(\vecy)\cup 
P_{\vecell}^{d-1}(\vecz)\bigr)
=\emptyset\Bigr\}\Bigr)
\end{align*}
where 
$\vecy=\sfrac4{b^2}\vece_1+\vece_2$, 
$\vecz=\sfrac{4a^2}{b^2}\vece_1-\vece_2$ and
$w=2^{-\frac d{d-1}}|b|^{\frac d{d-1}}{v}^{-\frac1{d-1}}$.
Now take $\delta=\max(1,\tau_b^{-2},\tau_a^2\tau_b^{-2})$;
then by Lemma \ref{PARABOLOIDINCLLEM} we have
\begin{align*}
w'P_\vecell^{d-1}(\vecy')\subset\delta w'P_\vecell^{d-1}(\vecy)
\qquad\text{and}\qquad
w'P_\vecell^{d-1}(\vecz')\subset\delta w'P_\vecell^{d-1}(\vecz),
\end{align*}
and thus $\Xi(a,b;\vech;v)\leq\Xi(a',b';\veck;v')$ holds
so long as $w\geq\delta w'$, viz.\
$v\leq\delta^{1-d}\tau_b^{-d}v'$,
which is certainly true for $v'=(1+c(|\tau_a-1|+|\tau_b-1|))v$ with $c$
sufficiently large.
\end{proof}

\subsection{\texorpdfstring{Approximating packing probabilites with $\Xi(a,b;\vech;v)$}{Approximating packing probabilites with Xi(a,b;h;v)}}
\label{APPRPARABOLOIDSEC}

Just as in the proof of Theorem \ref{PHIXIWASYMPTTHM} a crucial step was to
approximate the probability of a random lattice avoiding a given
cut ball by the corresponding probability for a cut paraboloid
(cf.\ Sec.\ \ref{PARABOLOIDAPPRSEC}),
so it will be important in the proof of Theorem \ref{PHI0XILARGETHM} 
to carry out the
corresponding approximation for a union of \textit{two} given cut balls.
Given $\vech\in\R^{d-1}\setminus\{\bn\}$ 
and $\vecw,\vecz\in\scrB_1^{d-1}$ we define
\begin{align}\label{CHVECWDEF}
\fC_\vech(\vecw):=(\vecw+\scrB_1^{d-1})\cap\R^{d-1}_{\vech-}
\end{align}
(this agrees with \eqref{CHWDEF} if $\vecw=w\vece_1$) and set
\begin{align}\label{CCZWDEF}
\fC_{\vech}(\vecz,\vecw):=\fC_\vech(\vecz)\cup\fC_\vech(\vecw).
\end{align}
Finally let $\Upsilon(\vecz,\vecw,\vech,v)$ be the 
probability that 
a random lattice of covolume $v$ has empty intersection with
$\fC_\vech(\vecz,\vecw)$, viz.\
\begin{align}\label{UPSILONDEF}
\Upsilon(\vecz,\vecw,\vech,v):=\mu\bigl(\bigl\{M\in \myX^{(d-1)}
\col\Z^{d-1}M\cap 
v^{-\frac1{d-1}}\fC_{\vech}(\vecz,\vecw)=\emptyset\bigr\}\bigr).
\end{align}
Our goal in the present section is to 
approximate the function $\Upsilon(\vecz,\vecw,\vech,v)$
from above and below with the function $\Xi(a,b;\vech;v)$.

When proving Theorem \ref{PHI0XILARGETHM} we may without loss of generality
assume $w\geq z$, since
$\Phi_\bn(\xi,w,z,\varphi)=\Phi_\bn(\xi,z,w,\varphi)$
and also the right hand side of \eqref{PHI0XILARGETHMRES} 
is symmetric in $w,z$
(as will be seen from the definition
\eqref{PHI0XILARGETHMFDDEF}, using \eqref{XIDM1SYMM}).
For given $0<z\leq w<1$ and $\varphi\in[0,\frac\pi2)$,
we now make an explicit choice of corresponding vectors
$\vecw=(w_1,\ldots,w_{d-1})$, $\vecz=(z_1,\ldots,z_{d-1})\in\scrB_1^{d-1}$.
\begin{align}\notag
&\text{If }\: z>w\cos\varphi: &&
\begin{cases}
{\displaystyle 
\vecz=\frac1{\sqrt{z^2+w^2-2zw\cos\varphi}}\bigl(zw(\sin\varphi)\vece_1+
z(z-w\cos\varphi)\vece_2\bigr)}
\\
{\displaystyle 
\vecw=\frac1{\sqrt{z^2+w^2-2zw\cos\varphi}}\bigl(zw(\sin\varphi)\vece_1-
w(w-z\cos\varphi)\vece_2\bigr)};
\end{cases}
\\[-10pt]\label{CASE12FORMULA}
&\:
\\\notag
&\text{if }\: z\leq w\cos\varphi:&&
\begin{cases}\vecz=z\vece_1
\\
\vecw=w(\cos\varphi)\vece_1-w(\sin\varphi)\vece_2.
\end{cases}
\end{align}
These $\vecz,\vecw$ are easily verified to satisfy
$\|\vecw\|=w$, $\|\vecz\|=z$ and $\varphi(\vecz,\vecw)=\varphi$.
Note that $z>w\cos\varphi$ holds if and only if 
the triangle $\triangle\bn\vecz\vecw$ is acute;
the point of our choice of $\vecz,\vecw$ in this case is to make
$z_1=w_1$ hold.

For $\vecz,\vecw$ given as above we now wish to approximate
$\fC_\vech(\vecz,\vecw)$ by two cut paraboloids.
In fact we will use two translates of the same paraboloid $P_{u,r}$
(cf.\ Sec.\ \ref{PARABOLOIDAPPRSEC}),
with $u=\sqrt{1-w_1^2}$ and appropriate $r$.
This means that near the origin, $\vecw+P_{u,r}$ looks very much like
$\vecw+\scrB_1^{d-1}$ (cf.\ Lemma \ref{PUVDEFLEM} and note $w_2\leq0$).
Also $\vecz+P_{u,r}$ is in many cases a good approximation of
$\vecz+\scrB_1^{d-1}$ near the origin; however
if $z_2$ is near $\sqrt{1-z_1^2}$ then
$\vecz+P_{u,r}$ may even fail to contain the origin, and a much better
approximation of $\vecz+\scrB_1^{d-1}$ is given by
$\vecz+\rho(P_{u,r})$,
where $\rho:\R^{d-1}\to\R^{d-1}$ denotes reflection in 
the hyperplane $\vece_2^\perp$
(viz.\ $\rho((x_1,\ldots,x_{d-1}))=(x_1,-x_2,x_3,\ldots,x_{d-1})$).
With $A=A(u,r)$, $B=B(u,r)$ as in 
Section \ref{PARABOLOIDAPPRSEC} we have
\begin{align}\label{VECZPMOTIVATION}
\vecz+%
\rho(P_{u,r})=\bigl(\vecz+\sfrac BA\vece_2\bigr)+P_{u,r}.
\end{align}
In fact we will use $\vecz+\rho(P_{u,r})$ in place of
$\vecz+P_{u,r}$ if and only if $\varphi^2\geq3(1-z)$.
In view of the above relation, it is convenient to introduce the point
\begin{align}\label{ZPDEF}
\vecz'=\vecz'_{z,\varphi,u,r}
=(z_1',z_2',\ldots,z_{d-1}')
:=\begin{cases}
\vecz+\frac BA\vece_2&\text{if }\:\varphi^2\geq3(1-z)
\\
\vecz&\text{otherwise.}\end{cases}
\end{align}
Note that $\varphi^2\geq3(1-z)$
can only happen in the first of the two cases in
\eqref{CASE12FORMULA}, since 
$z\leq w\cos\varphi$ implies
$z\leq w\cos\varphi<\cos\varphi<1-\frac13\varphi^2$.
Note also that if $r=-u$ then $B=0$ and $\vecz'=\vecz$,
and by Lemma \ref{PARABOLOIDBALLAPPRLEM}(i) we have
\begin{align}\label{CHZWEASYINCL}
\fC_\vech(\vecz,\vecw)\subset\bigl((\vecz'+P_{u,-u})\cup(\vecw+P_{u,-u})\bigr)
\cap\R_{\vech-}^{d-1}.
\end{align}

The following lemma gives conditions for the opposite inclusion to hold:
\begin{lem}\label{CUTCONTAINEDLEM2}
There is an absolute constant $c_\clowJ>0$ such that whenever
$\frac9{10}\leq z\leq w<1$, $0\leq\varphi\leq\frac1{10}$, 
$\vech=(h_1,\ldots,h_{d-1})\in\R_+^{d-1}$,  %
$0\leq u<r\leq\frac12$, and 
\begin{align*}
r\geq {c_\clowJ}\Bigl(\sqrt{1-z}+\varphi+\frac{\|\vech'\|}{\|\vech\|}\Bigr)
\qquad\bigl(\text{with }\: \vech':=(h_2,\ldots,h_{d-1})\bigr),
\end{align*}
then we have,
for $\vecw,\vecz$ as in \eqref{CASE12FORMULA}
and $\vecz'$ as in \eqref{ZPDEF},
\begin{align}\label{CUTCONTAINEDLEM2R}
\bigl((\vecz'+P_{u,r})\cup(\vecw+P_{u,r})\bigr)\cap\R_{\vech-}^{d-1}
\subset\fC_\vech(\vecz,\vecw).
\end{align}
\end{lem}
\begin{proof}
Recalling the definition of $\fC_\vech(\vecz,\vecw)$ we see that
it suffices to prove
\begin{align}\label{CUTCONTAINEDLEM2RES}
(\vecz'+P_{u,r})\cap\R_{\vech-}^{d-1}\subset\vecz+\scrB_1^{d-1}
\qquad\text{and}\qquad
(\vecw+P_{u,r})\cap\R_{\vech-}^{d-1}\subset\vecw+\scrB_1^{d-1}.
\end{align}
In view of \eqref{VECZPMOTIVATION} and
Lemma \ref{PARABOLOIDBALLAPPRLEM}(ii),
in order to prove the first inclusion in \eqref{CUTCONTAINEDLEM2RES}
it suffices to prove 
\begin{align}\label{CUTCONTAINEDLEM2PF1}
(\vecz'+P_{u,r})\cap\R_{\vech-}^{d-1}\subset
\bigl\{x_1<z_1-\sqrt{1-r^2}\bigr\}.
\end{align}
Without loss of generality we may rescale $\vech$ so that $h_1=1$.
But we have by \eqref{PUVDEF},
\begin{align*}
P_{u,r}=(C-\sfrac{B^2}{4A})\vece_1-\sfrac B{2A}\vece_2
+\bigl\{x_1>A(x_2^2+\ldots+x_{d-1}^2)\bigr\}.
\end{align*}
Also the image of the halfspace $\R_{\vech-}^{d-1}$ 
under the translation
$\vecx\mapsto\vecx-\vecz'-(C-\sfrac{B^2}{4A})\vece_1+\sfrac B{2A}\vece_2$
is computed to equal $\lambda\vece_1+\R_{\vech-}^{d-1}$, where
\begin{align*}
\lambda:=-z_1-C+\sfrac{B^2}{4A}+h_2(\sfrac B{2A}-z_2').
\end{align*}
Hence \eqref{CUTCONTAINEDLEM2PF1} is equivalent with
\begin{align}\label{CUTCONTAINEDLEM2PF1a}
\bigl\{x_1>A(x_2^2+\ldots+x_{d-1}^2)\bigr\}
\cap\bigl(\lambda\vece_1+\R_{\vech-}^{d-1}\bigr)\subset
\bigl\{x_1<-C+\sfrac{B^2}{4A}-\sqrt{1-r^2}\bigr\}.
\end{align}
Note that $0<A\asymp1$ (cf.\ \eqref{PUVDEF} and recall 
$0\leq u<r\leq\frac12$).
Using the we compute that the supremum of $x_1$ taken over all
points $\vecx$ lying in the set 
in the left hand side of \eqref{CUTCONTAINEDLEM2PF1a} is
\begin{align}\label{CUTCONTAINEDLEM2PF1b}
=(4A)^{-1}\Bigl(\|\vech'\|+(\|\vech'\|^2+4A\lambda)^{1/2}\Bigr)^2
\leq\lambda
+O\Bigl(\|\vech'\|^2+\|\vech'\|\sqrt{\max(0,\lambda)}\Bigr).
\end{align}
(This presupposes $\|\vech'\|^2+4A\lambda>0$;
in the opposite case the set
in the left hand side of \eqref{CUTCONTAINEDLEM2PF1a} is empty,
so that the desired inclusion holds trivially.)

Now note that if $z>w\cos\varphi$ then 
$0\leq w-z<w(1-\cos\varphi)\leq\frac12\varphi^2$, and hence
\begin{align}\label{NORMALIZEPUVLEMPF3}
z^2+w^2-2zw\cos\varphi
=(w-z)^2+2zw(1-\cos\varphi)
=\varphi^2\bigl(1+O\bigl(1-z+\varphi^2\bigr)\bigr).
\end{align}
The same computation also shows
\begin{align}\label{NORMALIZEPUVLEMPF3a}
z^2+w^2-2zw\cos\varphi\gg\varphi^2.
\end{align}
Using \eqref{NORMALIZEPUVLEMPF3}, \eqref{NORMALIZEPUVLEMPF3a}
and \eqref{CASE12FORMULA} we obtain
\begin{align}\label{ZWFACT1}
1-z_1\asymp1-z_1^2\ll 1-z+\varphi^2.
\end{align}
This bound is obviously also true when $z\leq w\cos\varphi$,
cf.\ \eqref{CASE12FORMULA}, i.e.\ it is true in general.
We also get from \eqref{CASE12FORMULA} 
(using $z-w\cos\varphi\leq w(1-\cos\varphi)$),
\begin{align}\label{ZWFACT2}
0\leq z_2\ll\varphi.
\end{align}
Also note that 
${\EE}\ll r^2$, $A\asymp1$, $|B|\ll r^3$ and $C=-1+O(r^2)$
(cf.\ \eqref{PUVDEF} and recall $0\leq u<r\leq\frac12$),
and $\frac B{2A}-z_2'=(\pm\frac B{2A})-z_2$.

It follows from the above observations that 
\eqref{CUTCONTAINEDLEM2PF1b} is
\begin{align}\label{CUTCONTAINEDLEM2PF1c}
\leq-C+\sfrac{B^2}{4A}-1+O\Bigl(
\bigl(\sqrt{1-z}+\varphi+\|\vech'\|\bigr)
\bigl(\sqrt{1-z}+\varphi+\|\vech'\|+r\bigr)\Bigr).
\end{align}
This has been proved under the assumptions
$\frac9{10}\leq z\leq w<1$, $0\leq\varphi\leq\frac1{10}$, 
$\vech\in\R_+^{d-1}$, $h_1=1$ and $0\leq u<r\leq\frac12$.
Now if we \textit{also} assume
${c_\clowJ}(\sqrt{1-z}+\varphi+\frac{\|\vech'\|}{\|\vech\|})\leq r\leq\frac12$
where ${c_\clowJ}$ is a sufficiently large constant
(this in particular forces $\frac{\|\vech'\|}{\|\vech\|}$ to be small;
hence $\|\vech\|\asymp h_1=1$), then it follows that
the big-$O$-term in \eqref{CUTCONTAINEDLEM2PF1b} is
$<\frac12r^2$,
and since $\frac12r^2<\frac{r^2}{1+\sqrt{1-r^2}}=1-\sqrt{1-r^2}$
this implies that \eqref{CUTCONTAINEDLEM2PF1a} holds.

This completes the proof of the first inclusion in 
\eqref{CUTCONTAINEDLEM2RES}; the second inclusion is proved by a 
completely similar argument.
\end{proof}

Next, in order to relate
\eqref{CHZWEASYINCL} and \eqref{CUTCONTAINEDLEM2R} to the 
function $\Xi(a,b;\vech;v)$, we need to transform 
$((\vecz'+P_{u,r})\cup(\vecw+P_{u,r}))\cap\R_{\vech-}^{d-1}$
by a linear map into a union of the form
$P^{d-1}_\vech(\bn)\cup P^{d-1}_\vech(\bn T_{a,b})$,
cf.\ \eqref{XIDM1DEFnew} and \eqref{XIDM1ABDEF}.
The following lemma gives a detailed description of the
parameters occurring in this transformation.

\begin{lem}\label{NORMALIZEPUVLEM}
There is an absolute constant ${c_\clowK}\in(0,\frac1{10}]$ such that
for any fixed $z,w,\varphi$ with  %
$1-{c_\clowK}\leq z\leq w<1$, $0\leq\varphi\leq {c_\clowK}$,
there exist a function $M:[-{c_\clowK},{c_\clowK}]\to\GL(d-1,\R)$
and $\C^1$ functions
$a,\alpha:[-{c_\clowK},{c_\clowK}]\to\R_{>0}$,
$b:[-{c_\clowK},{c_\clowK}]\to\R_{\geq0}$,
$\beta:[-{c_\clowK},{c_\clowK}]\to\R$,
such that for $\vecz,\vecw$ as in \eqref{CASE12FORMULA}, 
$\vecz'$ as in \eqref{ZPDEF},
$u=\sqrt{1-w_1^2}$ and arbitrary 
$r\in[-{c_\clowK},{c_\clowK}]$ and $\vech\in\R_+^{d-1}$, we have,
writing $M=M(r)$, 
$a=a(r)$, $b=b(r)$, $\alpha=\alpha(r)$, $\beta=\beta(r)$:
\begin{align}\label{NORMALIZEPUVLEMRES}
\Bigl(((\vecz'+P_{u,r})\cup(\vecw+P_{u,r}))\cap\R_{\vech-}^{d-1}\Bigr)M
=P^{d-1}_\veck(\bn)\cup P^{d-1}_\veck(\bn T_{a,b}
)
\end{align}
where
\begin{align}\label{NORMALIZEPUVLEMRESVDEF}
\veck=(\alpha h_1,2\beta h_1+h_2,h_3,\ldots,h_{d-1}),
\end{align}
and furthermore, with $R=(1-z)+\varphi^2+r^2$
and with absolute implied constants:
\begin{align}\notag
&a=\sqrt{\frac{1-w}{1-z}}\bigl(1+O(R)\bigr);
&&\biggl|\frac{\partial a}{\partial r}\biggr|
\ll\sqrt{\frac{1-w}{1-z}}\sqrt R;
\\\notag
&b=\frac{\varphi}{\sqrt{2(1-z)}}\bigl(1+O(R)\bigr);
&&\biggl|\frac{\partial b}{\partial r}\biggr|
\ll\frac{\varphi}{\sqrt{1-z}}\sqrt R;
\\\label{NORMALIZEPUVLEMREL2}
&\alpha=\sqrt{\sfrac12(1-z)}\bigl(1+O(R)\bigr);
&&\biggl|\frac{\partial\alpha}{\partial r}\biggr|\ll\sqrt{1-z}\sqrt R;
\\\notag
&\bigl|\beta\bigr|\ll\sqrt{1-z+\varphi^2};
&&\biggl|\frac{\partial\beta}{\partial r}\biggr|\ll R;
\\\notag
&0<\det M=2^{1-\frac d2}(1-z)^{-\frac d2}\bigl(1+O(R)\bigr).
\end{align}
\end{lem}

As an auxiliary lemma,
let us first note the following regarding the functions
$A=A(u,r)$, $B=B(u,r)$, $C=C(u,r)$, ${\EE}={\EE}(u,r)$
introduced in Section \ref{PARABOLOIDAPPRSEC}. %
\begin{lem}\label{ABCEDERLEM}
The following bounds holds uniformly over all $u,r$ with $|u|,|r|\leq\frac12$:
\begin{align*}
&\biggl|\frac{\partial {\EE}}{\partial r}\biggr|\ll |u|+|r|,\qquad
\biggl|\frac{\partial A}{\partial r}\biggr|\ll |u|+|r|,\qquad
\biggl|\frac{\partial B}{\partial r}\biggr|\ll |u|\bigl(|u|+|r|\bigr),\qquad
\biggl|\frac{\partial C}{\partial r}\biggr|\ll u^2\bigl(|u|+|r|\bigr).
\end{align*}
\end{lem}
\begin{proof}
The bound $|\frac{\partial {\EE}}{\partial r}|\ll |u|+|r|$ is proved by a 
direct computation and the other three bounds follow trivially from this.
\end{proof}

\begin{proof}[Proof of Lemma \ref{NORMALIZEPUVLEM}]
Let $z,w,\varphi$ (and thus $\vecw,\vecz,\vecz'$) 
be given as in the formulation of the 
lemma, and set $u=\sqrt{1-w_1^2}$.
Also set $\epsilon:=-1$ if $\varphi^2\geq3(1-z)$ and otherwise
$\epsilon:=1$,
so that $\vecz'=\vecz+\sfrac12(1-\epsilon)\sfrac BA\vece_2$.
Now for any $r\in [-{c_\clowK},{c_\clowK}]$ we let $A=A(u,r)$, $B=B(u,r)$, $C=C(u,r)$, ${\EE}={\EE}(u,r)$
be as in Section \ref{PARABOLOIDAPPRSEC},
define $T_1\in\AGL(d-1,\R)$ by
\begin{align}\label{T1DEF}
\vecx T_1=(x_1,\ldots,x_{d-1})T_1
:=\bigl(Ax_1-1-AC+\sfrac14B^2,Ax_2+\sfrac12B,Ax_3,\ldots,Ax_{d-1}\bigr),
\end{align}
and set
\begin{align}%
\label{NORMALIZEPUVLEMPF10}
&\vecy:=(-\vecz')T_1=\bigl(-Az_1-1-AC+\sfrac14B^2\bigr)\vece_1
+\bigl(\epsilon\sfrac12B-Az_2\bigr)\vece_2;
\\\notag
&\vecy':=(-\vecw)T_1=\bigl(-Aw_1-1-AC+\sfrac14B^2\bigr)\vece_1
+\bigl(\sfrac12B-Aw_2\bigr)\vece_2.
\end{align}
Finally define $\alpha=\alpha(r)$, $\beta=\beta(r)$, $a=a(r)$, $b=b(r)$
as in Lemma~\ref{TALFBETGENNORMALIZELEM},
applied with $\vecy,\vecy'$ as in \eqref{NORMALIZEPUVLEMPF10}
(we will see that $\vecy,\vecy'$ lie in $P^{d-1}$ provided that
$c_\clowK$ is sufficiently small; cf.\ \eqref{NORMALIZEPUVLEMPF8}
and \eqref{NORMALIZEPUVLEMPF9} below),
and set $T=T_1T_{\alpha,\beta}^{-1}$ and
\begin{align}
M=M(r):=\text{diag}[A,\ldots,A]M_{\alpha,\beta}^{-1}.
\end{align}
Then by construction we have $T=(M,\vecxi)$ for some $\vecxi\in\R^{d-1}$;
$(-\vecz')T=\bn$;
$(-\vecw)T=\bn T_{a,b}$,
and $P_{u,r}T=P^{d-1}$.
Also by a quick computation one checks that
$\R_{\vech-}^{d-1}M=\R_{\veck-}^{d-1}$ for all
$\vech\in\R^{d-1}_+$,
with $\veck$ as in \eqref{NORMALIZEPUVLEMRESVDEF}.
Hence \eqref{NORMALIZEPUVLEMRES} holds.

It remains to verify that the functions $M,a,b,\alpha,\beta$ have
the properties stated in \eqref{NORMALIZEPUVLEMREL2}.
Writing $\vecy=y_1\vece_1+y_2\vece_2$ and
$\vecy'=y_1'\vece_1+y_2'\vece_2$, we compute
\begin{align}\label{NORMALIZEPUVLEMPF1}
&1+y_1-y_2^2
=A(-z_1-Az_2^2+\epsilon Bz_2-C)
=\frac A{2w_1}\bigl(1+w_1^2-2w_1z_1-z_2^2-{\EE}(u+\epsilon z_2)^2\bigr);
\\\notag
&1+y_1'-{y_2'}^2
=A(-w_1-Aw_2^2+Bw_2-C)
=\frac A{2w_1}\bigl(1 %
-w^2-{\EE}(u+w_2)^2\bigr).
\end{align}
Note that $w_1\geq z_1$ always holds,
and
recall \eqref{ZWFACT1} in the proof of Lemma \ref{CUTCONTAINEDLEM2}.
It follows that
\begin{align}\label{NORMALIZEPUVLEMPF2}
u^2=1-w_1^2\ll 1-z+\varphi^2.
\end{align}
Hence assuming that $c_\clowK$ has been taken sufficiently small
we have $0\leq u\leq\frac12$,  %
and since also $|r|\leq c_\clowK\leq\frac1{10}$, \eqref{EDEF} implies
\begin{align}\label{NORMALIZEPUVLEMPF7}
{\EE}\ll u^2+r^2.
\end{align}
We also get (cf.\ \eqref{PUVDEF})
\begin{align}\label{NORMALIZEPUVLEMPF6}
A=\frac12+O(R)
\qquad\text{and}\qquad
\frac A{w_1}=\frac12+O(R).
\end{align}
Now if $z>w\cos\varphi$ then we obtain from \eqref{NORMALIZEPUVLEMPF1}:
\begin{align}
1+y_1-y_2^2=\frac A{2w_1}\bigl(1-z^2-{\EE}(u+\epsilon z_2)^2\bigr),
\end{align}
and here if $\epsilon=-1$ then (since $0<z_2<u$)
\begin{align}
u+\epsilon z_2=\frac{1-z^2}{\sqrt{1-z^2+z_2^2}+z_2}\leq\sqrt{1-z^2},
\end{align}
while if $\epsilon=1$ (thus $\varphi^2\ll 1-z$) then we still have
$u+\epsilon z_2<2u  %
\ll\sqrt{1-z}$.
Hence always when $z>w\cos\varphi$ we get, using also 
\eqref{NORMALIZEPUVLEMPF7} and \eqref{NORMALIZEPUVLEMPF6},
\begin{align}\label{NORMALIZEPUVLEMPF8}
1+y_1-y_2^2=\sfrac12(1-z)\bigl(1+O(R)\bigr).
\end{align}
On the other hand if $z\leq w\cos\varphi$ then 
we have $\varphi^2<3(1-z)$ %
as noted below \eqref{ZPDEF},
and also $z_2=0$ and
\begin{align}
1+y_1-y_2^2=\frac A{2w_1}\bigl((1-w_1)^2+2w_1(1-z)-{\EE}u^2\bigr).
\end{align}
Hence using \eqref{NORMALIZEPUVLEMPF2}, \eqref{NORMALIZEPUVLEMPF7}
and \eqref{NORMALIZEPUVLEMPF6} we see that
\eqref{NORMALIZEPUVLEMPF8} again holds;
thus \eqref{NORMALIZEPUVLEMPF8} is true \textit{in general}.
By a similar discussion, also using Lemma \ref{ABCEDERLEM}, we find that,
both when $z>w\cos\varphi$ and when $z\leq w\cos\varphi$:
\begin{align}\label{NORMALIZEPUVLEMPF13}
\biggl|\frac{\partial}{\partial r}\bigl(1+y_1-y_2^2\bigr)\biggr|
\ll(1-z)\sqrt R.
\end{align}

By similar computations 
(using $u+w_2=\frac{1-w^2}{\sqrt{1-w_1^2}+|w_2|}\leq\sqrt{1-w^2}$)
we also have
\begin{align}\label{NORMALIZEPUVLEMPF9}
1+y_1'-{y_2'}^2=\sfrac12(1-w)\bigl(1+O(R)\bigr);
\qquad
\biggl|\frac{\partial}{\partial r}\bigl(1+y_1'-{y_2'}^2\bigr)\biggr|
\ll(1-w)\sqrt R.
\end{align}

We next study the difference $y_2'-y_2$.
By \eqref{NORMALIZEPUVLEMPF10},
\begin{align}\label{NORMALIZEPUVLEMPF11}
y_2'-y_2=\sfrac12(1-\epsilon)B+A(z_2-w_2).
\end{align}
If $z>w\cos\varphi$, then assuming $c_\clowK$ sufficiently small and using 
\eqref{NORMALIZEPUVLEMPF3} and \eqref{NORMALIZEPUVLEMPF6} we have
\begin{align}
A(z_2-w_2)=A\sqrt{z^2+w^2-2zw\cos\varphi}=\frac12\varphi\bigl(1+O(R)\bigr).
\end{align}
Note also that if $\epsilon=-1$ then $1-z\ll\varphi^2$ and thus
by \eqref{PUVDEF}, \eqref{NORMALIZEPUVLEMPF2} and \eqref{NORMALIZEPUVLEMPF7},
$|B|\ll {\EE}u\ll R\varphi$.
Hence always when $z>w\cos\varphi$ we have
\begin{align}\label{NORMALIZEPUVLEMPF12}
y_2'-y_2=\frac12\varphi\bigl(1+O(R)\bigr).
\end{align}
On the other hand if $z\leq w\cos\varphi$ then $\epsilon=1$
and $z_2-w_2=w\sin\varphi=\varphi(1+O(R))$;
hence %
\eqref{NORMALIZEPUVLEMPF12} again holds,
i.e.\ \eqref{NORMALIZEPUVLEMPF12} is true in general.
By similar considerations, also using Lemma~\ref{ABCEDERLEM}, we find that,
both when $z>w\cos\varphi$ and when $z\leq w\cos\varphi$,
\begin{align}\label{NORMALIZEPUVLEMPF14}
\biggl|\frac{\partial}{\partial r}\bigl(y_2'-y_2\bigr)\biggr|
\ll\varphi\sqrt R.
\end{align}

Note that by assuming ${c_\clowK}$ to be sufficiently small we can force the
big-$O$ terms in \eqref{NORMALIZEPUVLEMPF8} and \eqref{NORMALIZEPUVLEMPF9}
to be less than $\frac12$ in absolute value.
Hence the first three lines of \eqref{NORMALIZEPUVLEMREL2} now follow from 
Lemma \ref{TALFBETGENNORMALIZELEM}
combined with
\eqref{NORMALIZEPUVLEMPF8}, \eqref{NORMALIZEPUVLEMPF13},
\eqref{NORMALIZEPUVLEMPF9}, \eqref{NORMALIZEPUVLEMPF12}, 
\eqref{NORMALIZEPUVLEMPF14}.
The fourth line of \eqref{NORMALIZEPUVLEMREL2} follows from 
$\beta=y_2=\epsilon\frac12 B-Az_2$,
$|B|\ll u$, $A\ll1$, $0\leq z_2<u\ll\sqrt{1-z+\varphi^2}$,
and Lemma \ref{ABCEDERLEM}.
Finally the last line of \eqref{NORMALIZEPUVLEMREL2} follows by also using
$\det M=A^{d-1}\alpha^{-d}=A^{d-1}(1+y_1-y_2^2)^{-\frac d2}$,
cf.\ \eqref{TALFBETDEF} and
Lemma \ref{TALFBETGENNORMALIZELEM}. %
\end{proof}

In the next two propositions we give the desired 
approximations of $\Upsilon(\vecz,\vecw,\vech,v)$ 
in terms of the $\Xi$-function. 
We start with the approximation from above, 
which is in some respects more complicated than the one from below.

\begin{prop}\label{UPSILONMAINUPPERBOUNDPROP}
There exist constants $c_\clowL\in(0,\frac1{10}]$ and $c_\clowM>1$ 
which only depend on $d$ 
such that for any fixed $z,w,\varphi$ with
$1-c_\clowL\leq z\leq w<1$, $0\leq\varphi\leq c_\clowL$, 
there exist $\C^1$ functions 
$\alpha:[0,c_\clowL]\to\R_{>0}$ and $\beta:[0,c_\clowL]\to\R$ 
which satisfy the bounds
\begin{align}\label{UPSILONMAINUPPERBOUNDPROPRES2}
&\sfrac12\sqrt{1-z}
<\alpha(s)=\sqrt{\sfrac12(1-z)}\Bigl(1+O\bigl(1-z+\varphi^2+s^2\bigr)\Bigr);
&&\bigl|\alpha'(s)\bigr|\ll\sqrt{1-z}\sqrt{1-z+\varphi^2+s^2};
\\\notag
&\bigl|\beta(s)\bigr|\ll\sqrt{1-z+\varphi^2};
&&\bigl|\beta'(s)\bigr|\ll1-z+\varphi^2+s^2,
\end{align}
for all $s\in[0,c_\clowL]$,
and which have the property that, for 
$\vecz,\vecw$ as in \eqref{CASE12FORMULA}, and for
all $v>0$ and all
$\vech\in\R_+^{d-1}$ with $\frac{\|\vech'\|}{\|\vech\|}\leq c_\clowL$,
\begin{align}\notag
\Upsilon(\vecz,\vecw,\vech,v)
\leq\Xi\biggl(\sqrt{\frac{1-w}{1-z}},
\frac{\varphi}{\sqrt{2(1-z)}};
\Bigl(\alpha\bigl(\sfrac{\|\vech'\|}{\|\vech\|}\bigr)h_1,
2\beta\bigl(\sfrac{\|\vech'\|}{\|\vech\|}\bigr)h_1+h_2,h_3,\ldots,h_{d-1}
\Bigr);
\\\label{UPSILONMAINUPPERBOUNDPROPRES}
2^{1-\frac d2}(1-z)^{-\frac d2}v\Bigl\{1+c_\clowM 
\Bigl(1-z+\varphi^2+\sfrac{\|\vech'\|^2}{\|\vech\|^2}\Bigr)\Bigr\}\biggr).
\end{align}
\end{prop}

\begin{proof}
We will use the constants $c_\clowK$ from Lemma \ref{NORMALIZEPUVLEM}
and $c_\clowJ$ from Lemma \ref{CUTCONTAINEDLEM2}.
After possibly enlargening ${c_\clowJ}$ we may assume that
$\sqrt{1-w_1^2}<{c_\clowJ}(\sqrt{1-z}+\varphi)$ whenever
$1-{c_\clowK}\leq z\leq w<1$, $0\leq\varphi\leq {c_\clowK}$
(cf.\ \eqref{NORMALIZEPUVLEMPF2}).
Take $c_\clowL\in(0,{c_\clowK}]$ so small that
${c_\clowJ}(\sqrt{c_\clowL}+2c_\clowL)\leq {c_\clowK}$.

Now let $z,w,\varphi$ (and thus $\vecw,\vecz$) be given subject to
$1-c_\clowL\leq z\leq w<1$ and $0\leq\varphi\leq c_\clowL$.
Let us write $a,\alpha_0:[-{c_\clowK},{c_\clowK}]\to\R_{>0}$,
$b:[-{c_\clowK},{c_\clowK}]\to\R_{\geq0}$,
$\beta_0:[-{c_\clowK},{c_\clowK}]\to\R$ for the functions provided by
Lemma \ref{NORMALIZEPUVLEM}.
Now for $s\in[0,c_\clowL]$ we set
\begin{align*}
r=r(s):={c_\clowJ}(\sqrt{1-z}+\varphi+s);
\qquad
\alpha_1(s)=\alpha_0(r(s));\qquad\beta_1(s)=\beta_0(r(s)).
\end{align*}
(Note that $r\leq {c_\clowK}\leq\frac1{10}$ for all $s\in[0,c_\clowL]$ because
of our choice of $c_\clowL$; in particular 
$\alpha_1(s)$, $\beta_1(s)$ are well-defined for all $s\in[0,c_\clowL]$.)
We set $u=\sqrt{1-w_1^2}$ as in Lemma \ref{NORMALIZEPUVLEM};
note that $u<r(s)$ for all $s\in[0,c_\clowL]$, because of our choice of ${c_\clowJ}$.

Now let $\vech\in\R_+^{d-1}$ be given, subject to
$\frac{\|\vech'\|}{\|\vech\|}\leq c_\clowL$.
Set $s=\frac{\|\vech'\|}{\|\vech\|}$, $r=r(s)$ and
$M=M(r(s))$, where $M:[-{c_\clowK},{c_\clowK}]\to\GL(d-1,\R)$ is as in 
Lemma \ref{NORMALIZEPUVLEM}; then
\begin{align*}
\Bigl(((\vecz'+P_{u,r})\cup(\vecw+P_{u,r}))\cap\R_{\vech-}^{d-1}\Bigr)M
=P^{d-1}_\veck(\bn)\cup P^{d-1}_\veck(\bn T_{a,b})
\end{align*}
with $a=a(r)$, $b=b(r)$ and
\begin{align*}
\veck:=\Bigl(\alpha_1(s)h_1,2\beta_1(s)h_1+h_2,h_3,\ldots,h_{d-1}\Bigr).
\end{align*}
Also by Lemma \ref{CUTCONTAINEDLEM2} we have
\begin{align*}
\bigl((\vecz'+P_{u,r})\cup(\vecw+P_{u,r})\bigr)\cap\R_{\vech-}^{d-1}
\subset\fC_\vech(\vecz,\vecw).
\end{align*}
Hence, recalling \eqref{XIDM1DEFnew}, \eqref{XIDM1ABDEF},
\eqref{UPSILONDEF} and the fact that $\mu$ is $G$-invariant,
\begin{align}\label{UPSILONMAINUPPERBOUNDPROPPF1}
\Upsilon(\vecz,\vecw,\vech,v)\leq\Xi(a,b;\veck;(\det M)v).
\end{align}
Using now also Lemma \ref{XIdm1CONTLEM3} 
it follows, assuming that $\tau_a:=a^{-1}\sqrt{\frac{1-w}{1-z}}$
and $\tau_b:=b^{-1}\frac{\varphi}{\sqrt{2(1-z)}}$
(if $\varphi=0$: $\tau_b:=1$)
both lie in $[\frac12,\frac32]$,
that %
\begin{align}
\Xi(a,b;\veck;(\det M)v)
\leq
\Xi\biggl(\sqrt{\frac{1-w}{1-z}},
\frac{\varphi}{\sqrt{2(1-z)}};
\Bigl(\tau_b^{-1} %
k_1,k_2,\ldots,k_{d-1}\Bigr);v'\biggr),
\end{align}
where $v'=(\det M)v(1+O(|\tau_a-1|+|\tau_b-1|))$.

But from Lemma \ref{NORMALIZEPUVLEM} we know that
$|\tau_a-1|,|\tau_b-1|\ll1-z+\varphi^2+r^2$;
hence $|\tau_a-1|,|\tau_b-1|\ll1-z+\varphi^2+s^2$,
and thus after possibly shrinking $c_\clowL$ we can ensure that
$|\tau_a-1|,|\tau_b-1|\leq\frac12$ always hold for our $\vecw,\vecz,\vech$.
Hence we see, using also the fact that
$\det M=2^{1-\frac d2}(1-z)^{-\frac d2}(1+O(1-z+\varphi^2+s^2))$
(cf.\ Lemma \ref{NORMALIZEPUVLEM}),
that \eqref{UPSILONMAINUPPERBOUNDPROPRES} holds,
if we define
\begin{align*}
\alpha(s)=\frac{\sqrt{2(1-z)}b(r(s))}{\varphi}\alpha_1(s);
\qquad\text{and}\quad
\beta(s)=\beta_1(s),\quad\forall s\in[0,c_\clowL].
\end{align*}
(If $\varphi=0$ then $b\equiv0$, and we set $\alpha(s)=\alpha_1(s)$.)
Finally the properties %
in \eqref{UPSILONMAINUPPERBOUNDPROPRES2}
follow directly from our definitions and
Lemma \ref{NORMALIZEPUVLEM} (taking $c_\clowL$ sufficiently small).
\end{proof}

\begin{prop}\label{UPSILONMAINLOWBOUNDPROP}
There exist constants $c_\clowN\in(0,\frac1{10}]$ and
$c_\clowO>1$ which only depend on $d$ such that
for any fixed $z,w,\varphi$ with
$1-c_\clowN\leq z\leq w<1$, $0\leq\varphi\leq c_\clowN$,
there exist real numbers $\alpha,\beta$ which satisfy the bounds
\begin{align}\label{UPSILONMAINLOWERBOUNDPROPRES2}
\sfrac12\sqrt{1-z}<\alpha=\sqrt{\sfrac12(1-z)}\Bigl(1+O(1-z+\varphi^2)\Bigr);
\qquad|\beta|\ll\sqrt{1-z+\varphi^2},
\end{align}
and which have the property that,
for $\vecz,\vecw$ as in \eqref{CASE12FORMULA}, and for
all $v>0$ and all $\vech\in\R_+^{d-1}$,
\begin{align}\notag
\Upsilon(\vecz,\vecw,\vech,v)
\geq\Xi\biggl(\sqrt{\frac{1-w}{1-z}},
\frac{\varphi}{\sqrt{2(1-z)}};
\bigl(\alpha h_1,2\beta h_1+h_2,h_3,\ldots,h_{d-1}\bigr);\hspace{30pt}
\\\label{UPSILONMAINLOWERBOUNDPROPRES}
2^{1-\frac d2}(1-z)^{-\frac d2}v
\Bigl(1-c_\clowO\bigl(1-z+\varphi^2\bigr)\Bigr)^+\biggr),
\end{align}
where $x^+:=\max(0,x)$.
\end{prop}

\begin{proof}
This is very similar to the proof of 
Proposition \ref{UPSILONMAINUPPERBOUNDPROP},
except that we apply Lemma~\ref{NORMALIZEPUVLEM} with
$r=-u=-\sqrt{1-w_1^2}$ (independent of $\vech$);
in place of Lemma \ref{CUTCONTAINEDLEM2} we simply use 
\eqref{CHZWEASYINCL} to obtain
$\Upsilon(\vecz,\vecw,\vech,v)\geq\Xi(a,b;\veck;(\det M)v)$;
and finally we apply
Lemma \ref{XIdm1CONTLEM3} with 
$\tau_a:=a\sqrt{\frac{1-z}{1-w}}$
and $\tau_b:=b\frac{\sqrt{2(1-z)}}{\varphi}$,
and $\sqrt{\frac{1-z}{1-w}}$, $\frac{\sqrt{2(1-z)}}{\varphi}$
in place of $a,b$.
\end{proof}

\section{\texorpdfstring{Asymptotics for $\Phi_\bn(\xi,\vecw,\vecz)$ as $\xi\to\infty$}{Asymptotics for Phi0(xi,w,z) as xi tends to infinity}}
\label{PHI0XIWZASYMPTSEC}

We now start with the proof of Theorem \ref{PHI0XILARGETHM}.
The proof involves approximating $\Phi_\bn(\xi,\vecw,\vecz)$ with an
integral involving the $\Upsilon$-function
(cf.\ \eqref{PFPHI0XILARGETHMLOWERBOUND} below, as well as \eqref{UPSILONDEF}),
which is then estimated from above and below in terms of the 
$\Xi$-function, using Propositions \ref{UPSILONMAINUPPERBOUNDPROP}
and \ref{UPSILONMAINLOWBOUNDPROP}.
The resulting integral is then made cleaner in a series of steps,
eventually resulting in the function
$F_{\bn,d}$ %
which we define in \eqref{PHI0XILARGETHMFDDEF} below.

\subsection{Initial reductions}
\label{PFPHI0XILARGETHMSEC}

Note that if $d=2$ then Theorem \ref{PHI0XILARGETHM}
(with $F_{\bn,2}$ as in \eqref{F02EXPL}) follows directly from the 
explicit formula in \cite{partIII}.
Hence we will from now on assume $d\geq3$.

As pointed out in Section \ref{APPRPARABOLOIDSEC}
we may assume $w\geq z$ without loss of generality.
Let us fix the constant
$c_\clowD$ so that $c_\clowD\geq\sqrt{\sigma_d(1,0)}$
and $\Phi_\bn(\xi,w,z,\varphi)=0$ whenever
$1-z\geq c_\clowD\xi^{-\frac2d}$
(cf.\ Proposition \ref{PHI0SUPPORTTHM});
these conditions are %
equivalent to the conditions imposed on
$c_\clowD$ at the start of Section \ref{SIMPLINTSEC}.
It will be clear from the definition of $F_{\bn,d}$ in
\eqref{PHI0XILARGETHMFDDEF}
that $F_{\bn,d}(t_1,t_2,\alpha)=0$ holds whenever
$t_1\geq\sqrt{\sigma_d(1,0)}$
(for recall \eqref{RHODDEF}, \eqref{FDTCOMPSUPP}, and
$\Xi(a,b;\vech;v)\leq\Xi(\frac{\|\vech'\|}{h_1};v)$).
Thus %
\eqref{PHI0XILARGETHMRES}
is automatic when $1-z\geq c_\clowD\xi^{-\frac2d}$.
Hence from now on we will assume $1-z<c_\clowD\xi^{-\frac2d}$.

Let $c_\clowP$ be a positive constant which is smaller than both
\label{CLOWPREQUIREMENTS}
$c_\clowL$ from Proposition \ref{UPSILONMAINUPPERBOUNDPROP} 
and $c_\clowN$ from Proposition \ref{UPSILONMAINLOWBOUNDPROP}.
(We will later impose some further conditions on $c_\clowP$ being sufficiently
small, but we will see that it can be fixed in a way which only depends on
$d$.)
If $c_\clowP<\varphi\leq\frac\pi2$
then by Theorem \ref{CYLINDER2PTSMAINTHM} and 
\eqref{PHI0XILARGETHMFDDEF} coupled with Lemma \ref{SABVBOUNDLEMCOR} below,
both $\Phi_\bn(\xi,w,z,\varphi)$ and the main term in
the right hand side of \eqref{PHI0XILARGETHMRES}
are $\ll\xi^{-3+\frac2{d-1}}$,
and thus \eqref{PHI0XILARGETHMRES} is automatically true.
Hence from now on we will assume $0\leq\varphi\leq c_\clowP$.

By Lemma \ref{A1LARGELEM} there is a constant $0<c_\clowE<\frac12$ 
\label{FINALFIXINGOFCLOWE}
which only depends on 
$d$ such that for any $\xi>0$ and any translate $\fZ$ of the cylinder
$\xi^{\frac1d}\fZ(0,1,1)$, we have that
$a_1>A:=c_\clowE\xi^{\frac1d}$ holds for all
$M\in\Si_d$ with $\Z^dM\cap\fZ=\emptyset$.
We will assume from start that 
$\xi>\max(1,(c_\clowD/c_\clowP)^{d/2},c_\clowE^{-d})$;
in particular we have $1-c_\clowP<z<1$ and $A>1$.
For later reference we recapitulate our main assumptions on $z,w,\varphi$:
\begin{align}\label{ZWPHIBOUNDS}
1-z\geq c_\clowD\xi^{-\frac2d};\qquad
1-c_\clowP<z\leq w<1;\qquad %
0\leq\varphi\leq c_\clowP.
\end{align}
We fix $\vecz,\vecw$ as in \eqref{CASE12FORMULA}, for our given $z,w,\varphi$.

Let $\F_d\subset\Si_d'$ be a fundamental region for 
$\Gamma\backslash G$ as in Lemma \ref{FDCONTAINMENTSLEM}
(applied with our $A={c_\clowE}\xi^{\frac1d}$).
By \cite[(7.32)]{lprob} we have
\begin{align}\label{PHI0FIRSTSPLIT}
\Phi_\bn(\xi,\vecw,\vecz)=
\sum_{\veck\in\widehat\Z^d}\nu_\vecy\bigl(\bigl\{
M\in G_{\veck,\vecy}\cap\F_d\col \Z^dM\cap\fZ=\emptyset\bigr\}\bigr)
\end{align}
where $\fZ=\xi^{\frac1d}(\fZ(0,1,1)+(0,\vecz))$, 
$\vecy=\xi^{\frac1d}(1,\vecz+\vecw)$,
$\widehat\Z^d$ is the set of primitive vectors in $\Z^d$, and
$G_{\veck,\vecy}=\{M\in G\col\veck M=\vecy\}$.
Using the bound \cite[Prop.\ 7.3]{lprob} 
on the contribution from all $\veck$ with $k_1\neq1$ in 
\eqref{PHI0FIRSTSPLIT}, we get
\begin{align}\label{PHI0FIRSTSPLITcons}
\Phi_\bn(\xi,\vecw,\vecz)=
\sum_{\veck'\in\Z^{d-1}}\nu_\vecy\bigl(\bigl\{
M\in G_{\veck,\vecy}\cap\F_d\col \Z^dM\cap\fZ=\emptyset\bigr\}\bigr)
+O(E_1),
\end{align}
where we write $\veck=(1,\veck')$, and where
\begin{align}\label{E1DEF}
E_1:=\begin{cases}
\xi^{-2}\log(2+\min(\xi,\varphi^{-1}))&\text{if }\:d=3
\\
\xi^{-2}\min\bigl(1,(\xi\varphi^{d-2})^{-\frac{d-3}{d-1}}\bigr)
&\text{if }\:d\geq4.
\end{cases}
\end{align}
Using now Lemma \ref{FDCONTAINMENTSLEM} and the fact that
$a_1>A$ holds for all $M\in\F_d$ with $\Z^dM\cap\fZ=\emptyset$, we get:
\begin{align}\notag
\Phi_\bn(\xi,\vecw,\vecz)=
\sum_{\veck'\in\Z^{d-1}}\nu_\vecy\bigl(\bigl\{
M\in G_{\veck,\vecy}\cap\FG_A\col \Z^dM\cap\fZ=\emptyset\bigr\}\bigr)
\hspace{160pt}
\\\label{PHI0FIRSTSPLIT2}
+O\biggl(\sum_{\veck'\in\Z^{d-1}}\nu_\vecy\bigl(\bigl\{
M\in G_{\veck,\vecy}\cap(\Si_d'\cup\FG_A)\col \Z^dM\cap\fZ=\emptyset,
\: a_2\geq (c_\clowH^{(d-1)})^{-1}A\bigr\}\bigr)\biggr)
+O(E_1).
\end{align}

Following \cite{lprob}, we parametrize $G_{\veck,\vecy}$ 
(for $\veck=(1,\veck')$) as follows.
For any $\tM=\nn(\tu)\aa(\ta)\tkk\in G^{(d-1)}$ 
and $\vecv\in\S_1^{d-1}$ with $\vecy\cdot\vecv>0$
there is a unique choice of $a_1>0$, $\vecu\in\R^{d-1}$ such that
$[a_1,\vecv,\vecu,\tM]\in G_{\veck,\vecy}$, namely:
\begin{align}\label{UFORMULAIFK1NEQ0}  
a_1=\vecy\cdot\vecv;\qquad
\vecu
&=a_1^{\frac1{d-1}}
\iota^{-1}\bigl((\vecy-a_1\vecv)f(\vecv)^{-1}\bigr)
\tkk^{-1}\aa(\ta)^{-1}-\veck'\nn(\tu).
\end{align}
We write $[\vecv,\tM]_{\veck,\vecy}$
for this element
$[a_1,\vecv,\vecu,\tM]\in G_{\veck,\vecy}$.
This gives a bijective map
\begin{align*}
\{\vecv\in\S_1^{d-1}\col\vecy\cdot\vecv>0\}\times G^{(d-1)}
\ni\langle\vecv,\tM\rangle
\mapsto[\vecv,\tM]_{\veck,\vecy}\in G_{\veck,\vecy}.
\end{align*}
Let $L_{\vecv,\tM}$ be the lattice
\begin{align}\label{LVMDEF}
L_{\vecv,\tM}:=\Z^d[\vecv,\tM]_{\veck,\vecy}.
\end{align}
For any given $\vecv,\tM$ as above, $L_{\vecv,\tM}$ is in fact
independent of $\veck'\in\Z^{d-1}$
(for note that $a_1,\vecv,\tM$ are independent of $\veck'$, and
so is the congruence class of $\vecu\mod\Z^{d-1}\nn(\tu)$;
hence the claim follows from $M=\nn(u)\aa(a)\kk$ and \eqref{NAKSPLIT}).
Note also that for any $\vecv,\tM$ as above there is a unique choice of
$\veck'\in\Z^{d-1}$ which yields $\vecu\in(-\frac12,\frac12]^{d-1}$.
Using now the definition of $\FG_A$ (cf.\ \eqref{FGDEF}),
$a_2=a_1^{-\frac1{d-1}}\ta_1$,
and the expression for the measure $\nu_\vecy$ in the
parameters $\vecv,\tM$ (\cite[Lemma~5.2]{lprob}),
we conclude
\begin{align}\notag
\Phi_\bn&(\xi,\vecw,\vecz)=\zeta(d)^{-1}\int_{S}
\mu^{(d-1)}\Bigl(\Bigl\{\tM\in\F_{d-1}\col L_{\vecv,\tM}\cap\fZ
=\emptyset\Bigr\}\Bigr)
\,\frac{d\vecv}{(\vecy\cdot\vecv)^{d}}
\\\label{PHI0LOWBOUND2pre}
&+O\biggl(A^{-d}\int_S\mu^{(d-1)}\Bigl(\Bigl\{\tM\in\Si_{d-1}\col 
L_{\vecv,\tM}\cap\fZ=\emptyset,
\:\:\ta_1\geq (c_\clowH^{(d-1)})^{-1}A^{\frac d{d-1}}\Bigr\}\Bigr)\,d\vecv\biggr)
+O(E_1),
\end{align}
where
\begin{align}\label{SDEFmin}
S=\{\vecv\in\HS\setminus\{\vece_1\}\col a_1=\vecy\cdot\vecv>A\}.
\end{align}
Here the first error term is
\begin{align}\label{FIRSTERRORTERM}
\ll A^{-d}\int_S\mu^{(d-1)}\Bigl(\Bigl\{\tM\in\Si_{d-1}\col 
\Z^{d-1}\tM\cap a_1^{\frac1{d-1}}\fZ_\vecv=\emptyset,
\:\:\ta_1\geq (c_\clowH^{(d-1)})^{-1}A^{\frac d{d-1}}\Bigr\}\Bigr)\,d\vecv,
\end{align}
where $\fZ_\vecv=\iota^{-1}(\fZ f(\vecv)^{-1})$.
(For note that $L_{\vecv,\tM}\cap\fZ=\emptyset$ implies
$a_1^{-\frac1{d-1}}\iota(\Z^{d-1}\tM)f(\vecv)\cap\fZ=\emptyset$,
i.e.\ $\Z^{d-1}\tM\cap a_1^{\frac1{d-1}}\fZ_\vecv=\emptyset$;
cf.\ \eqref{LATTICEINPARAM}.)
But for each $\vecv\in S$, the set $a_1^{\frac1{d-1}}\fZ_\vecv\subset\R^{d-1}$ 
contains an open $(d-1)$-dimensional right cone with $\bn$ in its base,
of radius 
$\gg A^{\frac1{d-1}}\xi^{\frac1d}\sin\omega_\vecz\gg\xi^{\frac1{d-1}}\sin\omega_\vecz$
and height $\gg\xi^{\frac1{d-1}}\sin^2\omega_\vecz$,
where $\omega_\vecz$ is the angle between
$\vecv'=(v_2,\ldots,v_d)$ and $\vecz$ in $\R^{d-1}$, by
\cite[Lemma 7.1]{lprob}.
Hence using \cite[Lemma 7.4]{lprob} 
and a parametrization similar to \eqref{VPARA} but rotated to have
$\omega_\vecz$ in place of $\omega$, we see that \eqref{FIRSTERRORTERM} is
\begin{align*}
\ll A^{-d}\int_0^{\pi/2}
A^{-d}(A^{\frac d{d-1}}\xi^{\frac{d-2}{d-1}}\omega_\vecz^{d-1})^{\frac2{d-1}-1}
\,\omega_\vecz^{d-3}\,d\omega_\vecz
\ll\xi^{-3+\frac2{d-1}}\int_0^{\pi/2}d\omega_\vecz
\ll\xi^{-3+\frac2{d-1}}.
\end{align*}
Hence, since $\xi^{-3+\frac2{d-1}}\ll E_1$, we conclude
\begin{align}\label{PHI0LOWBOUND2}
\Phi_\bn(\xi,\vecw,\vecz)=\zeta(d)^{-1}\int_{S}
\mu\Bigl(\Bigl\{\tM\in\F_{d-1}\col L_{\vecv,\tM}\cap\fZ
=\emptyset\Bigr\}\Bigr)\,\frac{d\vecv}{(\vecy\cdot\vecv)^{d}}+O(E_1).
\end{align}

Next, by \cite[Prop.\ 7.5]{lprob}, at the cost of an error which is
$O(E_1)$ we may restrict the range of integration in \eqref{PHI0LOWBOUND2}
to the set 
\begin{align}\label{SPDEF}
S'=\Bigl\{\vecv\in\HS\setminus\{\vece_1\}\col a_1=\vecy\cdot\vecv>A,\:\:
v_1> c_\clowP^{-2}(\varphi+\omega)^2\Bigr\},
\end{align}
where $\omega=\varphi(\vecv',\vece_1)$ as in \eqref{VPARA}.
Recall that $c_\clowP<\min(c_\clowL,c_\clowN)\leq\frac1{10}$
(cf.\ p.\ \pageref{CLOWPREQUIREMENTS}).
As before we write $\vecv'=(v_2,\ldots,v_d)\in\R^{d-1}$ and
$\vecv''=(v_3,\ldots,v_d)\in\R^{d-2}$;
then $\|\vecv'\|=\sin\varpi$ and
$\|\vecv''\|=\sin\varpi\sin\omega$ (cf.\ \eqref{VPARA}).
Note that
$\vecv\in S'$ forces
$\omega<c_\clowP\leq\frac1{10}$ and thus $v_2=\sin\varpi\cos\omega>0$.
We also have
\begin{align}\label{SPFACT0}
\frac{\|\vecv''\|}{\|\vecv'\|}
=\sin\omega<c_\clowP\leq\sfrac1{10},\qquad
\forall \vecv\in S',
\end{align}
and thus also
\begin{align}\label{SPFACT1}
\|\vecv''\|\leq\sfrac1{10};\quad \:\text{ and }\quad
v_1+2v_2\geq\sqrt{v_1^2+v_2^2}\geq\sqrt{\sfrac{99}{100}}>\sfrac9{10},
\qquad
\forall \vecv\in S'.
\end{align}
We now impose the condition that $c_\clowP$ should be so small that
\eqref{ZWPHIBOUNDS} forces 
\begin{align}\label{CCLOWPREQUIREMENT}
\|\vecz-\vece_1\|<\sfrac1{20}\quad\text{and}\quad
\|\vecw-\vece_1\|<\sfrac1{20}
\end{align}
(this is clearly possible, since both $\varphi(\vecz,\vece_1)$
and $\varphi(\vecw,\vece_1)$ are always $\leq\varphi$
in \eqref{CASE12FORMULA}).
Then also $\|\vecz+\vecw-2\vece_1\|<\frac1{10}$,
and since $\vecy=\xi^{\frac1d}(1,\vecz+\vecw)$ it
follows that, for all $\vecv\in S'$,
\begin{align}\notag
\xi^{-\frac1d}a_1=\xi^{-\frac1d}(\vecy\cdot\vecv)
=v_1+2\vece_1\cdot\vecv'+(\vecz+\vecw-2\vece_1)\cdot\vecv'
\hspace{50pt}
\\\label{A1LOWBOUND}
\geq v_1+2v_2-\sfrac1{10}\|\vecv'\|
>\sfrac9{10}-\sfrac1{10}=\sfrac45.
\end{align}
Note that this was derived without ever using the condition
$\vecy\cdot\vecv>A$ in \eqref{SPDEF};
hence that condition is in fact redundant
(since $A=c_\clowE\xi^{\frac1d}<\frac12\xi^{\frac1d}$), i.e.\ we have
\begin{align}\label{SPDEFnew}
S'=\Bigl\{\vecv\in\HS\col v_1> c_\clowP^{-2}(\varphi+\omega)^2\Bigr\}.
\end{align}

Recall that we always have
$L_{\vecv,\tM}\subset\cup_{n\in\Z}(na_1\vecv+\vecv^\perp)$,
cf.\ \eqref{LATTICECONTAINEMENT}.
The following lemma implies that
$(na_1\vecv+\vecv^\perp)\cap\fZ=\emptyset$ 
holds for all $\vecv\in S'$ and all $n\in\Z\setminus\{0,1\}$.
\begin{lem}\label{GOODINTLEM2}
For any $\vecv\in\S_1^{d-1}$ with $0<v_1<1$ and 
$\varphi(\vecv',\vece_1)\leq\frac1{10}$,
and for any $\vecz,\vecw\in\scrB_1^{d-1}$ with
$\|\vecz-\vece_1\|<\frac1{10}$, $\|\vecw-\vece_1\|<\frac1{10}$,
\begin{align}\label{GOODINTLEM2RES}
(n(1,\vecz+\vecw)+\vecv^\perp)\cap
(\fZ(0,1,1)+(0,\vecz))=\emptyset,
\qquad\forall n\in\Z\setminus\{0,1\}.
\end{align}
\end{lem}
\begin{proof}
Since $n(1,\vecz+\vecw)+\vecv^\perp$ 
has nonempty intersection with
$\fZ(0,1,1)+(0,\vecz)$ for $n=0,1$,
and $\fZ(0,1,1)+(0,\vecz)$ is convex, it suffices to prove that
\eqref{GOODINTLEM2RES} holds for $n=-1$ and for $n=2$.
We have $\fZ(0,1,1)+(0,\vecz)=\fZ'+(\frac12,\vecz)$
where $\fZ':=\fZ(-\frac12,\frac12,1)$; hence
$\fZ(0,1,1)+(0,\vecz)$ has nonempty intersection with
$n(1,\vecz+\vecw)+\vecv^\perp$ 
if and only if $\fZ'$ has nonempty intersection with
\begin{align*}
n(1,\vecz+\vecw)-(\sfrac12,\vecz)+\vecv^\perp
=\bigl((n-\sfrac12)v_1+(n-1)\vecz\cdot\vecv'+n\vecw\cdot\vecv'\bigr)\vecv
+\vecv^\perp.
\end{align*}
Hence by Lemma \ref{CYLPLANELEMMA} our task is to prove that
$\bigl|(n-\sfrac12)v_1+(n-1)\vecz\cdot\vecv'+n\vecw\cdot\vecv'\bigr|\geq
\sfrac12v_1+\|\vecv'\|$ for $n=-1,2$, i.e.\ it suffices to prove that
\begin{align}\label{GOODINTLEM2PF1}
v_1+\vecz\cdot\vecv'+2\vecw\cdot\vecv'\geq\|\vecv'\|
\qquad\text{and}\qquad
v_1+2\vecz\cdot\vecv'+\vecw\cdot\vecv'\geq\|\vecv'\|.
\end{align}
However this is clear since the assumptions of the lemma imply
$\varphi(\vecz,\vecv')<\frac13\pi$ and thus
$\vecz\cdot\vecv'>\frac12\|\vecz\|\|\vecv'\|>\frac9{20}\|\vecv'\|$,
and similarly $\vecw\cdot\vecv'>\frac9{20}\|\vecv'\|$.
\end{proof}

Let us also note that
\begin{align}\label{TWOSPECIALSUBLATTICES}
\vecv^\perp\cap L_{\vecv,\tM}
=a_1^{-\frac1{d-1}}\iota(\Z^{d-1}\tM)f(\vecv)
\quad\text{and}\quad
(a_1\vecv+\vecv^\perp)\cap L_{\vecv,\tM}
=\vecy+a_1^{-\frac1{d-1}}\iota(\Z^{d-1}\tM)f(\vecv).
\end{align}
Indeed these relations follow from \eqref{LATTICEINPARAM} with $n=0,1$, 
if we also note that 
for any $\veck=(1,\veck')$ ($\veck'\in\Z^{d-1}$),
if $M=[\vecv,\tM]_{\veck,\vecy}$ then
$(\vece_1+\iota(\Z^{d-1}))M=(\veck+\iota(\Z^{d-1}))M
=\vecy+\iota(\Z^{d-1})M$,
since $M\in G_{\veck,\vecy}$.
Taking also Lemma \ref{GOODINTLEM2} into account we conclude that
for any $\vecv\in S'$ we have the equivalence:
\begin{align}\label{DISJOINTIFFTWOPLANESDISJ}
L_{\vecv,\tM}\cap\fZ=\emptyset
\quad\Longleftrightarrow\quad
a_1^{-\frac1{d-1}}\iota(\Z^{d-1}\tM)f(\vecv)
\cap\bigl(\fZ\cup(\vecy-\fZ)\bigr)=\emptyset.
\end{align}
Note here that 
\begin{align*}
\vecy-\fZ=\xi^{\frac1d}\bigl(\fZ(0,1,1)+(0,\vecw)\bigr).
\end{align*}
Hence by the same argument as we used to get \eqref{CVRELATION},
we have for every $\vecv\in S'$:
\begin{align}\notag
\mu\Bigl(\Bigl\{\tM\in\F_{d-1}\col L_{\vecv,\tM}\cap\fZ=\emptyset\Bigr\}\Bigr)
\hspace{190pt}
\\\label{BASICWZINEQ}
\geq
\mu\Bigl(\Bigl\{M\in \myX^{(d-1)}\col
\Z^{d-1}M\cap(\xi^{\frac1d}v_1^{-\frac1{d-1}}a_1^{\frac1{d-1}})
\fC_{\vecv'}(\vecz,\vecw)=\emptyset\Bigr\}\Bigr)\hspace{30pt}
\\\notag
=\Upsilon\bigl(\vecz,\vecw,\vecv',\xi^{-1+\frac1d}v_1a_1^{-1}\bigr),
\end{align}
cf.\ \eqref{CCZWDEF} and \eqref{UPSILONDEF}.
Furthermore we have \textit{equality} in \eqref{BASICWZINEQ}
whenever both $\vecv^\perp\cap\fZ\subset\{x_1<\frac12\xi^{\frac1d}\}$
and $\vecv^\perp\cap(\vecy-\fZ)\subset\{x_1<\frac12\xi^{\frac1d}\}$ hold.

Now by Lemma \ref{GOODINTLEM} (applied after appropriate rotations),
if $\vecv\in S'$ does \textit{not} satisfy both
$\vecv^\perp\cap\fZ\subset\{x_1<\frac12\xi^{\frac1d}\}$
and $\vecv^\perp\cap(\vecy-\fZ)\subset\{x_1<\frac12\xi^{\frac1d}\}$,
then
\begin{align}\label{BADINTCOND}
\sfrac\pi2-\varpi\ll1-z+\omega_\vecz^2
\qquad\text{or}\qquad
\sfrac\pi2-\varpi\ll1-w+\omega_\vecw^2,
\end{align}
where from now on we write $\omega_\vecw$ for the angle
between $\vecv'$ and $\vecw$, and
(as before) $\omega_\vecz$ for the angle
between $\vecv'$ and $\vecz$.
Since $w\geq z$, we see that \eqref{BADINTCOND} forces 
\begin{align*}
\sfrac\pi2-\varpi\ll\max(\omega_\vecz,\omega_\vecw)^2
\qquad\text{or}\qquad
\sfrac\pi2-\varpi\ll1-z
\end{align*}
to hold (the implied constant depends only on $d$).
But recall the definition of $S'$, \eqref{SPDEF}, and note that
$\frac\pi2-\varpi\asymp v_1$ and
$\max(\omega_\vecz,\omega_\vecw)\leq\omega+\varphi$,
by the triangle inequality in $\S_1^{d-2}$.
Hence by choosing $c_\clowP$ sufficiently small we can ensure that
every $\vecv\in S'$ which does not satisfy both
$\vecv^\perp\cap\fZ\subset\{x_1<\frac12\xi^{\frac1d}\}$
and $\vecv^\perp\cap(\vecy-\fZ)\subset\{x_1<\frac12\xi^{\frac1d}\}$
must in fact satisfy $\frac\pi2-\varpi\ll1-z$.
Now the total contribution from these $\vecv\in S'$ to the
integral in \eqref{PHI0LOWBOUND2} can be bounded 
by following the proof of \eqref{CYLINDER2PTSMAINTHMRES} in 
\cite[Sec.\ 7.2]{lprob}
(the ``$\Sigma_1$-part'', for $\varphi\leq\frac\pi2$),
but restricting the integration by $\frac\pi2-\varpi\ll1-z$; 
cf.\ especially \cite[(7.22)-(7.23)]{lprob};
it follows that this contribution is:
\begin{align}\label{PHI0LOWBOUND2badcontr}
\ll(1-z)\xi^{-2+\frac2d}\min(1,(\xi\varphi^d)^{-1+\frac2{d(d-1)}}).
\end{align}
By Proposition \ref{PHI0SUPPORTTHM} we know that 
$\Phi_\bn(\xi,w,z,\varphi)>0$ implies
\begin{align}\label{PHI0SUPPORTTHMinterpr}
\max(1-z,1-w)\ll\xi^{-\frac2d}
\min\bigl(1,(\xi\varphi^d)^{-\frac2{d(d-1)}}\bigr).
\end{align}
Hence if $\Phi_\bn(\xi,\vecw,\vecz)>0$ then
\eqref{PHI0LOWBOUND2badcontr} is $\ll E_2$, where
\begin{align}\label{E2DEF}
E_2:=\xi^{-2}\min(1,(\xi\varphi^d)^{-1}),
\end{align}
and it follows that:
\begin{align}\label{PFPHI0XILARGETHMLOWERBOUND}
\Phi_\bn(\xi,\vecw,\vecz)
=\zeta(d)^{-1}\int_{\vecv\in S'}
\Upsilon(\vecz,\vecw,\vecv',\xi^{-\frac{d-1}d}v_1a_1^{-1})
\,\frac{d\vecv}{(\vecy\cdot\vecv)^d}+O(E_1)+O(E_2).
\end{align}
This is in fact true \textit{in general},
for in the remaining case when $\Phi_\bn(\xi,\vecw,\vecz)=0$, 
\eqref{PFPHI0XILARGETHMLOWERBOUND} is an obvious consequence of 
\eqref{PHI0LOWBOUND2} and the inequality \eqref{BASICWZINEQ}.
(In \eqref{PFPHI0XILARGETHMLOWERBOUND} we may note that 
$E_2\ll E_1$ if $d=3$, but if $d\geq4$ then
$E_2\ll E_1$ holds if and only if $\varphi\ll\xi^{-\frac1{d-2}}$
or $\varphi\gg\xi^{-\frac1{2d-3}}$.
However in the end we will just use $E_1,E_2\ll E$,
cf.\ \eqref{PHI0XILARGETHMEDEF}.)

Next, using \eqref{ZWFACT1}, \eqref{ZWFACT2}
and their analogues for $\vecw$, we have for every $\vecv\in S'$:
\begin{align}\label{A1ASYMPT}
a_1=\vecv\cdot\vecy
=\xi^{\frac1d}\bigl(v_1+2v_2+(z_1+w_1-2)v_2+(z_2+w_2)v_3\bigr)
\hspace{60pt}
\\\notag
=\xi^{\frac1d}\bigl(v_1+2v_2+O\bigl(1-z+\varphi^2+\|\vecv''\|\varphi
\bigr)\bigr).
\end{align}
Since $v_1+2v_2>\frac9{10}$ and $a_1>\frac45\xi^{\frac1d}$
for all $\vecv\in S'$
(cf.\ \eqref{SPFACT1} and \eqref{A1LOWBOUND}),
\eqref{A1ASYMPT} implies
\begin{align}\label{A1INVASYMPT}
a_1^{-1}=(\vecv\cdot\vecy)^{-1}=\xi^{-\frac1d}\frac1{v_1+2v_2}
\bigl(1+O(1-z+\varphi^2+\|\vecv''\|\varphi)\bigr).
\end{align}
Now because of
\eqref{ZWPHIBOUNDS}, \eqref{SPFACT0} and $c_\clowP<c_\clowL$,
Proposition \ref{UPSILONMAINUPPERBOUNDPROP} can be applied to
bound the integrand in \eqref{PFPHI0XILARGETHMLOWERBOUND} from above,
throughout the range of integration.
It follows (using also \eqref{A1INVASYMPT} and the monotonicity
property of $\Xi$ in its last argument, Lemma \ref{XIDM1INCRLEM})
that there is a constant $c_\clowQ$ which only depends on $d$ such that 
\begin{align}\notag
\Phi_\bn(\xi,\vecw,\vecz)
\leq\frac{\xi^{-1}}{\zeta(d)}\int_{\vecv\in S'}
\Xi\biggl(\sqrt{\frac{1-w}{1-z}},
\frac{\varphi}{\sqrt{2(1-z)}};\vech;
\hspace{120pt}
\\\label{PHI0XILARGETHMPF1}
2^{1-\frac d2}(1-z)^{-\frac d2}\xi^{-1}\frac{v_1}{v_1+2v_2}\Bigl\{1+c_\clowQ
\Bigl(1-z+\varphi^2+\sfrac{\|\vecv''\|^2}{\|\vecv'\|^2}\Bigr)\Bigr\}\biggr)
\hspace{40pt}
\\\notag
\times\Bigl(1+O\bigl(1-z+\varphi^2+\|\vecv''\|\varphi\bigr)\Bigr)
\,\frac{d\vecv}{(v_1+2v_2)^d}+O(E_1+E_2),
\end{align}
where 
\begin{align}\label{HCHOICE}
\vech=(h_1,\ldots,h_{d-1})=\Bigl(
\alpha\bigl(\sfrac{\|\vecv''\|}{\|\vecv'\|}\bigr)v_2,
2\beta\bigl(\sfrac{\|\vecv''\|}{\|\vecv'\|}\bigr)v_2+v_3,v_4,\ldots,v_d\Bigr),
\end{align}
and where $\alpha,\beta$ are functions in $\C^1([0,c_\clowL])$
satisfying \eqref{UPSILONMAINUPPERBOUNDPROPRES2}.

Similarly, by Proposition \ref{UPSILONMAINLOWBOUNDPROP},
after possibly increasing $c_\clowQ$ we \textit{also} have
\begin{align}\notag
\Phi_\bn(\xi,\vecw,\vecz)
\geq\frac{\xi^{-1}}{\zeta(d)}\int_{\vecv\in S'}
\Xi\biggl(\sqrt{\frac{1-w}{1-z}},
\frac{\varphi}{\sqrt{2(1-z)}};\vech;
\hspace{170pt}
\\\label{PHI0XILARGETHMPF1a}
2^{1-\frac d2}(1-z)^{-\frac d2}\xi^{-1}\frac{v_1}{v_1+2v_2}\Bigl\{1-c_\clowQ
\Bigl(1-z+\varphi^2+\|\vecv''\|\varphi\Bigr)\Bigr\}^+\biggr)
\hspace{40pt}
\\\notag
\times\Bigl(1-O\bigl(1-z+\varphi^2+\|\vecv''\|\varphi\bigr)\Bigr)
\,\frac{d\vecv}{(v_1+2v_2)^d}-O(E_1+E_2),
\end{align}
where 
\begin{align}
\vech=(h_1,\ldots,h_{d-1})=\bigl(
\alpha v_2,2\beta v_2+v_3,v_4,\ldots,v_d\bigr),
\end{align}
where, this time, 
$\alpha$ and $\beta$ are real \textit{numbers}
satisfying \eqref{UPSILONMAINLOWERBOUNDPROPRES2}.

\begin{remark}\label{LOWERBOUNDREMARK1}
In fact the error term ``$-O(E_1+E_2)$''
in \eqref{PHI0XILARGETHMPF1a} may be improved to 
``$-O(\xi^{-3+\frac2{d-1}})$''.
Indeed, note that the error term in \eqref{PHI0FIRSTSPLITcons}
is non-negative; hence by going through the proof of
\eqref{PHI0LOWBOUND2} but only aiming for a lower bound we obtain
\begin{align*}
\Phi_\bn(\xi,\vecw,\vecz)\geq\zeta(d)^{-1}\int_{S}
\mu\Bigl(\Bigl\{\tM\in\F_{d-1}\col L_{\vecv,\tM}\cap\fZ
=\emptyset\Bigr\}\Bigr)\,\frac{d\vecv}{(\vecy\cdot\vecv)^{d}}
-O(\xi^{-3+\frac2{d-1}}).
\end{align*}
Hence, using \eqref{BASICWZINEQ} we get
\begin{align*}
\Phi_\bn(\xi,\vecw,\vecz)
\geq\zeta(d)^{-1}\int_{\vecv\in S'}
\Upsilon(\vecz,\vecw,\vecv',\xi^{-\frac{d-1}d}v_1a_1^{-1})
\,\frac{d\vecv}{(\vecy\cdot\vecv)^d}-O(\xi^{-3+\frac2{d-1}}),
\end{align*}
and now our claim follows by applying Proposition
\ref{UPSILONMAINLOWBOUNDPROP}.
\end{remark}

\subsection{The main term}\label{PHI0MAINTERMSEC}

We wish to simplify the integrals in \eqref{PHI0XILARGETHMPF1}
and \eqref{PHI0XILARGETHMPF1a}.
We will first discuss \eqref{PHI0XILARGETHMPF1};
the treatment of \eqref{PHI0XILARGETHMPF1a} is similar but easier,
as we explain at the end of this section.
To start with, we introduce new variables of integration 
$\veceta=(\eta_1,\ldots,\eta_{d-1})$ via
\begin{align}\label{ETAFROMV}
\eta_j=\frac{2v_{j+1}}{v_1+2v_2}\qquad (j=1,\ldots,d-1).
\end{align}
A quick computation shows that this formula defines a diffeomorphism
\begin{align*}
\S_1^{d-1}\cap\{v_1,v_2>0\}\ni\vecv\mapsto\veceta
\in(0,1)\times\R^{d-2},
\end{align*}
with inverse given by
\begin{align}\label{VFROMETA}
\vecv=\frac1{\sqrt{4(1-\eta_1)^2+\|\veceta\|^2}}
\bigl(2(1-\eta_1),\eta_1,\eta_2,\eta_3,\ldots,\eta_{d-1}\bigr).
\end{align}
Now $\frac{\|\vecv''\|^2}{\|\vecv'\|^2}=
\frac{\|\veceta'\|^2}{\|\veceta\|^2}$ for all 
$\vecv\in\S_1^{d-1}\cap\{v_1,v_2>0\}$,
where we write $\veceta':=(\eta_2,\ldots,\eta_{d-1})$.
Also for all $\vecv\in S'$ we have (cf.\ \eqref{SPDEF},
and recall $\vecv\in S'\Rightarrow v_2>0$):
\begin{align*}
\frac{\|\veceta'\|^2}{1+\|\veceta'\|^2}<
\frac{\|\veceta'\|^2}{\|\veceta\|^2}
=\frac{\|\vecv''\|^2}{\|\vecv'\|^2}
<c_\clowP^2 v_1<\frac{1-\eta_1}{\sqrt{4(1-\eta_1)^2+\|\veceta\|^2}}
<\frac1{\sqrt{1+\|\veceta'\|^2}},
\end{align*}
and this forces
\begin{align*}
\|\veceta'\|<2.
\end{align*}
Hence $\sqrt{4(1-\eta_1)^2+\|\veceta\|^2}\asymp1$ and
$v_1\asymp1-\eta_1$ and $\|\vecv''\|\asymp\|\veceta'\|$
for all $\vecv\in S'$.

Also note that for $\vecv,\veceta$ related by the above diffeomorphism,
we have for any $j\geq2$, $\ell\geq1$:
\begin{align*}
\frac{\partial v_j}{\partial\eta_\ell}
=\frac{\partial}{\partial\eta_\ell}
\Bigl(\sfrac{\eta_{j-1}}{\sqrt{4(1-\eta_1)^2+\|\veceta\|^2}}\Bigr)
=\sfrac{\delta_{\ell,j-1}}{\sqrt{4(1-\eta_1)^2+\|\veceta\|^2}}
-\sfrac{\eta_{j-1}}{2(4(1-\eta_1)^2+\|\veceta\|^2)^{3/2}}
\left.\begin{cases}10\eta_1-8&\text{if }\ell=1
\\ 2\eta_\ell&\text{if }\ell\geq2\end{cases}\right\}\hspace{20pt}
\\
=\sfrac12(v_1+2v_2)\left(\delta_{\ell,j-1}
-v_j\left.\begin{cases}v_2-2v_1&\text{if }\ell=1
\\ v_{\ell+1}&\text{if }\ell\geq2\end{cases}\right\}\right).
\end{align*}
Hence for $\vecv\in S'$ we have,
using also \eqref{SPFACT1}
and $\frac{\partial v_2}{\partial\eta_1}=
\frac12(v_1+2v_2)^2v_1(1+\frac{\|\vecv''\|^2}{v_1(v_1+2v_2)})$:
\begin{align*}
\frac{\partial v_j}{\partial\eta_\ell}
=\begin{cases}
\frac12(v_1+2v_2)^2v_1
\bigl(1+O(v_1^{-1}\|\vecv''\|^2)\bigr)&\text{if }\: j=2,\ell=1
\\
O(\|\vecv''\|) &\text{if }\:[j=2,\ell>1]\text{ or }[j>2,\ell=1]
\\
\frac12(v_1+2v_2)(\delta_{\ell,j-1}+O(\|\vecv''\|^2))&\text{if }\: j>2,\ell>1.
\end{cases}
\end{align*}
Hence the Jacobian is:
\begin{align*}
\frac{\partial(v_2,\ldots,v_d)}{\partial(\eta_1,\ldots,\eta_{d-1})}
=2^{1-d}(v_1+2v_2)^dv_1\Bigl(1+O\Bigl(\frac{\|\vecv''\|^2}{v_1}\Bigr)\Bigr).
\end{align*}
Finally recall that if we parametrize $\HS$
with $v_2,\ldots,v_d$ then %
$d\vecv=v_1^{-1}dv_2dv_3\cdots dv_d$.

In view of the above observations, \eqref{PHI0XILARGETHMPF1} now becomes
\begin{align}\notag
\Phi_\bn(\xi,\vecw,\vecz)
\leq\frac{2^{1-d}\xi^{-1}}{\zeta(d)}\int_{S''}
\Xi\Bigl(a,b;\vech;
\kappa(1-\eta_1)
\Bigl\{1+c_\clowQ
\Bigl(1-z+\varphi^2+\sfrac{\|\veceta'\|^2}{\|\veceta\|^2}\Bigr)\Bigr\}
\Bigr)
\hspace{30pt}
\\\label{PHI0XILARGETHMPF2}
\times\Bigl(1+O\Bigl(1-z+\varphi^2+\frac{\|\veceta'\|^2}{1-\eta_1}\Bigr)\Bigr)\,
d\veceta
+O(E_1+E_2),
\end{align}
where from now on we write
\begin{align}\label{ABKAPPADEF}
a=\sqrt{\frac{1-w}{1-z}},
\qquad
b=\frac{\varphi}{\sqrt{2(1-z)}},
\qquad
\kappa=2^{1-\frac d2}\xi^{-1}(1-z)^{-\frac d2},
\end{align}
and where $d\veceta=d\eta_1\cdots d\eta_{d-1}$,
and $S''$ is the set of all \label{SPPDEF}
$\veceta\in(0,1)\times\R^{d-2}$ which correspond to $\vecv\in S'$
under our diffeomorphism.
Let us write
$\alpha=\alpha(\frac{\|\veceta'\|}{\|\veceta\|})$
and $\beta=\beta(\frac{\|\veceta'\|}{\|\veceta\|})$ for short, 
and recall \eqref{HCHOICE}. Also recall that the dependence of
$\Xi$ on its third argument is only up to proportionality.
Hence in \eqref{PHI0XILARGETHMPF2} we may take
\begin{align}\label{HCHOICE2}
\vech %
=(\eta_1,\alpha^{-1}(2\beta\eta_1+\eta_2),\alpha^{-1}\eta_3,\ldots,
\alpha^{-1}\eta_{d-1}).
\end{align}

We will next carry out one more substitution,
taking $h_1,\ldots,h_{d-1}$ as new variables of integration.
Note that \eqref{HCHOICE2} defines a $\C^1$ function
$\veceta\mapsto\vech$ for all $\veceta$ in the open cone
\begin{align}\label{OMEGAOPENCONE}
\Omega=\bigl\{\veceta\in\R^{d-1}\col
0<\eta_1<1,\:\|\veceta'\|<(c_\clowP^{-2}-1)^{-\frac12}\eta_1\bigr\}
\end{align}
(since $\|\veceta'\|<(c_\clowP^{-2}-1)^{-\frac12}\eta_1
\Rightarrow\frac{\|\veceta'\|}{\|\veceta\|}<c_\clowP<c_\clowL$).
It follows from \eqref{SPFACT0} that $S''\subset\Omega$.
We have
\begin{align*}
\frac{\partial}{\partial\eta_1}\Bigl(\frac{\|\veceta'\|}{\|\veceta\|}\Bigr)
=-\frac{\eta_1\|\veceta'\|}{\|\veceta\|^3}
\qquad\text{and}\qquad
\frac{\partial}{\partial\eta_j}\Bigl(\frac{\|\veceta'\|}{\|\veceta\|}\Bigr)
=\frac{\eta_1^2\eta_j}{\|\veceta\|^3\|\veceta'\|}
\quad(j\geq2).
\end{align*}
Hence for all $\veceta\in\Omega$ we have, 
using also \eqref{UPSILONMAINUPPERBOUNDPROPRES2}
and $\|\veceta\|\ll\eta_1$:
\begin{align*}
&\biggl|\frac{\partial}{\partial\eta_1}\biggl(\alpha
\Bigl(\frac{\|\veceta'\|}{\|\veceta\|}\Bigr)^{-1}\biggr)\biggr|
\ll (1-z)^{-\frac12}\Bigl(1-z+\varphi^2+\frac{\|\veceta'\|^2}{\|\veceta\|^2}
\Bigr)^{\frac12}\frac{\|\veceta'\|}{\|\veceta\|^2};
\\
&\biggl|\frac{\partial}{\partial\eta_j}\biggl(\alpha
\Bigl(\frac{\|\veceta'\|}{\|\veceta\|}\Bigr)^{-1}\biggr)\biggr|
\ll (1-z)^{-\frac12}\Bigl(1-z+\varphi^2+\frac{\|\veceta'\|^2}{\|\veceta\|^2}
\Bigr)^{\frac12}\frac{1}{\|\veceta\|}
&&(\forall j\geq2);
\\
&\biggl|\frac{\partial}{\partial\eta_j}\biggl(\beta
\Bigl(\frac{\|\veceta'\|}{\|\veceta\|}\Bigr)\biggr)\biggr|
\ll\Bigl(1-z+\varphi^2+\frac{\|\veceta'\|^2}{\|\veceta\|^2}\Bigr)
\frac{1}{\|\veceta\|}
&&(\forall j\geq2).
\end{align*}
Using these bounds and \eqref{UPSILONMAINUPPERBOUNDPROPRES2} we obtain,
for all $\veceta\in\Omega$:
\begin{align}\label{PHI0XILARGETHMPF3}
&\frac{\partial h_k}{\partial\eta_j}
=\sqrt{\frac2{1-z}}\Bigl(\delta_{kj}+O
\Bigl(1-z+\varphi^2+\frac{\|\veceta'\|^2}{\|\veceta\|^2}\Bigr)\Bigr)
&&(\forall k,j\geq2).
\end{align}
It follows that if $1-z+\varphi^2+\sup_{\veceta\in\Omega}
(\frac{\|\veceta'\|^2}{\|\veceta\|^2})$ is sufficiently small
then our map $\veceta\mapsto\vech$ is %
injective on all of $\Omega$.
Because of \eqref{ZWPHIBOUNDS} and \eqref{OMEGAOPENCONE} we can ensure this
by requiring the constant $c_\clowP$ to be sufficiently small.
In particular it now follows that the map $\veceta\mapsto\vech$
restricts to a diffeomorphism from $S''$ onto some open set
$S'''\subset(0,1)\times\R^{d-2}$. \label{SPPPDEF}

Using \eqref{PHI0XILARGETHMPF3} 
together with $\frac{\partial h_1}{\partial\eta_j}=\delta_{1j}$ we get
\begin{align}\label{PHI0XILARGETHMPF5} 
\frac{\partial(h_1,\ldots,h_{d-1})}{\partial(\eta_1,\ldots,\eta_{d-1})}
=\Bigl(\frac2{1-z}\Bigr)^{\frac d2-1}
\Bigl(1+O\Bigl(1-z+\varphi^2+\frac{\|\veceta'\|^2}{\|\veceta\|^2}\Bigr)\Bigr).
\end{align}
By requiring $c_\clowP$ to be sufficiently small we may assume that the
big-$O$ term in \eqref{PHI0XILARGETHMPF5} is $<\frac12$, say.
We now obtain from \eqref{PHI0XILARGETHMPF2}: 
\begin{align}\notag
\Phi_\bn(\xi,\vecw,\vecz)
\leq\frac{2^{2-\frac32d}\xi^{-1}(1-z)^{\frac d2-1}}{\zeta(d)}\int_{S'''}
\Xi\Bigl(a,b;\vech;
\kappa(1-h_1)\Bigl\{1+c_\clowQ
\Bigl(1-z+\varphi^2+\sfrac{\|\veceta'\|^2}{\|\veceta\|^2}\Bigr)\Bigr\}\Bigr)
\hspace{5pt}
\\\label{PHI0XILARGETHMPF4}
\times
\Bigl(1+O\Bigl(1-z+\varphi^2+\frac{\|\veceta'\|^2}{1-\eta_1}
+\frac{\|\veceta'\|^2}{\|\veceta\|^2}\Bigr)\Bigr)\,
d\vech+O(E_1+E_2).
\end{align}

The treatment of \eqref{PHI0XILARGETHMPF1a} is quite similar:
By exactly the same argument as that leading to \eqref{PHI0XILARGETHMPF2}
we obtain, after possibly increasing $c_\clowQ$,
\begin{align}\notag
\Phi_\bn(\xi,\vecw,\vecz)
\geq\frac{2^{1-d}\xi^{-1}}{\zeta(d)}\int_{S''}
\Xi\Bigl(a,b;\vech;\kappa(1-\eta_1)
\bigl(1-c_\clowQ\bigl(1-z+\varphi^2+\|\veceta'\|\varphi\bigr)\bigr)^+\Bigr)
\hspace{60pt}
\\\label{PHI0XILARGETHMPF2a}
\times\Bigl(1-O\Bigl(1-z+\varphi^2+\frac{\|\veceta'\|^2}{1-\eta_1}\Bigr)\Bigr)\,
d\veceta
-O(E_1+E_2),
\end{align}
where $S''$ is the same as in \eqref{PHI0XILARGETHMPF2},
and $\vech$ is as in \eqref{HCHOICE2}, where this time
$\alpha,\beta$ are real \textit{numbers}
satisfying \eqref{UPSILONMAINLOWERBOUNDPROPRES2}.
Just as before we now carry out one more substitution, taking 
$h_1,\ldots,h_{d-1}$ as new variables of integration.
Since $\alpha,\beta$ are now constants independent of $\veceta$,
this transformation is much less complicated than before;
in fact it is just a (non-singular) \textit{linear} map
$\R^{d-1}\ni\veceta\mapsto\vech\in\R^{d-1}$,
with Jacobian
\begin{align*}
\frac{\partial(h_1,\ldots,h_{d-1})}{\partial(\eta_1,\ldots,\eta_{d-1})}
=\alpha^{2-d}
=\Bigl(\frac2{1-z}\Bigr)^{\frac d2-1}
\bigl(1+O\bigl(1-z+\varphi^2\bigr)\bigr)
\end{align*}
(cf.\ \eqref{UPSILONMAINLOWERBOUNDPROPRES2}).
Hence we get
\begin{align}\notag
\Phi_\bn(\xi,\vecw,\vecz)
\geq\frac{2^{2-\frac32d}\xi^{-1}(1-z)^{\frac d2-1}}{\zeta(d)}\int_{S'''}
\Xi\Bigl(a,b;\vech;
\kappa(1-h_1)\bigl(1-c_\clowQ\bigl(1-z+\varphi^2+\|\veceta'\|\varphi\bigr)\bigr)^+\Bigr)
\hspace{5pt}
\\\label{PHI0XILARGETHMPF4a}
\times
\Bigl(1-O\Bigl(1-z+\varphi^2+\frac{\|\veceta'\|^2}{1-\eta_1}\Bigr)\Bigr)\,
d\vech-O(E_1+E_2),
\end{align}
where (again) $S'''\subset(0,1)\times\R^{d-2}$ is the image
of $S''$ under our (this time linear) map $\veceta\mapsto\vech$.

\subsection{\texorpdfstring{Bounds on some integrals involving $\Xi$}{Bounds on some integrals involving Xi}}

In order to bound the contribution from the various error terms in
\eqref{PHI0XILARGETHMPF4} we first prove some auxiliary bounds on 
integrals involving the $\Xi$-function.

\begin{lem}\label{SABVBOUNDLEM}
We have, uniformly over all $a>0$, $b\geq0$, $v>0$,
\begin{align}\label{SABVBOUNDLEMRES}
\int_{\R^{d-2}}\Xi\bigl(a,b;(1,\vecu);v\bigr)\,d\vecu
\ll\min\Bigl\{v^{1-\frac2d},%
v^{2-\frac2{d-1}}b^{-d+\frac2{d-1}}
\Bigr\}.
\end{align}
(When $b=0$ the right hand side should be interpreted as ``$v^{1-\frac2d}$''.)
\end{lem}
\begin{remark}
Note that $v^{1-\frac2d}\leq v^{2-\frac2{d-1}}b^{-d+\frac2{d-1}}$
if and only if $b\leq v^{\frac1d}$.
\end{remark}
\begin{proof}
By Lemma \ref{XIDm1TRIVBOUNDCOR} we have
\begin{align*}
\Xi(a,b;(1,\vecu);v)\ll
\min\biggl\{1,\Bigl(\Bigl(
\max(u_1^2,(2b+u_1)^2)
+u_2^2+\ldots+u_{d-2}^2\Bigr)^{-\frac d2}v\Bigr)^{2-\frac2{d-1}}
\biggr\}.
\end{align*}
Let us first assume $d\geq4$.
By noting the symmetry
$u_1\leftrightarrow-2b-u_1$, and
using polar coordinates for the remaining variables
$(u_2,\ldots,u_{d-2})=r\vecomega$ ($\vecomega\in\S_1^{d-4}$),
we get
\begin{align*}
&\int_{\R^{d-2}}\Xi\bigl(a,b;(1,\vecu);v\bigr)\,d\vecu
\ll
\int_{b}^\infty\int_0^\infty
\min\biggl\{1,\Bigl(\bigl(u_1^2+r^2\bigr)^{-\frac d2}v\Bigr)^{2-\frac2{d-1}}
\biggr\}r^{d-4}\,dr\,du_1
\\
&\ll\int_0^\infty\int_{r+b}^\infty
\min\Bigl\{1,\bigl(x^{-d}v\bigr)^{2-\frac2{d-1}}\Bigr\}\,dx\,r^{d-4}\,dr
\\
&\ll\int_0^\infty v^{\frac1d}\min\Bigl\{1,
\bigl((r+b)v^{-\frac1d}\bigr)^{-\frac{2d^2-5d+1}{d-1}}\Bigr\}\,r^{d-4}\,dr.
\end{align*}
This is always
\begin{align*}
&\ll v^{\frac1d}\int_0^\infty 
\bigl((r+b)v^{-\frac1d}\bigr)^{-\frac{2d^2-5d+1}{d-1}}\,r^{d-4}\,dr
\\
&\ll v^{2-\frac{2}{d-1}}\biggl(
\int_0^{b}b^{-\frac{2d^2-5d+1}{d-1}}r^{d-4}\,dr
+\int_{b}^\infty r^{-\frac{2d^2-5d+1}{d-1}+d-4}\,dr\biggr)
\ll v^{2-\frac{2}{d-1}}b^{-d+\frac2{d-1}},
\end{align*}
where we used $-\frac{2d^2-5d+1}{d-1}+d-3=-d+\frac2{d-1}<0$ in the last step.
On the other hand if $b<v^{\frac1d}$ then we can do better as follows:
\begin{align*}
\ll v^{\frac1d}\int_0^{v^{\frac1d}}r^{d-4}\,dr
+v^{2-\frac2{d-1}}\int_{v^{\frac1d}}^\infty
r^{-\frac{2d^2-5d+1}{d-1}+d-4}\,dr
\ll v^{1-\frac2d}.
\end{align*}

In the remaining case $d=3$ we get instead
\begin{align*}
&\int_{\R^{d-2}}\Xi\bigl(a,b;(1,\vecu);v\bigr)\,d\vecu
\ll
\int_{b}^\infty
\min\bigl\{1,u_1^{-3}v\bigr\}\,du_1\ll\min\bigl(v^{\frac13},vb^{-2}\bigr),
\end{align*}
which again agrees with \eqref{SABVBOUNDLEMRES}.
\end{proof}

\begin{lem}\label{SABVBOUNDLEMCOR}
We have, uniformly over all $a>0$, $b\geq0$, $v>0$,
\begin{align*}
\int_{\vech\in(0,1)\times\R^{d-2}}
\Xi(a,b;\vech;(1-h_1)v)\,d\vech
\ll\min\Bigl\{v^{1-\frac2d},
v^{2-\frac2{d-1}}b^{-d+\frac2{d-1}}
\Bigr\}.
\end{align*}
\end{lem}
\begin{proof}
Writing $\vech=h_1(1,\vecu)$ we get
\begin{align*}
\int_{\vech\in(0,1)\times\R^{d-2}}
\Xi(a,b;\vech;(1-h_1)v)\,d\vech
=\int_0^1h_1^{d-2}
\int_{\R^{d-2}}\Xi(a,b;(1,\vecu);(1-h_1)v)\,d\vecu\,dh_1.
\end{align*}
Hence the lemma follows from Lemma \ref{SABVBOUNDLEM}.
\end{proof}
\begin{lem}\label{TOUGHXIINTLEM}
We have, uniformly over all
$a>0$, $b\geq0$, $v\geq1$, $0\leq \rho\leq\min(b^{-2},v^{-\frac2d})$,
\begin{align}\notag
\int_{(0,1)\times\R^{d-2}}
\Xi\Bigl(a,b;\vech;v(1-h_1)\Bigr)
\min\Bigl(1,\rho\frac{(bh_1)^2+\|\vech\|^2}{h_1^2(1-h_1)}\Bigr)\,d\vech
\\\label{TOUGHXIINTLEMRES}
\ll\begin{cases}
\rho v\log\Bigl(2+\rho^{-1}\min(b^{-2},v^{-2/3})\Bigr)&\text{if }\: d=3
\\
\rho v\min\Bigl(1,(vb^{-d})^{1-\frac2{d-1}}\Bigr)
&\text{if }\: d\geq4.
\end{cases}
\end{align}
\end{lem}
\begin{proof}
We first assume $d\geq4$.
Substituting $\vech=(1-t)(1,u_1,r\vecomega)$ with 
$\vecomega\in\S_1^{d-4}$,
and using the bound on $\Xi$ from Lemma \ref{XIDm1TRIVBOUNDCOR},
similarly as in the proof of Lemma \ref{SABVBOUNDLEM},
we see that the left hand side of \eqref{TOUGHXIINTLEMRES} is
\begin{align}\notag
\ll\int_0^\infty\int_{-\infty}^\infty\int_0^1
\min\biggl\{1,\Bigl(\Bigl(\max(u_1^2,(2b+u_1)^2)+r^2
\Bigr)^{-\frac d2}vt\Bigr)^{2-\frac2{d-1}}\biggr\}
\hspace{100pt}&
\\\label{TOUGHXIINTLEMPF1}
\times\min\biggl\{1,\frac{\rho(1+b^2+u_1^2+r^2)}t\biggr\}\,dt\,du_1
\,r^{d-4}\,dr.
\end{align}
We now prove an auxiliary result:
\begin{lem}\label{TECHNICALLEM}
For any fixed $\delta>0$ we have, uniformly over all $A>0$, $B\geq0$:
\begin{align}\label{TECHNICALLEMRES}
\int_0^1\min\Bigl\{1,At^\delta\Bigr\}
\min\Bigl\{1,\frac Bt\Bigr\}\,dt
\asymp\min\bigl(1,A\bigr)\min\Bigl\{1,B\log\Bigl(2+\frac1{A^{-1}+B}\Bigr)\Bigr\}.
\end{align}
\end{lem}
\begin{proof}
First assume $A\leq10^\delta$.
Then the left hand side of \eqref{TECHNICALLEMRES} is
\begin{align*}
\asymp A\int_0^1 t^\delta\min\Bigl\{1,\frac Bt\Bigr\}\,dt
\asymp A\min(1,B),
\end{align*}
i.e.\ \eqref{TECHNICALLEMRES} holds.
Next assume $A>10^\delta$.
Then the left hand side of \eqref{TECHNICALLEMRES} is
\begin{align*}
\asymp A\int_0^{A^{-1/\delta}}t^\delta\min\Bigl\{1,\frac Bt\Bigr\}\,dt
+\int_{A^{-1/\delta}}^1\min\Bigl\{1,\frac Bt\Bigr\}\,dt.
\end{align*}
If $B\leq A^{-1/\delta}$ then this is $\asymp B+B\log A\asymp B\log A$;
if $A^{-1/\delta}\leq B\leq\frac1{10}$ then it is
$\asymp A^{-1/\delta}+B\log(B^{-1})\asymp B\log(B^{-1})$,
and finally if $B\geq\frac1{10}$ then it is 
$\asymp A^{-1/\delta}+1\asymp 1$.
Hence \eqref{TECHNICALLEMRES} holds in all cases.
\end{proof}

We now continue onwards with the proof of Lemma \ref{TOUGHXIINTLEM}.
By Lemma \ref{TECHNICALLEM}
(used together with $\log(2+\frac1{A^{-1}+B})\leq\log(2+A)$), 
\eqref{TOUGHXIINTLEMPF1} is
\begin{align}\notag
\ll\int_0^\infty\int_{-\infty}^\infty
\min\biggl\{1,\bigl((b+|u_1|+r)^{-d}v\bigr)^{2-\frac2{d-1}}\biggr\}
\hspace{160pt}
\\\label{TOUGHXIINTLEMPF2}
\times\min\biggl\{1,\rho\bigl(1+b^2+u_1^2+r^2\bigr)\log\Bigl(
2+(b+|u_1|+r)^{-d}v\Bigr)\biggr\}
\,du_1\,r^{d-4}\,dr.
\end{align}
Let us first assume $b\geq1$.
Then $1+b^2+u_1^2+r^2\ll(b+|u_1|+r)^2$ for all $u_1\in\R$,
and we thus get, upon setting $s=b+|u_1|+r$,
\begin{align}\label{TOUGHXIINTLEMPF5}
\ll\int_{b}^\infty
\min\biggl\{1,\bigl(s^{-d}v\bigr)^{2-\frac2{d-1}}\biggr\}
\min\biggl\{1,\rho s^2\log\bigl(2+s^{-d}v\bigr)\biggr\}
\,s^{d-3}\,ds.
\end{align}
By our assumptions we have $b\leq\rho^{-\frac12}$; hence the above is
\begin{align*}
\ll\int_{b}^{\rho^{-\frac12}}\rho s^2\log(2+s^{-d}v)
\min\biggl\{1,\bigl(s^{-d}v\bigr)^{2-\frac2{d-1}}\biggr\}\,s^{d-3}\,ds
+\int_{\rho^{-\frac12}}^\infty 
\bigl(s^{-d}v\bigr)^{2-\frac2{d-1}}\,s^{d-3}\,ds.
\end{align*}
Also by our assumptions we have $v^{\frac1d}\leq\rho^{-\frac12}$. Using
this and the fact that 
$d-1-d(2-\frac2{d-1})<-1$ for $d\geq4$,
we find by a quick computation that the above is
\begin{align}\label{TOUGHXIINTLEMPF4}
\ll\rho v\min\Bigl(1,(vb^{-d})^{1-\frac2{d-1}}\Bigr)+
v^{2-\frac2{d-1}}\rho^{\frac d2-\frac1{d-1}}
\ll\rho v\min\Bigl(1,(vb^{-d})^{1-\frac2{d-1}}\Bigr).
\end{align}
(The last step follows since
$\rho v\geq v^{2-\frac2{d-1}}\rho^{\frac d2-\frac1{d-1}}
\Leftrightarrow \rho\leq v^{-\frac2d}$
and
$\rho v^{2-\frac2{d-1}}b^{-d(1-\frac2{d-1})}
\geq v^{2-\frac2{d-1}}\rho^{\frac d2-\frac1{d-1}}
\Leftrightarrow \rho\leq b^{-2}$, both of which are true by assumption.)

It now remains to treat the case $b<1$.
In this case $1+b^2+u_1^2+r^2\ll(b+|u_1|+r)^2$ 
still holds whenever $|u_1|\geq1$ or $r\geq1$, and hence
the contribution from all such $\langle u_1,r\rangle$ 
in \eqref{TOUGHXIINTLEMPF2} is still bounded by \eqref{TOUGHXIINTLEMPF5},
and hence also bounded by \eqref{TOUGHXIINTLEMPF4}.
Hence it only remains to treat the contribution 
from $\langle u_1,r\rangle$ with $|u_1|<1$ and $r<1$.
But for these $\langle u_1,r\rangle$ we have
$1+b^2+u_1^2+r^2\ll1$ and hence the contribution from
these $\langle u_1,r\rangle$ in \eqref{TOUGHXIINTLEMPF2} is
\begin{align*}
\ll\int_0^1\int_0^1\rho\log\bigl(2+(u_1+r)^{-d}v\bigr)\,du_1\,r^{d-4}\,dr
\ll\rho\int_0^2\log(2+s^{-d}v)\,s^{d-3}\,ds\hspace{40pt}
\\
\ll\rho\log(2+v)\ll\rho v.
\end{align*}
Hence \eqref{TOUGHXIINTLEMRES} holds also when $b<1$.

Finally we treat the case $d=3$.
In this case the left hand side of \eqref{TOUGHXIINTLEMRES} is
\begin{align}
\ll\int_{-\infty}^\infty\int_0^1
\min\biggl\{1,\Bigl(\max(u_1^2,(2b+u_1)^2)\Bigr)^{-\frac 32}vt\biggr\}
\min\biggl\{1,\frac{\rho(1+b^2+u_1^2)}t\biggr\}\,dt\,du_1.
\end{align}
If $b\geq1$ then arguing as before we get
\begin{align*}
\ll\int_{b}^\infty
\min\bigl\{1,s^{-3}v\bigr\}
\min\Bigl\{1,\rho s^2\log\bigl(2+s^{-3}v\bigr)\Bigr\}\,ds.
\end{align*}
This is the same as ``\eqref{TOUGHXIINTLEMPF5} with $d=3$'',
and the analysis goes through as before except that there is an
extra logarithm factor (since $d-1-d(2-\frac2{d-1})=-1$ for $d=3$),
and we obtain the bound in \eqref{TOUGHXIINTLEMRES}.
The extension to the case $b<1$ works as before.
\end{proof}

\subsection{\texorpdfstring{Proof of Theorem \ref*{PHI0XILARGETHM}}{Proof of Theorem 1.7}}

We now bound the contribution from the error term in 
\eqref{PHI0XILARGETHMPF4}.
First, %
it follows from Lemma \ref{SABVBOUNDLEMCOR} 
that the contribution from
``$O(1-z+\varphi^2)$'' in %
\eqref{PHI0XILARGETHMPF4} is
\begin{align}\label{E3DEF}
\ll E_3:=
\xi^{-2+\frac2d}\min\Bigl\{1,
(\xi\varphi^d)^{-1+\frac2{d(d-1)}}\Bigr\}(1-z+\varphi^2).
\end{align}
We next consider the contribution from 
$O(\frac{\|\veceta'\|^2}{1-\eta_1}+\frac{\|\veceta'\|^2}{\|\veceta\|^2})$.
Note that 
$\frac{\|\veceta'\|^2}{1-\eta_1}+\frac{\|\veceta'\|^2}{\|\veceta\|^2}
\ll\min(1,\frac{\|\veceta'\|^2}{\eta_1^2(1-\eta_1)})$
for all $\vech\in S'''$.
Furthermore
\begin{align}\label{ETAINH}
\veceta=(h_1,\alpha h_2-2\beta h_1,\alpha h_3,\ldots,\alpha h_{d-1}),
\end{align}
where we recall that
$\alpha=\alpha(\frac{\|\veceta'\|}{\|\veceta\|})$,
$\beta=\beta(\frac{\|\veceta'\|}{\|\veceta\|})$.
Hence
\begin{align}\label{PHI0XILARGETHMPF13}
\frac{\|\veceta'\|^2}{\eta_1^2(1-\eta_1)}
=\frac{(\alpha h_2-2\beta h_1)^2+(\alpha h_3)^2+\cdots+(\alpha h_{d-1})^2}
{h_1^2(1-h_1)}
\ll\alpha^2\frac{(\frac\beta\alpha)^2h_1^2+\|\vech'\|^2}{h_1^2(1-h_1)}.
\end{align}
Here $\alpha^2\ll1-z$
and $|\frac\beta\alpha|\ll1+\frac{\varphi}{\sqrt{1-z}}\ll1+b$,
by \eqref{UPSILONMAINUPPERBOUNDPROPRES2}
(also recall \eqref{ABKAPPADEF}).
Hence if apply Lemma \ref{TOUGHXIINTLEM} with 
$\rho=1-z$ and an appropriate choice of
$v\asymp(1-z)^{-\frac d2}\xi^{-1}$,
using Lemma \ref{XIDM1INCRLEM}, then we conclude 
(since $\xi^{-1}(1-z)^{\frac d2-1}\rho v\asymp\xi^{-2}$)
that the contribution from 
$O(\frac{\|\veceta'\|^2}{1-\eta_1}+\frac{\|\veceta'\|^2}{\|\veceta\|^2})$
in \eqref{PHI0XILARGETHMPF4} is
\begin{align}\label{TOUGHXIINTLEMAPPL}
\ll E:=\begin{cases}
\xi^{-2}\log(2+\min(\xi,\varphi^{-1}))&\text{if }\:d=3
\\
\xi^{-2}\min\bigl(1,(\xi\varphi^d)^{-\frac{d-3}{d-1}}\bigr)
&\text{if }\:d\geq4.\end{cases}
\end{align}
(This is the same $E$ as in \eqref{PHI0XILARGETHMEDEF}.)
One checks by inspection that $E_1+E_2\ll E$,
cf.\ \eqref{E1DEF}, \eqref{E2DEF}.
Furthermore, using $1-z<c_\clowD\xi^{-\frac2d}$ we see that
(cf.\ \eqref{E3DEF})
\begin{align}\label{PHI0XILARGETHMPF14}
E_3
\ll\xi^{-2}\min\Bigl\{1,
(\xi\varphi^d)^{-1+\frac2{d(d-1)}}\Bigr\}
+\xi^{-2}\min\Bigl\{(\xi\varphi^d)^{\frac2d},
(\xi\varphi^d)^{-\frac{d-3}{d-1}}\Bigr\}
\ll E.
\end{align}
Hence we conclude
\begin{align}\notag
\Phi_\bn(\xi,\vecw,\vecz)
\leq\frac{2^{2-\frac32d}\xi^{-1}(1-z)^{\frac d2-1}}{\zeta(d)}
\int_{S'''}
\Xi\Bigl(a,b;\vech;
\kappa(1-h_1)
\hspace{150pt}
\\\label{PHI0XILARGETHMPF6}
\times
\Bigl\{1+c_\clowQ
\Bigl(1-z+\varphi^2+\sfrac{\|\veceta'\|^2}{\|\veceta\|^2}\Bigr)\Bigr\}\Bigr)
\,d\vech+O(E).
\end{align}

Here %
note that (cf.\ \eqref{ETAINH})
\begin{align*}
\frac{\|\veceta'\|^2}{\|\veceta\|^2}
\leq\alpha^2\frac{(h_2-2(\beta/\alpha)h_1)^2+h_3^2+\cdots+h_{d-1}^2}{h_1^2}
\ll\alpha^2\frac{(b+1)^2h_1^2+\|\vech'\|^2}{h_1^2}
\\
\ll1-z+\varphi^2+(1-z)\frac{\|\vech'\|^2}{h_1^2}.
\end{align*}
Also recall \eqref{ZWPHIBOUNDS} and
$\frac{\|\veceta'\|}{\|\veceta\|}<c_\clowP$ %
(cf.\ \eqref{OMEGAOPENCONE}).
Hence by requiring $c_\clowP$ to be sufficiently small 
we can force 
$c_\clowQ(1-z+\varphi^2+\sfrac{\|\veceta'\|^2}{\|\veceta\|^2})
\leq\frac12$
to hold for all $\vech\in S'''$.
Now from \eqref{PHI0XILARGETHMPF6} we see that
there is a constant $c_\clowR>0$ which only depends on $d$
such that 
\begin{align}\notag
\Phi_\bn(\xi,\vecw,\vecz)
\leq\frac{2^{2-\frac32d}\xi^{-1}(1-z)^{\frac d2-1}}{\zeta(d)}
\int_{S'''}
\Xi\Bigl(a,b;\vech;
\kappa(1-h_1)\Bigl\{1+M\bigl(h_1^{-1}\|\vech'\|\bigr)\Bigr\}\Bigr)\,d\vech+O(E)
\end{align}
where
\begin{align*}
M(\ell):=\min\bigl\{\sfrac12,c_\clowR
\bigl(\varphi^2+(1-z)(1+\ell^2)\bigr)\bigr\}.
\end{align*}

Recall that $S'''\subset(0,1)\times\R^{d-2}$;
hence the above inequality remains true if we replace the range of
integration by $(0,1)\times\R^{d-2}$.
Writing $\vech=h_1(1,\vecu)$ we thus get
\begin{align}\notag
\Phi_\bn(\xi,\vecw,\vecz)
\leq\frac{2^{2-\frac32d}\xi^{-1}(1-z)^{\frac d2-1}}{\zeta(d)}
\int_{\R^{d-2}}\int_0^1
\Xi\Bigl(a,b;(1,\vecu);
\kappa(1-h_1)\Bigl(1+M(\|\vecu\|)\Bigr)\Bigr)
\hspace{30pt}
\\\label{PHI0XILARGETHMPF7}
\times h_1^{d-2}\,dh_1\,d\vecu+O(E).
\end{align}
Now in the inner integral in \eqref{PHI0XILARGETHMPF7} we
substitute
$h_1=1-(1-t)(1+M(\|\vecu\|))^{-1}$, $t\in[-M(\|\vecu\|),1]$.
Using $M(\|\vecu\|)\leq\frac12$ and Lemma \ref{XIDM1INCRLEM}
we see that the contribution from $t<0$ is
\begin{align}\notag
&\ll\xi^{-1}(1-z)^{\frac d2-1}\int_{\R^{d-2}}
\Xi\bigl(a,b;(1,\vecu);2\kappa\bigr)M(\|\vecu\|)\,d\vecu
\\\label{PHI0XILARGETHMPF8}
&\ll\xi^{-1}(1-z)^{\frac d2-1}\int_{(\frac14,\frac12)\times\R^{d-2}}
\Xi\bigl(a,b;\vech;4\kappa(1-h_1)\bigr)
\,\min\Bigl(1,\varphi^2+(1-z)\frac{\|\vech\|^2}{h_1^2}\Bigr)
\,d\vech
\ll E,
\end{align}
where the last bound follows from Lemma \ref{TOUGHXIINTLEM}
and Lemma \ref{SABVBOUNDLEMCOR}.
We also have, for all $t\in[0,1]$:
\begin{align*}
\bigl|h_1^{d-2}-t^{d-2}\bigr|\ll\bigl|h_1-t\bigr|
\ll\bigl|1-(1+M(\|\vecu\|))^{-1}\bigr|\ll M(\|\vecu\|).
\end{align*}
Hence we obtain
\begin{align}\notag
\Phi_\bn(\xi,\vecw,\vecz)
\leq\frac{2^{2-\frac32d}\xi^{-1}(1-z)^{\frac d2-1}}{\zeta(d)}
\int_{\R^{d-2}}\int_0^1
\Xi\bigl(a,b;(1,\vecu);\kappa(1-t)\bigr)
\hspace{90pt}
\\\notag
\times\bigl(t^{d-2}+O(M(\|\vecu\|))\bigr)\,dt\,d\vecu+O(E).
\end{align}
Here the contribution from the $M(\|\vecu\|)$-term is again $\ll E$,
since it is bounded above by the first line in \eqref{PHI0XILARGETHMPF8}.
Hence, setting $\vech=t(1,\vecu)$, we have finally proved:
\begin{align}\label{PHI0XILARGETHMPF9}
\Phi_\bn(\xi,\vecw,\vecz)
\leq\frac{2^{2-\frac32d}\xi^{-1}(1-z)^{\frac d2-1}}{\zeta(d)}
\int_{(0,1)\times\R^{d-2}}
\Xi\bigl(a,b;\vech;\kappa(1-h_1)\bigr)\,d\vech+O(E).
\end{align}

In a very similar way we also obtain a similar \textit{lower} bound:
First, in almost exactly the same way as
we got \eqref{PHI0XILARGETHMPF6} from \eqref{PHI0XILARGETHMPF4},
\eqref{PHI0XILARGETHMPF4a} leads to
\begin{align}\notag
\Phi_\bn(\xi,\vecw,\vecz)
\geq\frac{2^{2-\frac32d}\xi^{-1}(1-z)^{\frac d2-1}}{\zeta(d)}
\int_{S'''}
\Xi\Bigl(a,b;\vech;\kappa(1-h_1)\bigl(1-c_\clowQ
\bigl(1-z+\varphi^2+\|\veceta'\|\varphi\bigr)\bigr)^+\Bigr)\,d\vech
\\\label{PHI0XILARGETHMPF6a}
-O(E).
\end{align}

We will prove that the range of integration in \eqref{PHI0XILARGETHMPF6a}
may be replaced by $(0,1)\times\R^{d-2}$, 
at the cost of an error which is $\ll E$.
Recall that we have a bijection
$\vech\leftrightarrow\veceta\leftrightarrow\vecv$ between
$\vech\in(0,1)\times\R^{d-2}$ and $\vecv\in\S_1^{d-1}\cap\{v_1,v_2>0\}$,
and by definition $S'''$ is the set of all 
$\vech\in(0,1)\times\R^{d-2}$ which correspond to $\vecv\in S'$.
Hence for any $\vech\in((0,1)\times\R^{d-2})\setminus S'''$ the corresponding
vector $\vecv$ has $v_1\leq c_\clowP^{-2}(\varphi+\omega)^2$,
cf.\ \eqref{SPDEFnew}.
Using also 
$v_1\asymp\frac{1-\eta_1}{\sqrt{1+\|\veceta'\|^2}}$
and 
$\omega\ll\sin\omega=\frac{\|\vecv''\|}{\|\vecv'\|}
=\frac{\|\veceta'\|}{\|\veceta\|}$
(cf.\ \eqref{VFROMETA}),
we see that there is a constant $c_\clowT>0$ such that 
the error caused by replacing $S'''$ by $(0,1)\times\R^{d-2}$
in \eqref{PHI0XILARGETHMPF6a} is
\begin{align}\notag
\leq\xi^{-1}(1-z)^{\frac d2-1}\int_{(0,1)\times\R^{d-2}}
\Xi\bigl(a,b;\vech;\kappa\bigr)
\,\biggl\{
I\biggl(\frac{\|\veceta'\|^2}{\|\veceta\|^2}
\frac{\sqrt{1+\|\veceta'\|^2}}{1-\eta_1}>c_\clowT\biggr)
\hspace{70pt}
\\\label{PHI0XILARGETHMPF11}
+I\biggl(\varphi^2\frac{\sqrt{1+\|\veceta'\|^2}}{1-\eta_1}>c_\clowT\biggr)
\biggr\}\,d\vech.
\end{align}
Using $\sqrt{1+\|\veceta'\|^2}\leq1+\|\veceta'\|^2\leq\eta_1^{-2}\|\veceta\|^2$
we get
\begin{align}\label{PHI0XILARGETHMPF12}
\ll\xi^{-1}(1-z)^{\frac d2-1}\int_{(0,1)\times\R^{d-2}}
\Xi\bigl(a,b;\vech;\kappa\bigr)
\,\biggl\{
\min\biggl(1,\frac{\|\veceta'\|^2}{\eta_1^2(1-\eta_1)}\biggr)
+I\biggl(\frac{\varphi^2}{1-\eta_1}>\sfrac12 c_\clowT\biggr)
\biggr\}\,d\vech,
\end{align}
where to bound the second indicator function in 
\eqref{PHI0XILARGETHMPF11} we used the fact that if
$\|\veceta'\|>1$ then
$\frac{\|\veceta'\|^2}{\eta_1^2(1-\eta_1)}>1$.
The contribution from the ``min''-term in \eqref{PHI0XILARGETHMPF12}
is seen to be $\ll E$ using Lemma~\ref{TOUGHXIINTLEM} and
\eqref{PHI0XILARGETHMPF13} with our present constants $\alpha,\beta$,
and the contribution from the
``$I(\cdot)$''-term in \eqref{PHI0XILARGETHMPF12}
is, using $\eta_1=h_1$ and Lemma \ref{SABVBOUNDLEM}:
\begin{align*}
\ll\varphi^2\xi^{-1}(1-z)^{\frac d2-1}
\min\{\kappa^{1-\frac2d},
\kappa^{2-\frac2{d-1}}b^{-d+\frac2{d-1}}\}
\ll\varphi^2
\min\{\xi^{-2+\frac2d},\xi^{-3+\frac2{d-1}}\varphi^{-d+\frac2{d-1}}\}
\\
\ll E_3\ll E
\end{align*}
(cf.\ \eqref{E3DEF} and \eqref{PHI0XILARGETHMPF14}).

Hence we have proved that the range of integration in
\eqref{PHI0XILARGETHMPF6a} may indeed be replaced by $(0,1)\times\R^{d-2}$.
Now by the same argument as when going from 
\eqref{PHI0XILARGETHMPF6} to \eqref{PHI0XILARGETHMPF9}
we obtain  %
\begin{align}\label{PHI0XILARGETHMPF10}
\Phi_\bn(\xi,\vecw,\vecz)
\geq\frac{2^{2-\frac32d}\xi^{-1}(1-z)^{\frac d2-1}}{\zeta(d)}
\int_{(0,1)\times\R^{d-2}}
\Xi\bigl(a,b;\vech;\kappa(1-h_1)\bigr)\,d\vech-O(E).
\end{align}
Together, \eqref{PHI0XILARGETHMPF9} and \eqref{PHI0XILARGETHMPF10}
imply that the relation \eqref{PHI0XILARGETHMRES} in
Theorem \ref{PHI0XILARGETHM} holds, with
\begin{align}\label{PHI0XILARGETHMFDDEF}
F_{\bn,d}(t_1,t_2,\alpha)=\frac{2^{2-\frac 32d}t_1^{\frac d2-1}}{\zeta(d)}
\int_{\vech\in(0,1)\times\R^{d-2}}
\Xi\biggl(t_1^{-\frac 12}t_2^{\frac 12},\frac{\alpha}{\sqrt{2t_1}};\vech;
2^{1-\frac d2}t_1^{-\frac d2}(1-h_1)\biggr)\,d\vech.
\end{align}
The fact that $F_{\bn,d}$ is uniformly bounded follows from
Lemma \ref{SABVBOUNDLEMCOR}.
Furthermore for any $\delta>0$ there exists a
bounded set $C\subset\R^{d-2}$ such that for all
$\langle t_1,t_2,\alpha\rangle\in[\delta,\infty)\times\R_{>0}\times\R_{\geq0}$
the support of the integrand in \eqref{PHI0XILARGETHMFDDEF} is contained in
$(0,1)\times C$
(this follows from Lemma \ref{XISUPBOUNDLEM},
since $\Xi(a,b;\vech;v)\leq\Xi(\frac{\|\vech'\|}{h_1},v)$).
Hence Lemma \ref{XIDM1CONTLEM} implies that $F_{\bn,d}$ is continuous.
This completes the proof of Theorem \ref{PHI0XILARGETHM}.
\hfill$\square$ $\square$ $\square$

\section{\texorpdfstring{On the support of $\Phi_\bn(\xi,\vecw,\vecz)$}{On the support of Phi0(xi,w,z)}}
\label{SUPPORTSECTION}

\subsection{The functions $\sigma_d(r,\alpha)$ and $\xi_0(w,z,\varphi)$}
\label{XI0WZPHISEC}

We continue to keep $d\geq3$.
Recall that we have defined, for $a>0$, $b\in\R$ (cf.\ \eqref{RHODABDEF}):
\begin{align*}
\rho(a,b)=\inf\bigl\{v>0\col \exists \vech\in\R_+^{d-1}:\:
\Xi(a,b;\vech;v)>0\bigr\}.
\end{align*}
Note that $\rho(a,-b)=\rho(a,b)$, 
immediately from \eqref{XIDM1ABDEF} and \eqref{XIDTKINV}.
Also by \eqref{XIDM1SYMM} we have the symmetry relation
\begin{align}\label{RHOSYMM}
\rho(a,b)=a^{d}\rho(a^{-1},a^{-1}b).
\end{align}
It follows from \eqref{PHI0XILARGETHMFDDEF} that
\begin{align}\label{SUPPORTprelim3}
F_{\bn,d}(t_1,t_2,\alpha)>0
\Longleftrightarrow
\rho\Bigl(t_1^{-\frac12}t_2^{\frac12},
\frac{\alpha}{\sqrt{2t_1}}\Bigr)<2^{1-\frac d2}t_1^{-\frac d2}.
\end{align}
In order to express this relation in a slightly cleaner way we introduce the
function
\begin{align}\label{SUPPORTprelim4}
\sigma_d(r,\alpha):=2^{\frac4d-2}r\,
\rho\bigl(r^{\frac12},2^{-\frac12}(r\alpha)^{\frac14}\bigr)^{-\frac4d}
\qquad
(r>0,\:\alpha\geq0).
\end{align}
Then \eqref{RHOSYMM} translates into the symmetry relation
\begin{align}\label{SIGMADSYMM}
\sigma_d(r,\alpha)=\sigma_d(r^{-1},\alpha),
\end{align}
and \eqref{SUPPORTprelim3} translates into \eqref{SUPPORTprelim2},
i.e.\ 
\begin{align*}
F_{\bn,d}(t_2,t_1,\alpha)>0\Longleftrightarrow
t_1t_2<\sigma_d\Bigl(\frac{t_2}{t_1},\frac{\alpha^4}{t_1t_2}\Bigr).
\end{align*}

We remark that we will prove below in Corollary \ref{BASICRHOLOWBOUNDCOR}
that $\rho(a,b)\asymp\max(1,b)$ holds uniformly over
$0<a\leq1$, $b\geq0$.
This translates into the relation
\begin{align}\label{SIGMADASYMP}
\sigma_d(r,\alpha)\asymp r\min\bigl(1,(r\alpha)^{-\frac1d}\bigr),
\qquad\forall 0<r\leq1,\:\alpha\geq0.
\end{align}

We next prove the existence of the continuous function
$\xi_0:[0,1)\times[0,1)\times[0,\pi]\to\R_{>0}$ %
as stated in Theorem \ref{PHI0SUPPORTASYMPTTHM}.
\label{XI0WZPHIDISC}
Let us fix any $\vecw,\vecz\in\scrB_1^{d-1}$.
Now the function $\xi\mapsto\Phi_\bn(\xi,\vecw,\vecz)$ is 
continuous, decreasing
(cf.\ \cite[Lemma 7.11]{lprob}),
positive for $\xi$ small (e.g.\ by Theorem \ref{PHI0ZEROSMALLTHM})
and vanishing for all sufficiently large $\xi$
(e.g.\ by \cite[Prop.\ 1.9]{lprob});
hence there exists a unique number $\xi_0>0$ such that
$\Phi_\bn(\xi,\vecw,\vecz)>0 \Leftrightarrow\xi<\xi_0$.
This proves that there exists a unique function
$\xi_0:[0,1)\times[0,1)\times[0,\pi]\to\R_{>0}$ 
such that $\Phi_\bn(\xi,w,z,\varphi)>0\Leftrightarrow\xi<\xi_0(w,z,\varphi)$.
Since $\Phi_\bn(\xi,\vecw,\vecz)$ is
continuous (jointly in all three variables) it follows that
$\xi_0(w,z,\varphi)$ is lower semicontinuous in
$(w,z,\varphi)\in[0,1)\times[0,1)\times[0,\pi]$.
Finally the fact that 
$\xi_0(w,z,\varphi)$ is upper semicontinuous 
(and hence continuous) follows from the following lemma,
which is a slight generalization of \cite[Lemma 7.11]{lprob}:
\begin{lem}\label{PHI0INCRGENLEM}
For any $\vecw,\vecz,\vecw',\vecz'\in\scrB_1^{d-1}$ and $\xi,\xi'>0$,
$\Phi_\bn(\xi,\vecw,\vecz)\geq\Phi_\bn(\xi',\vecw',\vecz')$ holds
whenever $\|\vecz'-\vecz\|<1-\|\vecz\|$,
$\|\vecw'-\vecw\|<1-\|\vecw\|$ and
\begin{align*}
\xi'\geq\max\Bigl(\Bigl(1-\frac{\|\vecz'-\vecz\|}{1-\|\vecz\|}\Bigr)^{1-d},
\Bigl(1-\frac{\|\vecw'-\vecw\|}{1-\|\vecw\|}\Bigr)^{1-d}\Bigr)\,\xi.
\end{align*}
\end{lem}
\begin{proof}
Follow the proof of \cite[Lemma 7.11]{lprob}, but replace the matrix
$T$ therein by
\begin{align*}
T=\matr\alpha{\xi^{-1}(\vecw'+\vecz'-\alpha^{-\frac1{d-1}}(\vecw+\vecz))}
{\trans\bn}{\alpha^{-\frac1{d-1}}1_{d-1}}\in G
\qquad(\alpha=\xi'/\xi).
\end{align*}
\end{proof}

\subsection{Bound from below on the support of $\Phi_\bn$}

The following proposition gives one half of
Theorem \ref{PHI0SUPPORTASYMPTTHM}. 

\begin{prop}\label{PHI0SUPPORTASYMPTTHMHALFPROP}
We have
\begin{align}\label{PHI0SUPPORTASYMPTTHMHALFPROPRES}
\xi_0(w,z,\varphi)\geq2^{1-\frac d2}(1-z)^{-\frac d2}
\rho\biggl(\sqrt{\frac{1-w}{1-z}},
\frac{\varphi}{\sqrt{2(1-z)}}\biggr)^{-1}
\Big\{1-O\Bigl(\max(1-z,1-w)+\varphi^2\Bigr)\Bigr\},
\end{align}
uniformly over all $z,w\in[0,1)$, $\varphi\in[0,\frac\pi2]$.
\end{prop}
The proof depends on the following lemma
(with constants $c_\clowN$, $c_\clowO$ as in 
Proposition \ref{UPSILONMAINLOWBOUNDPROP}).\enlargethispage{10pt}
\begin{lem}\label{PHI0SUPPORTASYMPTTHMHALFPROPLEM}
Let $z,w,\varphi$ be given with
$1-c_\clowN\leq z\leq w<1$ and $0\leq\varphi\leq c_\clowN$.
Then there exist numbers
$\alpha,\beta$ satisfying \eqref{UPSILONMAINLOWERBOUNDPROPRES2}
and which have the property that
for any $\xi>0$ with $\Phi_\bn(\xi,z,w,\varphi)=0$ 
and for $\vecz,\vecw$ as in \eqref{CASE12FORMULA}, we have
\begin{align}\notag
\Xi\biggl(\sqrt{\frac{1-w}{1-z}},
\frac{\varphi}{\sqrt{2(1-z)}};
(\alpha v_2,2\beta v_2+v_3,v_4,\ldots,v_d);
\hspace{120pt}
\\\label{PHI0SUPPORTASYMPTTHMHALFPROPLEMRES}
\frac{2^{1-\frac d2}(1-z)^{-\frac d2}\xi^{-1}v_1}{\vecv\cdot(1,\vecz+\vecw)}
\Bigl\{1-c_\clowO(1-z+\varphi^2)\Bigr\}^+\biggr)=0
\end{align}
for all $\vecv=(v_1,\ldots,v_d)\in\S_1^{d-1}$ 
with $v_1>\frac{199}{200}$ and $v_2>0$.
\end{lem}
\begin{proof}
Let $\xi,z,w,\varphi$ be given with $\xi>0$, 
$1-c_\clowN\leq z\leq w<1$ and $0\leq\varphi\leq c_\clowN$,
and assume $\Phi_\bn(\xi,z,w,\varphi)=0$.
Take $\vecz,\vecw$ as in \eqref{CASE12FORMULA},
and set $\fZ=\xi^{\frac1d}(\fZ(0,1,1)+(0,\vecz))$ and 
$\vecy=\xi^{\frac1d}(1,\vecz+\vecw)$ as in the
previous section.
Then since $G$ is covered by a countable number of $\F_d$-translates
we must have
$\nu_\vecy\bigl(\bigl\{
M\in G_{\veck,\vecy}\col \Z^dM\cap\fZ=\emptyset\bigr\}\bigr)=0$
for every $\veck\in\widehat\Z^d$.
In particular this holds for $\veck=\vece_1$, and recalling the
definition of $L_{\vecv,\tM}$ in \eqref{LVMDEF} it follows that
$L_{\vecv,\tM}\cap\fZ=\emptyset$ for almost all
$\langle\vecv,\tM\rangle\in(\S_1^{d-1}\cap\R_{\vecy+}^d)\times G^{(d-1)}$.
Let us write $S(\frac{199}{200})$ for the set of all
$\vecv\in\S_1^{d-1}$ with $v_1>\frac{199}{200}$ and $v_2>0$.
Then for every $\vecv\in S(\frac{199}{200})$ we have
$\|\vecv'\|=\sqrt{1-v_1^2}<\frac1{10}$ so that
$\vecv\cdot\vecy>\xi^{\frac1d}(\frac{199}{200}-\frac15)>0$,
and thus we have $L_{\vecv,\tM}\cap\fZ=\emptyset$
for almost all $\langle\vecv,\tM\rangle\in S(\frac{199}{200})\times G^{(d-1)}$.

We next note that for every $\vecv\in S(\frac{199}{200})$
the conclusion of Lemma \ref{GOODINTLEM2} holds, viz.\ for all 
$n\in\Z\setminus\{0,1\}$ we have $(na_1\vecv+\vecv^\perp)\cap\fZ=\emptyset$,
with $a_1=\vecy\cdot\vecv$.
Indeed, as in the proof of that lemma it suffices to check that
\eqref{GOODINTLEM2PF1} holds, and this is clear since 
$\|\vecv'\|<\frac1{10}$ for $\vecv\in S(\frac{199}{200})$.
It now follows as in Section \ref{PFPHI0XILARGETHMSEC} that
\eqref{BASICWZINEQ} holds for all $\vecv\in S(\frac{199}{200})$.
But \textit{also}, by Proposition~\ref{UPSILONMAINLOWBOUNDPROP},
there exist numbers $\alpha,\beta$ which only
depend on $z,w,\varphi$ (and $d$) and which satisfy
\eqref{UPSILONMAINLOWERBOUNDPROPRES2}, such that
$\Upsilon(\vecz,\vecw,\vecv',\xi^{-1+\frac1d}v_1a_1^{-1})$
is larger than or equal to the left hand side of 
\eqref{PHI0SUPPORTASYMPTTHMHALFPROPLEMRES} for all 
$\vecv\in S(\frac{199}{200})$.

These observations together imply that
the left hand side of 
\eqref{PHI0SUPPORTASYMPTTHMHALFPROPLEMRES} must vanish for almost
all $\vecv\in S(\frac{199}{200})$,
and thus since $\Xi$ is continuous (Lemma \ref{XIDM1CONTLEM}),
it must vanish for \textit{all} $\vecv\in S(\frac{199}{200})$.
\end{proof}

\begin{proof}[Proof of Proposition \ref{PHI0SUPPORTASYMPTTHMHALFPROP}]
Because of $\xi_0(w,z,\varphi)=\xi_0(z,w,\varphi)$ and 
\eqref{RHOSYMM}, we may assume $z\leq w$ without loss of generality.
Note that \eqref{PHI0SUPPORTASYMPTTHMHALFPROPRES} 
is content-free unless both $\varphi$ and $1-z$ are small; we may thus 
assume $1-c_\clowN\leq z\leq w<1$ and $0\leq\varphi\leq c_\clowN$.
Given $z,w,\varphi$ we let 
$\vecz,\vecw$ be the corresponding points as in \eqref{CASE12FORMULA}. 
Also let $\alpha,\beta$ be the corresponding numbers as in 
Lemma~\ref{PHI0SUPPORTASYMPTTHMHALFPROPLEM}.

Set $a=\sqrt{\frac{1-w}{1-z}}\in(0,1]$ and
$b=\frac{\varphi}{\sqrt{2(1-z)}}\in\R_{\geq0}$.
Fix any number $\rho'>\rho(a,b)$.
Then there is some $\vech=(h_1,\ldots,h_{d-1})\in\R_+^{d-1}$,
which we fix from now on, such that
$\Xi(a,b;\vech;v)>0$ for $v=\rho'$ and thus for all $v\geq\rho'$.
Now let $t$ be a small positive parameter which we will later take to
tend to $0$,
and set $\vecv:=(v_1,\vecv')$ where
\begin{align*}
\vecv':=t(\alpha^{-1}h_1,h_2-2(\beta/\alpha)h_1,h_3,\ldots,h_{d-1});
\qquad
v_1:=\sqrt{1-\|\vecv'\|^2}.
\end{align*}
Clearly for all sufficiently small $t$ the vector 
$\vecv$ is well-defined, lies in $\S_1^{d-1}$,
and has $v_1>\frac{199}{200}$ and $v_2>0$.
Also note that, for all $v\geq\rho'$,
\begin{align*}
\Xi\bigl(a,b;(\alpha v_2,2\beta v_2+v_3,v_4,\ldots,v_d);v\bigr)
=\Xi(a,b;t\vech;v)=\Xi(a,b;\vech;v)>0.
\end{align*}
Hence Lemma \ref{PHI0SUPPORTASYMPTTHMHALFPROPLEM}
implies that $\Phi_\bn(\xi,z,w,\varphi)>0$ for all $\xi>0$ with
\begin{align*}
\xi\leq{\rho'}^{-1}
\frac{2^{1-\frac d2}(1-z)^{-\frac d2}v_1}{\vecv\cdot(1,\vecz+\vecw)}\Bigl\{1-c_\clowO(1-z+\varphi^2)\Bigr\}^+.
\end{align*}
Letting now $t\to0$, and then using the fact that 
$\rho'$ was arbitrary with $\rho'>\rho(a,b)$, it follows that
$\Phi_\bn(\xi,z,w,\varphi)>0$ for all $\xi>0$ with
\begin{align*}
\xi<\rho(a,b)^{-1}2^{1-\frac d2}(1-z)^{-\frac d2}
\Bigl\{1-c_\clowO(1-z+\varphi^2)\Bigr\}^+.
\end{align*}
This concludes the proof.
\end{proof}

\begin{cor}\label{BASICRHOBOUNDCOR}
We have $\rho(a,b)\gg1+|b|$, uniformly over all $0<a\leq1$, $b\in\R$.
\end{cor}
\begin{proof}
We may assume $b\geq0$, since $\rho(a,-b)=\rho(a,b)$.
Given any $a\in(0,1]$ and $b\geq0$ we may find $z\leq w<1$ and $\varphi\geq0$
satisfying $a=\sqrt{\frac{1-w}{1-z}}$, $b=\frac{\varphi}{\sqrt{2(1-z)}}$, and 
with both $1-z$ and $\varphi$ %
arbitrarily small.
Now Proposition \ref{PHI0SUPPORTTHM} says that
$\xi_0(w,z,\varphi)\asymp(1-z)^{-\frac d2}(1+b)^{-1}$,
and the corollary follows from this %
combined with
Proposition \ref{PHI0SUPPORTASYMPTTHMHALFPROP}.
\end{proof}

\subsection{\texorpdfstring{An exact formula for $\Phi_\bn(\xi,\vecw,\vecz)$ when $\xi(1-z)^{\frac{d-1}2}$ is large}{An exact formula for Phi0(xi,w,z) when *** is large}}

We will now prove that \textit{if 
$\xi(1-z)^{\frac{d-1}2}$ is sufficiently large, then
the formula \eqref{PFPHI0XILARGETHMLOWERBOUND},
with range of integration $S$ in place of $S'$,
holds without error terms.}
We will use this result to complete the proof of 
Theorem~\ref{PHI0SUPPORTASYMPTTHM}, but it is clearly also of independent
interest; for example we expect that when $d=3$ 
this result could be used to find 
completely explicit formulas for $\Phi_\bn$ in certain parameter
regimes with $\xi$ large
(we stress however that we anticipate any such explicit formula to be
rather complicated).
In view of this independent interest we allow a more general choice
of $\vecw,\vecz$ than in \eqref{CASE12FORMULA} when stating the result;
this does not cause any extra difficulties in the proof.

\begin{prop}\label{EXACTFORMULAPROP}
There is a constant $c_\clowW>0$ which only depends on $d$ such that 
for any $z,w,\varphi$ with $0\leq z\leq w<1$,
$0\leq\varphi\leq\frac\pi2$,
any $\xi\geq c_\clowW(1-z)^{\frac{1-d}2}$,
and any $\vecw,\vecz\in\scrB_1^{d-1}$ with $\|\vecw\|=w$,
$\|\vecz\|=z$, $\varphi(\vecw,\vecz)=\varphi$ and
$\varphi(\vece_1,\vecz)\leq\varphi$, $\varphi(\vece_1,\vecw)\leq\varphi$,
we have
\begin{align}\label{EXACTFORMULAPROPRES}
\Phi_\bn(\xi,\vecw,\vecz)
=\zeta(d)^{-1}\int_{S}
\Upsilon\bigl(\vecz,\vecw,\vecv',\xi^{-\frac{d-1}d}v_1(\vecy\cdot\vecv)^{-1}
\bigr) \,\frac{d\vecv}{(\vecy\cdot\vecv)^d},
\end{align}
where $\vecy=\xi^{\frac1d}(1,\vecz+\vecw)$ and where
$S$ is as in \eqref{SDEFmin}.
\end{prop}
(To be more precise: $S=\{\vecv\in\S_1^{d-1}\col 0<v_1<1,\: \vecy\cdot\vecv>
c_\clowE\xi^{\frac1d}\}$, where the constant $c_\clowE\in(0,\frac12)$ is as on
p.\ \pageref{FINALFIXINGOFCLOWE}.)

The proof of Proposition \ref{EXACTFORMULAPROP}
basically consists in going through the reductions carried out in 
Section \ref{PFPHI0XILARGETHMSEC}, checking that at each step the 
error is in fact zero, provided that
$\xi(1-z)^{\frac{d-1}2}$ is sufficiently large.
We start by proving a couple of auxiliary lemmas.

\begin{lem}\label{TRIVLINALGLEM2}
Assume $0<b_1\leq b_2\leq\ldots\leq b_d$, let
$B$ be the ellipsoid
\begin{align}
B=\Bigl\{(x_1,\ldots,x_d)\col \Bigl(\frac{x_1}{b_1}\Bigr)^2+\ldots+
\Bigl(\frac{x_d}{b_d}\Bigr)^2\leq1\Bigr\}
\end{align}
(i.e.\ ON half axes $b_1,\ldots,b_d$),
and let $\Pi\subset\R^d$ be an arbitrary linear subspace of dimension $k$.
Then $\vol_k(\Pi\cap B)$ is larger than or equal to the volume of a
$k$-dimensional ellipsoid with (ON) half axes $b_1,b_2,\ldots,b_k$.
\end{lem}

\begin{proof}
Indeed, a simple application of the min-max principle in linear algebra 
shows that
$\Pi\cap B$ is an ellipsoid with ON half axes $0<b_1'\leq\cdots\leq b_k'$
satisfying $b_j'\geq b_j$, $j=1,\ldots,k$.
\end{proof}

\begin{lem}\label{A2SMALLLEM}\enlargethispage{20pt}
For any $\vecz\in\scrB_1^{d-1}$ and $\xi>0$,
if $M\in G$ satisfies
$\Z^dM\cap\xi^{\frac1d}(\fZ(0,1,1)+(0,\vecz))=\emptyset$
and furthermore is generic in the sense that
$(\Z\vece_3+\cdots+\Z\vece_d)M\cap\vece_1^\perp=\{\bn\}$,
then we necessarily have $a_1a_2\ll
\xi^{\frac{2-d}d}(1-z)^{\frac{1-d}2}$
in the Iwasawa decomposition of $M$
(cf.\ \eqref{IWASAWA}, \eqref{ADEF}).
\end{lem}
\begin{proof}
Let $\vecz,\xi,M$ satisfy the stated assumptions.
After a rotation we may assume $\vecz=z\vece_1$, $0\leq z<1$.
Set $\fZ=\xi^{\frac1d}(\fZ(0,1,1)+(0,\vecz))$.
Note that $\Z^dM\cap\fZ=\emptyset$ implies
$\Z^dM\cap(-\fZ)=\emptyset$.
Hence in view of the genericity assumption we have
\begin{align*}
(\Z\vece_3+\cdots+\Z\vece_d)M\cap\bigl(\fZ\cup\vece_1^\perp\cup(-\fZ)\bigr)
=\{\bn\}.
\end{align*}
Now if $B$ denotes the ellipsoid
\begin{align*}
B=\Bigl\{(x_1,\ldots,x_d)\col(2x_1)^2+\Bigl(\frac{2x_2}{1-z}\Bigr)^2
+\frac{4d}{1-z}(x_3^3+\ldots+x_d^2)\leq1\Bigr\}.
\end{align*}
then
\begin{align*}
\xi^{\frac1d}B\subset\fZ\cup\vece_1^\perp\cup(-\fZ).
\end{align*}
Indeed, if $(x_1,\ldots,x_d)\in B$ then
$|x_1|\leq\frac12$, $|x_2|\leq\frac12(1-z)$ and 
$|x_j|\leq(4d)^{-\frac12}(1-z)^{\frac12}$ for $j=3,\ldots,d$,
and thus also
\begin{align*}
\bigl\|(x_2,\ldots,x_d)-\vecz\bigr\|
<\bigl(z+\sfrac12(1-z)\bigr)^2+(d-2)\frac{1-z}{4d}
<\sfrac14(1+z)^2+\sfrac14(1-z)
<1,
\end{align*}
which proves the claim.

It follows that
\begin{align*}
(\Z\vece_3+\cdots+\Z\vece_d)M\cap\xi^{\frac1d}B=\{\bn\}.
\end{align*}
However $(\Z\vece_3+\cdots+\Z\vece_d)M$ is a lattice of covolume
$a_3\cdots a_d$ in the $(d-2)$-dimensional subspace
$(\R\vece_3+\cdots+\R\vece_d)M\subset\R^d$ (cf.\ \eqref{IWASAWA});
furthermore $(\R\vece_3+\cdots+\R\vece_d)M\cap\xi^{\frac1d}B$
is an ellipsoid centered at $\bn$ which by
Lemma \ref{TRIVLINALGLEM2} has volume 
$\gg\xi^{\frac{d-2}d}(1-z)^{\frac{d-1}2}$.
Hence by Minkowski's Theorem (cf., e.g., \cite[Thm.\ 10]{siegel})
we must have $a_3\cdots a_d\gg\xi^{\frac{d-2}d}(1-z)^{\frac{d-1}2}$.
This proves the lemma, since $a_1a_2=(a_3\cdots a_d)^{-1}$.
\end{proof}

\begin{proof}[Proof of Proposition \ref{EXACTFORMULAPROP}]
Recall equation \eqref{PHI0FIRSTSPLIT} in Section \ref{PFPHI0XILARGETHMSEC}.
It was proved in \cite[Prop.~7.3]{lprob} that if
$c_\clowW$ is sufficiently large (which we assume from now on) then
our assumption $\xi\geq c_\clowW(1-z)^{\frac{1-d}2}$ implies that
all terms with $k_1\neq1$ in \eqref{PHI0FIRSTSPLIT} \textit{vanish.}
Hence
\begin{align}\label{EXACTFORMULAPROPPF1}
\Phi_\bn(\xi,\vecw,\vecz)=
\sum_{\veck'\in\Z^{d-1}}\nu_\vecy\bigl(\bigl\{
M\in G_{\veck,\vecy}\cap\F_d\col \Z^dM\cap\fZ=\emptyset\bigr\}\bigr),
\end{align}
where we write $\veck=(1,\veck')$.

We next use Lemma \ref{FDCONTAINMENTSLEM} with $A=c_\clowE\xi^{\frac1d}$.
(Note $\xi\geq c_\clowW$,
so that $A>1$ certainly holds provided $c_\clowW$ is sufficiently large.)
For any $\veck=(1,\veck')$, by Lemma \ref{A2SMALLLEM} we have
$a_1a_2\ll\xi^{\frac{2-d}d}(1-z)^{\frac{1-d}2}$ for $\nu_\vecy$-almost
all $M\in G_{\veck,\vecy}$ with $\Z^dM\cap\fZ=\emptyset$;
thus by taking $c_\clowW$ sufficiently large we can force
$a_1a_2<(c_1^{(d-1)})^{-1}A^2$ to hold for these $M$,
and it follows that the set $\FC$ in \eqref{FDCONTAINMENTSLEMRES}
satisfies $\nu_\vecy(\{M\in G_{\veck,\vecy}\cap\FC
\col\Z^dM\cap\fZ=\emptyset\})=0$.
Hence, recalling \eqref{FGDEF}
and the discussion between \eqref{UFORMULAIFK1NEQ0}  and
\eqref{PHI0LOWBOUND2pre}, we have
\begin{align}\label{EXACTFORMULAPROPPF2}
\Phi_\bn(\xi,\vecw,\vecz)=
\zeta(d)^{-1}\int_S
\mu^{(d-1)}\Bigl(\Bigl\{\tM\in\F_{d-1}\col
L_{\vecv,\tM}\cap\fZ=\emptyset\Bigr\}\Bigr)\,
\frac{d\vecv}{(\vecy\cdot\vecv)^d}.
\end{align}
Thus, in order to prove \eqref{EXACTFORMULAPROPRES} it now suffices to prove
that, with $a_1=\vecy\cdot\vecv$,
\begin{align}\label{EXACTFORMULAPROPPF7}
\mu^{(d-1)}(\{\tM\in\F_{d-1}\col L_{\vecv,\tM}\cap\fZ=\emptyset\})=
\Upsilon(\vecz,\vecw,\vecv',\xi^{-\frac{d-1}d}v_1a_1^{-1})
\end{align}
holds for all $\vecv\in S$.

Note that if $\vecv$ has the property that
$a_1^{-\frac1{d-1}}\iota(\Z^{d-1}\tM)f(\vecv)
\cap\bigl(\fZ\cup(\vecy-\fZ)\bigr)\neq\emptyset$
for $\mu^{(d-1)}-$almost all $\tM\in\F_{d-1}$ then
$\mu^{(d-1)}(\{\tM\in\F_{d-1}\col L_{\vecv,\tM}\cap\fZ=\emptyset\})=0$
and also $\Upsilon(\vecz,\vecw,\vecv',\xi^{-\frac{d-1}d}v_1a_1^{-1})=0$,
cf.\ \eqref{TWOSPECIALSUBLATTICES} and the argument we used to get
\eqref{CVRELATION}.
Hence from now on we may assume that $\vecv$ does not
have the above property, i.e.\ we may assume that
$\vecv$ satisfies
\begin{align}\label{EXACTFORMULAPROPPF3}
\mu^{(d-1)}\Bigl(\Bigl\{\tM\in\F_{d-1}\col
\Z^{d-1}\tM\cap a_1^{\frac1{d-1}}\bigl(\fZ_\vecv\cup\fZ_\vecv'\bigr)
=\emptyset
\Bigr\}\Bigr)>0,
\end{align}
where
$\fZ_\vecv=\iota^{-1}(\fZ f(\vecv)^{-1})$ and
$\fZ_\vecv'=\iota^{-1}((\vecy-\fZ) f(\vecv)^{-1})$.

We keep $\vecv\in S$,
and as before we write $\omega_\vecw=\varphi(\vecv',\vecw)$ 
and $\omega_\vecz=\varphi(\vecv',\vecz)$.
By \cite[Lemma 7.1]{lprob}, $\fZ_\vecv$ contains an open
right $(d-1)$-dimensional cone $\fC_\vecz$ 
with $\bn$ in its base, which has 
radius $r_\vecz$, height $h_\vecz$ and edge ratio $e_\vecz$, where
\begin{align}\label{EXACTFORMULAPROPPF4}
r_\vecz\asymp\xi^{\frac1d}(1-z+\sin^2\omega_\vecz)^{\frac12},
\qquad
h_\vecz\asymp\xi^{\frac1d}\min\Bigl(1,\frac{1-z+\omega_\vecz^2}{v_1}\Bigr),
\qquad
e_\vecz\asymp\min(1,\frac{1-z}{\sin^2\omega_\vecz}),
\end{align}
and, since $\vecy-\fZ=\xi^{\frac1d}(\fZ(0,1,1)+(0,\vecw))$,
$\fZ_\vecv'$ contains an open
right $(d-1)$-dimensional cone $\fC_\vecw$ 
with $\bn$ in its base, which
has radius $r_\vecw$ and height $h_\vecw$, where
\begin{align}\label{EXACTFORMULAPROPPF4a}
r_\vecw\asymp\xi^{\frac1d}(1-w+\sin^2\omega_\vecw)^{\frac12},
\qquad
h_\vecw\asymp\xi^{\frac1d}\min\Bigl(1,\frac{1-w+\omega_\vecw^2}{v_1}\Bigr).
\end{align}

Now by \cite[Cor.\ 1.4]{lprob} applied with $\fC_\vecz$,
and using $a_1\gg\xi^{\frac1d}$, %
we see that \eqref{EXACTFORMULAPROPPF3} forces
$e_\vecz^{\frac{d-1}2}\xi^{\frac1d}h_\vecz r_\vecz^{d-2}\ll1$.
If $v_1<1-z+\omega_\vecz^2$ then
\eqref{EXACTFORMULAPROPPF4} would give
\begin{align*}
1\gg e_\vecz^{\frac{d-1}2}\xi^{\frac1d}h_\vecz r_\vecz^{d-2}
\gg\xi(1-z)^{\frac{d-1}2} %
(1-z+\sin^2\omega_\vecz)^{-\frac12}
\gg\xi(1-z)^{\frac{d-1}2},
\end{align*}
which is impossible if $c_\clowW$ is sufficiently large.
Hence we must have $v_1\geq1-z+\omega_\vecz^2$
(in particular $\omega_\vecz<1$).
We now obtain
\begin{align*}
1\gg e_\vecz^{\frac{d-1}2}\xi^{\frac1d}h_\vecz r_\vecz^{d-2}
\gg \xi(1-z)^{\frac{d-1}2}\frac{\sqrt{1-z}+\omega_\vecz}{v_1}.
\end{align*}
Hence we conclude that,
for any $\vecv\in S$ satisfying 
our assumption \eqref{EXACTFORMULAPROPPF3},
\begin{align}\label{EXACTFORMULAPROPPF5}
\omega_\vecz\ll\xi^{-1}(1-z)^{\frac{1-d}2}v_1
\qquad\text{and}\qquad (1-z)^{\frac12}\ll\xi^{-1}(1-z)^{\frac{1-d}2}v_1.
\end{align}
In particular by taking $c_\clowW$ large we can force
$\omega_\vecz$ to be less than any fixed small constant of our choice.
Also since 
$\omega_\vecw\leq\varphi+\omega_\vecz\leq\frac\pi2+\omega_\vecz$
we may from now on assume $\omega_\vecw<\frac34\pi$.

Note from the proof of \cite[Lemma 7.1]{lprob}
that the heights of the cones $\fC_\vecz$ and $\fC_\vecw$ are both parallel
to the line
\begin{align*}
L=\iota^{-1}\bigl((\vecv^\perp\cap\text{Span}\{\vece_1,\vecv\})f(\vecv)^{-1}\bigr)
=\R\,\iota^{-1}\bigl((\|\vecv'\|^2,-v_1\vecv')f(\vecv)^{-1}\bigr)
\subset\R^{d-1}.
\end{align*}
Thus if we let
$T\in G^{(d-1)}$ be the linear map which acts by scalar multiplication by
$(h_\vecw/r_\vecw)^{\frac1{d-1}}$ on every vector in $L^\perp\subset\R^{d-1}$ 
and multiplication
by $(h_\vecw/r_\vecw)^{\frac{2-d}{d-1}}$ on every vector in $L$, then
$a_1^{\frac1{d-1}}\fC_\vecw T$ is a cone which has both height and
radius $=a_1^{\frac1{d-1}}h_\vecw^{\frac1{d-1}}r_\vecw^{\frac{d-2}{d-1}}$,
and hence by Lemma~\ref{A1LARGELEM} any $\tM\in\F_{d-1}\subset\Si_{d-1}$ with 
$\Z^{d-1}\tM\cap a_1^{\frac1{d-1}}C_\vecw T=\emptyset$ must have
$\ta_1\gg a_1^{\frac1{d-1}}h_\vecw^{\frac1{d-1}}r_\vecw^{\frac{d-2}{d-1}}$.
Also $a_1^{\frac1{d-1}}\fC_\vecz T$ is a cone with radius $r_\vecz'$
and height $h_\vecz'$, where
\begin{align}\label{EXACTFORMULAPROPPF6}
r_\vecz'=a_1^{\frac1{d-1}}\Bigl(\frac{h_\vecw}{r_\vecw}\Bigr)^{\frac1{d-1}}
r_\vecz,
\qquad
h_\vecz'=a_1^{\frac1{d-1}}\Bigl(\frac{h_\vecw}{r_\vecw}\Bigr)^{\frac{2-d}{d-1}}
h_\vecz,
\end{align}
and edge ratio $e_\vecz$ as before.
Let us temporarily assume $\omega_\vecw^2>1-z+\omega_\vecz^2$.
We then claim that $h_\vecz'\ll r_\vecz'$.
Indeed, using \eqref{EXACTFORMULAPROPPF4} and \eqref{EXACTFORMULAPROPPF4a},
and recalling that $v_1\geq1-z+\omega_\vecz^2$
and $\omega_\vecw<\frac34\pi$,
this claim is seen to be equivalent with
\begin{align*}
\frac{(1-z+\omega_\vecz^2)^{\frac12}(1-w+\omega_\vecw^2)^{\frac12}}{v_1}
\ll\min\Bigl(1,\frac{1-w+\omega_\vecw^2}{v_1}\Bigr),
\end{align*}
which is true because of
$1-z+\omega_\vecz^2<1-w+\omega_\vecw^2\ll1$ and \eqref{EXACTFORMULAPROPPF5}.
Now since $h_\vecz'\ll r_\vecz'$ we may just as well assume
$h_\vecz'\leq r_\vecz'$,
for if $h_\vecz'>r_\vecz'$ then we may shrink 
the cone $C_\vecz$ by keeping the base fixed while
decreasing $h_\vecz$ until
$h_\vecz'=r_\vecz'$, and \eqref{EXACTFORMULAPROPPF4} remains true.
Now \eqref{EXACTFORMULAPROPPF3} implies
\begin{align*}
\mu^{(d-1)}\Bigl(\Bigl\{\tM\in\F_{d-1}\col
\Z^{d-1}\tM\cap a_1^{\frac1{d-1}}\bigl(\fC_\vecz\cup\fC_\vecw\bigr)T
=\emptyset
\Bigr\}\Bigr)>0,
\end{align*}
and thus \cite[Lemma 7.4]{lprob} applies to give
$e_\vecz\ll((a_1^{\frac1{d-1}}h_\vecw^{\frac1{d-1}}r_\vecw^{\frac{d-2}{d-1}})
h_\vecz'{r_\vecz'}^{d-3})^{-\frac2{d-1}}$.
Hence using \eqref{EXACTFORMULAPROPPF6}, \eqref{EXACTFORMULAPROPPF4},
\eqref{EXACTFORMULAPROPPF4a} 
we conclude
\begin{align}\label{EXACTFORMULAPROPPF8}
\omega_\vecw\ll\xi^{-1}(1-z)^{\frac{1-d}2}v_1.
\end{align}
This is of course true also when
$\omega_\vecw^2\leq1-z+\omega_\vecz^2$, by \eqref{EXACTFORMULAPROPPF5};
hence \eqref{EXACTFORMULAPROPPF8} holds for all 
$\vecv\in S$ satisfying our assumption \eqref{EXACTFORMULAPROPPF3}.

By \eqref{EXACTFORMULAPROPPF5} and \eqref{EXACTFORMULAPROPPF8}
we can force both $\omega_\vecz$ and $\omega_\vecw$ to be less than
any fixed small constant of our choice, by taking $c_\clowW$
large. Hence also $\varphi$ and $\varphi(\vecv',\vece_1)$ are 
forced to be small,
since $\varphi\leq\omega_\vecz+\omega_\vecw$
and $\varphi(\vecv',\vece_1)\leq\omega_\vecz+\varphi(\vecz,\vece_1)
\leq\omega_\vecz+\varphi$.
Furthermore both $z$ and $w$ must be near $1$, 
by \eqref{EXACTFORMULAPROPPF5} and using $z\leq w<1$.
Hence by Lemma \ref{GOODINTLEM2}, if $c_\clowW$ is sufficiently large
then $(na_1\vecv+\vecv^\perp)\cap\fZ=\emptyset$ holds for all
$\vecv\in S$
satisfying \eqref{EXACTFORMULAPROPPF3}
and all $n\in\Z\setminus\{0,1\}$;
and hence
\begin{align}\notag
\mu^{(d-1)}\Bigl(\Bigl\{\tM\in\F_{d-1}\col L_{\vecv,\tM}\cap\fZ=\emptyset
\Bigr\}\Bigr)
\hspace{150pt}
\\\label{EXACTFORMULAPROPPF9}
=\mu^{(d-1)}\Bigl(\Bigl\{\tM\in\F_{d-1}\col 
\Z^{d-1}\tM\cap a_1^{\frac1{d-1}}\bigl(\fZ_\vecv\cup\fZ_\vecv'\bigr)
=\emptyset\Bigr\}\Bigr).
\end{align}
Furthermore by \eqref{EXACTFORMULAPROPPF5} 
we have $\frac{1-z+\omega_\vecz^2}{v_1}\ll\frac{\sqrt{1-z}+\omega_\vecz}{v_1}
\ll\xi^{-1}(1-z)^{\frac{1-d}2}$,
and hence by taking $c_\clowW$ large we can force the ratio
$\frac{1-z+\omega_\vecz^2}{v_1}$ to be smaller than any fixed constant
of our choice;
similarly by \eqref{EXACTFORMULAPROPPF8} we can also force
$\frac{1-w+\omega_\vecw^2}{v_1}$ to be small.
Hence by Lemma \ref{GOODINTLEM} (applied after appropriate rotations), 
if $c_\clowW$ is sufficiently large
then both $\vecv^\perp\cap\fZ\subset\{x_1<\frac12\xi^{\frac1d}\}$
and $\vecv^\perp\cap(\vecy-\fZ)\subset\{x_1<\frac12\xi^{\frac1d}\}$
must hold, for all
$\vecv\in S$
satisfying \eqref{EXACTFORMULAPROPPF3}.
Hence by \eqref{EXACTFORMULAPROPPF9} and the same argument as we used to get 
\eqref{CVRELATION}, it follows that \eqref{EXACTFORMULAPROPPF7} holds for 
all such $\vecv$,
and we are done.
\end{proof}

\subsection{Bound from above on the support of $\Phi_\bn$}

We will now prove an upper bound on $\xi_0(w,z,\varphi)$ which
together with Proposition \ref{PHI0SUPPORTASYMPTTHMHALFPROP} will
complete the proof of Theorem \ref{PHI0SUPPORTASYMPTTHM}.
We first prove a weak form of the desired statement.
\begin{prop}\label{PHI0SUPUPBDWEAKPROP}
There is a constant $c_\clowX>1$ which only depends on $d$ such that
for any $0\leq z\leq w<1$, $0\leq\varphi\leq\pi$
and any  %
\begin{align*} %
\xi\geq \max\biggl\{c_\clowX(1-z)^{\frac{1-d}2},
2^{1-\frac d2}(1-z)^{-\frac d2}
\rho\biggl(\sqrt{\frac{1-w}{1-z}},
\frac{\varphi}{\sqrt{2(1-z)}}\biggr)^{-1}
\Bigl(1+c_\clowX\xi^{-2}(1-z)^{1-d}\Bigr)
\biggr\},
\end{align*}
we have %
$\Phi_\bn(\xi,w,z,\varphi)=0$.
\end{prop}
\begin{proof}
Assume $0\leq z\leq w<1$, $0\leq\varphi\leq\pi$,
$\xi\geq c_\clowX(1-z)^{\frac{1-d}2}$
(where we will successively impose conditions on $c_\clowX$ being sufficiently
large), and $\Phi_\bn(\xi,w,z,\varphi)>0$.
Then by Proposition~\ref{PHI0SUPPORTTHM} we have
$(1-z+\varphi^2)^{\frac12}\ll\xi^{-1}(1-z)^{\frac{1-d}2}$.
Hence if %
$c_\clowX$ is sufficiently large then 
$1-c_\clowL\leq z\leq w<1$ and $0\leq\varphi\leq c_\clowL$,
where $c_\clowL$ is the constant in
Proposition \ref{UPSILONMAINUPPERBOUNDPROP}.
We also require $c_\clowX\geq c_\clowW$; then by 
Proposition \ref{EXACTFORMULAPROP} we have
\begin{align}\label{PHISUPUPPERBDPF1}
\Phi_\bn(\xi,\vecw,\vecz)
=\zeta(d)^{-1}\int_{S}
\Upsilon\bigl(\vecz,\vecw,\vecv',\xi^{-\frac{d-1}d}v_1(\vecy\cdot\vecv)^{-1}
\bigr) \,\frac{d\vecv}{(\vecy\cdot\vecv)^d},
\end{align}
where we now take $\vecw,\vecz\in\scrB_1^{d-1}$ as in \eqref{CASE12FORMULA}.

We saw in the proof of 
Proposition \ref{EXACTFORMULAPROP} that 
$\Upsilon(\vecz,\vecw,\vecv',\xi^{-\frac{d-1}d}v_1(\vecy\cdot\vecv)^{-1})>0$
can only hold for $\vecv\in S$ with
$\omega_\vecz,\omega_\vecw\ll\xi^{-1}(1-z)^{\frac{1-d}2}$,
and then we must also have
$\varphi<c_\clowY\xi^{-1}(1-z)^{\frac{1-d}2}$ and
$\varphi(\vecv',\vece_1)<c_\clowY\xi^{-1}(1-z)^{\frac{1-d}2}$
(where $c_\clowY>0$ is some constant which only depends on $d$),
since $\varphi\leq\omega_\vecz+\omega_\vecw$ and
$\varphi(\vecv',\vece_1)\leq\omega_\vecz+\varphi(\vecz,\vece_1)
\leq\omega_\vecz+\varphi$.
Thus by requiring $c_\clowX>c_\clowY/c_\clowL$ 
we force $\varphi(\vecv',\vece_1) %
<c_\clowL$
(and hence afortiori $\frac{\|\vecv''\|}{\|\vecv'\|}=
\sin\varphi(\vecv',\vece_1)<c_\clowL$) to hold 
for all $\vecv\in S$ with 
$\Upsilon(\vecz,\vecw,\vecv',\xi^{-\frac{d-1}d}v_1(\vecy\cdot\vecv)^{-1})>0$.
It now follows from our assumption $\Phi_\bn(\xi,\vecw,\vecz)>0$
together with \eqref{PHISUPUPPERBDPF1} and 
Proposition \ref{UPSILONMAINUPPERBOUNDPROP} that
there exists some $\vecv\in S$
satisfying $\varphi(\vecv',\vece_1)<c_\clowY\xi^{-1}(1-z)^{\frac{1-d}2}$
and
\begin{align}\label{PHISUPUPPERBDPF2}  %
2^{1-\frac d2}(1-z)^{-\frac d2}\xi^{-\frac{d-1}d}\frac{v_1}{\vecy\cdot\vecv}
\Bigl\{1+c_\clowM\Bigl(1-z+\varphi^2+\sfrac{\|\vecv''\|^2}{\|\vecv'\|^2}\Bigr)
\Bigr\}
>\rho\biggl(\sqrt{\frac{1-w}{1-z}},
\frac{\varphi}{\sqrt{2(1-z)}}\biggr).
\end{align}
Note that the computation in \eqref{A1ASYMPT} applies,
and since both $(1-z+\varphi^2)^{\frac12}$ and
$\|\vecv''\|$ are $\ll\xi^{-1}(1-z)^{\frac{1-d}2}$,
and $v_1+2v_2\gg1$ (as follows from $v_1>0$ and
$\varphi(\vecv',\vece_1)<c_\clowL$), we get
\begin{align*}
\vecy\cdot\vecv=\xi^{\frac1d}(v_1+2v_2)(1+O(\xi^{-2}(1-z)^{1-d})).
\end{align*}
We require that $c_\clowX$ is so large that the big $O$-term
in the last expression has absolute value $\leq\frac12$; it then follows that 
the left hand side of \eqref{PHISUPUPPERBDPF2} is
$<2^{1-\frac d2}(1-z)^{-\frac d2}\xi^{-1}
(1+O(\xi^{-2}(1-z)^{1-d}))$,
and thus
\begin{align*}
\xi<2^{1-\frac d2}(1-z)^{-\frac d2}
\rho\biggl(\sqrt{\frac{1-w}{1-z}},
\frac{\varphi}{\sqrt{2(1-z)}}\biggr)^{-1}
\Bigl(1+O\bigl(\xi^{-2}(1-z)^{1-d}\bigr)\Bigr).
\end{align*}
Hence the statement of the proposition follows, after increasing
$c_\clowX$ if necessary so as to be larger than or equal to
the implied constant in the last big $O$-term.
\end{proof}

\begin{cor}\label{BASICRHOLOWBOUNDCOR}
We have $\rho(a,b)\asymp1+|b|$, uniformly over all $0<a\leq1$, $b\in\R$.
\end{cor}
\begin{proof}
Because of Corollary \ref{BASICRHOBOUNDCOR} and $\rho(a,-b)=\rho(a,b)$,
it suffices to prove $\rho(a,b)\ll1+b$ for all $0<a\leq1$, $b\geq0$.
Given any $a\in(0,1]$ and $b\geq0$ we may find $z\leq w<1$ and $\varphi\geq0$
satisfying $a=\sqrt{\frac{1-w}{1-z}}$,
$b=\frac{\varphi}{\sqrt{2(1-z)}}$, and 
with both $1-z$ and $\varphi$ %
arbitrarily small.
Now Proposition~\ref{PHI0SUPPORTTHM} says that
$\xi_0(w,z,\varphi)\asymp(1-z)^{-\frac d2}(1+b)^{-1}$;
thus if we take $1-z$ sufficiently small (for our fixed $a,b$)
we have $\xi_0(w,z,\varphi)>c_\clowX(1-z)^{\frac{1-d}2}$.
But $\Phi_\bn(\xi,w,z,\varphi)>0$ for all $\xi<\xi_0(w,z,\varphi)$,
and hence Proposition \ref{PHI0SUPUPBDWEAKPROP} implies that
$\xi_0(w,z,\varphi)\leq2^{1-\frac d2}(1-z)^{-\frac d2}\rho(a,b)^{-1}
(1+c_\clowX^{-1})$.
Using here 
$\xi_0(w,z,\varphi)\asymp(1-z)^{-\frac d2}(1+b)^{-1}$
we obtain $\rho(a,b)\ll1+b$, as desired.
\end{proof}

Using Corollary \ref{BASICRHOLOWBOUNDCOR} we are now able to make
Proposition \ref{PHI0SUPUPBDWEAKPROP} a bit more precise, as follows:
\begin{prop}\label{PHI0SUPUPBDPROP}
We have
\begin{align}\notag %
\xi_0(w,z,\varphi)\leq2^{1-\frac d2}(1-z)^{-\frac d2}
\rho\biggl(\sqrt{\frac{1-w}{1-z}},
\frac{\varphi}{\sqrt{2(1-z)}}\biggr)^{-1}
\Bigl\{1+O\Bigl(\max(1-z,1-w)+\varphi^2\Bigr)\Bigr\},
\end{align}
uniformly over all $z,w\in[0,1)$, $\varphi\in[0,\pi]$.
\end{prop}
\begin{proof}
Because of $\xi_0(w,z,\varphi)=\xi_0(z,w,\varphi)$ and 
\eqref{RHOSYMM}, we may assume $z\leq w$ without loss of generality.
Let us write $a=\sqrt{\frac{1-w}{1-z}}$ and
$b=\frac{\varphi}{\sqrt{2(1-z)}}$, as usual.
It follows from Corollary \ref{BASICRHOLOWBOUNDCOR} that there is a constant
$0<c<1$ which only depends on $d$ such that
$2^{1-\frac d2}(1-z)^{-\frac d2}\rho(a,b)^{-1}%
>c_\clowX(1-z)^{\frac{1-d}2}$ holds whenever 
$1-z\leq c$ and $0\leq\varphi\leq c$.
Hence for any such $z,w,\varphi$, if we let $\xi_1$ be the unique
real positive solution to the equation
\begin{align*}
\xi_1=2^{1-\frac d2}(1-z)^{-\frac d2}
\rho(a,b)^{-1}(1+c_\clowX\xi_1^{-2}(1-z)^{1-d}),
\end{align*}
then $\xi_1>c_\clowX(1-z)^{\frac{1-d}2}$
and hence by Proposition \ref{PHI0SUPUPBDWEAKPROP} we have
$\xi_0(w,z,\varphi)\leq\xi_1$.
Also $\xi_1\asymp(1-z)^{-\frac d2}(1+\frac{\varphi}{\sqrt{1-z}})^{-1}$;
thus $\xi_1^{-2}(1-z)^{1-d}\asymp1-z+\varphi^2$ and
\begin{align*}
\xi_1=2^{1-\frac d2}(1-z)^{-\frac d2}
\rho(a,b)^{-1}(1+O(1-z+\varphi^2)),
\end{align*}
which implies the desired bound. %
In the remaining case when $1-z>c$ or $\varphi>c$,
the desired bound follows automatically from
Corollary \ref{BASICRHOLOWBOUNDCOR} and
Proposition \ref{PHI0SUPPORTTHM}.
\end{proof}

Note that Theorem \ref{PHI0SUPPORTASYMPTTHM} follows from
Proposition \ref{PHI0SUPUPBDPROP} and
Proposition \ref{PHI0SUPPORTASYMPTTHMHALFPROP},
together with the discussion at the end of Section \ref{XI0WZPHISEC}.
\hfill $\square$ $\square$ $\square$

\vspace{5pt}

We may now also give the simple proof of Corollary \ref{PHIXIWSUPCOR}:
\begin{proof}[Proof of Corollary \ref{PHIXIWSUPCOR}]
The existence and continuity of $\xi_0(w)$ is proved by a similar
argument as for $\xi_0(w,z,\varphi)$
(cf.\ the end of Section \ref{XI0WZPHISEC}),
working directly from the definition of $\Phi(\xi,\vecw)$, \eqref{PHIXIWDEF}.
In particular, the upper semicontinuity of $\xi_0(w)$ follows from the fact
that $\Phi(\xi,w)\geq\Phi(\xi',w')$ holds whenever
$\xi'\geq k^{d-1}\xi$, where $k=\frac{1+\ve w}{1+\ve w'}$
with $\ve=\sgn(w-w')$.
This in turn follows from the inclusion
$w\vece_2+\fZ(0,\xi,1)\subset 
(w'\vece_2+\fZ(0,\xi,1))D$, where $D=\text{diag}(1,k,k,\ldots,k)$.

Next, in order to prove the asymptotic formula \eqref{PHIXIWSUPCORRES} we
note that, by \eqref{PHIFROMPHIZERO},
\begin{align}\label{PHIXIWSUPCORPF1}
\xi_0(w)=\sup\bigl\{\xi_0(w,z,\varphi)\col
z\in[0,1),\:\varphi\in[0,\pi]\bigr\},
\qquad\forall w\in[0,1).
\end{align}
It follows from Proposition \ref{PHI0SUPPORTTHM}
that there exists a constant $c>1$ which only depends on $d$
such that, for any $w\in[0,1)$,
the supremum in \eqref{PHIXIWSUPCORPF1} remains unchanged if we restrict
to the set
\begin{align*}
S_w:=\bigl\{(z,\varphi)\in[0,1)\times[0,\pi]\col 1-z\leq c(1-w),\:
\varphi\leq c\sqrt{1-w}\bigr\}.
\end{align*}
Hence by Theorem \ref{PHI0SUPPORTASYMPTTHM}, as $w\to1^-$ we have
\begin{align}\label{PHIXIWSUPCORPF2}
\xi_0(w)
=\biggl(\sup_{(z,\varphi)\in S_w}
2^{1-\frac d2}(1-w)^{-\frac d2}
\rho\Bigl(\sqrt{\frac{1-z}{1-w}},\frac\varphi{\sqrt{2(1-w)}}\Bigr)^{-1}
\biggr)\bigl\{1+O(1-w)\bigr\}.
\end{align}
However it is immediate from the definition \eqref{XIDM1DEFnew} that
$\Xi(\bn,\vecy;\vech;v)\leq\Xi(\bn,\bn;\vech;v)$;
hence $\rho(a,b)\geq\rho(1,0)$ for all $a>0$, $b\in\R$
(cf.\ \eqref{XIDM1ABDEF} and \eqref{RHODABDEF}),
so that the supremum in \eqref{PHIXIWSUPCORPF2} is attained at
$z=w$, $\varphi=0$.
Hence \eqref{PHIXIWSUPCORRES} holds.
\end{proof}

\subsection{\texorpdfstring{Proof of Proposition \ref*{PHIZEROASYMPT}}{Proof of Proposition 1.15}}

Changing to a slightly different topic, we now give the quick proof of
Proposition \ref{PHIZEROASYMPT}.

It follows e.g.\ from \eqref{PHIXIAVFORMULA3} and \cite[Lemma 7.11]{lprob}
that $\Phi_\bn(\xi)$ is a continuous and decreasing function of $\xi$;
furthermore \eqref{PHIFROMPHIZERO} and \eqref{PHIXIAVFORMULA3} 
imply $\Phi(\xi,\bn)=\int_\xi^\infty\Phi_\bn(\eta)\,d\eta$.
Hence we have $\Phi_\bn(\xi)>0$ if and only if $0<\xi<\xi_0(0)$,
and by \eqref{PHIXIWDEF} this holds if and only if
\begin{align}\label{PHIZEROASYMPTPF1}
\mu(\{M\in X_1\col\Z^dM\cap\fZ(0,\xi,1)=\emptyset\})>0.
\end{align}
Using $-\Z^dM=\Z^dM$ and the fact that 
$\Z^dM\cap\vece_1^\perp=\{\bn\}$ holds for 
$\mu$-almost all $M\in X_1$, we see that \eqref{PHIZEROASYMPTPF1}
holds if and only if
\begin{align}\label{PHIZEROASYMPTPF2}
\mu(\{M\in X_1\col\Z^dM\cap\fZ(-\xi,\xi,1)=\{\bn\}\})>0.
\end{align}
But note that $M\in X_1$ satisfies
$\Z^dM\cap\fZ(-\xi,\xi,1)=\{\bn\}$ if and only if
$\Z^dM$ is a packing lattice of 
$\frac12\fZ(-\xi,\xi,1)$  %
(cf., e.g., \cite[Sec.\ 20, Thm 1]{GL}).
Hence \eqref{PHIZEROASYMPTPF2} holds if and only if
$\vol(\frac12\fZ(-\xi,\xi,1))<\delta_d^*(\fZ)$,
i.e.\ if and only if
$2^{1-d}v_{d-1}\xi<\delta_d^*(\fZ)$.
Hence $\xi_0(0)=2^{d-1}v_{d-1}^{-1}\delta^*_d(\fZ)$.
\hfill$\square$

\section{\texorpdfstring{Asymptotics for $\frac{\partial}{\partial\xi}\Phi(\xi,\vecw)$ derived from Theorem \ref*{PHI0XILARGETHM}}{Asymptotics for (d/dxi)Phi(xi,w) derived from Theorem 1.7}}
\label{PARTIALPHIXIZSEC}
In this section we use Theorem \ref{PHI0XILARGETHM} and
\eqref{PHIFROMPHIZERO} to derive an asymptotic formula for
$\frac{\partial}{\partial\xi}\Phi(\xi,\vecw)$ as $\xi\to\infty$.
In particular this results in a new proof of Theorem \ref{PHIXIWASYMPTTHM}
(except for a slightly worse $\log$-factor), %
and it also gives an internal check of consistency of our asymptotic
formulas.

\begin{thm}\label{PPHIXIWASYMPTTHM}
For any $d\geq3$ we have
\begin{align}\notag
-\frac{\partial}{\partial\xi}\Phi(\xi,w)
=\xi^{-3+\frac2d}G_d\bigl(\xi^{\frac2d}(1-w)\bigr)
\hspace{210pt}
\\
+O\biggl(\xi^{-3}\log\Bigl(2+\min(\xi,\xi^{-\frac2d}(1-w)^{-1})\Bigr)
\left.\begin{cases}
\log\xi&\text{if }\:d=3
\\\label{PPHIXIWASYMPTTHMRES2}
1&\text{if }\:d\geq4
\end{cases}\right\}\biggr),
\end{align}
as $\xi\to\infty$, uniformly over all $0\leq w<1$, where
\begin{align}\notag
G_d(t)=(2-\sfrac2d)F_d(t)-\sfrac2dF_d'(t)t
\hspace{240pt}
\\\label{PPHIXIWASYMPTTHMRES2a}
=\frac{2^{3(1-\frac d2)}\pi^{\frac d2-1}t^{\frac d2-1}}
{\Gamma(\frac d2-1)\zeta(d)}
\int_0^1\int_0^\infty\Xi\Bigl({\sigma},
2^{1-\frac d2}t^{-\frac d2}y\Bigr){\sigma}^{d-3}y(1-y)^{d-2}\,d{\sigma}\,dy.
\end{align}
The function $G_d(t)$ is a bounded continuous function from
$\R_{>0}$ to $\R_{\geq0}$.
\end{thm}
Note that the second equality in \eqref{PPHIXIWASYMPTTHMRES2a} 
follows immediately from the definition of $F_d(t)$
(cf.\ \eqref{PHIXIWASYMPTTHMRES}) and the fact that
\begin{align}\notag
\frac d{dt}\int_0^1\int_0^\infty\Xi\Bigl({\sigma},
2^{1-\frac d2}t^{-\frac d2}y\Bigr){\sigma}^{d-3}(1-y)^{d-1}\,d{\sigma}\,dy
\hspace{160pt}
\\\label{PPHIXIWASYMPTTHMRES2apf} 
=\frac d2t^{-1}\int_0^1\int_0^\infty\Xi\Bigl({\sigma},
2^{1-\frac d2}t^{-\frac d2}y\Bigr){\sigma}^{d-3}
(1-dy)(1-y)^{d-2}\,d{\sigma}\,dy
\end{align}
for all $t>0$.
This last identity is proved by substituting $y=t^{\frac d2}x$
in the outer integral in the left hand side, then carrying out the 
differentiation with respect to $t$, and finally substituting back
$x=t^{-\frac d2}y$.
The fact that $G_d(t)$ is bounded and continuous is proved by the
same argument as below \eqref{PHIXIWASYMPTTHMRES},
and similar considerations also justify the differentiation in
\eqref{PPHIXIWASYMPTTHMRES2apf}.
Note in particular that it follows that $F_d(t)$ is $\C^1$ on all $\R_{>0}$.

Note also that the main term in \eqref{PPHIXIWASYMPTTHMRES2}
is what is obtained by differentiation of the main term in
\eqref{PHIXIWASYMPTTHMRES2}.
Hence Theorem \ref{PPHIXIWASYMPTTHM} indeed implies
Theorem \ref{PHIXIWASYMPTTHM} (except for a slightly worse $\log$-factor)
upon integrating over $\xi\in[\xi_0,\infty)$,
and using the fact that for any fixed $w$ we have $\Phi(\xi,w)=0$
for all sufficiently large $\xi$.

\vspace{5pt}

The key fact needed for the deduction of Theorem \ref{PPHIXIWASYMPTTHM}
from Theorem \ref{PHI0XILARGETHM} is the following:
\begin{lem}\label{XIAVERAGELEM}
For any $\vech\in\R_+^{d-1}$ and $v>0$ we have
(writing $\vech'=(h_2,\ldots,h_{d-1})$ as usual)
\begin{align}\label{XIAVERAGELEMRES}
\int_{P^{d-1}}\Xi(\bn,\vecy;\vech;v)\,d\vecy
=v\,\Xi\Bigl(\frac{\|\vech'\|}{h_1},v\Bigr).
\end{align}
\end{lem}
\begin{proof}
By definition we have
\begin{align*}
\Xi(\bn,\vecy;\vech;v)
=\int_{\myX^{(d-1)}}I\Bigl((v^{\frac1{d-1}}\Z^{d-1}M)\cap P_\vech^{d-1}(\bn)=\emptyset\Bigr)
I\Bigl((v^{\frac1{d-1}}\Z^{d-1}M)\cap P_\vech^{d-1}(\vecy)=\emptyset\Bigr)\,d\mu(M).
\end{align*}
Hence by Fubini's Theorem the left hand side of \eqref{XIAVERAGELEMRES} equals
\begin{align*}
\int_{\myX}I\Bigl((v^{\frac1{d-1}}\Z^{d-1}M)\cap P_\vech^{d-1}(\bn)=\emptyset\Bigr)
\vol_{\R^{d-1}}\Bigl(
\bigl\{\vecy\in P^{d-1}\col 
(v^{\frac1{d-1}}\Z^{d-1}M)\cap P_\vech^{d-1}(\vecy)=\emptyset\bigr\}\Bigr)
\,d\mu(M).
\end{align*}
But for $\mu$-almost every $M\in \myX$ we have
$\Z^{d-1}M\cap\vech^\perp=\{\bn\}$, and for each such $M$ the set
\begin{align}\label{XIAVERAGELEMPF1}
\bigl\{\vecy\in P^{d-1}\col 
(v^{\frac1{d-1}}\Z^{d-1}M)\cap P_\vech^{d-1}(\vecy)=\emptyset\bigr\}
\end{align}
is in fact a fundamental domain for $\R^{d-1}/(v^{\frac1{d-1}}\Z^{d-1}M)$.
Indeed, for every $\vecx\in\R^{d-1}$ the set
$\Omega_\vecx=P^{d-1}\cap(\vecx+v^{\frac1{d-1}}\Z^{d-1}M)$ is infinite, and 
there is exactly one $\vecy\in \Omega_\vecx$ for which
$(v^{\frac1{d-1}}\Z^{d-1}M)\cap P_\vech^{d-1}(\vecy)=\emptyset$,
namely that $\vecy\in \Omega_\vecx$ for which $\vecy\cdot\vech$ is minimal
(the uniqueness is guaranteed since $\Z^{d-1}M\cap\vech^\perp=\{\bn\}$).
Hence the volume of the set in \eqref{XIAVERAGELEMPF1} equals $v$,
and the lemma follows.
\end{proof}

\begin{proof}[Proof of Theorem \ref{PPHIXIWASYMPTTHM}]
As in section \ref{SIMPLINTSEC} we fix the constant
$c_\clowD$ so that 
$c_\clowD\geq\sqrt{\sigma_d(1,0)}$ (cf.\ \eqref{FDTCOMPSUPP}) and
$\Phi(\xi,w)=0$ whenever $1-w\geq c_\clowD\xi^{-\frac2d}$.
Then $G_d(t)=0$ for all $t\geq c_\clowD$, and hence
\eqref{PPHIXIWASYMPTTHMRES2} is automatic whenever
$1-w\geq c_\clowD\xi^{-\frac2d}$.
Hence from now on we will assume $1-w<c_\clowD\xi^{-\frac2d}$.

By \eqref{PHIFROMPHIZERO} we have:
\begin{align}\label{PPHIXIWASYMPTTHMPF2}
-\frac{\partial}{\partial\xi}\Phi(\xi,w)=
\int_{\scrB_1^{d-1}}\Phi_\bn(\xi,w\vece_1,\vecz)\,d\vecz.
\end{align}
Let us write $\varphi:=\varphi(\vecz,\vece_1)$ and $z=\|\vecz\|$.
By Proposition \ref{PHI0SUPPORTTHM} there is a constant $c>0$
which only depends on $d$ such that $\Phi_\bn(\xi,w\vece_1,\vecz)=0$
holds for all $\vecz\in\scrB_1^{d-1}\setminus(U_1\cup U_2)$, where
\begin{align*}
&U_1:=\Bigl\{\vecz\in\scrB_1^{d-1}\col
\varphi\leq\sfrac\pi2;\:\:
\varphi<c\xi^{-1}(1-\min(w,z))^{\frac{1-d}2};\:\:
1-\min(z,w)<c\xi^{-\frac2d}\Bigr\};
\\
&U_2:=\Bigl\{\vecz\in\scrB_1^{d-1}\col
\varphi>\sfrac\pi2;\:\:
1-\min(z,w)<c\Bigl(\xi^{-\frac2{d-2}}
+%
\bigl(\xi/(\pi-\varphi)\bigr)^{-\frac2{d-1}}\Bigr)\Bigr\}.
\end{align*}
Now for $\vecz\in U_1$ we apply Theorem \ref{PHI0XILARGETHM} to
$\Phi_\bn(\xi,w\vece_1,\vecz)=\Phi_\bn(\xi,w,z,\varphi)
=\Phi_\bn(\xi,z,w,\varphi)$,
while for $\vecz\in U_2$ we apply the bound from
Theorem \ref{CYLINDER2PTSMAINTHM}. This gives

\begin{align}\notag
-\frac{\partial}{\partial\xi}\Phi(\xi,w)
=\int_{U_1}\biggl(\frac{2^{2-\frac32d}(1-w)^{\frac d2-1}\xi^{-1}}{\zeta(d)}
\int_{(0,1)\times\R^{d-2}}
\Xi\biggl(\sqrt{\frac{1-z}{1-w}},
\frac{\varphi}{\sqrt{2(1-w)}};\vech;
\hspace{50pt}
\\\label{PPHIXIWASYMPTTHMPF1}
2^{1-\frac d2}(1-w)^{-\frac d2}\xi^{-1}(1-h_1)\biggr)\,d\vech
+O(E)\biggr)\,d\vecz\hspace{10pt}
\\\notag
+O\biggl(\int_{U_2}\xi^{-2}
\min\Bigl\{1,
(\xi(\pi-\varphi)^{d-2})^{-1+\frac2{d-1}}\Bigr\}\,d\vecz\biggr),
\end{align}
where %
$E$ is as in \eqref{PHI0XILARGETHMEDEF}. 

We parametrize $\vecz\in\scrB_1^{d-1}$ as
\begin{align*}
(0,1)\times(0,\pi)\times\S_1^{d-3}\ni\langle z,\varphi,\vecomega\rangle
\mapsto\vecz=\bigl(z\cos\varphi,z(\sin\varphi)\vecomega\bigr)\in\scrB_1^{d-1}.
\end{align*}
Then
\begin{align}\label{DZEXPR}
d\vecz=z^{d-2}(\sin\varphi)^{d-3}\,dz\,d\varphi\,d\vecomega.
\end{align}
Let us first consider the contribution from the error term $O(E)$ in 
\eqref{PPHIXIWASYMPTTHMPF1}. %
For $\vecz\in U_1$ we have
$0\leq\varphi<\min(\frac\pi2,c\xi^{-1}(1-w)^{\frac{1-d}2})$
and %
$0<1-z<\min(c\xi^{-\frac2d},(c^{-1}\xi\varphi)^{-\frac2{d-1}})
\ll\xi^{-\frac2d}\min(1,(\xi\varphi^d)^{-\frac2{d(d-1)}})$.
Hence the contribution from the $O(E)$-term in \eqref{PPHIXIWASYMPTTHMPF1}
is:
\begin{align*}
\ll\int_0^{\min(\frac\pi2,c\xi^{-1}(1-w)^{\frac{1-d}2})}
\xi^{-\frac2d}\min(1,(\xi\varphi^d)^{-\frac2{d(d-1)}})E\varphi^{d-3}\,d\varphi.
\end{align*}
We may assume that $c$ is so large that
$c^{\frac2{d-1}}>c_\clowD$.
Then $\xi^{-\frac1d}<c\xi^{-1}(1-w)^{\frac{1-d}2}$,
since $1-w<c_\clowD\xi^{-\frac2d}$, and hence if $d\geq4$ then we get
(cf.\ \eqref{TOUGHXIINTLEMAPPL})
\begin{align*}
=\int_0^{\xi^{-1/d}}\xi^{-2-\frac2d}\varphi^{d-3}\,d\varphi
+\int_{\xi^{-1/d}}^{\min(\frac\pi2,c\xi^{-1}(1-w)^{\frac{1-d}2})}
\xi^{-3}\varphi^{-1}\,d\varphi
\ll\xi^{-3}\log\Bigl(2+\min(\xi,\xi^{-\frac2d}(1-w)^{-1})\Bigr).
\end{align*}
When $d=3$ we get the same bound except for an extra factor $\log\xi$.
The integral over $U_2$ in \eqref{PPHIXIWASYMPTTHMPF1} is
easily seen to be $\ll\xi^{-3}$.
(This bound was also pointed out in \cite[Cor.\ 1.10]{lprob}.)
Hence we have
\begin{align}\notag
-\frac{\partial}{\partial\xi}\Phi(\xi,w)=
\frac{2^{2-\frac32d}}{\zeta(d)}(1-w)^{\frac d2-1}\xi^{-1}
\int_{U_1}\int_{(0,1)\times\R^{d-2}}
\Xi\biggl(\sqrt{\frac{1-z}{1-w}},
\frac{\varphi}{\sqrt{2(1-w)}};\vech;
\hspace{60pt}
\\\label{PPHIXIWASYMPTTHMPF3}
2^{1-\frac d2}(1-w)^{-\frac d2}\xi^{-1}(1-h_1)\biggr)\,d\vech\,d\vecz
+O(E'),
\end{align}
where we write $E'$ for the error majorant in the second line of
\eqref{PPHIXIWASYMPTTHMRES2}.

Next note that in \eqref{DZEXPR} we have
\begin{align*}
\bigl|z^{d-2}(\sin\varphi)^{d-3}-\varphi^{d-3}\bigr|
\ll(1-z+\varphi^2)\varphi^{d-3}
\ll(\xi^{-\frac2d}+\varphi^2)\varphi^{d-3},
\end{align*}
uniformly over all $\vecz\in U_1$.
Using this fact together with Lemma \ref{SABVBOUNDLEMCOR} we see that
the error caused by replacing $d\vecz$ by
$\varphi^{d-3}\,dz\,d\varphi\,d\vecomega$ in \eqref{PPHIXIWASYMPTTHMPF3} is
\begin{align*}
\ll\int_{U_1}
\min\Bigl(\xi^{-2+\frac2d},\xi^{-3+\frac2{d-1}}\varphi^{-d+\frac2{d-1}}\Bigr)
(\xi^{-\frac2d}+\varphi^2)\varphi^{d-3}\, dz\,d\varphi\,d\vecomega,
\end{align*}
and this is seen to be $\ll E'$ by a computation which is very similar
to the computation above bounding the contribution from $O(E)$ in
\eqref{PPHIXIWASYMPTTHMPF1}.
Hence
\begin{align}\notag
-\frac{\partial}{\partial\xi}\Phi(\xi,w)=
\frac{2^{2-\frac32d}}{\zeta(d)}(1-w)^{\frac d2-1}\xi^{-1}
\int_{U_1}\int_{(0,1)\times\R^{d-2}}
\Xi\biggl(\sqrt{\frac{1-z}{1-w}},
\frac{\varphi}{\sqrt{2(1-w)}};\vech;
\hspace{30pt}
\\\label{PPHIXIWASYMPTTHMPF4}
2^{1-\frac d2}(1-w)^{-\frac d2}\xi^{-1}(1-h_1)\biggr)\,d\vech\,
\varphi^{d-3}\,dz\,d\varphi\,d\vecomega+O(E').
\end{align}
Now by Corollary \ref{BASICRHOBOUNDCOR} and \eqref{RHOSYMM}
we see that by requiring that the constant $c$
has been chosen sufficiently large (in a way that only depends on $d$)
we have that any tuple 
$\langle z',w',\varphi,\xi\rangle\in
(\R_{>0})^4$
for which
\begin{align}\notag
\int_{(0,1)\times\R^{d-2}}
\Xi\biggl(\sqrt{\frac{z'}{w'}},
\frac{\varphi}{\sqrt{2w'}};\vech;
2^{1-\frac d2}{w'}^{-\frac d2}\xi^{-1}(1-h_1)\biggr)\,d\vech>0
\end{align}
must satisfy both
$\varphi<c\xi^{-1}\max(z',w')^{\frac{1-d}2}$ and
$\max(z',w')<c\xi^{-\frac2d}$.
From now on we assume that $\xi$ is so large that
$c\xi^{-\frac2d}<1$.
It then follows that in the outer integral in \eqref{PPHIXIWASYMPTTHMPF4}
we may extend the range $U_1$ to the larger set of all
\begin{align}\label{PPHIXIWASYMPTTHMPF5}
\langle z,\varphi,\vecomega\rangle
\in(-\infty,1)\times
\bigl(0,\min(\sfrac\pi2,c\xi^{-1}(1-w)^{\frac{1-d}2})\bigr)
\times\S_1^{d-3}
\end{align}
without changing the value of the integral.
In fact if $c\xi^{-1}(1-w)^{\frac{1-d}2}\leq\frac\pi2$ then
we may extend the range all the way to
$(-\infty,1)\times\R_{>0}\times\S_1^{d-3}$
without changing the value of the integral; 
on the other hand if $c\xi^{-1}(1-w)^{\frac{1-d}2}>\frac\pi2$ then
we may extend the range to $(-\infty,1)\times\R_{>0}\times\S_1^{d-3}$
at the cost of an error which is
\begin{align*}
\ll\int_{\frac\pi2}^\infty
\min\Bigl(\xi^{-2+\frac2d},\xi^{-3+\frac2{d-1}}\varphi^{-d+\frac2{d-1}}\Bigr)
\cdot\xi^{-\frac2d}\min\bigl(1,(\xi\varphi^d)^{-\frac2{d(d-1)}}\bigr)
\varphi^{d-3}\,d\varphi
\ll\xi^{-3}\int_{\frac\pi2}^\infty\varphi^{-3}\,d\varphi
\\
\ll\xi^{-3}\ll E'.
\end{align*}
Hence we obtain, after substituting $\varphi=\sqrt{2(1-w)}r$
and using \eqref{XIDM1ABDEF}
(note that the following is correct also for $d=3$,
with the convention that $\vol_{\S_1^{0}}\bigl(\S_1^{0}\bigr)=2$):
\begin{align}\notag
-\frac{\partial}{\partial\xi}\Phi(\xi,w)=
\frac{2^{1-d}\vol_{\S_1^{d-3}}(\S_1^{d-3})}{\zeta(d)}(1-w)^{d-2}\xi^{-1}
\int_{-\infty}^1\int_0^\infty
\int_{(0,1)\times\R^{d-2}}
\Xi\Bigl(\bn,
\hspace{90pt}
\\\label{PPHIXIWASYMPTTHMPF6}
\Bigl(\frac{1-z}{1-w}+r^2-1\Bigr)\vece_1+r\vece_2;
\vech;
2^{1-\frac d2}(1-w)^{-\frac d2}\xi^{-1}(1-h_1)\Bigr)\,d\vech\,
r^{d-3}\,dr\,dz
\\\notag
+O(E').
\end{align}
Next, using \eqref{XIDTKINV} we see that
$\int_{(0,1)\times\R^{d-2}}\Xi(\bn,x_1\vece_1+x_2\vece_2;\vech;v)\,d\vech
=\int_{(0,1)\times\R^{d-2}}\Xi(\bn,x_1\vece_1+x_2\vecomega;\vech;v)\,d\vech$
for any $\vecomega\in\S_1^{d-3}$
and any $x_1,x_2\in\R$ with $x_1>x_2^2-1$.
Integrating over all $\vecomega\in\S_1^{d-3}$ we get that the
main term in \eqref{PPHIXIWASYMPTTHMPF6} equals
\begin{align*}
\frac{2^{1-d}}{\zeta(d)}(1-w)^{d-2}\xi^{-1}
\int_{-\infty}^1\int_{\vece_1^\perp}\int_{(0,1)\times\R^{d-2}}
\Xi\Bigl(\bn,\Bigl(\frac{1-z}{1-w}+\|\vecu\|^2-1\Bigr)\vece_1+\vecu
;\vech;
\hspace{60pt}
\\\notag
2^{1-\frac d2}(1-w)^{-\frac d2}\xi^{-1}(1-h_1)\Bigr)
\,d\vech\,d\vecu\,dz,
\end{align*}
where $\vece_1^\perp$ is the orthogonal complement of $\vece_1$ in $\R^{d-1}$.
Substituting now $z=1-(1-w)y_1$ and then letting
$\vecy:=y_1\vece_1+\vecu$, we get
\begin{align*}
=\frac{2^{1-d}}{\zeta(d)}(1-w)^{d-1}\xi^{-1}
\int_{P^{d-1}}\int_{(0,1)\times\R^{d-2}}
\Xi\Bigl(\bn,\vecy;\vech;
2^{1-\frac d2}(1-w)^{-\frac d2}\xi^{-1}(1-h_1)\Bigr)
\,d\vech\,d\vecy.
\end{align*}
Applying now Lemma \ref{XIAVERAGELEM} and Fubini's Theorem we get
\begin{align*}
&=\frac{2^{2-\frac 32d}}{\zeta(d)}(1-w)^{\frac d2-1}\xi^{-2}
\int_{(0,1)\times\R^{d-2}}(1-h_1)\Xi\Bigl(\frac{\|\vech'\|}{h_1},
2^{1-\frac d2}(1-w)^{-\frac d2}\xi^{-1}(1-h_1)\Bigr)\,d\vech
\\
&=\frac{2^{3(1-\frac d2)}\pi^{\frac d2-1}}
{\Gamma(\frac d2-1)\zeta(d)}(1-w)^{\frac d2-1}\xi^{-2}
\int_0^1\int_0^\infty\Xi\Bigl({\sigma},
2^{1-\frac d2}(1-w)^{-\frac d2}\xi^{-1}y\Bigr)
{\sigma}^{d-3}y(1-y)^{d-2}\,d{\sigma}\,dy,
\end{align*}
where we substituted $\vech=(1-y)(1,{\sigma}\vecomega)$
($0<y<1$, ${\sigma}>0$, $\vecomega\in\S_1^{d-3}$) and used the fact that
$\vol(\S_1^{d-3})=\frac{2\pi^{\frac d2-1}}{\Gamma(\frac d2-1)}$.
Hence \eqref{PPHIXIWASYMPTTHMRES2} is proved.
\end{proof}

\section*{Index of notations}

\begin{center}
\begin{footnotesize}
\begin{longtable}{llr}
$[a_1,\vecv,\vecu,\tM]$ & the $\SL(d,\R)$-matrix defined by \eqref{NAKSPLIT},
\eqref{M1DEF} & \pageref{NAKSPLIT}
\\
$[\vecp,\vecx]_\vecy$ & the $3\times 3$ matrix defined by
\eqref{MYQP}, \eqref{QYPX}
&
\pageref{MYQP}
\\
$\aa(a)$ & the diagonal matrix in \eqref{ADEF} & \pageref{ADEF}
\\
$\scrB_r^d$ & open ball in $\R^d$ of radius $r$, centered at the origin
& \pageref{scrBDEF}
\\
$\fC_\vech(w)$ & %
the cut ball in \eqref{CHWDEF} & \pageref{CHWDEF}
\\
$\fC_\vech(\vecw)$ & the cut ball in \eqref{CHVECWDEF} & \pageref{CHVECWDEF}
\\
$\fC_{\vech}(\vecz,\vecw)$ & $=\fC_\vech(\vecw)\cup\fC_\vech(\vecz)$
& \pageref{CCZWDEF}
\\
$\vece_1$ & $=(1,0,\ldots,0)$
\\
$\vece_2$ & $=(0,1,0,\ldots,0)$
\\
$\F_d$ & a fundamental domain for $\Gamma^{(d)}\backslash G^{(d)}$
& \pageref{SCRFDEF}
\\
$F(t)$ & $=\pi-\arccos(t)+t\sqrt{1-t^2}$ & \pageref{FDEF}
\\
$F_d(t)$ & the function in \eqref{PHIXIWASYMPTTHMRES} (for $d\geq3$)
& \pageref{PHIXIWASYMPTTHMRES}
\\
$F_{\bn,d}(t_1,t_2,\alpha)$ 
& the function in \eqref{PHI0XILARGETHMFDDEF} (for $d\geq3$)
& \pageref{PHI0XILARGETHMFDDEF}
\\
$G$, $G^{(d)}$ & $=\SL(d,\R)$ & \pageref{GMUX1DEF}
\\
$\FG_A$ & the subset of $\SL(d,\R)$ in \eqref{FGDEF} & \pageref{FGDEF}
\\
$H$ & $=\{g\in\SL(d,\R)\col\vece_1g=\vece_1\}$
& \pageref{HDEF}
\\
$M_{\alpha,\beta}$ & the linear map in \eqref{TALFBETDEF}
& \pageref{TALFBETDEF}
\\
$\nn(u)$ & the upper triangular matrix in \eqref{NUDEF} & \pageref{NUDEF}
\\
$P^{d-1}$ & $=\{(x_1,\ldots,x_{d-1})\in\R^{d-1}
\col x_1>x_2^2+\ldots+x_{d-1}^2-1\}$  %
& \pageref{PDM1DEF}
\\
$P^{d-1}_\vech$ & the cut paraboloid $P^{d-1}\cap \R_{\vech-}^{d-1}$
& \pageref{PDM1HDEF}
\\
$P^{d-1}_\vech(\vecy)$ & the cut paraboloid in \eqref{PDM1CYDEF}
& \pageref{PDM1CYDEF}
\\
$P_{u,r}$ & the paraboloid in \eqref{PUVDEF} & \pageref{PUVDEF}
\\
$\vecq_{\vecy,\vecp}(\vecx)$ & the $\R^3$-vector in \eqref{QYPX}
& \pageref{QYPX}
\\
$\R_+^{d-1}$ & $=\{(h_1,\ldots,h_{d-1})\in\R^{d-1}\col h_1>0\}$
& \pageref{RPLUSDEF}
\\
$\R_{\vech-}^{d-1}$ & $\{\vecx\in\R^{d-1}\col\vecx\cdot\vech<0\}$
& \pageref{RHMINUSDEF}
\\
$\Si_d$ & the Siegel set in \eqref{SIDEF} & \pageref{SIDEF}
\\
$\Si_d'$ & the subset of $\Si_d$ defined in \eqref{SIDPDEF} & \pageref{SIDPDEF}
\\
$\S_1^{d-1}$ & unit sphere in $\R^d$
\\
$\HS$ & the hemisphere $\{(v_1,\ldots,v_d)\in\S_1^{d-1}\col v_1>0\}$
& \pageref{HSDEF}
\\
$S$ & the set defined in \eqref{SDEFmin} & \pageref{SDEFmin}
\\
$S'$ & the set defined in \eqref{SPDEF1}
(in section \ref{PRELIMINARIESSEC})
or \eqref{SPDEF} (in section \ref{PHI0XIWZASYMPTSEC})
& \pageref{SPDEF1}, \pageref{SPDEF}
\\
$S''$ & the subset of $(0,1)\times\R^{d-2}$ defined just below
\eqref{ABKAPPADEF} & \pageref{SPPDEF}
\\
$S'''$ & the subset of $(0,1)\times\R^{d-2}$ defined just below
\eqref{PHI0XILARGETHMPF3} & \pageref{SPPPDEF}
\\
$T_{\alpha,\beta}$ & the affine linear map in \eqref{TALFBETDEF}
& \pageref{TALFBETDEF}
\\
$v_{d-1}$ & volume of the unit ball in $\R^{d-1}$ 
& \pageref{ranPhi}
\\
$X_1$, $X_1^{(d)}$ & $=\SL(d,\Z)\backslash\SL(d,\R)$, space of lattices 
& \pageref{GMUX1DEF}
\\
$X_1(\vecy)$ & $=\{M\in X_1\col\vecy\in\Z^dM\}$,
a submanifold of $X_1$
& \pageref{X1YDEF}
\\
$X_1(\veck,\vecy)$ & $=\{\Gamma M\in X_1\col M\in G,\:\veck M=\vecy\}$
& \pageref{X1KYDEF}
\\
$\fZ(c_1,c_2,r)$ & the cylinder in \eqref{FZC1C2RDEF} & \pageref{FZC1C2RDEF}
\\
$\fZ_\vecv$ & $=\iota^{-1}(\fZ f(\vecv)^{-1})\subset\R^{d-1}$ 
(for various sets $\fZ\subset\R^d$) 
& \pageref{ZVDEF}
\\
$\delta_d^*(\fZ)$ & the maximal lattice packing density of a cylinder in $\R^d$
& \pageref{DELTASTARDEF}
\\
$\Gamma$, $\Gamma^{(d)}$ & $=\SL(d,\Z)$ & \pageref{GMUX1DEF}
\\
$\iota$ & the embedding $\R^{d-1}\ni (x_1,\ldots,x_{d-1})
\mapsto (0,x_1,\ldots,x_{d-1})\in\R^d$
& \pageref{IOTADEF}
\\
$\mu$, $\mu^{(d)}$ & Haar measure on $G$, probability measure on $X_1$
& \pageref{GMUX1DEF}
\\
$\nu_\vecy$ & a probability measure on $X_1(\vecy)$ 
(cf.\ \cite[Sec.\ 7]{partI}, \cite[Sec.\ 5]{lprob})
& \pageref{NUYDEF}
\\
$\xi_0(w)$ & the function defined in Corollary \ref{PHIXIWSUPCOR}
& \pageref{PHIXIWSUPCOR}
\\
$\xi_0(w,z,\varphi)$ & the function defined in Theorem \ref{PHI0SUPPORTASYMPTTHM}
& \pageref{PHI0SUPPORTASYMPTTHM}
\\
$\xi_1(\vecw,\vecz)$ & the function $\scrB_1^2\times\scrB_1^2\to\R$
defined just before Theorem \ref{D3EXPLTHM}
&
\pageref{D3EXPLTHM}
\\
$\Xi(\sigma,v)$ & 
the lattice probability in \eqref{XIDM1CVDEF}
& \pageref{XIDM1CVDEF}
\\
$\Xi(\vecy,\vecy';\vech;v)$ &
the lattice probability in \eqref{XIDM1DEFnew}
& \pageref{XIDM1DEFnew}
\\
$\Xi(\vecy;\vech;v)$ &
$=\Xi(\vecy,\vecy;\vech;v)$ & \pageref{XI3DEF}
\\
$\Xi(a,b;\vech;v)$ &
$=\Xi(\bn,(a^2+b^2-1)\vece_1+b\vece_2;\vech;v)$
& \pageref{XIDM1ABDEF}
\\
$\rho(a,b)$
& $=\inf\{v>0\col \exists \vech\in\R_+^{d-1}:\:\Xi(a,b;\vech;v)>0\}$
& \pageref{RHODABDEF}
\\
$\sigma_d(r,\alpha)$ & the function in \eqref{SUPPORTprelim4}
& \pageref{SUPPORTprelim4}
\\
$\Upsilon(\vecz,\vecw,\vech,v)$
& $=\mu(\{M\in \myX\col\Z^{d-1}M\cap 
v^{-\frac1{d-1}}\fC_{\vech}(\vecz,\vecw)=\emptyset\})$
& \pageref{UPSILONDEF}
\\
$\Phi(\xi)$ & 
limiting distribution for the free path length
& \pageref{PHIXIAVFORMULA}
\\
$\overline\Phi_\bn(\xi)$ &
limiting distribution for the free path length %
& \pageref{PHIXIAVFORMULA2}
\\
$\Phi_\bn(\xi)$ & 
limiting distribution for the free path length
& \pageref{PHIXIAVFORMULA3}
\\
$\Phi(\xi,\vecw)$ & the collision kernel function defined in \eqref{PHIXIWDEF} 
& \pageref{PHIXIWDEF}
\\
$\Phi(\xi,w)$ & $=\Phi(\xi,\vecw)$ with $w=\|\vecw\|$
& \pageref{PHIXISIMPWDEF}
\\
$\Phi_\bn(\xi,\vecw,\vecz)$ & 
the collision kernel function defined in \eqref{PHI0ZERODEF} 
& \pageref{PHI0ZERODEF}
\\
$\Phi_\bn(\xi,w,z,\varphi)$ &
$=\Phi_\bn(\xi,\vecw,\vecz)$ with $w=\|\vecw\|$, $z=\|\vecz\|$ and 
$\varphi=\varphi(\vecw,\vecz)$
& \pageref{PHIXIWZPHIDEF}
\end{longtable}
\end{footnotesize}
\end{center}

\end{document}